\documentclass[11pt,a4paper]{article}
\usepackage[margin=2.2cm]{geometry}
\usepackage{setspace}
\usepackage{graphicx}
\usepackage{soul}
\usepackage{enumerate,longtable,mathtools}

\usepackage{amssymb,amsmath,mathrsfs,amsthm}
\usepackage{dsfont,textcomp}
\usepackage{leftidx}
\usepackage[toc,page]{appendix}
\usepackage{comment}
\usepackage{todonotes}

\usepackage{tikz}
\usepackage{tikz-cd}
\usetikzlibrary{shapes,hobby}

\usepackage{pgfplots}
 
\pgfplotsset{compat=newest}

\usepackage{doi}

\usepackage{hyperref}
\hypersetup{
	breaklinks=true,   
	colorlinks=true,   
}
\usepackage{xurl}

\numberwithin{equation}{section}

\makeatletter
\newcommand{\raisemath}[1]{\mathpalette{\raisem@th{#1}}}
\newcommand{\raisem@th}[3]{\raisebox{#1}{$#2#3$}}
\makeatother

\newcommand{\CC}{\mathbb{C}}
\newcommand{\NN}{\mathbb{N}}
\newcommand{\RR}{\mathbb{R}}

\newcommand{\HH}{\mathcal{H}}
\newcommand{\KK}{\mathcal{K}}
\newcommand{\MM}{\mathcal{M}}

\newcommand{\Af}{\mathscr{A}}

\newcommand{\bs}{\boldsymbol{s}}
\newcommand{\bt}{\boldsymbol{t}}
\newcommand{\bx}{\boldsymbol{x}}
\newcommand{\won}{\text{\textwon}}

\DeclareMathOperator{\Bor}{Bor}

\DeclareMathOperator{\id}{id}

\DeclareMathOperator{\EOM}{EOM}

\DeclareMathOperator{\Intr}{Int}

\DeclareMathOperator{\Lin}{Lin}

\DeclareMathOperator*{\slim}{s-lim}
\DeclareMathOperator{\Weyl}{Weyl}

\DeclareMathOperator{\Ad}{Ad}
\DeclareMathOperator{\Aut}{Aut}

\newcommand{\one}{\mathds{1}}

\newcommand{\Sol}{\mathrm{Sol}}

\newcommand{\Floc}{\mathcal{F}_{\textup{loc}}}
\newcommand{\Fsloc}{\mathcal{F}_{\textup{sloc}}}
\newcommand{\SFloc}{\mathcal{SF}_{\textup{loc}}}
\newcommand{\SFsloc}{\mathcal{SF}_{\textup{sloc}}}

\newcommand{\loc}{\textup{loc}}

\renewcommand{\sc}{\textnormal{sc}}
\newcommand{\rel}{\textup{rel}}

\newcommand{\diff}{\textup{d}}
\newcommand{\dVol}{\textup{dVol}}
\newcommand{\nml}{\textup{n}}
\newcommand{\tang}{\textup{t}}
\newcommand{\ol}[1]{\overline{#1}}
\newcommand{\udl}[1]{\underline{#1}}
\newcommand{\ipc}[2]{\left\langle\,#1\,,\,#2\,\right\rangle} 
\DeclareMathOperator{\supp}{supp}
\DeclareMathOperator{\Data}{Data}

\newcommand{\PJ}{\textnormal{PJ}}

\newtheorem{theorem}{Theorem}[section]
\newtheorem{corollary}[theorem]{Corollary}
\newtheorem{proposition}[theorem]{Proposition}
\newtheorem{definition}[theorem]{Definition}
\newtheorem{lemma}[theorem]{Lemma}
\newtheorem{example}[theorem]{Example}
\newtheorem{rem}[theorem]{Remark}

\usepackage{scalerel}
\usetikzlibrary{svg.path}

\definecolor{orcidlogocol}{HTML}{A6CE39}
\tikzset{
  orcidlogo/.pic={
    \fill[orcidlogocol] svg{M256,128c0,70.7-57.3,128-128,128C57.3,256,0,198.7,0,128C0,57.3,57.3,0,128,0C198.7,0,256,57.3,256,128z};
    \fill[white] svg{M86.3,186.2H70.9V79.1h15.4v48.4V186.2z}
                 svg{M108.9,79.1h41.6c39.6,0,57,28.3,57,53.6c0,27.5-21.5,53.6-56.8,53.6h-41.8V79.1z M124.3,172.4h24.5c34.9,0,42.9-26.5,42.9-39.7c0-21.5-13.7-39.7-43.7-39.7h-23.7V172.4z}
                 svg{M88.7,56.8c0,5.5-4.5,10.1-10.1,10.1c-5.6,0-10.1-4.6-10.1-10.1c0-5.6,4.5-10.1,10.1-10.1C84.2,46.7,88.7,51.3,88.7,56.8z};
  }
}

\newcommand\orcidicon[1]{\href{https://orcid.org/#1}{\mbox{\scalerel*{
\begin{tikzpicture}[yscale=-1,transform shape]
\pic{orcidlogo};
\end{tikzpicture}
}{|}}}}

\let\restriction\relax
\newcommand{\restriction}{\!\upharpoonright}

\title{Semi-local observables, edge modes and quantum reference frames in quantum electromagnetism: an algebraic approach}

\author{Christopher J. Fewster{\orcidicon{0000-0001-8915-5321}}${}^{1,2}$\thanks{\tt chris.fewster@york.ac.uk},
Daan W. Janssen\orcidicon{0000-0001-7809-5044}${}^{1}$\thanks{\tt science@daanwjanssen.nl}, and Kasia Rejzner\orcidicon{0000-0001-7101-5806}${}^{1,2}$\thanks{\tt kasia.rejzner@york.ac.uk} \\[6pt]  
	\small ${}^{1}$ Department of Mathematics, Ian Wand Building, Deramore Lane,
University of York, \\ \small York YO10 5GH, United Kingdom.\\[4pt]
	\small ${}^{2}$ York Centre for Quantum Technologies, University of York, York YO10 5DD, United Kingdom.}

\date{\today}
\begin{document}
\maketitle

\begin{abstract}
 Boundaries and corners of spacetime play a vital role in understanding physical concepts including entanglement entropy, the infrared problem in QFT and quantum gravity. Standard local quantum field theory struggles to accommodate such boundary-sensitive observables. In this paper we develop an algebraic framework for \emph{semi-local quantum electromagnetism} on finite Cauchy lenses: a class of compact spacetimes with boundaries and corner. At the classical level, we establish a decomposition of the reduced covariant phase space into bulk closed-loop and surface sectors and demonstrate how the covariant phase space approach relates to the Peierls bracket construction commonly used in perturbative algebraic quantum field theory. Upon quantisation, we obtain a Weyl $C^{*}$-algebra of semi-local observables transforming non-trivially under large gauge transformations (those with non-trivial boundary contribution). To recover gauge invariance, we invoke the notion of \emph{quantum reference frames} (QRFs) and construct a relativisation map, where we treat auxiliary surface degrees of freedom as QRFs for the large gauge transformations. The relativisation map is constructed directly on the level of $C^{*}$-algebras, making our construction state-independent. The QRF viewpoint on semi-local observables provides new tools for understanding gauge theories on manifolds with boundary, including the problem of gluing theories on Cauchy lenses with common boundaries.
 \end{abstract}

{
  \hypersetup{linkcolor=black!30!blue}
  \tableofcontents
}
\section{Introduction}
The interplay between gauge theories and the presence of boundaries has become an increasingly prominent theme in theoretical physics. From the study of holographic dualities and infrared phenomena to new approaches to quantum gravity, boundaries and corners of spacetime manifolds have emerged not merely as geometric complications, but as rich sources of physical structure for matter systems and (quantum) gravity \cite{wenTheoryEdgeStates1992,balachandranEdgeStatesGravity1996,donnellyLocalSubsystemsGauge2016,dybalskiMechanismHolographyNoninteracting2019,mertensMinimalFactorizationChernSimons2026}. The formulation of gauge theories in the presence of boundaries and corners has a long history \cite{reggeRoleSurfaceIntegrals1974,sniatyckiBoundaryConditionsYangMills1991,cattaneoClassicalBVTheories2014,strohmaierClassicalQuantumPhoton2021,rielloHamiltonianGaugeTheory2024}. In gauge theories, and in particular in electromagnetism, the presence of boundaries introduces novel degrees of freedom—\emph{edge modes} or \emph{surface degrees of freedom}—which are intricately linked to large gauge transformations, i.e., those gauge symmetries that do not vanish at the boundary. These structures play a crucial role in diverse areas such as condensed matter physics, topological phases, and quantum gravity.
One can also consider boundaries at infinity, replacing boundary degrees of freedom by asymptotic degrees of freedom \cite{brownCentralChargesCanonical1986,herdegenLongrangeEffectsAsymptotic1995,giuliniAsymptoticSymmetryGroups1995,ashtekarSymplecticGeometryRadiative1997,barnichCovariantTheoryAsymptotic2002,stromingerAsymptoticSymmetriesYangMills2014,stromingerLecturesInfraredStructure2018,rejznerAsymptoticSymmetriesBVBFV2021,araujo-regadoSoftEdgesMany2025,borsboomGlobalGaugeSymmetries2026}. In quantum field theory (QFT) these issues are related to the infrared problem \cite{kulishAsymptoticConditionsInfrared1970,frohlichInfraredProblemSpontaneous1979,herdegenSemidirectProductCCR1998}, superselection sectors \cite{buchholzLocalityStructureParticle1982,buchholzPhysicalStateSpace1982} and memory effects \cite{staruszkiewiczGAUGEINVARIANTSURFACE1981,satishchandranAsymptoticBehaviorMassless2019,rielloNullHamiltonianYangMills2025}.

Despite their physical importance, the mathematical formalisation of edge modes and related observables within rigorous QFT remains subtle. 
Notably, algebraic QFT (AQFT) \cite{haagLocalQuantumPhysics1996,buchholzUniversalCalgebraElectromagnetic2016}, focused on local algebras generated by gauge-invariant observables with compact support, fails to account for observables associated with electric fluxes and surface-related quantities that naturally emerge in gauge theories defined on spacetimes with boundary. These \emph{semi-local observables}, which straddle the boundary between purely local and global quantities, demand a novel approach, which we term \emph{semi-local quantum physics}.

This paper initiates a comprehensive algebraic approach to quantum electromagnetism in spacetimes with boundaries and corners, providing a rigorous construction of semi-local observables and clarifying their transformation properties under large gauge transformations. As a concrete setting, we give a detailed analysis of pure electromagnetism with an external conserved current on a class of Lorentzian manifolds that we term \emph{finite Cauchy lenses} (see Def.~\ref{def:cauchy_lens}). These are compact regions $\ol{N}$ of spacetime bounded between two spacelike Cauchy surfaces whose boundaries intersect at a codimension-2 corner $\angle\ol{N}$. This geometric setting serves both to model realistic physical situations and to sharpen our understanding of gauge theories on manifolds with boundary. An example is given by the subset
of Minkowski space given by
\begin{equation}\label{eq:lens_example}
	|t|\le b-\sqrt{a^2+\|x\|^2}, \qquad \|x\|^2\le b^2-a^2
\end{equation}
in standard inertial coordinates $(t,x)$, where $b>a>0$, for which the corner is a sphere of radius $\sqrt{b^2-a^2}$ in the $t=0$ hypersurface.
A cross-section of this subset is shown in Figure~\ref{fig:FinCauchLens0}.
In general, finite Cauchy lenses may have multiple connected components, and the components may have nontrivial topology. 
\begin{figure}
	\centering
	\begin{tikzpicture}[auto, scale=4, node distance=2cm,domain=0:2]
		\draw plot (\x,{sqrt(5)-sqrt(4+(\x-1)^2)});
		\draw plot (\x,{sqrt(4+(\x-1)^2)-sqrt(5)});
		\draw (1.5,-.3) node {$\overline{\Sigma}_-$};
		\draw (.5,.3) node {$\overline{\Sigma}_+$};
		\draw (1,0) node[above]{$\overline{\Sigma}$};
		\node at (0,0)[circle,fill,inner sep=1.5pt]{};
		\node at (2,0)[circle,fill,inner sep=1.5pt]{};
		\node at (0,0)[left]{$\angle \overline{N}$};
		\node at (2,0)[right]{$\angle \overline{N}$};
		\draw[dashed] (0,0)--(2,0);
	\end{tikzpicture}
	\caption{Cross-section of a finite Cauchy lens $\ol{N}$.}
	\label{fig:FinCauchLens0}
\end{figure}

At the classical level, we develop the covariant phase space formalism for electromagnetism on finite Cauchy lenses. We follow the approaches of \cite{kijowskiSymplecticFrameworkField1979,zuckermanActionPrinciplesGlobal1987,ashtekarPhaseSpaceFormulation1987,crnkovicCovariantDescriptionCanonical1987,crnkovicSymplecticGeometryCovariant1988,leeLocalSymmetriesConstraints1990,harlowCovariantPhaseSpace2020,girelliNoetherCovariantStudy2025}, with similar ideas previously presented in \cite{peierlsCommutationLawsRelativistic1952,souriauStructureSystemesDynamiques1970,chernoffPropertiesInfiniteDimensional1974}. We also discuss how the covariant phase space formalism is related to the pAQFT approach \cite{rejznerPerturbativeAlgebraicQuantum2016}, where the appropriate space of observables is equipped with the Peierls bracket.

Pure electromagnetism on a finite Cauchy lens $\ol{N}$ with external current $J$ has configurations given by smooth $1$-forms $A$. The Maxwell equations are 
\begin{equation}\label{eq:Maxwell}
	-\diff^*\diff A = J
\end{equation}
where $\diff^*$ is the codifferential (see Sec.~\ref{sec:set-up}). Although the Maxwell
equations~\eqref{eq:Maxwell} are invariant under $A\mapsto A+\diff\Lambda$ for any smooth function $\Lambda$, the covariant phase space formalism for gauge theories on spacetimes with boundaries and corners \cite{harlowCovariantPhaseSpace2020} does not regard all gauge transformations as redundancies, but only those \emph{bulk gauge transformations} for which, up to a locally constant function, $\Lambda=0$ at the corner. It is these that form the radical of the pre-symplectic form on solutions to the homogeneous equations (see Sec.~\ref{sec:cov_phas}). The on-shell phase space consists of bulk gauge orbits of solutions and defines a symplectic manifold $\Sol^J_\mathscr{G}(\ol{N})$. The tangent space to an orbit $[A]$ may be identified with the bulk gauge orbits of homogeneous solutions. It therefore contains nontrivial \emph{large gauge directions} which are bulk gauge equivalence classes of arbitrary exact $1$-forms.

We show in Prop.~\ref{prop:ext_ini_dat} that the covariant phase space $\Sol^J_\mathscr{G}(\ol{N})$ is in one-to-one correspondence with initial data on well-behaved Cauchy surfaces of finite Cauchy lenses. The Hodge decomposition applied to initial data leads to a decomposition 
\begin{equation}\label{eq:SolJG_decomp}
	\Sol^J_{\mathscr{G}}(\ol{N})\cong V^C(\ol{\Sigma})\times V^S(\ol{N}),
\end{equation}
of the classical phase space
into two symplectic vector spaces, namely, the \emph{closed loop} space $V^C(\ol{\Sigma})$ and the \emph{surface} configuration space $V^S(\ol{N})$, the latter being sensitive to boundary data and encapsulating the edge mode content (see Prop.~\ref{prop:sympl_decomp}).  
Concretely, $V^C(\ol{\Sigma})$ is a phase space of coclosed data with vanishing normal component at $\partial\ol{\Sigma}=\angle\ol{N}$ while
$V^S(\ol{N})$ is a phase space built from a subspace of smooth functions on $\angle\ol{N}$.

Classical observables are smooth functionals on $\Sol^J_{\mathscr{G}}(\ol{N})$. A natural class of observables, studied in Sec.~\ref{sec:class_loc_obs}, is the Poisson algebra of \emph{local observables} $\Floc(\ol{N})$ consisting of
polynomials in smearings of the vector potential against coclosed test $1$-forms that are compactly supported in the interior of $\ol{N}$. Observables of this type appear in treatments of electromagnetism on globally hyperbolic spacetimes without boundaries \cite{sandersElectromagnetismLocalCovariance2014}. They are invariant under large gauge transformations. However, a larger class of observables may be defined by a symplectic smearing procedure using arbitrary solutions to the Maxwell equations on $\ol{N}$, leading to a Poisson algebra of \emph{semi-local observables} $\Fsloc(\ol{N})$ discussed in Sec.~\ref{sec:class_semi-loc}. The  semi-local observables contain \emph{edge mode observables}, i.e., observables sensitive to large gauge transformations. 
Choosing a Cauchy surface $\ol{\Sigma}\subset\ol{N}$, the decomposition~\eqref{eq:SolJG_decomp} allows the identification of two mutually Poisson commuting Poisson subalgebras 
$\Fsloc^C(\ol{\Sigma})$ and $\Fsloc^S(\ol{\Sigma})$ that together generate
$\Fsloc(\ol{N})$ (Theorem~\ref{thm:class_obs_decomp}). These algebras are generated by observables that 
smear the initial data of $\Sol_{\mathscr{G}}^J(\ol{N})$ at $\ol{\Sigma}$ in different ways. To be specific (and with $J=0$ for simplicity) 
let $\mathbf{A}$ and $\mathbf{E}$ denote the vector potential and electric field at $\ol{\Sigma}$ of a field configuration in $\Sol^J_{\mathscr{G}}(\ol{N})$.
Then $\Fsloc^C(\ol{\Sigma})$ is generated by observables $\int_{\ol{\Sigma}} F\wedge \star \mathbf{A}$ and $\int_{\ol{\Sigma}} F\wedge \star \mathbf{E}$ 
for arbitrary coclosed $1$-form $F$ on $\ol{\Sigma}$ with vanishing normal component at $\partial\ol{\Sigma}$. Intuitively, one may regard such
observables as smeared out versions of 
closed Wilson loop observables and their normal derivative at $\ol{\Sigma}$,
\begin{equation}
	\oint_\gamma \mathbf{A}_i\diff\gamma^i,\qquad \oint_\gamma \mathbf{E}_i\diff\gamma^i,
\end{equation}
where $\gamma:S^1\to \ol{\Sigma}$ is a closed loop in $\ol{\Sigma}$. The analogy is that the closed loop observables are automatically gauge invariant under any gauge transformation, which holds for the surface integral of $\mathbf{A}$ by virtue of $F$ being coclosed. On the other hand,  
$\Fsloc^S(\ol{\Sigma})$ is generated by observables that are smearings of electric flux through $\ol{\Sigma}$ and integrals $\int_{\ol{\Sigma}} F\wedge\star \mathbf{A}$ where $F$ is exact and coclosed, with prescribed normal component at $\partial\ol{\Sigma}$. Such an $F$ is norm-minimizing among co-closed forms on $\ol{\Sigma}$ with fixed boundary normal component, and as such the latter observables have an intuitive interpretation as  
smeared analogues of geodesic (i.e.~minimal distance) Wilson line observables
\begin{equation}
	\int_\gamma \mathbf{A}_i\diff\gamma^i 
\end{equation}
where $\gamma$ is a (directed) geodesic with endpoints in $\partial\ol{\Sigma}$. For these reasons, we refer to elements of $\Fsloc^C(\ol{\Sigma})$ and $\Fsloc^C(\ol{\Sigma})$ as \emph{closed loop observables} and \emph{surface observables} respectively.

Quantisation in Section~\ref{sec:quant_alg} proceeds via the Weyl $C^*$-algebra construction applied to semi-local observables, generalizing the approach of \cite{dimockQuantizedElectromagneticField1992}, which only covers the local case, see Prop.~\ref{prop:CCR_BPI}. This algebra captures both bulk and boundary contributions and admits a decomposition mirroring the classical one into closed loop and surface observables at a Cauchy surface $\ol{\Sigma}\subset\ol{N}$,
\begin{equation}
	\label{eq:Afdecomp}
	\Af(\ol{N})\cong \Af^C(\ol{\Sigma})\otimes \Af^S(\ol{N})=\Af^{C\oplus S}(\ol{\Sigma}),
\end{equation}
While local observables  $\Af_{\textup{loc}}(\ol{N})\subset \Af(\ol{N})$ are invariant under all gauge transformations, semi-local observables transform non-trivially under large gauge transformations. This motivates a deeper investigation into their structure and their role in the quantum theory.

As a $C^*$-algebra, $\Af(\ol{N})$ automatically admits states and Hilbert space representations; however, it is not guaranteed that they have good physical properties. 
To begin an investigation of this issue, we prove that $\Af(\ol{N})$ admits quasifree states and corresponding Fock space representations. In particular, using the decomposition~\eqref{eq:Afdecomp}, we construct an \emph{$L^2$-representation} of $\Af(\ol{N})$ at a Cauchy surface $\ol{\Sigma}\subset\ol{N}$, see Thm.~\ref{thm:L2rep}. 
Note, however, that the associated Fock vacuum states are not expected to have satisfy the important \emph{Hadamard} condition that would allow for the construction of Wick ordered observables by point-splitting. As will be shown elsewhere, Hadamard states do exist, and examples can be constructed \cite{fewsterHadamardStatesSemilocal}.

One of the main results of our paper is to provide a fundamental understanding of edge modes as \emph{quantum reference frames} (QRFs) for large gauge transformations, using ideas proposed in \cite{donnellyLocalSubsystemsGauge2016,gomesObserversGhostFieldspace2017,carrozzaEdgeModesReference2022,kabelQuantumReferenceFrames2023,araujo-regadoSoftEdgesMany2025,araujo-regadoRelationalEntanglementEntropies2025a}. Inspired by recent developments in the operational and algebraic formulations of QRFs \cite{caretteOperationalQuantumReference2025,fewsterQuantumReferenceFrames2025}, we introduce an auxiliary system—the \emph{surface field}—which transforms covariantly under large gauge transformations and serves as a QRF (this makes rigorous the ideas presented in \cite{donnellyLocalSubsystemsGauge2016}). 
The surface fields have previously been interpreted as Goldstone modes associated with breaking of large gauge symmetry, because of their transformation behaviour \cite{stromingerLecturesInfraredStructure2018,ballDynamicalEdgeModes2024,araujo-regadoSoftEdgesMany2025}. The classical treatment of such extended system is discussed in Section~\ref{sec:classical_sfe}, where we  follow~\cite{donnellyLocalSubsystemsGauge2016} to construct the \emph{fusion product} of the original electromagnetic phase space with the surface field phase space. It consists of orbits under the diagonal action of large gauge transformations on the two phase spaces, so large gauge transformations are redundancies. The main take-home message of the classical theory is the \emph{equivalence between the description of electromagnetism using semi-local observables in which large gauge transformations act nontrivially, and the one in which they are treated as redundancies by invoking additional surface fields} (see Prop.~\ref{prop:class_edge_equiv}). Similar observations have been made in \cite{araujo-regadoSoftEdgesMany2025}, however here, they appear in a precise and rigorous form for the class of finite Cauchy lenses. 

On the quantum level, by considering the joint observables of the electromagnetic field and the surface field (elements of $\Af(\ol{N})\otimes \Af^S(\ol{N})$) that are invariant under the diagonal action of large gauge transformations, we construct a new algebra $\widetilde{\Af}^{\mathscr{LG}}(\ol{N})$ (the subalgebra of invariants) that is related to the original semi-local algebra via a \emph{relativisation map}:
\begin{equation}
	\label{eq:rel_intro}
	\yen:\Af(\ol{N})\to \widetilde{\Af}^{\mathscr{LG}}(\ol{N}),
\end{equation}
with range $\Af_{\rel}(\ol{N})\subset\widetilde{\Af}^{\mathscr{LG}}(\ol{N})$. We can define $\yen$ as a homomorphism acting as a dressing operation (by the surface field), converting gauge-covariant quantities into gauge-invariant ones. Formally, this can be understood through the action of $\yen$ on infinitesimal generators of $\Af(\ol{N})$ (well-defined for sufficiently regular representations of these algebras)
\begin{equation}
    \yen\left(\int_{\ol{\Sigma}} [\mathbf{A}_iH^i+\mathbf{E}_iF^i]\,\dVol_{\ol{\Sigma}}\right)=\int_{\ol{\Sigma}} [\mathbf{A}_iH^i+\mathbf{E}_iF^i]\,\dVol_{\ol{\Sigma}}-\int_{\partial\ol{\Sigma}}\varphi n^iH_i\,\dVol_{\partial\ol{\Sigma}},
\end{equation}
where $\mathbf{A}$ and $\mathbf{E}$ are the vector potential and electric field on a Cauchy surface $\ol{\Sigma}\subset\ol{N}$ and $\varphi$ the surface field on the boundary $\partial\ol{\Sigma}$ (or equivalently the corner of $\ol{N}$), $F$ and $H$ are $1$-form fields on $\ol{\Sigma}$ (with indices raised by the metric on $\ol{\Sigma}$) with $\nabla^iH_i=0$, and $n$ is the outward pointing normal vector field on $\partial\ol{\Sigma}\subset\ol{\Sigma}$. As a result, this relativised observable is invariant under joint large gauge transformations
\begin{equation}
    (\mathbf{A},\varphi)\mapsto (\mathbf{A}+\diff\Lambda,\varphi+\Lambda\restriction_{\partial\ol{\Sigma}}),
\end{equation}
where $\Lambda$ a smooth function on $\ol{\Sigma}$. Following the analogy between smearing with co-closed forms and smearing along Wilson lines, we can sketch the action of $\yen$ even more formally by viewing it as mapping
\begin{equation}
	\int_\gamma \mathbf{A}_i\diff\gamma^i \mapsto \int_\gamma \mathbf{A}_i\diff\gamma^i-[\varphi(\gamma(s))]_{s=s_{i}}^{s_{f}},
\end{equation}
where $\gamma:[s_i,s_f]\to \ol{\Sigma}$ a curve with $\gamma(s_{i/f})\in \partial\ol{\Sigma}$ as illustrated in Fig.~\ref{fig:Wilson_relativise_endpoints}. A rigorous definition of $\yen$ is given by Eq.~\eqref{eq:C*QRF} and Lem.~\ref{lem:relativise}. We emphasize that here we construct the relativization map already \emph{on the level of $C^*$ algebras}, which is a representation-independent version of the relativisation map used previously in the context of operational QRFs, including our previous work \cite{fewsterQuantumReferenceFrames2025}.
\begin{figure}[ht]
\includegraphics[width=\textwidth]{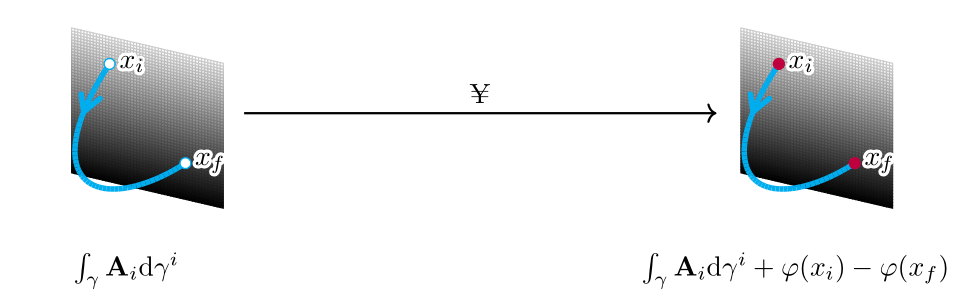}
\caption{The relativisation map dressing Wilson lines with surface fields.}
\label{fig:Wilson_relativise_endpoints}
\end{figure}

We show in Thm.~\ref{thm:Weyl_comm_thm} that $\widetilde{\Af}^{\mathscr{LG}}(\ol{N})$ is generated by the relativised observables $\Af_{\rel}(\ol{N})$ 
together with the unitaries implementing surface gauge transformations. This result provides a $C^*$-algebraic analogue of the characterisation of invariant von Neumann algebras in the presence of a quantum reference frame as given in \cite{fewsterQuantumReferenceFrames2025}. 

Another crucial result of our paper is the link between QRFs and superselection sectors. In Sec.~\ref{sec: superselection}, we show that projections $\Gamma_{\Phi}:\widetilde{\Af}^{\mathscr{LG}}(\ol{N})\to
\Af_{\rel}(\ol{N})$ onto the relativised observable algebra, labelled by 
the flux function $\Phi$ \emph{outside} the corner of $\ol{N}$,
define superselection sectors of $\widetilde{\Af}^{\mathscr{LG}}(\ol{N})$. We propose a superselection criterion (formula \eqref{eq: sups criterion}) and in  Thm.~\ref{thm:sup_sel} we show that all sufficiently regular superselection sectors of $\widetilde{\Af}^{\mathscr{LG}}(\ol{N})$ may be constructed this way. Our analysis differs from the classical analogue presented in \cite{rielloHamiltonianGaugeTheory2024,rielloNullHamiltonianYangMills2025,rielloSymplecticReductionYangMills2021}, since the latter would also yield superselection sectors labelled by the fluxes associated to the \emph{inside} of the corner. In our setting, these are not superselected a priori, since one can construct \emph{semi-local edge mode observables} that interpolate between states corresponding to such internal fluxes (see example~\ref{Ex:superselection_decomp}). Our approach is similar in spirit to that of \cite{herdegenSemidirectProductCCR1998}, where the asymptotic observables in QED are also quantised to yield a different superselection structure from that of algebras of local observables \cite{buchholzPhysicalStateSpace1982}. At the end of Sec.~\ref{sec: superselection}, we comment on how concrete models for electromagnetism inside a charged shell could shine light on the question of whether these semi-local observables can be measured. Such models may in particular exhibit breaking of superselection rules associated with local quantum physics.

Finally, to cement the link between the $C^*$-algebraic relativisation map and the operational QRF formalism, we show that relativisation map Eq.~\eqref{eq:rel_intro} can alternatively be constructed from large gauge covariant projection valued measures, arising naturally for sufficiently regular representations of the algebra $\Af^S(\ol{N})$, see Thm.~\ref{thm:edge_mode_QRF}.

The QRF perspective on semi-local observables is not only conceptually appealing, but also has practical advantages, e.g. in the treatement of 
gluing algebras of semi-local observables and their corresponding states across a common boundary (see Sect.~\ref{sec:gluing}).  Here, we consider a globally hyperbolic spacetime without boundary $M$, with a Cauchy surface $\Sigma$ that is decomposed as
$\Sigma=\ol{\Sigma}_1\cup\ol{\Sigma}_2$, where $\ol{\Sigma}_i$ are Cauchy surfaces of finite Cauchy lenses $\ol{N}_i$ whose corners coincide.
The semi-local algebras $\Af(\ol{N}_i)$ can be glued across the common corner to produce a glued algebra $\Af_{\textnormal{glue}}(\ol{N}_1;\ol{N}_2)$ consisting of joint large gauge invariant observables. In a precise way, one may regard $\Af(\ol{N}_2)$ as providing a reference frame for $\Af(\ol{N}_1)$ and vice versa. The corresponding relativisation maps give parameterisations of the glued algebra, and can be used to construct glued states on $\Af_{\textnormal{glue}}(\ol{N}_1;\ol{N}_2)$ from states on $\Af(\ol{N}_1)$ and $\Af(\ol{N}_2)$. Meanwhile, the gauge invariant algebra of $M$ can be embedded in a natural way within $\Af_{\textnormal{glue}}(\ol{N}_1;\ol{N}_2)$. This perspective on gluing is similar in spirit to the concept of extrinsic reference frames for gauge theories as discussed in \cite{carrozzaEdgeModesReference2022,carrozzaEdgeModesDynamical2024,araujo-regadoSoftEdgesMany2025,araujo-regadoRelationalEntanglementEntropies2025a}.

\paragraph{Structure of the paper}

The paper is organized as follows. Section~\ref{sec:set-up} includes the main conventions, the precise definition of Cauchy lenses and the necessary background on differential forms on manifolds with boundaries and corners. Section~\ref{sec:classical} develops the classical phase space of electromagnetism on Cauchy lenses, detailing both local and semi-local observables. Section~\ref{sec:quant_alg} presents the algebraic quantisation via the Weyl C*-algebra construction and examines its representations. Section~\ref{sec:qrfs} discusses the treatment of large gauge transformations, the construction of quantum reference frames and in Section~\ref{sec:gluing} we discuss gluing. We conclude in Sec.~\ref{sec:conclusion} with a discussion of the implications of our results and potential directions for future work. Appendices contain various technical details; in particular, Appendix~\ref{apx:symbols} gives a table of symbols and spaces to assist the reader.

\section{Set-up, notations and conventions}
\label{sec:set-up}

\subsection{Main conventions} 

We use a mostly minus spacetime metric signature and adopt a rationalised unit system in which $\hbar$, $c$ and the vacuum permittivity $\varepsilon_0$ are set to unity and Maxwell's equations with current $J_\nu$ take the form 
	\begin{equation}
		\nabla^\mu F_{\mu\nu} = J_\nu, 
	\end{equation}
	where $\nabla$ is the Levi--Civita derivative and the Faraday tensor is given in terms of the electromagnetic potential by $F_{\mu\nu} = \nabla_\mu A_\nu-\nabla_\nu A_\mu$,
	or $F=\diff A$ using differential form notation. The Lorentz $4$-force on a test particle with charge $q$ and $4$-velocity $u^\mu$ is $f_\mu = q F_{\mu\nu}u^\nu$. Charge is dimensionless in these units, as are the forms $A=A_\mu \diff x^\mu$ and $F=\tfrac{1}{2}F_{\mu\nu}\diff x^\mu\wedge \diff x^\nu$, though their components $A_\mu$ and $F_{\mu\nu}$ have dimensions of mass and mass-squared respectively.
\subsection{Manifolds with corners, Lorentzian geometry and Cauchy lenses}

In this paper we consider electromagnetism on a class of compact 3+1 dimensional (Lorentzian) spacetime backgrounds with corners, which we call \emph{finite Cauchy lenses} (see Def.~\ref{def:cauchy_lens}). 
Following \cite{leeIntroductionSmoothManifolds2012}, $n$-dimensional manifolds with corners are spaces locally modelled on $\ol{\RR}_+^n:=[0,\infty)^n$. Notions of smoothness (smooth curves, maps, tensor fields etc.) can be described locally in terms of restrictions of smooth objects on $\RR^n$ to $\ol{\RR}_+^n$, see Appx.~\ref{apx:corners} for more details. The boundary $\partial\ol{\RR}_+^n$ consists of points in $\ol{\RR}_+^n$ with at least one vanishing coordinate, while the (first order) corner $\angle\ol{\RR}_+^n$ comprises those points with at least two vanishing coordinates -- we will not consider higher order corners in this paper. 
As described in Appx.~\ref{apx:corners}, $\partial\ol{\RR}_+^n$ and $\angle\ol{\RR}_+^n$ locally model boundary $\partial\ol{M}$ and the corner $\angle\ol{M}$ of a manifold with corners $\ol{M}$.

Lorentzian geometry may be developed on manifolds with corners as usual -- see Appx.~\ref{apx:corners}. For convenience, we shall only consider Cauchy surfaces that are smooth spacelike hypersurfaces. In particular, in a time-oriented Lorentzian manifold with corners $(\ol{M},g)$, we define
a \emph{Cauchy surface with boundary} to be a topologically closed acausal smooth spacelike codimension-$1$ submanifold with boundary $\partial\ol{\Sigma}\subset\partial\ol{M}$, with domain of dependence $\mathscr{D}(\ol{\Sigma})=\ol{M}$. Aside from globally hyperbolic spacetimes without boundary or corners, our main examples of spacetimes possessing such Cauchy surfaces are of the following type 
(a specific example was given in~\eqref{eq:lens_example}).
\begin{definition}
    \label{def:cauchy_lens}
    A \emph{finite Cauchy lens} $\ol{N}$ is a compact (time)-oriented Lorentzian manifold with corners, such that 
    \begin{itemize}
        \item $\angle\ol{N}\subset\ol{N}$ is a smooth compact proper submanifold without boundaries whose causal future/past are causally closed, i.e., $\mathscr{J}^\pm(\angle \ol{N})=\angle \ol{N}$,
        \item the boundary decomposes as $\partial \ol{N}=\ol{\Sigma}_-\cup\ol{\Sigma}_+$, where $\ol{\Sigma}_\pm$ are Cauchy surfaces with boundaries, $\ol{\Sigma}_-\cap\ol{\Sigma}_+=\angle \ol{N}$ and $\ol{\Sigma}_\pm\subset \mathscr{J}^\pm(\ol{\Sigma}_\mp)$,
        \item every Cauchy surface with boundary $\ol{\Sigma}$ in $\ol{N}$
        has $\partial\ol{\Sigma}=\angle\ol{N}$,
        \item there exists a Cauchy surface with boundary $\ol{\Sigma}\subset \ol{N}$ that is \emph{regular}, meaning
        \begin{itemize}
            \item $\Sigma=\Intr(\ol{\Sigma})$ is a Cauchy surface of $N=\Intr(\ol{N})$,
            \item $\mathscr{J}^+(\ol{\Sigma})$ and $\mathscr{J}^-(\ol{\Sigma})$ are manifolds with corners where $\partial\mathscr{J}^\pm(\ol{\Sigma})=\ol{\Sigma}^\pm\cup \ol{\Sigma}$,
        \end{itemize}
        \item there exists a causally convex isometric embedding $\iota:\ol{N}\to M$ for some globally hyperbolic spacetime without boundaries $M$.
    \end{itemize}
\end{definition}
For an embedding $\iota:\ol{N}\to M$ as in Def.~\ref{def:cauchy_lens}, and $\ol{\Sigma}\subset \ol{N}$ a Cauchy surface with boundary, there always exists a Cauchy surface $\tilde{\Sigma}\subset M$ such that $\iota(\ol{\Sigma})\subset \tilde{\Sigma}$. This follows by \cite[Thm.~1.1]{bernalFurtherResultsSmoothability2006} 
because $\iota(\ol{\Sigma})\subset M$ is a spacelike acausal compact co-dimension 1 submanifold with boundaries.

To study electromagnetism on finite Cauchy lenses, we first recall some features of differential forms on (compact) manifolds with corners.

\subsection{Differential forms and integration on manifolds with corners}
\label{sec:forms}

For any oriented $n$-manifold $M$ (possibly with corners), $\Omega^k(M)\subset \mathfrak{X}^{(0,k)}(M)$ ($0\le k\le n$) denotes the space of \emph{differential $k$-forms} i.e., smooth anti-symmetric tensor fields of rank $(0,k)$. The exterior derivative (or differential) $\diff:\Omega^k(M)\to \Omega^{k+1}(M)$, satisfies $\diff^2=0$ and the graded Leibniz rule $\diff (\alpha\wedge \beta)=\diff \alpha\wedge \beta+(-1)^k \alpha\wedge \diff \beta$ for $\alpha\in \Omega^k(M)$ and $\beta\in \Omega^l(M)$, where $\wedge$ is the wedge product;
its kernel defines the space of \emph{closed $k$-forms}, $\Omega^k_{\diff}(M)$. 
In particular, $\Omega^0_{\diff}(M)$ is the space of locally constant functions on $M$.
See e.g.~\cite[Ch.~6]{abrahamManifoldsTensorAnalysis1988} and~\cite{gurerDifferentialFormsManifolds2019} for details concerning differential forms on manifolds with boundaries and corners, respectively. 

For $S\subset M$ a submanifold and $\iota:S\to M$ the identity embedding, we write $\alpha\restriction_S=\iota^*\alpha$ as an alternative notation for the pulled-back form.\footnote{This is not to be confused with restriction of the a form on $M$ to $S$, which retains how $\alpha$ acts on vectors that are not tangent to $S$.} If $M$ is a compact oriented manifold with boundaries and (possibly) corners, a version of Stokes theorem  applies (see e.g.~\cite[Thm.~7.2.20]{abrahamManifoldsTensorAnalysis1988} or \cite[Prop.~16.21]{leeIntroductionSmoothManifolds2012}), namely
\begin{equation}
    \int_{M} \diff \alpha=\int_{\partial M}\alpha:=\sum_{i=1}^n\int_{S_i}\alpha\restriction_{S_i},
\end{equation}
for every $\alpha\in \Omega^{n-1}(M)$, 
where $\partial M=\bigcup_{i=1}^n S_i$ is a decomposition of $\partial M$ into appropriately oriented smooth manifolds with boundaries (and possibly corners) $S_i$  with mutually disjoint interiors.

For the rest of this subsection, unless otherwise stated, $M$ will be an oriented $n$-manifold with boundaries and/or corners equipped with a non-degenerate metric $g$, whose signature (typically Riemannian or Lorentzian) will be specified where relevant and is otherwise general. The corresponding volume form will be denoted $\diff\textup{Vol}_g$.  

\paragraph{Hodge star} The Hodge star $\star:\Omega^k(M)\to\Omega^{n-k}(M)$ is the linear map such that
\begin{equation}
    \alpha\wedge \star\beta=(\,\alpha\,,\,\beta\,)_g\diff\textup{Vol}_g,
\end{equation}
for each $\alpha,\beta\in \Omega^k(M)$, where, in components induced by local coordinates $(x^\mu)$,
\begin{equation}
    (\,\alpha\,,\,\beta\,)_g=\frac{1}{k!}g^{\mu_1\nu_1}\cdots g^{\mu_k\nu_k}\alpha_{\mu_1...\mu_k}\beta_{\nu_1...\nu_k},
\end{equation}
using the Einstein summation convention and recalling that $\alpha$ agrees locally with $\frac{1}{k!}\alpha_{\mu_1...\mu_k}\diff x^{\mu_1}\wedge...\wedge \diff x^{\mu_k}$ (and similarly for $\beta$). The Hodge star is invertible, with $(\star)^{-1}=\star$
in the Riemannian case, whereas for Lorentzian metrics in the mostly minus sign convention one has $(\star)^{-1}=(-1)^{n-k-1}\star$ when evaluated on $k$-forms. We refer to $\star\alpha\in\Omega^{n-k}(M)$ as  the Hodge dual of $\alpha\in \Omega^k(M)$. The codifferential   $\diff^*:=(-1)^{k}(\star)^{-1}\circ\diff \circ\star:\Omega^k(M)\to\Omega^{k-1}(M)$ satisfies ${\diff^*}^2=0$, and its kernel $\Omega^k_{\diff^*}(M)$ is the space of \emph{co-closed $k$-forms}.
The Hodge star defines a natural  pairing $\ipc{\,\cdot\,}{\,\cdot\,}_M$ on $\Omega^k(M)$ by
\begin{equation}\label{eq:Hodge_pairing}
    \ipc{\alpha}{\beta}_M=\int_M \alpha\wedge \star\beta,
\end{equation}
which is an inner product if $M$ is compact and Riemannian. 

\paragraph{Laplace-de Rham operator} The Laplace-de Rham operator $\Delta_{k}:\Omega^k(M)\to\Omega^k(M)$, is  $\Delta_{k}=(\diff\diff^*+\diff^*\diff)=(\diff+\diff^*)^2$. In our sign convention, the zero-form Laplace-de Rham operator equals the sign-reversed Laplace-Beltrami operator 
\begin{equation}
    \Delta_{0}=-\frac{1}{\sqrt{\vert\det(g)\vert}}\partial_\mu \sqrt{\vert\det(g)\vert}g^{\mu\nu}\partial_\nu.
\end{equation} 
Keeping this convention in mind, we shall drop the subscript and simply denote the Laplace-de Rham operator on $k$-forms as $\Delta$.

\paragraph{Normals} If $S\subset M$ is a codimension-$1$ oriented submanifold with non-degenerate metric $g_S$ (not necessarily induced from $M$) defining a Hodge star $\star_S$ and a codifferential $\diff^*_S$, then the normal $\nml_S:\Omega^k(M)\to\Omega^{k-1}(S)$ is defined by
\begin{equation}
    \nml_S\alpha=(\star_S)^{-1}[\star\alpha]\restriction_S.
\end{equation}

A useful identity is that
\begin{equation}
\label{eq:nmlds}
    \diff_S^* \nml_S\alpha = -\nml_S\diff^*\alpha.
\end{equation}

If $S\subset \partial M$, one can assign a boundary orientation to $S$, see e.g.~\cite[Def.~7.2.7.]{abrahamManifoldsTensorAnalysis1988}, such that the normal $\nml_S$ can be realised as a contraction with an outward pointing normal vector $n$ on $S$, given in local coordinates by
\begin{equation}
    (\nml_S\alpha)_{\mu_2...\mu_k}=n^{\mu_1}\alpha_{\mu_1\mu_2...\mu_k}.
\end{equation} 
In particular, if $M$ has Riemannian metric $g$ inducing a Riemannian metric $g_S$ on $S$, we can consider local geodesic collar coordinates  $(x^1,x^2,...,x^n):U\to[0,\varepsilon)\times \RR^{n-1}$ on some neigbourhood $U$ of $p\in S$, such that $x^1$ is the minimal geodesic distance to $S$ and the metric $g$ takes the form
\begin{equation}
    g\restriction_{U}=(\diff x^1)^2+\sum_{i,j=2}^ng_{ij}\diff x^i\diff x^j,
\end{equation}
we then find (noting that $x^1$ increases into the interior of $M$)
\begin{equation}
    \nml_{S}\diff x^1\wedge \diff x^{i_2}\wedge...\wedge \diff x^{i_k}=-(\diff x^{i_2}\wedge...\wedge \diff x^{i_k})\restriction_{S},\qquad \nml_{S}\diff x^{i_1}...\wedge \diff x^{i_k}=0,\qquad \text{for }i_1,...,i_k\in\{2,...,n\}.
\end{equation}

Suppose $M$ is compact and the induced metric $g_\partial$ on $\partial M\subset M$ (with boundary orientiation) is non-degenerate. Then one has by Stokes' theorem
\begin{equation}
\label{eq:diffcodiffbound}
    \ipc{\diff\alpha}{\beta}_M-\ipc{\alpha}{\diff^*\beta}_M=\ipc{\alpha}{\nml_{\partial M}\beta}_{\partial M},
\end{equation}
for any $\alpha\in \Omega^k(M)$ and $\beta\in \Omega^{k+1}(M)$. 
The case $k=0$, $\alpha\equiv 1$ gives the divergence theorem
\begin{equation}
    \int_M -\star\diff^*\beta = \int_{\partial M} \nml_{\partial M}\beta
\end{equation}
for $\beta\in \Omega^{1}(M)$.

The Laplace-de Rham operator acts symmetrically on appropriate subspaces of $\Omega^k(M)$. For instance, let $\Omega^k_0(M)$ be the space of smooth $k$-forms with compact support on the interior of $M$. Then for $\alpha,\beta\in \Omega^k_{0}(M)$, the identity~\eqref{eq:diffcodiffbound} yields
\begin{equation}
    \ipc{\alpha}{\Delta\beta}_M=\ipc{\diff \alpha}{\diff\beta}_M+\ipc{\diff^* \alpha}{\diff^*\beta}_M=\ipc{\Delta\alpha}{\beta}_M.
\end{equation}

Now suppose that $M$ is furthermore Lorentzian and (time)-oriented, and $\Sigma\subset M$ is a spacelike submanifold with embedding map $\iota:\Sigma\to M$. Then
$g_\Sigma=-\iota^*g$ is a Riemannian metric on $\Sigma$, and we orient $\Sigma$ so that $\nml_{\Sigma}\alpha>0$ for some (and hence all) everywhere future-directed $\alpha\in \Omega^1(M)$, which we refer to as the \emph{future-directed orientation} of $\Sigma$. 
In particular, for $\ol{N}$ a finite Cauchy lens,  Eq.~\eqref{eq:diffcodiffbound} can be reexpressed as
\begin{equation}
    \ipc{\diff\alpha}{\beta}_{\ol{N}}-\ipc{\alpha}{\diff^*\beta}_{\ol{N}}=\ipc{\alpha\restriction_{\Sigma_+}}{\nml_{\Sigma_+}\beta}_{\Sigma_+}-\ipc{\alpha\restriction_{\Sigma_-}}{\nml_{\Sigma_-}\beta}_{\Sigma_-}.
\end{equation}
We can now define an orientation on $\angle{\ol{N}}$ as the boundary orientation to $\partial\ol{\Sigma}$ with respect to the future directed orientation of a regular Cauchy surface $\ol{\Sigma}$ (for which $\partial\ol{\Sigma}=\angle{\ol{N}}$). This definition turns out to be independent of the choice of regular Cauchy surface and can therefore be called the \emph{corner orientation}.
\begin{lemma}
    Let $\ol{N}$ be a finite Cauchy lens and $\ol{\Sigma},\ol{\Sigma}'\subset\ol{N}$ two regular Cauchy surfaces, then the boundary orientations on $\partial\ol{\Sigma}=\partial\ol{\Sigma}'=\angle\ol{N}$ w.r.t.~the future directed orientation of $\ol{\Sigma}$ and $\ol{\Sigma}'$ coincide.
\end{lemma}
\begin{proof}
    Suppose $\ol{N}$ is $n$-dimensional. The orientation of $\angle\ol{N}$ is uniquely specified by the resulting integration map $\int_{\angle\ol{N}}:\Omega^{n-2}(\angle\ol{N})\to \RR$. Given the orientation on $\angle\ol{N}$ arising from the boundary orientation of $\ol{\Sigma}$, we have for each $\alpha\in\Omega^{n-2}(\ol{N})$ that
    \begin{equation}
        \int_{\angle\ol{N}}\alpha=\int_{\ol{\Sigma}}\diff\alpha=\int_{\ol{\Sigma}^+}\diff\alpha-\int_{J^+(\ol{\Sigma})}\diff^2\alpha=\int_{\ol{\Sigma}^+}\diff\alpha,
    \end{equation}
    having used that given the future directed orientation on $\ol{\Sigma}$ and $\ol{\Sigma}^+$ we have for each $\beta\in \Omega^{n-1}(\ol{N})$
    \begin{equation}
        \int_{J^+(\ol{\Sigma})}\diff\beta=\int_{\ol{\Sigma}^+}\beta-\int_{\ol{\Sigma}}\beta.
    \end{equation}
    It thus follows that the value (and in particular the sign) of $\int_{\angle\ol{N}}\alpha$ is independent on the choice of $\ol{\Sigma}$, and thus the orientation of $\angle\ol{N}$ is uniquely defined.
\end{proof}
The corner orientation of $\angle\ol{N}$, together with the induced Riemannian metric on $\angle\ol{N}$, defines a Hodge star $\star_{\angle\ol{N}}:\Omega^{k}(\angle\ol{N})\to \Omega^{\dim(\ol{N})-2-k}(\angle\ol{N})$ and a \emph{corner normal} $\nml_{\angle\ol{N}}^{(2)}:\Omega^k(\ol{N})\to \Omega^{k-2}(\angle\ol{N})$ given by
\begin{equation}
    \nml_{\angle\ol{N}}^{(2)}\alpha=(\star_{\angle\ol{N}})^{-1}[\star\alpha]\restriction_{\angle\ol{N}}.
\end{equation}
For any regular Cauchy surface $\ol{\Sigma}$ with future-directed orientation, one may straightforwardly compute
\begin{equation}
\label{eq:corner_nml}
    \nml_{\angle\ol{N}}^{(2)}=\nml_{\partial\ol{\Sigma}}\circ\nml_{\ol{\Sigma}}.
\end{equation}

\paragraph{Hodge decomposition} We will employ a number of orthogonal decompositions.
\begin{lemma}\label{lem:Hodge}
Let $M$ be a compact Riemannian manifold with boundary. For $k>0$, $\Omega^k(M)$ admits the following orthogonal decompositions with respect to the $L^2$ inner product~\eqref{eq:Hodge_pairing}:
\begin{align}    
\Omega^k(M)&=\Omega^k_{\tang \diff^*}(M)\oplus\{\diff\lambda:\lambda\in\Omega^{k-1}(M),~\diff^*\diff\lambda=0\}\oplus \{\diff\lambda:\lambda\in\Omega^{k-1}(M),~\lambda|_{\partial M}=0\} \\
&=\Omega^k_{\diff^*}(M)\oplus\{\diff\lambda:\lambda\in\Omega^{k-1}(M),~\lambda|_{\partial M}=0\} \label{eq:nontang_Hodge}\\
&=\Omega^k_{\tang\diff^*}(M)\oplus\diff\Omega^{k-1}(M), \label{eq:tang_Hodge}
\end{align}
where
\begin{equation}\label{eq:Omegaktdstar_def}
    \Omega^k_{\tang \diff^*}(M)=\{\alpha\in \Omega^k(M):\diff^*\alpha=0,\,\nml_{\partial M}\alpha=0\}.
\end{equation}
In addition, if $\partial M\neq\emptyset$, $\Omega^0(\partial M)$ admits the orthogonal decomposition
\begin{equation}\label{eq:nmlcoclosed}
    \Omega^0(\partial M)=\nml_{\partial M}\Omega^1_{\diff^*}(M)\oplus \left(\Omega^0_{\diff}(M)\restriction_{\partial M}\right).
\end{equation} 
\end{lemma}
\begin{proof}
    For the $k>0$ statements, see \cite[Cor. 2.4.9]{schwarzHodgeDecompositionMethod1995}; for~\eqref{eq:nmlcoclosed}, see Lem.~\ref{lem:nmlcoclosed}.
\end{proof}
Note that $\Omega^0_{\diff}(M)\restriction_{\partial M}\subset \Omega^0_\diff(\partial M)$, i.e., the restriction of a locally constant function is locally constant, but the inclusion is strict if one or more of the components of $M$ has a disconnected boundary.
 
An important consequence of Lem.~\eqref{lem:Hodge} is that any $1$-form $\mathbf{A}\in \Omega^1(M)$ 
admits a unique \emph{tangential Hodge--Helmholtz decomposition}
\begin{equation}
    \mathbf{A}=\mathbf{A}^{\tang}+\diff \alpha,
\end{equation}
where $\mathbf{A}^{\tang}\in \Omega^1_{\tang \diff^*}(M)$ (i.e., $\nabla^i \mathbf{A}^\tang_i=0$ and $n^i\mathbf{A}^\tang_i|_{\partial\ol{\Sigma}}=0$) and $\alpha\in \Omega^0(M)$ solves the boundary value problem
\begin{equation}
    \diff^*\diff \alpha=\diff^*\mathbf{A},\qquad \nml_{\partial M} \diff \alpha=\nml_{\partial M} \mathbf{A},
\end{equation}
which determines $\alpha$ up to the addition of a locally constant function.
If $\partial M$ is nonempty, $\alpha$ may be chosen so that $\alpha\restriction_{\partial M}\in \nml_{\partial M}\Omega^1_{\diff^*}(M)$, by~\eqref{eq:nmlcoclosed}, which is equivalent to the condition that
$\ipc{c\restriction_{\partial M}}{\alpha\restriction_{\partial M}}_{\partial M}=0$ for all locally constant $c\in\Omega^0(M)$, or equally that $\int_{\partial M'}\star\alpha\restriction_{\partial M'}=0$ for each connected component $M'\subset M$.

\section{Classical electromagnetism on finite Cauchy lenses} \label{sec:classical}

\subsection{The covariant phase space formalism and the reduced phase space}
\label{sec:cov_phas}
Starting from an action functional,
the covariant phase space formalism provides an algorithm to define a symplectic structure on a space of (on-shell) field configurations -- see e.g.,~\cite{harlowCovariantPhaseSpace2020}.
In outline, if $\mathcal{E}$ is a space of (off-shell) field configurations on a suitable spacetime and $S:\mathcal{E}\to \RR$ is a local action functional, the covariant phase space formalism for a (classical) field theory assigns: 
\begin{itemize}
    \item an equation of motion $EOM:\mathcal{E}\to\mathcal{E}^*$, where $\mathcal{E}^*$ the topological dual of $\mathcal{E}$, specifying an on-shell configuration submanifold $\Sol=EOM^{-1}(\{0\})$,
    \item a (pre-)symplectic form $\mathbf{\Omega}$ on $\Sol$, where the vector fields $\mathfrak{g}\in \mathfrak{X}(\Sol)$ with $\iota_\mathfrak{g}\mathbf{\Omega}=0$ are referred to as the \emph{infinitesimal gauge transformations},
    \item a symplectic manifold $(\Sol_{\mathscr{G}},\mathbf{\Omega})$ obtained through the gauge reduction of $\Sol$.
\end{itemize}

A systematic treatment of this framework for general models requires a theory of infinite dimensional manifolds, allowing one to make precise what constitutes the space $\mathfrak{X}(\Sol)$ of smooth vectorfields on solution submanifolds and what constitutes a pre-symplectic form on $\Sol$. While general treatise of such infinite dimensional manifolds are indeed available, see e.g.~\cite{krieglConvenientSettingGlobal1997}, the case of (pure) electromagnetism on the spacetimes we consider is sufficiently simple to avoid these issues. We refer to~\cite{harlowCovariantPhaseSpace2020} for applications of the covariant phase space formalism to various classical field theories, including pure electromagnetism, on a more general class of spacetime backgrounds that may have time-like boundaries. Below, we give a detailed formulation of pure electromagnetism on an $n+1$ dimensional finite Cauchy lens $\ol{N}$, as a foundation for the discussion of semi-local observables and their quantisation.

\subsubsection{Action, field equation and gauge invariance}

The off-shell configuration space on $\ol{N}$ will be $\mathcal{E}=\Omega^1(\ol{N})$,\footnote{More generally, one may consider connections on some principal $U(1)$-bundle over $\ol{N}$, which may be locally identified with $1$-forms following a choice of trivialisation, see e.g.~\cite{beniniCalgebraQuantizedPrincipal2014}.} while the action $S^J:\mathcal{E}\to\RR$ for
pure electromagnetism with conserved external current $J\in \Omega^1_{\diff^*}(\ol{N})$ is 
\begin{equation}
    {S^J}(A)=-\int_{\ol{N}}\frac{1}{2}\diff A\wedge \star\diff A+J\wedge \star A=-\left(\ipc{\diff A}{\diff A}_{\ol{N}}+\ipc{J}{A}_{\ol{N}}\right).
\end{equation}
The equation of motion is obtained as a map $\EOM^J:\Omega^1(\ol{N})\to \Omega^1(\ol{N})$ (where the co-domain can be identified with $\Omega^1(\ol{N})^*$ through the pairing of forms) by writing the directional functional derivative
\begin{equation}
    \langle S^{J\,(1)}(A),h\rangle:=\partial_t S^{J}(A+th)\restriction_{t=0} 
\end{equation}
for all variations compactly supported away from the boundary, $h\in \Omega^1_0(N)$, in the form
\begin{equation}\label{eq:EOM_def}
    \langle S^{J\,(1)}(A),h\rangle=\ipc{\EOM^J(A)}{h}_{\ol{N}},
\end{equation}
which gives
\begin{equation}
    \EOM^J(A)=-\diff^*\diff A-J.
\end{equation} 
The \emph{on-shell configuration space} or \emph{pre-phase space} $\Sol^J(\ol{N})$ is then given by 
\begin{equation}\label{eq:SolJN_def}
    \Sol^J(\ol{N})
    =\{A\in \Omega^1(\ol{N}):-\diff^*\diff A=J\}.
\end{equation}

In the case of pure electromagnetism, 
$\Sol^J(\ol{N})$ is an affine space and its underlying vectors space is the vector space of solutions to the linearised equation of motion
\begin{equation}
\label{eq:solspace_lin}
    \Sol(\ol{N})
    =\{\udl{A}\in \Omega^1(\ol{N}):\langle S^{J\,(2)}(A),\udl{A}\otimes h\rangle=0\text{ for all }h\in \Omega^1_{0}(N)\}.
\end{equation}
$\Sol^J(\ol{N})$ is an infinite dimensional Fr\'echet manifold (see e.g.~\cite{krieglConvenientSettingGlobal1997}), modelled on the space $\Sol(\ol{N})$ with a Fr\'echet topology inherited from $\Omega^1(\ol{N})$. Accordingly, the kinematical tangent space 
at any $A\in \Sol^J(\ol{N})$, and the total manifold of the kinematical tangent bundle are 
\begin{equation}\label{eq:TSolJNbar_def}
    T_A\Sol^J(\ol{N})\cong \Sol(\ol{N}), \qquad T\Sol^J(\ol{N})\cong \Sol^J(\ol{N})\times \Sol(\ol{N}).
\end{equation}
For the purposes of this work, the spaces $T\Sol^J(\ol{N})$ and $\Sol^J(\ol{N})\times \Sol(\ol{N})$ are identified.

A major theme of this work is the gauge symmetries of the action $S^J$ under the transformation $A\mapsto A + s\diff\Lambda$ for $\Lambda\in\Omega^0(\ol{N})$, $s\in\RR$.
Noting that $S^J(A+s\diff\Lambda)-S^J(A)= s\int_{\ol{N}} \diff (\Lambda J)$, Noether's theorem yields corresponding currents given by 
\begin{equation}
    j_\Lambda = \Lambda J - \iota_{(\diff\Lambda)^\sharp} F,
\end{equation}
which are conserved on-shell, i.e., for any solution $A\in\Sol^J(\ol{N})$. On manifolds with boundary and corners, these gauge symmetries are classified in various ways,
according to their behaviour at the boundary and corner.

\subsubsection{Symplectic potential and pre-symplectic form}

By definition, the quantity $$\int_{\ol{N}}EOM^J(A)\wedge \udl{A}-\langle S^{J\,(1)}(A),\udl{A}\rangle$$ vanishes for $\udl{A}\in\Omega_0^1(N)$.
In general, any map $\Omega^1(\ol{N})\owns A\mapsto \mathbf{\Theta}_A\in\Lin(\Omega^1(\ol{N}),\Omega^3(\ol{N}))$ such that
\begin{equation}
    \int_{\ol{N}}EOM^J(A)\wedge \udl{A}-\langle S^{J\,(1)}(A),\udl{A}\rangle=\int_{\partial\ol{N}}\mathbf{\Theta}_A(\udl{A})
\end{equation}
for all $A,\udl{A}\in\Omega^1(\ol{N})$, will be called a \emph{symplectic potential}. A minimal choice is obtained by requiring
$\mathbf{\Theta}_A(\udl{A})(x)$ to be continuous in $x\in \ol{N}$ and depend on $\udl{A}$ linearly via its value $\udl{A}(x)$, and on $A$ via the $k$-jet prolongation at $x$, with $k\in\NN$ chosen as small as possible.
For electromagnetism, this gives 
\begin{equation}
    \mathbf{\Theta}_A(\udl{A})(x)=(\star\diff A)(x)\wedge \udl{A}(x).
\end{equation}
For each  $\ol{\Sigma}$, a Cauchy surface with boundary, $\int_{\ol{\Sigma}}\mathbf{\Theta}$ is a smooth one-form on $\Omega^1(\ol{N})$ (understood as an infinite dimensional manifold) and we define a pre-symplectic form $\mathbf{\Omega}^{\ol{\Sigma}}$ as a $2$-form on  $\Omega^1(\ol{N})$, satisfying $\mathbf{\Omega}^{\ol{\Sigma}}=\delta \int_{\ol{\Sigma}}\mathbf{\Theta}$, where $\delta$ is the configuration space exterior derivative. Explicitly, 
\begin{align}
       \mathbf{\Omega}^{\ol{\Sigma}}_A(\udl{A}_1\otimes \udl{A}_2)&= \partial_t\int_{\ol{\Sigma}}\mathbf{\Theta}_{A+t\udl{A}_1}(\udl{A}_2)-\mathbf{\Theta}_{A+t\udl{A}_2}(\udl{A}_1)\restriction_{t=0} \nonumber \\
       &=\ipc{\udl{A}_2}{\nml_{\ol{\Sigma}} \diff \udl{A}_1}_{\ol{\Sigma}}-\ipc{\udl{A}_1}{\nml_{\ol{\Sigma}} \diff \udl{A}_2}_{\ol{\Sigma}}
\end{align}
for $A,\udl{A}_1,\udl{A}_2\in \Omega^1(\ol{N})$. Clearly, $\mathbf{\Omega}^{\ol{\Sigma}}_A$ is independent of the base-point $A\in\Omega^1(\ol{N})$ and is (Fr\'echet) smooth on $\Sol(\ol{N})^{\otimes2}$.

As a shorthand, we write
$\mathbf{\Omega}^\pm = \mathbf{\Omega}^{\ol{\Sigma}^\pm}$ for the pre-symplectic forms induced by the bounding Cauchy surfaces $\ol{\Sigma}^\pm$ of $\ol{N}$. A crucial observation is summarised by the proposition below.
\begin{proposition}
   Let $\ol{\Sigma}\subset\ol{N}$ be a co-dimension 1 submanifold such that for every $\alpha\in \Omega^n(\ol{N})=\Omega^{\textup{top}-1}(\ol{N})$ we have 
\begin{equation}
\label{eq:SigmaSigmaMin}
    \int_{\mathscr{J}^-(\ol{\Sigma})}\diff\alpha=\int_{\ol{\Sigma}}\alpha-\int_{\ol{\Sigma}^-}\alpha.
\end{equation} 
(In particular, this holds when $\ol{\Sigma}$ is a regular Cauchy surface with boundary.)
 For $A\in\Sol^J(\ol{N})$, $\udl{A}_1,\udl{A}_2\in\Sol(\ol{N})$, one has 
\begin{equation}
 \mathbf{\Omega}^{\ol{\Sigma}}_A(\udl{A}_1\otimes \udl{A}_2)= \mathbf{\Omega}^{+}_A(\udl{A}_1\otimes \udl{A}_2)=\mathbf{\Omega}^{-}_A(\udl{A}_1\otimes \udl{A}_2)\,.
\end{equation}
\end{proposition}
\begin{proof}
We prove this for $\mathbf{\Omega}^{+}$ and $\mathbf{\Omega}^{-}$, but the same argument applies to a general  $\ol{\Sigma}\subset\ol{N}$ satisfying the condition above.
\begin{align}
    \mathbf{\Omega}^{+}_A(\udl{A}_1\otimes \udl{A}_2)-\mathbf{\Omega}^{-}_A(\udl{A}_1\otimes \udl{A}_2)=&\partial_t\int_{\partial\ol{N}}\mathbf{\Theta}_{A+t\udl{A}_1}(\udl{A}_2)-\mathbf{\Theta}_{A+t\udl{A}_2}(\udl{A}_1)\restriction_{t=0}\nonumber\\
    =&\langle S^{J\,(2)}(A),\udl{A}_2\otimes \udl{A}_1\rangle-\langle S^{J\,(2)}(A),\udl{A}_1\otimes \udl{A}_2\rangle\nonumber\\
    &+\partial_t\int_{\ol{N}}EOM^J(A+t\udl{A}_1)\wedge \udl{A}_2-EOM^J(A+t\udl{A}_2)\wedge \udl{A}_1\nonumber\\
    =&\partial_t\int_{\ol{N}}EOM^J(A+t\udl{A}_1)\wedge \udl{A}_2-EOM^J(A+t\udl{A}_2)\wedge \udl{A}_1.
\end{align}
By Eq.~\eqref{eq:solspace_lin}, it follows that for $(A,\udl{A})\in T\Sol^J(\ol{N})$ we have
\begin{equation}
    \partial_tEOM^J(A+t\udl{A})\restriction_{t=0}=0.
\end{equation}
Hence for $A\in\Sol^J(\ol{N})$, $\udl{A}_1,\udl{A}_2\in\Sol(\ol{N})$ (i.e.~$(A,\udl{A}_1\otimes \udl{A}_2)\in T^2\Sol^J(\ol{N})$) we have
\begin{equation}
    \mathbf{\Omega}^{+}_A(\udl{A}_1\otimes \udl{A}_2)-\mathbf{\Omega}^{-}_A(\udl{A}_1\otimes \udl{A}_2)=0.
\end{equation}
\end{proof}

Putting this together with the base-point independence of $\mathbf{\Omega}^{\ol{\Sigma}}$ mentioned above, we see that there exists a natural pre-symplectic form $\mathbf{\Omega}$ on $\Sol^J(\ol{N})$ and pre-symplectic structure $\sigma$ on $\Sol(\ol{N})$ such that
\begin{equation}
\label{eq:symp_struct}
    \mathbf{\Omega}_{A}(\udl{A}_1,\udl{A}_2)=\sigma(\udl{A}_1,\udl{A}_2)=\ipc{\udl{A}_2}{\nml_{\ol{\Sigma}} \diff \udl{A}_1}_{\ol{\Sigma}}-\ipc{\udl{A}_1}{\nml_{\ol{\Sigma}} \diff \udl{A}_2}_{\ol{\Sigma}},
\end{equation}
for all $A\in \Sol^J(\ol{N})$ and $\udl{A}_1,\udl{A}_2\in \Sol(\ol{N})$, and any regular Cauchy surface with boundaries $\ol{\Sigma}\subset \ol{N}$. 
Written in terms of configuration space one-forms, one has
\begin{equation}
    \mathbf{\Omega}=\int_{\ol{\Sigma}}\delta \mathbf{A}\wedge \star_{\Sigma}\delta \mathbf{E},
\end{equation}
where (suppressing the $\ol{\Sigma}$ dependence in the notation) we have for $(A,\udl{A})\in T\Sol^J(\ol{N})$\begin{equation}
    (\delta\mathbf{A})_{A}(\udl{A})=\udl{A}\restriction_{\ol{\Sigma}},\qquad (\delta\mathbf{E})_{A}(\udl{A})=-\nml_{\ol{\Sigma}}\diff \udl{A}.
\end{equation}

\subsubsection{Pre-symplectic reduction and large gauge transformations}

The reduced phase space is obtained by pre-symplectic reduction of the pre-phase space $(\Sol^J(\ol{N}),\mathbf{\Omega})$ with respect to the radical of $\sigma$, 
\begin{equation}\label{eq:GNbardef}
    \mathscr{G}(\ol{N}):=\{\udl{A}\in \Sol(\ol{N}):\sigma(\udl{A},\udl{A}')=0\text{ for all }\udl{A}'\in \Sol(\ol{N})\}.
\end{equation}
Note that this is a proper subspace of the full space $\diff\Omega^1(\ol{N})$.
\begin{definition}
\label{def:red_phas}
    Let $\ol{N}$ be a finite Cauchy lens and $J\in \Omega^1_{\diff^*}(\ol{N})$, then the \emph{reduced phase space} of electromagnetism with background current $J$ is given by
    \begin{equation}\label{eq:SolJGN_def}
        \Sol_{\mathscr{G}}^J(\ol{N})=\Sol^J(\ol{N})/\mathscr{G}(\ol{N}),
    \end{equation}
    Furthermore, denoting $\Sol_\mathscr{G}(\ol{N}):=\Sol^0_\mathscr{G}(\ol{N})$ we define the symplectic structure on $\Sol_\mathscr{G}(\ol{N})$ by 
    \begin{equation}
        \sigma([\udl{A}_1]\otimes [\udl{A}_2])=\ipc{\udl{A}_2}{\nml_{\ol{\Sigma}} \diff \udl{A}_1}_{\ol{\Sigma}}-\ipc{\udl{A}_1}{\nml_{\ol{\Sigma}} \diff \udl{A}_2}_{\ol{\Sigma}},
    \end{equation}
    where $\ol{\Sigma}\subset\ol{N}$ is a regular Cauchy surface with boundary. 
    
    Identifying, similarly as above, $T^n\Sol^J_{\mathscr{G}}(\ol{N})=\Sol_\mathscr{G}^J(\ol{N})\times(\Sol_\mathscr{G}(\ol{N})^{\otimes n}$ we define the symplectic form $\mathbf{\Omega}:T^2\Sol_\mathscr{G}^J(\ol{N})\to \RR$ via
    \begin{equation}
        ([A],[\udl{A}_1]\otimes [\udl{A}_2])\mapsto \mathbf{\Omega}_{[A]}([\udl{A}_1]\otimes [\udl{A}_2]):=\sigma([\udl{A}_1]\otimes [\udl{A}_2]).
    \end{equation}
\end{definition}
Here we have denoted $[A]:=A+\mathscr{G}(\ol{N})\in \Sol^J_{\mathscr{G}}(\ol{N})$ for $A\in \Sol^J_{\mathscr{G}}$. As is customary, we use the same symbols $\mathbf{\Omega}$ and $\sigma$ for both the pre-symplectic form and structure associated with the pre-phase space as well as the symplectic form and structure associated with the reduced phase space, distinguishing meaning by context. 

\begin{rem}
The gauge reduction described above relies on the affine structure of $(\Sol^J(\ol{N}),\mathbf{\Omega})$. For a general pre-symplectic manifold $(M,\mathbf{\Omega})$, one quotients out the flow of vector fields $\mathfrak{g}\in \mathfrak{X}(M)$ for which $\iota_{\mathfrak{g}}\mathbf{\Omega}=0$.
\end{rem}

In Prop.~\ref{prop:sym_red} it is shown that
\begin{equation}\label{eq:GNbar:result}
    \mathscr{G}(\ol{N})=\{\diff\Lambda:\Lambda\in \Omega^0(\ol{N}),\,\Lambda\restriction_{\angle \ol{N}}=0\},
\end{equation} 
from which it follows that $\Lambda\in\Omega^0(\ol{N})$ obeys $\diff\Lambda\in\mathscr{G}(\ol{N})$ if and only if
\begin{equation}
\label{eq:G_angle}
    \Lambda\restriction_{\angle \ol{N}}\in  \mathscr{G}_{\angle}(\ol{N}):=   \Omega^0_{\diff}(\ol{N})\restriction_{\angle\ol{N}}.
\end{equation} 
Clearly, gauge reduction with respect to $\mathscr{G}(\ol{N})$ only quotients out a subspace of exact $1$-forms. The remaining exact $1$-forms lead to equivalence classes in $\Sol_\mathscr{G}(\ol{N})$ that we call \emph{large gauge directions}
\begin{equation}\label{eq:large_gauge_def}
    \mathscr{LG}(\ol{N}):=\{[\diff\Lambda]:\Lambda\in \Omega^0(\ol{N})\}\subset \Sol_\mathscr{G}(\ol{N}).
\end{equation}
Elements of $\mathscr{LG}(\ol{N})$ are generators for gauge transformations that in \cite{stromingerLecturesInfraredStructure2018} are referred to as \emph{large gauge transformations}.\footnote{Note that this terminology does not coincide with the notion of large gauge symmetries used in the context of disconnected Lie groups (see e.g.~\cite{gomesLargeGaugeTransformations2020}).} Note that $\mathscr{LG}(\ol{N})$ 
is nontrivial whenever $\angle\ol{N}\neq\emptyset$.

\begin{rem}
Elsewhere in the literature, e.g.~\cite{donnellyLocalSubsystemsGauge2016,rielloHamiltonianGaugeTheory2024}, the group generated by $\mathscr{LG}(\ol{N})$ has been referred to as the group of \emph{surface symmetries} or \emph{the flux gauge group}. Similar symmetry groups, for theories on manifolds with either proper or asymptotic boundaries/corners, have been studied by many authors in the context of Yang-Mills, QED, gravitational and topological theories, see e.g.~\cite{reggeRoleSurfaceIntegrals1974,brownCentralChargesCanonical1986,giuliniAsymptoticSymmetryGroups1995,balachandranEdgeStatesGravity1996,barnichCovariantTheoryAsymptotic2002,araujo-regadoSoftEdgesMany2025}. 
\end{rem} 
 
\begin{proposition}\label{prop:large_gauge_boundary}
    For any finite Cauchy lens $\ol{N}$ with corner $\angle\ol{N}\subset\ol{N}$, the map
    \begin{align}\label{eq:LG_iso}
    \mathscr{LG}(\ol{N})& \to \Omega^0(\angle{\ol{N}})/\mathscr{G}_{\angle}(\ol{N})\nonumber\\
        [\diff\Lambda]&\mapsto \Lambda\restriction_{\angle \ol{N}}+\mathscr{G}_{\angle}( \ol{N}),
    \end{align}
    defines a linear isomorphism $\mathscr{LG}(\ol{N})\cong \Omega^0(\angle{\ol{N}})/\mathscr{G}_{\angle}(\ol{N})$. Consequently, defining $\mathscr{LG}_\angle(\ol{N})$
    as the orthogonal complement of $\mathscr{G}_{\angle}( \ol{N})$ in $\Omega^0(\angle\ol{N})$, 
    there is a unique linear isomorphism
    \begin{equation}
        \mathfrak{G}:\mathscr{LG}_{\angle}( \ol{N})\to \mathscr{LG}(\ol{N}) 
    \end{equation}
    so that, for all $\Lambda\in\Omega^0(\ol{N})$, $\mathfrak{G}^{-1}[\diff\Lambda]$ is the representative of $\Lambda\restriction_{\angle \ol{N}}+\mathscr{G}_{\angle}( \ol{N})$ that is 
    orthogonal to $\mathscr{G}_{\angle}( \ol{N})$.
\end{proposition}
\begin{proof}
    By definition, $[\diff\Lambda]=[\diff\Lambda']$ for $\Lambda,\Lambda'\subset \Omega^0(\ol{N})$ if and only if $\diff(\Lambda-\Lambda')\in \mathscr{G}(\ol{N})$; by \eqref{eq:G_angle}, this holds if and only if $(\Lambda-\Lambda')\restriction_{\angle\ol{N}}\in \mathscr{G}_{\angle}(\ol{N})$. Hence the map~\eqref{eq:LG_iso} is well-defined and injective; linearity is obvious
    and surjectivity holds because any smooth function on $\angle\ol{N}$ extends smoothly to $\ol{N}$ (see e.g.~\cite[Lem.~5.34]{leeIntroductionSmoothManifolds2012}) and
    elements of $\mathscr{G}_{\angle}(\ol{N})$ may be extended to locally constant functions. The consequence follows from~\eqref{eq:nmlcoclosed} because $\Omega^0(\angle{\ol{N}})/\mathscr{G}_{\angle}(\ol{N})\cong\mathscr{G}_{\angle}(\ol{N})^\perp$.
\end{proof}
Consequently, the space $\mathscr{LG}(\ol{N})$ depends not only on the geometry of $\angle{\ol{N}}$ but also on topological features of $\ol{N}$ (particularly on the connectedness of its corner components). A similar observation has previously been made from a Hamiltonian perspective in \cite[Sec.~5.4]{rielloHamiltonianGaugeTheory2024}.
We will use the isomorphism $\mathfrak{G}$ to parametrise the large gauge directions $\mathscr{LG}(\ol{N})$ by $\mathscr{LG}_\angle(\ol{N})$.

\subsubsection{Initial value problem}

We finish this subsection with the existence and uniqueness (modulo appropriate gauge transformations) of the initial value problem for the equation of motion $EOM^J(A)=0$. One can use this to define a symplectomorphism between symplectic manifolds of solutions and initial data respectively.

\begin{proposition}
\label{prop:ext_ini_dat}
    Let $\ol{N}$ be a finite Cauchy lens, $\ol{\Sigma}\subset \ol{N}$ be a regular Cauchy surface with boundaries and $J\in \Omega^1_{\diff^*}(\ol{N})$. Set $\rho=\nml_{\ol{\Sigma}}J\in \Omega^0(\ol{\Sigma})$ and consider the spaces
    \begin{align}
        \mathscr{G}(\ol{\Sigma})&=\diff \{\Lambda\in \Omega^0(\ol{\Sigma}):\Lambda\restriction_{\partial\ol{\Sigma}}=0\} = \mathscr{G}(\ol{N})\restriction_{\ol{\Sigma}} \label{eq:Gsigma_def}\\
        \mathcal{A}(\ol{\Sigma})&=\Omega^1(\ol{\Sigma})/\mathscr{G}(\ol{\Sigma}) \label{eq:Asigma_def} \\
        \mathcal{E}^\rho(\ol{\Sigma})&=\{\mathbf{E}\in \Omega^1(\ol{\Sigma}):-\diff^*_{\ol{\Sigma}}\mathbf{E}=\rho\}. \label{eq:Erho_def}
    \end{align}
    Then the map
    \begin{equation}
        \mathfrak{I}^J_{\ol{\Sigma}}:\Sol_{\mathscr{G}}^J(\ol{N})\to \mathcal{A}(\ol{\Sigma})\times \mathcal{E}^\rho(\ol{\Sigma}),
    \end{equation}
    given by
    \begin{equation}
        \mathfrak{I}^J_{\ol{\Sigma}}([A])=(A\restriction_{\ol{\Sigma}}+\mathscr{G}(\ol{\Sigma}),-\nml_{\ol{\Sigma}}\diff A),
    \end{equation} 
    is a bijection.
\end{proposition}
The proof relies on results shown in Appx~\ref{apx:IV_prob_lens}.
\begin{proof}
    Clearly the map $\mathfrak{I}^J_{\ol{\Sigma}}:\Sol_\mathscr{G}^J(\ol{N})\to \mathcal{A}(\ol{\Sigma})\times\mathcal{E}^\rho(\ol{\Sigma})$ is well defined. Surjectivity follows directly from Lem.~\ref{lem:IV_constr}. To prove injectivity, note that for any $A,A'\in \Sol^J(\ol{N})$, with $\mathfrak{I}^J_{\ol{\Sigma}}([A])=\mathfrak{I}^J_{\ol{\Sigma}}([A'])$, and $\tilde{A}\in \Sol(\ol{N})$ that
    \begin{equation}
        \sigma(\tilde{A},A-A')=\int_{\ol\Sigma}\left((A-A')\wedge \star_{\ol{\Sigma}}\nml_{\ol{\Sigma}}\diff \tilde{A}- \tilde{A}\wedge \star_{\ol{\Sigma}}\nml_{\ol{\Sigma}}\diff (A-A')\right)=\int_{\ol\Sigma}g\wedge \star_{\ol{\Sigma}}\nml_{\ol{\Sigma}}\tilde{A}=0,
    \end{equation}
    for some $g\in \mathscr{G}(\ol{\Sigma})$. Thus, by Prop.~\ref{prop:sym_red}, one has $A-A'\in \mathscr{G}(\ol{N})$ and thus $[A]=[A']$.
\end{proof}
The space $\mathcal{A}(\ol{\Sigma})\times \mathcal{E}^\rho(\ol{\Sigma})$ can be treated as an affine space with respect to the vector space $\mathcal{A}(\ol{\Sigma})\oplus\mathcal{E}(\ol{\Sigma})$, where $\mathcal{E}(\ol{\Sigma})=\mathcal{E}^0(\ol{\Sigma})$. In particular, for any $(A,\udl{A})\in T\Sol^J(\ol{N})$, we have
\begin{equation}
    \mathfrak{I}^J_{\ol{\Sigma}}([A+\udl{A}])=\mathfrak{I}^J_{\ol{\Sigma}}([A])+\mathfrak{I}_{\ol{\Sigma}}([\udl{A}]),
\end{equation}
where $\mathfrak{I}_{\ol{\Sigma}}=\mathfrak{I}^0_{\ol{\Sigma}}$. Therefore $\mathfrak{I}^J_{\ol{\Sigma}}$ is an affine map. 

$\Sol^J_{\mathscr{G}}(\ol{N})$ is an affine symplectic space, i.e. an affine space whose associated vector space is symplectic. Similarly, we can endow $\mathcal{A}(\ol{\Sigma})\oplus\mathcal{E}(\ol{\Sigma})$ with a symplectic structure such that $\mathcal{A}(\ol{\Sigma})\times \mathcal{E}^\rho(\ol{\Sigma})$ is an affine symplectic space. Since $\mathfrak{I}_{\ol{\Sigma}}$ is a linear isomorphism and $\Sol_{\mathscr{G}}(\ol{N})$ is endowed with a symplectic structure $\sigma$, we may define a symplectic structure $\sigma_{\ol{\Sigma}}$ on $\mathcal{A}(\ol{\Sigma})\oplus\mathcal{E}(\ol{\Sigma})$ by 
\begin{equation}
    \sigma_{\ol{\Sigma}}=\sigma\circ\left(\left(\mathfrak{I}_{\ol{\Sigma}}\right)^{-1}\times \left(\mathfrak{I}_{\ol{\Sigma}}\right)^{-1}\right).
\end{equation}
It is then straightforward to compute for $(\mathbf{A}_i+\mathscr{G}(\ol{\Sigma}))\in \mathcal{A}(\ol{\Sigma})$ and $ \mathbf{E}_i\in \mathcal{E}(\ol{\Sigma})$ with $i\in \{1,2\}$ that
\begin{equation}
        \label{eq:kin_symp}\sigma_{\ol{\Sigma}}((\mathbf{A}_1+\mathscr{G}(\ol{\Sigma}))\oplus \mathbf{E}_1,(\mathbf{A}_2+\mathscr{G}(\ol{\Sigma}))\oplus \mathbf{E}_2)=\ipc{\mathbf{A}_1}{\mathbf{E}_2}_{\ol{\Sigma}}-\ipc{\mathbf{A}_2}{\mathbf{E}_1}_{\ol{\Sigma}}.
    \end{equation}
    By construction, this elevates $\mathfrak{I}_{\ol{\Sigma}}$ to a symplectomorphism.
\begin{corollary}
    Let $\ol{N}$ be a finite Cauchy lens, $\ol{\Sigma}\subset \ol{N}$ be a regular Cauchy surface with boundaries. Endow $\mathcal{A}(\ol{\Sigma})\oplus\mathcal{E}(\ol{\Sigma})$ with the symplectic structure
    \begin{equation}
        \sigma_{\ol{\Sigma}}:\left(\mathcal{A}(\ol{\Sigma})\oplus\mathcal{E}(\ol{\Sigma})\right)^2\to \RR,
    \end{equation}
    given by Eq.~\eqref{eq:kin_symp}. Then the map $\mathfrak{I}_{\ol{\Sigma}}:\Sol_{\mathscr{G}}(\ol{N})\to \mathcal{A}(\ol{\Sigma})\oplus\mathcal{E}(\ol{\Sigma})$ is a linear symplectomorphism.
\end{corollary}

\subsubsection{Decomposition into closed-loop and surface configurations}
\label{sec:init_dat_hodge}

Given a choice of Cauchy surface, the reduced phase space of pure electromagnetism naturally decomposes as a product of two symplectic manifolds, see for example \cite{gomesQuasilocalDegreesFreedom2021,rielloHamiltonianGaugeTheory2024}. 
This occurs in our setting as follows.
\begin{definition}
\label{def:hodge_spaces}
    For $\ol{\Sigma}\subset\ol{N}$ a finite Cauchy lens with regular Cauchy surface with boundary, we define the symplectic spaces $(V^C(\ol{\Sigma}),\sigma^C_{\ol{\Sigma}})$ and $(V^S(\ol{N}),\sigma^S_{\ol{N}})$ by
    \begin{align}
         V^C(\ol{\Sigma})=\Omega^1_{\tang\diff^*}(\ol{\Sigma})^{\oplus2},\qquad V^S(\ol{N})=\mathscr{LG}_\angle(\ol{N})^{\oplus2} \label{eq:VCVS_def}\\
         \sigma^C_{\ol{\Sigma}}(F_1\oplus H_1,F_2\oplus H_2)=\ipc{F_1}{H_2}_{\ol{\Sigma}}-\ipc{F_2}{H_1}_{\ol{\Sigma}},\\
         \sigma^S_{\ol{N}}(f_1\oplus h_1,f_2\oplus h_2)=\ipc{f_1}{h_2}_{\angle\ol{N}}-\ipc{f_2}{h_1}_{\angle\ol{N}}.
    \end{align}
\end{definition}
It is important to note that $\mathscr{LG}_\angle(\ol{N})$, defined in Prop.~\ref{prop:large_gauge_boundary} as the orthogonal complement of $\mathscr{G}_\angle(\ol{N})=\Omega^0_{\diff}(\ol{N})\restriction_{\angle\ol{N}}=\Omega^0_{\diff}(\ol{\Sigma})\restriction_{\partial\ol{\Sigma}}$ in $\Omega^0(\angle\ol{N})=\Omega^0(\partial\ol{\Sigma})$, coincides with $\nml_{\partial\ol{\Sigma}}\Omega^1_{\diff^*}(\ol{\Sigma})$ by Eq.~\eqref{eq:nmlcoclosed}. 
We now give the decomposition of $\Sol^J_{\mathscr{G}}(\ol{N})$.
\begin{proposition}
\label{prop:sympl_decomp}
Let $\ol{\Sigma}$ be a regular Cauchy surface with boundary for a finite Cauchy lens $\ol{N}$, let $J\in \Omega^1_{\diff^*}(\ol{N})$ and $\rho=\nml_{\ol{\Sigma}}J$. Then there exist bijective symplectomorphisms
\begin{equation}
    \mathfrak{K}_{\ol{\Sigma}}:(\Sol_\mathscr{G}(\ol{N}),\sigma)\to (V^C(\ol{\Sigma})\oplus V^S(\ol{N}),\sigma^C_{\ol{\Sigma}}\oplus \sigma^S_{\ol{N}}),
\end{equation}
\begin{equation}
    \mathfrak{K}_{\ol{\Sigma}}^J:(\Sol_{\mathscr{G}}^J(\ol{N}),\sigma)\to (V^C(\ol{\Sigma})\times V^S(\ol{N}),\sigma^C_{\ol{\Sigma}}\oplus \sigma^S_{\ol{N}}),
\end{equation}
where $\mathfrak{K}_{\ol{\Sigma}}$ is linear and $\mathfrak{K}_{\ol{\Sigma}}^J$  is affine in the sense that for $([A],[\udl{A}])\in T\Sol_\mathscr{G}^J(\ol{N})$
\begin{equation}\label{eq:linearised_map_def}
    \mathfrak{K}_{\ol{\Sigma}}^J([A]+[\udl{A}])=\mathfrak{K}_{\ol{\Sigma}}^J([A])+\mathfrak{K}_{\ol{\Sigma}}([\udl{A}]).
\end{equation}
\end{proposition}
We say that $\mathfrak{K}_{\ol{\Sigma}}$ is the linearisation of the affine transformation $\mathfrak{K}_{\ol{\Sigma}}^J$,
which we call the \emph{Cauchy--Hodge--Helmholtz (CHH) decomposition}.

Prop.~\ref{prop:sympl_decomp} is proved in Appx~\ref{apx:proof_of_VCVS_decomp}.
The explicit formulae for $\mathfrak{K}_{\ol{\Sigma}}^J$ and $\mathfrak{K}_{\ol{\Sigma}}^J$ are given as follows. 
Starting with $[A]\in\Sol_\mathscr{G}^J(\ol{N})$, let $(\mathbf{A}+\mathscr{G}(\ol{\Sigma}),\mathbf{E})=\mathfrak{I}^J_{\ol{\Sigma}}([A])\in \mathcal{A}(\ol{\Sigma})\times \mathcal{E}^\rho(\ol{\Sigma})$ be the corresponding Cauchy data. By~\eqref{eq:nontang_Hodge}, 
$\mathbf{A}+\mathscr{G}(\ol{\Sigma})$ has a unique representative in $\Omega^1_{\diff^*}(\ol{\Sigma})$ (which vanishes if $\mathbf{A}\sim 0$), so we may assume $\diff^*\mathbf{A}=0$ without loss.
Applying the 
tangential Hodge--Helmholtz decomposition gives
\begin{align}
    \mathbf{A}=&\mathbf{A}^{\tang}+\diff\alpha,\nonumber\\
    \mathbf{E}=&\mathbf{E}^{\tang}+\diff \varepsilon,
\end{align}
where $\mathbf{A}^{\tang},\mathbf{E}^{\tang}\in \Omega^1_{\tang\diff^*}(\ol{\Sigma})$,
$\alpha\in \Omega^0(\ol{\Sigma})$ is chosen to obey $\alpha\restriction_{\partial\ol{\Sigma}}\in\nml_{\partial\ol{\Sigma}}\Omega^1_{\diff^*}(\ol{\Sigma})$ and $\Delta \alpha=\diff^*\mathbf{A}=0$, while $\varepsilon\in \Omega^0(\ol{\Sigma})$ solves $-\diff^*\diff\varepsilon=\rho$.  
Finally, let 
 $f^\rho\in\Omega^0(\partial\ol{\Sigma})$ be a locally constant function equal, on the boundary of each component of $\ol{\Sigma}$, to the total charge in that component 
 divided by the area of its boundary.
Then
\begin{equation}\label{eq:KJSigma}
 \mathfrak{K}_{\ol{\Sigma}}^J(A)  = \left(\mathbf{A}^{\tang},\mathbf{E}^{\tang};\alpha\restriction_{\partial\ol{\Sigma}},\nml_{\partial\ol{\Sigma}}\diff\varepsilon-f^\rho\right) = 
  \left(\mathbf{A}^{\tang},\mathbf{E}^{\tang};\alpha\restriction_{\partial\ol{\Sigma}},\nml_{\partial\ol{\Sigma}}\mathbf{E}-f^\rho\right)
\end{equation}
The linearised map $\mathfrak{K}_{\ol{\Sigma}}$ is obtained from~\eqref{eq:linearised_map_def}.  

The decomposition via $\mathfrak{K}^J_{\ol{\Sigma}}$ is a special case of those considered in \cite{rielloHamiltonianGaugeTheory2024}, used for a local characterisation of the reduced phase space of Yang-Mills theories, where the decomposition of the electric field given above is called the `radiative-Coulombic decomposition'. As will be explained in Sec.~\ref{sec:class_semi-loc}, we will instead say that  $V^C(\ol{\Sigma})$ comprises \emph{closed loop} configurations and $V^S(\ol{N})$ \emph{surface} configurations. We stress that the decomposition relies on the choice of $\ol{\Sigma}$ and is non-local, as computing $\mathfrak{K}_{\ol{\Sigma}}$ involves solving an elliptic boundary problem on $\ol{\Sigma}$.  

The expression for $\mathfrak{K}_{\ol{\Sigma}}$ simplifies for large gauge directions. If  $\Lambda\in\Omega^0(\Sigma)$ then $\mathfrak{I}^J_{\ol{\Sigma}}([\diff\Lambda])=(\diff\Lambda\restriction_{\ol{\Sigma}}+\mathscr{G}(\ol{\Sigma}),0)$ and the tangential component of the Hodge--Helmholtz decomposition vanishes, while
the scalar part $\alpha$ is such that $\alpha\restriction_{\partial\ol{\Sigma}}$ is
the representative of $(\diff\Lambda+\mathscr{G}(\ol{\Sigma}))\restriction_{\partial\ol{\Sigma}}$ that is orthogonal to $\mathscr{G}_\angle(\ol{N})=\mathscr{G}(\ol{\Sigma})\restriction_{\partial\ol{\Sigma}}$, i.e., $\mathfrak{G}^{-1}[\diff\Lambda]$. Therefore 
\begin{equation}
    \mathfrak{K}_{\ol{\Sigma}}(\mathfrak{G}(\lambda))=0\oplus 0\oplus \lambda\oplus 0
\end{equation}
for all $\lambda\in \mathscr{LG}_\angle(\ol{N})$. 
 
Given the reduced phase space $(\Sol^J_{\mathscr{G}},\mathbf{\Omega})$, 
an algebra of classical observables is a commutative algebra $\mathcal{F}$ of sufficiently regular functionals on $\Sol^J(\ol{N})$ that is closed under a Poisson bracket associated to $\mathbf{\Omega}$. We discuss two such Poisson algebras, namely, the algebra of \emph{semi-local} observables $\Fsloc$ and the subalgebra of \emph{local} observables $\Floc$.

\subsection{Classical local observables}\label{sec:class_loc_obs}

Local observables for electromagnetism have been considered in various contexts, see e.g.~\cite{dimockQuantizedElectromagneticField1992,sandersElectromagnetismLocalCovariance2014,buchholzUniversalCalgebraElectromagnetic2016,strohmaierClassicalQuantumPhoton2021}. Our Poisson algebra
$\Floc$ coincides with the observables considered in \cite{sandersElectromagnetismLocalCovariance2014} and is generated by smeared fields. We begin by introducing Green operators for the wave equation.

\begin{proposition}
\label{prop:green}
    Let $\ol{N}$ be a finite Cauchy lens with interior $N=\Intr(\ol{N})$. There exist unique retarded/advanced Green operators $G^\pm:\Omega_0^k(N)\to \Omega^k(\ol{N})$ for the wave operator $-(\diff+\diff^*)^2$ so that  
    \begin{equation}
        -(\diff^*\diff +\diff\diff^*)G^\pm f =f,\qquad \supp(G^\pm f)\subset \mathscr{J}^\pm(\supp(f)).
    \end{equation}
    The difference $G^{\PJ}:=G^--G^+$ defines the \emph{Pauli--Jordan operator}. One has $\diff G^{\#}=G^\#\diff$ and $\diff^*G^{\#}=G^\#\diff^*$ for $\#\in\{+,-,\PJ\}$.
\end{proposition}
\begin{proof}
    Let $\iota:\ol{N}\to M$ be a globally hyperbolic extension of $\ol{N}$ with $\partial M=\emptyset$. The wave equation on $M$ has unique Green operators $G^\pm_M:\Omega_0^k(M)\to \Omega^k(M)$ because it is normally hyperbolic -- see e.g.~\cite[Cor.~3.4.3]{barWaveEquationsLorentzian2007}) -- and we set 
    \begin{equation}
        G^\pm = \iota^* G_M^\pm \iota_*,
    \end{equation}
    where the push-forward $\iota_*f\in\Omega_0^k(M)$ is the extension by zero of $f\in \Omega_0^k(N)$. These operators have the required properties. Uniqueness follows 
    of uniqueness of Green operators on $N$. The last statement follows from the intertwining identities of $\diff$ and $\diff^*$ -- see e.g.,~\cite{sandersElectromagnetismLocalCovariance2014}.
\end{proof}

The Pauli-Jordan operator plays a key role in the definition of the Peierls bracket on local observables \cite{peierlsCommutationLawsRelativistic1952}, defining a Poisson algebra.
\begin{definition}
\label{def:class-loc-obs}
    For $\ol{N}$ a finite Cauchy lens and $J\in \Omega^1_{\diff^*}(\ol{N})$, the \emph{algebra of local observables} $\Floc^{}
    (\ol{N})\subset \mathcal{C}^\infty(\Sol_{\mathscr{G}}^J(\ol{N}), \RR)$ is the unital algebra of functionals generated by \emph{smeared fields} $A(f)\in\mathcal{C}^\infty(\Sol_{\mathscr{G}}^J(\ol{N}),\RR)$ with $f\in \Omega^1_{0\diff^*}(\ol{N})$, where\footnote{The definition of $A(f)([A])$ is independent of the choice of representative $A\in [A]$, as for each $f\in \Omega^1_{0\diff^*}(\ol{N})$ and $\diff\lambda\in \mathscr{G}(\ol{N})$ we have $\ipc{f}{\diff\lambda}=\ipc{\diff^*f}{\lambda}=0$.}
    \begin{equation}
        A(f)([A]):=\ipc{f}{A}_N,
    \end{equation}
    and the Poisson bracket on $\Floc^{}(\ol{N})$  (the Peierls bracket) is defined by
    \begin{equation}
        \{A(f),A(h)\}=-\ipc{f}{G^{\PJ}h}_N \one.
    \end{equation}
\end{definition}
The dependence of $\Floc^{}(\ol{N})$ on the background current $J$ has been suppressed in the notation, but the generators $A(f)$ are sensitive to the background current via the relation $A(-\diff^*\diff h)=\ipc{h}{J}\one$ for each $h\in \Omega^1_{0\diff^*}(\ol{N})$. 

\subsubsection{Local symplectically smeared fields}

In what follows, we will show how the Peierls bracket on $\Floc^{}(\ol{N})$ is related to the canonical structure obtained via the covariant phase space formalism, as described in Sec.~\ref{sec:cov_phas}.\footnote{For a general discussion of the relation between Peierls' method and the covariant phase space formalism see e.g.~\cite{forgerCovariantPoissonBrackets2005,khavkineCovariantPhaseSpace2014}.} This can be shown by expressing smeared fields evaluated on $\Sol_{\mathscr{G}}(\ol{N})$ in terms of symplectically smeared fields. We first introduce the following spaces.

\begin{definition}
\label{def:sol_sc}
    For $\ol{N}$ a finite Cauchy lens with interior $N$, we denote forms of \emph{spatially compact support} by
    \begin{equation}
        \Omega^k_{\sc}(\ol{N}):=\{\alpha\in \Omega^k(\ol{N}):\supp(\alpha)\subset \mathscr{J}(K)\text{ for some compact }K\subset N).
    \end{equation}
    Now let
    \begin{equation}
        \Sol_{\sc}(\ol{N})=\Sol(\ol{N})\cap \Omega^1_{\sc}(\ol{N}),
    \end{equation}
    and
    \begin{equation}
        \Sol_{\sc,\mathscr{G}}(\ol{N})=\{[A]\in \Sol_{\mathscr{G}}(\ol{N}):A\in \Sol_{\sc}(\ol{N})\}.
    \end{equation}
\end{definition}
Using the spatially compact solution orbits $\Sol_{\sc,\mathscr{G}}(\ol{N})$, we define the local symplectically smeared fields.
\begin{definition}\label{def:SFloc}
     Let $\ol{N}$ be a finite Cauchy lens and $J\in \Omega^1_{\diff^*}(\ol{N})$. We define the  \emph{local symplectically smeared fields}, $\SFloc(\ol{N})$, 
     by
    \begin{equation}
        \SFloc(\ol{N}):=\{O_{A}([\udl{A}]): A\in\Sol^J(\ol{N}),\ [\udl{A}]\in \Sol_{\sc,\mathscr{G}}(\ol{N})\},
    \end{equation}
    where $O_{A}([\udl{A}]):\Sol^J_{\mathscr{G}}(\ol{N})\to \RR$ is given by
\begin{equation}
    O_{A}([\udl{A}])([\tilde{A}])=\sigma(\udl{A},\tilde{A}-A).
\end{equation}
\end{definition}
\begin{rem}
\label{rem:symp_func_dep} The symplectically smeared fields are not functionally independent, in fact we have
\begin{equation}
    \label{eq:symp_shift}
        O_{A'}([\udl{A}])=O_{A}([\udl{A}])+\sigma(\udl{A},A-A')\one,
\end{equation}
for each $A,A'\in \Sol^J(\ol{N})$ and $\udl{A}\in \Sol_{\sc}(\ol{N})$. In particular, we see that $O_A([\udl{A}])=O_{A'}([\udl{A}])$ if $A-A'\in \mathscr{G}(\ol{N})$.
\end{rem}
A priori it may not seem obvious that the observables introduced in Def.~\ref{def:SFloc} are in fact local, this is shown below. In fact, similarly to the smeared fields, the local symplectically smeared fields form a generating set of the local observables.

\begin{proposition}
\label{prop:Floc_sympl_smear}
    Let $\ol{N}$ be a finite Cauchy lens and $J\in \Omega^1_{\diff^*}(\ol{N})$. Then:
    \begin{enumerate}[a)]
    \item $\SFloc(\ol{N})\subset\Floc(\ol{N})$, i.e.~the local symplectically smeared fields are local observables, and each local observable can be written as a polynomial in local symplectically smeared fields.
    \item The Poisson bracket on $\Floc(\ol{N})$ evaluates on local symplectically smeared fields as
    \begin{equation}
    \label{eq:symp_pois_bra}
        \{O_{A}([\udl{A}]),O_{A'}([\udl{A}'])\}=\sigma(\udl{A},\udl{A}')\one,
    \end{equation}
    for each $A,A'\in\Sol^J(\ol{N})$ and $[\udl{A}],[\udl{A}']\in \Sol_{\sc,\mathscr{G}}(\ol{N})$.
    \item For each $\psi\in \Floc^{}(\ol{N})$ there exists a unique Hamiltonian vector field $X_\psi\in \mathfrak{X}(\Sol_\mathscr{G}^J(\ol{N}))$. That is, for each $[\tilde{A}]\in\Sol_\mathscr{G}^J(\ol{N})$ there exists a unique $X_\psi([\tilde{A}])\in \Sol_{\mathscr{G}}(\ol{N})$ such that
\begin{equation}
    \langle \psi^{(1)}([\tilde{A}]),[\udl{A}]\rangle=\mathbf{\Omega}_{[\tilde{A}]}(X_\psi([\tilde{A}])\otimes [\udl{A}]),
\end{equation}
for all $[\udl{A}]\in\Sol_\mathscr{G}(\ol{N})$.
\item  For $\varphi,\psi\in \Floc^{}(\ol{N})$,
\begin{equation}
    \{\varphi,\psi\}([\tilde{A}])=\mathbf{\Omega}_{[\tilde{A}]}(X_\varphi([\tilde{A}])\otimes X_\psi([\tilde{A}])).
\end{equation}
    \end{enumerate}
\end{proposition}

In particular, point \emph{d)} shows that the Peierls bracket and the Poisson bracket associated with $\mathbf{\Omega}$ agree on the local observables, as similarly discussed in \cite{forgerCovariantPoissonBrackets2005,khavkineCovariantPhaseSpace2014,harlowCovariantPhaseSpace2020}. To prove Prop.~\ref{prop:Floc_sympl_smear}, we need two preparatory lemmas, adapting standard results to our setting (see e.g.~\cite[Lem.~2.4 \& Thm.~3.3]{sandersElectromagnetismLocalCovariance2014}). Lem.~\ref{lem:st_to_symp_smear} is used to show that one can rewrite smeared fields in symplectically smeared fields, whereas Lem.~\ref{lem:green_sc} shows that each local symplectically smeared field may be obtained this way.

\begin{lemma}\label{lem:localvssymp}
\label{lem:st_to_symp_smear}
If $\ol{N}$ is a finite Cauchy lens, then  
\begin{equation}
\label{eq:st_to_symp_smear}
    A(f)([\udl{A}])=-\sigma(G^{\PJ}f,\udl{A})
\end{equation}
for all $\udl{A}\in \Sol(\ol{N})$ and $f\in \Omega^1_{0\diff^*}(\ol{N})$.
\end{lemma}
\begin{proof}
One observes that
\begin{equation}
A(f)([\udl{A}])=\ipc{f}{\udl{A}}_N=-\ipc{\diff^*\diff G^-(f)}{\udl{A}}_{N\cap J^+(\Sigma)}-\ipc{\diff^*\diff G^+(f)}{\udl{A}}_{N\cap J^-(\Sigma)}
\end{equation}
and uses Stokes' theorem in the form~\eqref{eq:diffcodiffbound} to compute
\begin{align}
-\ipc{\diff^*\diff G^\pm(f)}{\udl{A}}_{N\cap J^\mp(\Sigma)} &= \pm \ipc{\nml_{\Sigma}\diff G^{\pm}(f)}{\udl{A}}_{\Sigma} - \ipc{\diff G^\pm(f)}{\diff \udl{A}}_{N\cap J^\mp(\Sigma)} \\
&= \pm \ipc{\nml_{\Sigma}\diff G^{\pm}(f)}{\udl{A}}_{\Sigma}\mp  \ipc{G^{\pm}(f)}{\nml_{\Sigma} \diff\udl{A}}_{\Sigma}
- \ipc{G^\pm(f)}{\diff^*\diff \udl{A}}_{N\cap J^\mp(\Sigma)} 
\nonumber
\end{align}
where the normal is chosen to be future-pointing in both cases. Adding,
\begin{equation}
\ipc{f}{\udl{A}}_N =-\ipc{\nml_{\Sigma}\diff G^{\PJ}(f)}{\udl{A}}_{\Sigma}+\ipc{G^{\PJ}(f)}{\nml_{\Sigma}\diff \udl{A}}_{\Sigma} = -\sigma(G^{\PJ}(f),\udl{A}),
\end{equation}
as required.
\end{proof}

\begin{lemma}
\label{lem:green_sc}
    Let $\ol{N}$ be a finite Cauchy lens. Then
    \begin{equation}\label{eq:Solsc_calc}
         \Sol_{\sc,\mathscr{G}}(\ol{N})=\{[G^{\PJ}f]:f\in \Omega^1_{0\diff^*}(\ol{N})\}.
    \end{equation}
    Moreover, let $K\subset U \subset \Intr(\ol{N})$ 
    with $K$ compact, and $U$ open and causally convex. Then for each $\udl{A}\in \Sol(\ol{N})$ with $\supp(\udl{A})\subset \mathscr{J}(K)$, there exists an $f\in \Omega^1_{0\diff^*}(U)$ with $[\udl{A}]=[G^{\PJ}f]$. 
\end{lemma}
\begin{proof}
    For each $f\in \Omega^1_{0\diff^*}(\ol{N})$, one has $\diff^*G^{\PJ}f=G^{\PJ}\diff^*f=0$, and thus $-\diff^*\diff G^{\PJ}f=-(\diff^*\diff +\diff\diff^*)G^{\PJ}f=0$. It thus follows that $G^{\PJ}f\in \Sol(\ol{N})$. Since by Prop.~\ref{prop:green}, one has $\supp(G^{\PJ}f)\subset \mathscr{J}(\supp(f))$, we find $[G^{\PJ}f]\in \Sol_{\sc,\mathscr{G}}(\ol{N})$, so the right-hand side of~\eqref{eq:Solsc_calc} is contained in the left.

    The reverse inclusion follows from the final statement, so it suffices to prove that. With $K$ and $U$ as in the hypotheses, consider an $\udl{A}\in \Sol_{\sc}(\ol{N})$ with $\supp(\udl{A})\subset \mathscr{J}(K)$. We will show that $[\udl{A}]=[G^{\PJ}f]$ for some  $f\in \Omega^1_{0\diff^*}(U)\subset \Omega^1_{0\diff^*}(\ol{N})$. 
    
    Choose a precompact open set $V$ with $K\subset V\subset U$ and choose $\chi\in \Omega^0(N)$ such that $\chi=1$ on $N\setminus \mathscr{J}^-(V)$ and $\chi=0$ on $\mathscr{J}^-(K)$. Defining $f=\diff^*\diff \chi \udl{A}$, one has $\supp(f)\subset \mathscr{J}(K)\cap (\mathscr{J}^-(V)\setminus \mathscr{J}^-(K))\subset \mathscr{J}^+(K)\cap \mathscr{J}^-(V)\subset U$ by causal convexity of $U$. This means $f\in \Omega^1_{0\diff^*}(U)$. By Lem.~\ref{lem:st_to_symp_smear}, we have for each $\udl{A}'\in \Sol(\ol{N})$
    \begin{align}
        \sigma(G^{\PJ}f,\udl{A}')=&-\ipc{f}{\udl{A}'}_{\ol{N}}=\ipc{-\diff^*\diff \chi \udl{A}}{\udl{A}'}_{\ol{N}}\nonumber\\
        =&\ipc{\nml_{\partial\ol{N}}\diff\chi\udl{A}}{\udl{A}'}_{\partial\ol{N}}-\ipc{\chi\udl{A}}{\nml_{\partial\ol{N}}\diff \udl{A}'}_{\partial\ol{N}}\nonumber\\
        =&\ipc{\nml_{\ol{\Sigma}^+}\diff\udl{A}}{\udl{A}'}_{\ol{\Sigma}^+}-\ipc{\udl{A}}{\nml_{\ol{\Sigma}^+}\diff \udl{A}'}_{\ol{\Sigma}^+}=\sigma(\udl{A},\udl{A}'),
    \end{align}
    where we have used Stokes' theorem, $-\diff^*\diff\udl{A}'=0$, and the behaviour of $\chi$ in neighbourhoods of $\supp(\udl{A})\cap\partial\ol{N}$. Since $\sigma$ is non-degenerate on $\Sol_{\mathscr{G}}(\ol{N})$, it follows that $[G^{\PJ}f]=[\udl{A}]$.
\end{proof}

Having establish these results, we can complete the following proof.
\begin{proof}[Proof of Prop.~\ref{prop:Floc_sympl_smear}]
Fix a background configuration $A'\in \Sol^J(\ol{N})$, then for each $[A'']\in \Sol_\mathscr{G}^J(\ol{N})$ and $f\in \Omega^1_{0\diff^*}(\ol{N})$ one has
\begin{equation}
\label{eq:sttosymp}
    A(f)([A''])=A(f)([A''-A'])+A(f)([A'])=-\sigma(G^{\PJ}(f),A''-A')+A(f)([A'])\,,
\end{equation}
where we applied Lemma~\ref{lem:localvssymp} to $A(f)([A''-A'])$. Using the definition of $O_{A'}$, we can then write
\begin{equation}
    A(f)=-O_{A'}([G^{\PJ}f])+A(f)([A'])\one,
\end{equation}
from which it follows that $O_{A'}([G^{\PJ}f])\in \Floc$ and that each functional in $\Floc$ can be expressed as a polynomial in symplectically smeared fields. By Lem.~\ref{lem:green_sc}, all local symplectically smeared observables are of the form $O_{A}([G^{\PJ}f])$, thus $\SFloc\subset \Floc$.

b) We calculate 
\begin{equation}
    \{O_{A}([G^{\PJ}f]),O_{A'}([G^{\PJ}h])\}=\{A(f),A(h)\}=-\ipc{f}{G^{\PJ}h}=\sigma(G^{\PJ}f,G^{\PJ}h)\one,
\end{equation}
where we have once again used Lem.~\ref{lem:localvssymp}.

c,d) For $\psi=O_{A}([\udl{A}])$, one finds that
\begin{equation}
    X_{\psi}([\tilde{A}])=[\udl{A}].
\end{equation}
By Prop.~\ref{prop:Floc_sympl_smear}, any $\psi\in \Floc^{}(\ol{N})$ can be written as a polynomial in symplectically smeared fields, and hence $X_\psi$ is fixed by linearity and the Leibniz rule. The resulting expression for the Poisson bracket can be read off from Eq.~\eqref{eq:symp_pois_bra}.
\end{proof}

\subsubsection{Background shift transformations}
\label{sec:field_shift}
We isolate an important subgroup of diffeomorphisms of the infinite dimensional affine manifold  $\Sol^J_{\mathscr{G}}(\ol{N})$, which will be called background shifts. 
\begin{definition}
    Let $\ol{N}$ be a finite Cauchy lens and $J\in \Omega^1_{\diff^*}(\ol{N})$. Each $[\udl{A}]\in \Sol_{\mathscr{G}}(\ol{N})$ defines a background shift   $\xi([\udl{A}]):\Sol^J_{\mathscr{G}}(\ol{N})\to \Sol^J_{\mathscr{G}}(\ol{N})$ by
    \begin{equation}
        \xi([\udl{A}])([A])=[A+\udl{A}].
    \end{equation}
If $[\udl{A}]\in \mathscr{LG}(\ol{N})$ then the background shift $\xi([\udl{A}])$ is a large gauge transformation on the phase space.
\end{definition}
Background shifts act on observables by pullback, and this action
preserves both the set of local observables and the set of symplectically smeared fields.
\begin{proposition}
    Let $\ol{N}$ be a finite Cauchy lens, $J\in \Omega^1_{\diff^*}(\ol{N})$ and $[\udl{A}']\in \Sol_{\mathscr{G}}(\ol{N})$. If $\psi\in \Floc(\ol{N})$, the background shifted observable $\xi([\udl{A}'])^*\psi=\psi\circ \xi([\udl{A}'])$ is local and $\xi([\udl{A}'])^*:\Floc(\ol{N})\to \Floc(\ol{N})$ is an isomorphism. Furthermore, for each $A\in \Sol^J(\ol{N})$ and $[\udl{A}]\in \Sol_{\sc,\mathscr{G}}(\ol{N})$, we have
    \begin{equation}
\label{eq:symp_fieldshift}\xi([\udl{A}'])^*O_A([\udl{A}])=O_{A-A'}([\udl{A}])=O_{A}([\udl{A}])+\sigma(\udl{A},\udl{A}')\one.
    \end{equation}
\end{proposition}
\begin{proof}
    It is enough to verify Eq.~\eqref{eq:symp_fieldshift}, because
    pullback maps define algebra homomorphisms on the space of functionals and $\xi([\udl{A}'])^*$ has an inverse $\xi([-\udl{A}'])^*$. Let $[A']\in \Sol^J_{\mathscr{G}}(\ol{N})$, then
    \begin{equation}
        (\xi([\udl{A}'])^*O_A([\udl{A}]))([A'])=O_A([\udl{A}])([A'+\udl{A}'])=\sigma(\udl{A},A'+\udl{A}'-A)=O_{A-\udl{A}'}([\udl{A}])([A']),
    \end{equation}
    Hence $(\xi([\udl{A}'])^*O_A([\udl{A}]))=O_{A-\udl{A}'}([\udl{A}])\in \SFloc(\ol{N})$. Furthermore
    \begin{equation}
        O_{A-\udl{A}'}([\udl{A}])([A'])=\sigma(\udl{A},A'-(A-\udl{A}'))=\sigma(\udl{A},A'-A)+\sigma(\udl{A},\udl{A}')=O_{A'}([\udl{A}])([A'])+\sigma(\udl{A},\udl{A}'),
    \end{equation}
    and thus $O_{A-\udl{A}'}([\udl{A}])=O_{A'}([\udl{A}])+\sigma(\udl{A},\udl{A}')\one$.
\end{proof}
\subsubsection{Large gauge invariance of local observables and Poisson degeneracies}
\label{sec:Pois_deg}

Let $\Lambda\in\Omega^0(\ol{N})$ and $\udl{A}'\in \Sol_{\sc}(\ol{N})$. Then one may compute
\begin{equation}\label{eq:sigmadLambdaAp}
        \sigma(\diff\Lambda,\udl{A}')=\ipc{\diff\Lambda\restriction_{\ol{\Sigma}}}{-\nml_{\ol{\Sigma}}\diff \udl{A}'}_{\ol{\Sigma}}=-\ipc{\Lambda\restriction_{\partial{\Sigma}}}{\nml^{(2)}_{\angle\ol{N}}\diff\udl{A}'}_{\partial\ol{\Sigma}}=0,
    \end{equation}
using the fact that $\diff^*\diff\udl{A}'=0$ and that $\udl{A}'$ vanishes near the $\partial\ol{\Sigma}$. This calculation has two consequences. First, local observables are invariant under large gauge transformations.
\begin{proposition}
\label{prop:c_loc_inv}
    Let $\psi\in \Floc$, then for each $[\diff\Lambda]\in \mathscr{LG}(\ol{N})$, we have 
    \begin{equation}
        \xi([\diff\Lambda])^*\psi=\psi.
    \end{equation}
\end{proposition} 
\begin{proof}
    By Prop.~\ref{prop:Floc_sympl_smear}, it is enough to establish the case where $\psi=O_{A'}([\udl{A}'])$ for
    any $A'\in\Sol^J(\ol{N})$ and $[\udl{A}']\in  \Sol_{\sc,\mathscr{G}}(\ol{N})$. Let $[A]\in \Sol^J_{\mathscr{G}}(\ol{N})$, then
    \begin{equation}
        ((\xi([\diff\Lambda])^*-\id)O_{A'}([\udl{A}']))([A]) = \sigma(\udl{A}',\diff\Lambda) = 0
    \end{equation}
    by~\eqref{eq:sigmadLambdaAp},
    choosing representatives $\udl{A}'\in \Sol_\sc(\ol{N})$ and $\Lambda\in\Omega^0(\ol{N})$.
\end{proof}
Accordingly, the algebra of observables must be extended in order to consider observables sensitive to the `large gauge degrees of freedom' (or edge modes). This is done in the next subsection.

Second, the calculation~\eqref{eq:sigmadLambdaAp} can be used to provide examples of Poisson-degenerate observables, i.e.~non vanishing functionals $\varphi\in \Floc^{}(\ol{N})$ such that $\{\varphi,\psi\}=0$ for all $\psi\in \Floc^{}(\ol{N})$.
\begin{proposition}
\label{prop:loc_poiss_degen}
    Let $[\udl{A}]\in \mathscr{LG}(\ol{N})\cap \Sol_{\sc,\mathscr{G}}(\ol{N})$, then for any background $A\in \Sol^J(\ol{N})$, the observable $O_{A}([\udl{A}])\in \Floc^{}(\ol{N})$ is Poisson degenerate.
\end{proposition}
\begin{proof}
Choose $\Lambda\in\Omega_\sc^0(\ol{N})$ such that $\diff\Lambda\in [\udl{A}]$. Then for 
arbitrary $A'\in \Sol^J(\ol{N})$ and $[\udl{A}']\in \Sol_{\sc,\mathscr{G}}(\ol{N})$,~\eqref{eq:sigmadLambdaAp} gives
\begin{equation}
        \{O_{A}([\udl{A}]),O_{A'}([\udl{A}'])\}=\sigma(\diff\Lambda,\udl{A}') \one=0.
    \end{equation}
    By Prop.~\ref{prop:Floc_sympl_smear} we have $\{O_{A}([\udl{A}]),\psi\}=0$ for every $\psi\in \Floc$, i.e., $O_{A}([\udl{A}])$ is Poisson degenerate.
\end{proof}

We conclude this subsection by recalling a prototypical example of a local Poisson degenerate observable (see e.g.\ \cite[Sec.~3.3]{sandersElectromagnetismLocalCovariance2014}). 

\begin{example}
\label{Ex:Pois_deg}
    Consider a finite Cauchy lens $\ol{N}$ with Cauchy surface with boundaries $\ol{\Sigma}\cong S^2\times [0,1]$, 
so the corner $\angle\ol{N}=C_1\sqcup C_2$ has components $C_i\cong S^2$. Set $J=0$ for convenience and choose a function $\Lambda\in \Omega^0(\ol{N})$ with $\Lambda\restriction_{U_1}=0$ and $\Lambda\restriction_{U_2}=1$ for some open neighbourhoods $U_i\subset\ol{N}$ of $C_i$ as illustrated in Fig.~\ref{fig:degen_obs}. Then $[\diff\Lambda]\in \mathscr{LG}(\ol{N})\cap \Sol_{\sc,\mathscr{G}}(\ol{N})$ and for any configuration $A\in \Sol_\mathscr{G}(\ol{N})$ we can evaluate (using the first part of~\eqref{eq:sigmadLambdaAp})
\begin{equation}
\label{eq:C2-flux}
    O_0([\diff\Lambda])([A])=\ipc{1}{\nml_{C_2}\mathbf{E}}_{C_2},
\end{equation}
where $\mathbf{E}=-\nml_{\ol{\Sigma}}\diff A$. By Prop.~\ref{prop:loc_poiss_degen}, this observable is Poisson degenerate. The observable $O_0([\diff\Lambda])$ measures the total outwards electric flux through $C_2$, which due to the Gauss constraint $\ipc{1}{\nml_{C_2}\mathbf{E}}_{C_2}+\ipc{1}{\nml_{C_1}\mathbf{E}}_{C_1}=0$ is equivalent to the total total inwards electric flux through $C_1$. Similarly, this electric flux can be re-expressed in terms of the total electric flux through any compact (co-dimension 1) surface $C\subset \Intr(\ol{\Sigma})$ for which $\ipc{1}{\nml_{C_2}\mathbf{E}}_{C_2}=\ipc{1}{\nml_{C}\mathbf{E}}_{C}$, as illustrated by Fig.~\ref{fig:degen_obs}.\footnote{$\ol{\Sigma}\setminus C$ must have a connected component containing $C_2$, whose closure is a manifold with boundary $C\sqcup C_2$.}
\end{example} 

\begin{figure}
    \centering
    \begin{tikzpicture}
\node[cylinder, draw, shape aspect=3, cylinder uses custom fill,  minimum height=10cm,  minimum width=2cm, cylinder body fill=gray!25, opacity=0.5,left color=white,
  right color=blue!70!black,
  cylinder end fill=gray!50] (c) {};
\draw[bend right, thick, red]
 (c.north) to (c.south)  node [at start, above, red] {$C$};
 \draw[->,thick, red] (-.2,-.5)--(.2,-.5) node [right] {$\nml_C$};

\draw[black]  (c.bottom) node [left] {$C_1$} (c.top) node [right] {$C_2$};

\draw (c.top)+(2,0) node[rectangle, draw, opacity=0.5,bottom color=white,
  top color=blue!70!black, minimum height=2cm,  minimum width=.1cm] (r) {};
\draw[black] (r.south) node [below] {$0$} (r.north) node [above] {$1$} (r.east) node [right] {$\Lambda$};
 \end{tikzpicture}
    \caption{A colour plot on the Cauchy surface $\ol{\Sigma}$ (with one dimension suppressed) for a choice of $\Lambda$ yielding the Poisson degenerate observable $O_0([\diff\Lambda]):A\mapsto \ipc{1}{\nml_{C}\mathbf{E}}$}
    \label{fig:degen_obs}
\end{figure}

\subsection{Semi-local classical observables on a finite Cauchy lens}
\label{sec:class_semi-loc}
While a restriction to local observables is in line with the usual philosophy of local (quantum) physics (see e.g.~\cite{haagLocalQuantumPhysics1996}), 
we have seen that they are insensitive to large gauge transformations. Therefore we need to extend $\Floc$ to include \emph{semi-local observables}, thus making contact with notions of boundary observables that play a role in the literature on large gauge transformations (see e.g.~\cite{donnellyLocalSubsystemsGauge2016,stromingerLecturesInfraredStructure2018,herdegenAsymptoticStructureElectrodynamics2017,gomesQuasilocalDegreesFreedom2021,rielloHamiltonianGaugeTheory2024}) and quantum reference frames (e.g.~\cite{carrozzaEdgeModesReference2022,kabelQuantumReferenceFrames2023}). 
\subsubsection{Semi-local symplectically smeared fields}
\label{sec:semlocsymp}
We introduce semi-local observables directly using symplectically smeared fields, analogous to the description of $\Floc$ as given in 
Prop.~\ref{prop:Floc_sympl_smear}.

\begin{definition}
\label{def:class-sloc-obs}
    Let $\ol{N}$ be a finite Cauchy lens and $J\in\Omega^1_{\diff^*}(\ol{N})$. We denote by $\Fsloc^{}(\ol{N})$ the unital real algebra of \emph{semi-local observables} on $\Sol_{\mathscr{G}}^J(\ol{N})$ generated by 
    \begin{equation}
        \SFsloc^{}(\ol{N})=\{O_{A}([\udl{A}]):(A,\udl{A})\in T\Sol^J(\ol{N})\},
    \end{equation}
    where  $O_{A}([\udl{A}]):\Sol_{\mathscr{G}}^J(\ol{N})\to \RR$ is given by
    \begin{equation}
        O_{A}([\udl{A}])([\tilde{A}])=\sigma(\udl{A},\tilde{A}-A).
    \end{equation}
    This algebra is given a Poisson bracket specified by
    \begin{equation}
        \{O_{A}([\udl{A}]),O_{A'}(\udl{A}')\}=\sigma(\udl{A},\udl{A}') \one.
    \end{equation}
\end{definition}

Here we have once again suppressed the dependence of $\Fsloc^{}(\ol{N})$ on the background current $J$ in the notation. One clearly has $\Floc^{}(\ol{N})\subset \Fsloc^{}(\ol{N})$.

Just as for $\Floc^{}(\ol{N})$, the Poisson structure on $\Fsloc^{}(\ol{N})$ corresponds to the symplectic form $\mathbf{\Omega}$ on $\Sol_{\mathscr{G}}^J(\ol{N})$: specifically, one has
\begin{equation}
    \{\varphi,\psi\}([\tilde{A}])=\mathbf{\Omega}_{[\tilde{A}]}(X_\varphi([\tilde{A}])\otimes X_\psi([\tilde{A}])),
\end{equation}
for arbitrary $\varphi,\psi\in \Fsloc^{}(\ol{N})$ and $[\tilde{A}]\in \Sol_{\mathscr{G}}^J(\ol{N})$, where $X_\varphi,X_\psi$ are the Hamiltonian vector fields of $\varphi,\psi$, defined so that $X_{\psi}([\tilde{A}])=[\udl{A}]$ in the case $\psi=O_{A}([\udl{A}])$. 

Furthermore, the action of background shifts on local observables discussed in Sec.~\ref{sec:field_shift} extends straightforwardly to semi-local observables. For each $[\udl{A}']\in \Sol_{\mathscr{G}}(\ol{N})$ the pullback of \begin{equation}
    \xi([\udl{A}']):\Sol^J_{\mathscr{G}}(\ol{N})\to \Sol^J_{\mathscr{G}}(\ol{N}),\qquad [A]\mapsto [A]+[\udl{A}],
\end{equation} defines an automorphism $\xi([\udl{A}'])^*:\Fsloc(\ol{N})\to \Fsloc(\ol{N})$, for which
\begin{equation}
\label{eq:sym_field_shift}
    \xi([\udl{A}'])^*O_{A}([\udl{A}])=O_{A-\udl{A'}}([\udl{A}])=O_{A}([\udl{A}])+\sigma(\udl{A},\udl{A'})\one.
\end{equation}

Some remarks are in order concerning the relation between semi-local observables and the notion of \emph{Wilson lines}. For simplicity, set $J=0$, let $\ol{\Sigma}\subset \ol{N}$ be a regular Cauchy surface with boundaries and $F\in \Omega^1_{\diff^*}(\ol{\Sigma})$. Denote $[A]=(\mathfrak{I}_{\ol{\Sigma}})^{-1}(0\oplus -F)$. Then by Eq.~\eqref{eq:symp_struct} and Prop.~\ref{prop:ext_ini_dat}, we have
\begin{equation}
    O_0([A])=\ipc{F}{\cdot}_{\ol{\Sigma}},
\end{equation}
as a semi-local observable on $\Sol_{\mathscr{G}}(\ol{N})$.
The map 
\begin{equation}
    \Omega^1(\ol{\Sigma})\ni \mathbf{A}\mapsto \ipc{F}{\mathbf{A}}_{\ol{\Sigma}},
\end{equation}
can be viewed as a smeared analogue of a spacelike Wilson line observable  
\begin{equation}
    \Omega^1(\ol{\Sigma})\ni \mathbf{A}\mapsto \int_0^1 \mathbf{A}_i(\gamma(t)) \dot{\gamma}^i\diff t.
\end{equation}
Here $\gamma:[0,1]\to \ol{\Sigma}$ is either a closed smooth curve for which we say that the observable above is a \emph{Wilson loop}, or $\gamma:[0,1]\to \ol{\Sigma}$ connects to the boundary, i.e.~$\gamma(0),\gamma(1)\in \partial\ol{\Sigma}$. Such observables, in particular the loop variant, were introduced in \cite{wilsonConfinementQuarks1974} 
in the context of lattice QFT. As a one-form distribution, the map 
\begin{equation}
    \psi_\gamma:\Omega^1_0(\Sigma)\to \RR,\qquad \psi_\gamma(u)=\int_0^1 u_i(\gamma(t)) \dot{\gamma}^i\diff t
\end{equation}
for $\gamma$ as above, is weakly co-closed, in the sense that for $h\in \Omega^0_0(\Sigma)$, one has $(\diff^*\psi_\gamma)(h):=\psi_\gamma(\diff h)=h(\gamma(1))-h(\gamma(0))=0$. Thus Wilson line observables are invariant under $\mathscr{G}(\ol{N})$. However, these observables are too singular to be evaluated on distributional configurations (such as considered in \cite{sandersElectromagnetismLocalCovariance2014}), which suggests that their quantisation may be ill-behaved. For this reason, it is preferable to work with their smeared analogues $O_0([A])$; nonetheless, the Wilson line picture can provide useful intuition.

\subsubsection{Edge mode observables}
\label{sec:edge_mode_obs}
Unlike the local observables, some semi-local observables transform non-trivially under large gauge transformations. Let us refer to such observables as follows.
\begin{definition}
    For $\ol{N}$ a finite Cauchy lens and $J\in \Omega^1_{\diff^*}(\ol{N})$, we say that $\psi\in \Fsloc^{}(\ol{N})$ is an \emph{edge-mode observable} if it is not invariant under large gauge transformations, i.e., there exists a $[\diff\Lambda]\in \mathscr{LG}(\ol{N})$ such that
    \begin{equation}
        \xi([\diff\Lambda])^*\psi\neq \psi.
    \end{equation}
\end{definition}

Using the ($\ol{\Sigma}$ dependent) decomposition $\mathfrak{K}_{\ol{\Sigma}}^J:\Sol_\mathscr{G}^J(\ol{N})\to V^C(\ol{\Sigma})\times V^S(\ol{N})$ as in Prop.~\ref{prop:sympl_decomp}, we can define two subalgebras of $\Fsloc^{}(\ol{N})$. In particular, in defining these subalgebras we will single out particular edge mode observables.
\begin{definition}
    \label{def:class_obs_decomp}
    Let $\ol{\Sigma}\subset \ol{N}$ be a regular Cauchy surface with boundaries of a finite Cauchy lens, and let $J\in \Omega^1_{\diff^*}(\ol{N})$. Let $\iota^{C/S}:V^{C/S}(\ol{\Sigma}/\ol{N})\to V^{C}(\ol{\Sigma})\times V^S(\ol{N}) $ be the injections defined by 
    $\iota^C(F\oplus H)=(F\oplus H;0\oplus 0)$, $\iota^S(f\oplus h)=(0\oplus 0;f\oplus h)$ and define
    \begin{equation}
        o^{C/S} = O_A\circ \left(\mathfrak{K}_{\ol{\Sigma}}\right)^{-1} \circ \iota^{C/S}: V^{C/S}(\ol{\Sigma}/\ol{N})\to \Fsloc^{}(\ol{N}),
    \end{equation}
    where $A\in \Sol^J(\ol{N})$ obeys $\mathfrak{K}_{\ol{\Sigma}}^J([A])=0$.\footnote{Note that the map $O_A:\Sol_{\mathscr{G}}(\ol{N})\to \Fsloc(\ol{N})$ is uniquely fixed by the gauge orbit $[A]\in \Sol_{\mathscr{G}}^J(\ol{N})$. This can be seen from Eq.~\eqref{eq:symp_shift}.}
    Then $\Fsloc^{C/S}(\ol{\Sigma})$ denotes the unital subalgebra of $\Fsloc^{}(\ol{N})$ generated by all observables of the form $o^{C/S}(\udl{v})$ for $\udl{v}\in V^{C/S}(\ol{\Sigma}/\ol{N})$.
\end{definition}
From this definition, we obtain the following identity.
    \begin{lemma}
    \label{lem:symp_to_CS}
        Let $\ol{\Sigma}\subset \ol{N}$ and $J\in \Omega^1_{\diff^*}(\ol{N})$ be as in Def.~\ref{def:class_obs_decomp}. Let $(A',\udl{A})\in T\Sol^J(\ol{N})$, then
        \begin{equation}
        \label{eq:symp_to_CS}
            O_{A'}([\udl{A}])=o^C(F\oplus H)+o^S(f\oplus h)-(o^C(F\oplus H)([A'])+o^S(f\oplus h)([A']))\one,
        \end{equation}
        with 
        \begin{equation}
        \label{eq:o_sigma}
            o^C(F\oplus H)([A'])=\sigma^C(F\oplus H, F'\oplus H'),\qquad o^S(f\oplus h)([A'])=\sigma^S(f\oplus h, f'\oplus h'),
        \end{equation}
        where $(F\oplus H\oplus f\oplus h)=\mathfrak{K}_{\ol{\Sigma}}([\udl{A}])$ and $(F'\oplus H'\oplus f'\oplus h')=\mathfrak{K}_{\ol{\Sigma}}^J([A'])$.
    \end{lemma}
    \begin{proof}
    Note that
    \begin{align}
        \sigma^C(F\oplus H) + \sigma^S(f\oplus h) &= 
        O_A((\mathfrak{K}_{\ol{\Sigma}})^{-1}(F\oplus H\oplus 0\oplus 0) +O_A((\mathfrak{K}_{\ol{\Sigma}})^{-1}(0\oplus 0\oplus f\oplus h) \nonumber\\ &=
        O_A((\mathfrak{K}_{\ol{\Sigma}})^{-1}(F\oplus H\oplus f\oplus h)),
    \end{align}
    which together with Eq.~\eqref{eq:symp_shift} implies Eq.~\eqref{eq:symp_to_CS}.
    Furthermore
    \begin{align}
        O_A((\mathfrak{K}_{\ol{\Sigma}})^{-1}(F\oplus H\oplus f\oplus h))&((\mathfrak{K}_{\ol{\Sigma}}^J)^{-1}(F'\oplus H'\oplus f'\oplus h'))\nonumber\\
        =&\sigma((\mathfrak{K}_{\ol{\Sigma}})^{-1}(F\oplus H\oplus f\oplus h),(\mathfrak{K}_{\ol{\Sigma}}^J)^{-1}(F'\oplus H'\oplus f'\oplus h')-A)\nonumber\\
        =&\sigma((\mathfrak{K}_{\ol{\Sigma}})^{-1}(F\oplus H\oplus f\oplus h),(\mathfrak{K}_{\ol{\Sigma}})^{-1}(F'\oplus H'\oplus f'\oplus h'))\nonumber\\=&\sigma^{C\oplus S}(F\oplus H\oplus f\oplus h,F'\oplus H'\oplus f'\oplus h'),
    \end{align}
    from which one can read off Eq.~\eqref{eq:o_sigma}. 
    \end{proof}
One now confirms the existence of edge mode observables, recalling that
$\mathscr{LG}_\angle(\ol{N})$ is non-trivial.
\begin{proposition}
\label{prop:edge_existence}
    Let $\ol{N}$ be a finite Cauchy lens and $J\in \Omega^1_{\diff^*}(\ol{N})$. Every observable 
    $o^S(0\oplus f)\in \Fsloc(\ol{N})$ with $f\in \mathscr{LG}_\angle(\ol{N})\setminus\{0\}$ is an edge mode observable.
\end{proposition}
\begin{proof}
    By Lem.~\ref{lem:symp_to_CS}, and recalling that $\mathfrak{K}_{\ol{\Sigma}}(\mathfrak{G}(f))=0\oplus 0\oplus f\oplus 0$, we compute
    \begin{equation}
        o^S(0\oplus f)([A]+\mathfrak{G}(f))-o^S(0\oplus f)([A])=\sigma^S(0\oplus f;f\oplus 0)=-\ipc{f}{f}_{\partial\ol{\Sigma}}<0.
    \end{equation}
    Hence $\xi(\mathfrak{G}(f))^*o^S(0\oplus f)\neq o^S(0\oplus f)$ and thus $o^S(0\oplus f)\in \Fsloc(\ol{N})$ is an edge mode observable.    
\end{proof}

The decomposition of the symplectic manifold $\Sol^J_{\mathscr{G}}(\ol{N})$ given in Prop.~\ref{prop:sympl_decomp} is inherited by the observables via a decomposition of Poisson algebras (see also \cite{gomesQuasilocalDegreesFreedom2021,rielloEdgeModesEdge2021,rielloHamiltonianGaugeTheory2024}). In particular, the subalgebras introduced in Def.~\ref{def:class_obs_decomp} are closed under the Poisson bracket, mutually Poisson commuting and generate the full algebra of semi-local observables.
\begin{theorem}
\label{thm:class_obs_decomp}
    Let $\ol{\Sigma}\subset \ol{N}$ be a regular Cauchy surface with boundaries of a finite Cauchy lens and $J\in \Omega^1_{\diff^*}(\ol{N})$. Then $\Fsloc^{}(\ol{N})$ is generated by $\Fsloc^{C}(\ol{\Sigma})\cup\Fsloc^{S}(\ol{\Sigma})$, and  
    \begin{equation}
        \Fsloc^{C}(\ol{\Sigma})\cap \Fsloc^{S}(\ol{\Sigma})=\CC\cdot \one,
    \end{equation} 
    with $1$ the unit of $\Fsloc^{}(\ol{N})$.
    Moreover, $\Fsloc^C(\ol{\Sigma})$ and $\Fsloc^S(\ol{\Sigma})$ are Poisson subalgebras
    obeying
    \begin{equation}
        \{\Fsloc^C(\ol{\Sigma}),\Fsloc^S(\ol{\Sigma})\}=\{0\}
    \end{equation}  
    and 
    \begin{equation}
        \{o^{C/S}(\udl{v}_1^{C/S}),o^{C/S}(\udl{v}_2^{C/S})\}=\sigma^{C/S}_{\ol{\Sigma}}(\udl{v}_1^{C/S},\udl{v}_2^{C/S})\cdot \one
    \end{equation}
    for each $\udl{v}_1^{C/S},\udl{v}_s^{C/S}\in V^{C/S}(\ol{\Sigma}/\ol{N})$.  
\end{theorem}
\begin{proof}
    The first statement is immediate from~\eqref{eq:symp_to_CS}. Now suppose that $\psi\in \Fsloc^{C}(\ol{\Sigma})\cap \Fsloc^{S}(\ol{\Sigma})$, and for any $[\tilde{A}]\in\Sol^J_{\mathscr{G}}(\ol{N})$ define 
    \begin{equation}
        [\tilde{A}]^{C/S} = (\mathfrak{K}^J_{\ol{\Sigma}})^{-1}(p^{C/S}\mathfrak{K}^J_{\ol{\Sigma}}([\tilde{A}])), 
    \end{equation}
    where $p^{C/S}:V^C(\ol{\Sigma)})\times V^S(\ol{N})
    \to V^{C/S}(\ol{\Sigma}/\ol{N})$
    are the canonical projections. As $\psi\in \Fsloc^{C}(\ol{\Sigma})$ one has $\psi([\tilde{A}])=\psi([\tilde{A}]^C)$ and as $\psi\in \Fsloc^{S}(\ol{\Sigma})$ one has $\psi([\tilde{A}])=\psi([\tilde{A}]^S)$. As a result $\psi([\tilde{A}])=\psi(([\tilde{A}]^C)^S)=\psi((\mathfrak{K}^J_{\ol{\Sigma}})^{-1}(0,0,0,0))$ for all $[\tilde{A}]$ and hence $\psi\in\CC\one$. The Poisson bracket relations follow from a direct computation.
\end{proof}

The observable $o^C(F\oplus H)\in \Fsloc^{C}(\ol{\Sigma})$ for $F,H\in \Omega^1_{\tang\diff^*}(\ol{\Sigma})$, evaluated on $[\tilde{A}]\in \Sol_\mathscr{G}^J(\ol{N})$ with $\mathbf{A}=[\tilde{A}]\restriction_{\ol{\Sigma}}$ and $\mathbf{E}=-\nml_{\ol{\Sigma}}\diff [\tilde{A}]$, yields
\begin{equation}
    o^C(F\oplus H)([\tilde{A}])=\ipc{F}{\mathbf{E}}_{\ol{\Sigma}}-\ipc{H}{\mathbf{A}}_{\ol{\Sigma}}.
\end{equation}
Following the analogy with Wilson lines as described earlier in Sec.~\ref{sec:class_semi-loc}, one can thus view the algebra  $\Fsloc^{C}(\ol{\Sigma})$ as generated by smeared analogues of closed loop observables
\begin{equation}
    c_1\oint_{\gamma_1} \mathbf{A}_i\diff\gamma^i_1+c_2 \oint_{\gamma_2} \mathbf{E}_i\diff\gamma^i_2,
\end{equation}
for smooth closed curves $\gamma_{1,2}:S^1\to\ol{\Sigma}$ and $c_1,c_2\in \RR$. In particular, both the loop integral observables and $o^C(F\oplus H)$ vanish when evaluated on any $[\tilde{A}]\in \Sol_\mathscr{G}^J(\ol{N})$ for which $\mathbf{A},\mathbf{E}$ as defined above are exact forms (i.e.~$[\tilde{A}]$ is a surface configuration at $\ol{\Sigma}$). As such, we will refer to $\Fsloc^{C}(\ol{\Sigma})$ as \emph{closed loop observables}. Note that this means that $\Fsloc^{C}(\ol{\Sigma})$ are invariant under large gauge transformations, meaning that they do not contain any edge mode observables. We note here that although the closed loop observables of the vector potential (Wilson loops) are widely regarded as natural objects to study, whereas closed loop observables of the electric field arise as time-derivatives of such spacelike Wilson loops.

The observables $o^S(f\oplus h)\in \Fsloc^{S}(\ol{\Sigma})$ for $f,h\in \nml_{\partial\ol{\Sigma}}\Omega^1_{\diff^*}(\ol{\Sigma})$ evaluate on $[\tilde{A}]\in \Sol^J_G(\ol{N})$ as
\begin{equation}
    o^S(f\oplus h)([\tilde{A}])=\ipc{f}{\nml_{\partial\ol{\Sigma}}\mathbf{E}}_{\partial\ol{\Sigma}}-\ipc{\diff \varphi}{\mathbf{A}}_{\ol{\Sigma}},
\end{equation}
where $\varphi\in \Omega^0(\ol{\Sigma})$ such that $\Delta \varphi=0$ and $\nml_{\partial\ol{\Sigma}}\diff \varphi=h$. Note that, as a consequence of the Hodge decompostion, we have
\begin{equation}
    \ipc{\diff \varphi}{\diff \varphi}_{\ol{\Sigma}}\leq \ipc{H}{H}_{\ol{\Sigma}}
\end{equation}
for any $H\in \Omega^1_{\diff^*}(\ol{\Sigma})$ with $\nml_{\partial\ol{\Sigma}}H=h$. In fact $H=\diff \varphi$ is the unique one-form in this class minimizing $\ipc{H}{H}_{\ol{\Sigma}}$. We can view $[\tilde{A}]\mapsto \ipc{H}{\mathbf{A}}_{\ol{\Sigma}}$ as a smeared analogue of a Wilson line associated to a curve $\gamma$ connecting two points on the boundary $\partial\ol{\Sigma}$. In particular, for such a curve $\gamma$ the map \begin{equation}
    [\udl{A}]\mapsto \int \mathbf{A}_i\diff\gamma^i,
\end{equation} can always be approximated by a smearing with co-closed form $H$ such that
\begin{equation}
    \ipc{H}{H}_{\ol{\Sigma}}\propto \int_0^1 g_{ij}(\gamma(t))\dot{\gamma}^i(t)\dot\gamma^j(t)\diff t,
\end{equation}
i.e. $\ipc{H}{H}_{\ol{\Sigma}}$ is proportional to the length of the curve $\gamma$ (with the proportionality constant diverging in the limit $\ipc{H}{\mathbf{A}}_{\ol{\Sigma}}\to\int \mathbf{A}_i\diff\gamma^i$). As such, we can loosely view the observable $[\tilde{A}]\mapsto \ipc{\diff\varphi}{\mathbf{A}}_{\ol{\Sigma}}$ as a smeared analogue of Wilson line observable along a path of minimal length, i.e.~a geodesic.\footnote{See  \cite{alexanderGeodesicsRiemannianManifoldswithBoundary1981} for a discussion of geodesics on Riemannian manifolds with boundary.} Such geodesic Wilson line observables are considered for instance in the context of holography, see e.g.~\cite{ammonWilsonLinesEntanglement2013}. 

The observable $[\tilde{A}]\mapsto \ipc{f}{\nml_{\partial\ol{\Sigma}}\mathbf{E}}_{\partial\ol{\Sigma}}$ is a smearing of the electric flux density through the surface $\partial\ol{\Sigma}$. In the context of large gauge transformations, these observables are sometimes referred to as surface charges, see e.g.~\cite{donnellyLocalSubsystemsGauge2016}. As such, we will refer to $\Fsloc^{S}(\ol{\Sigma})$, i.e.~the algebra of observables generated by surface charges $[\tilde{A}]\mapsto \ipc{f}{\nml_{\partial\ol{\Sigma}}\mathbf{E}}_{\partial\ol{\Sigma}}$ as well as the edge mode observable $[\tilde{A}]\mapsto \ipc{\diff\varphi}{\mathbf{A}}_{\ol{\Sigma}}$, as the \emph{surface observables}.

It should be pointed out that, even though one may associate $\Fsloc^{S}(\ol{\Sigma})$ with the boundary $\partial\ol{\Sigma}$, in general the algebra $\Fsloc^{S}(\ol{\Sigma})\cap\Floc^{}(\ol{N})$ is non-trivial. In particular, all observables of the form described in Prop.~\ref{prop:loc_poiss_degen} are contained in this subalgebra. Moreover, it can be seen that in general $\{\Floc^{}(\ol{N}),\Fsloc^{S}(\ol{\Sigma})\}\neq \{0\}$. This is interpreted as meaning that the observables in $\Fsloc^{S}(\ol{\Sigma})$ are generally not localisable on $\partial\ol{\Sigma}$. 

Similarly to the decomposition of Prop.~\ref{prop:sympl_decomp}, the decomposition of observables made above depends on the choice of Cauchy surface $\ol{\Sigma}$ and thus the notion of closed loop and surface observables are not locally covariant. Nevertheless, the decomposition of semi-local observables into surface and closed loop observables will prove convenient in the discussion of quantisation of the semi-local observables introduced above, particularly in the construction of Fock space representations of the algebra of observables for the resulting quantum field theory, see Sec.~\ref{sec:reps}.

\section{Semi-local observables in quantum electromagnetism}
\label{sec:quant_alg}

A quantum field theory (QFT) can be obtained from a (well-behaved) classical field theory through various quantization schemes. A fairly general scheme applicable to (perturbative) algebraic quantum field theory is deformation quantization (see e.g.~\cite{fedosovSimpleGeometricalConstruction1994,duetschPerturbativeAlgebraicField2001,hawkinsStarProductInteracting2020}). However, for sufficiently simple field theories (for instance certain affine theories, see e.g.~\cite{fewsterLocallyCovariantQuantum2015}), one can construct an algebra of quantized observables more directly.

Our goal is to construct a (non-commutative) *-algebra of semi-local observables for electromagnetism, which is representable on a Hilbert space. Concretely, for a finite Cauchy lens $\ol{N}$ and background current $J\in \Omega^1_{\diff^*}(\ol{N})$, we wish to construct a unital algebra containing self-adjoint observables $\mathbf{O}_A([\udl{A}])$ for $(A,\udl{A})\in T\Sol^J(\ol{N})$ (interpreted as quantized symplectically smeared fields), satisfying the canonical commutation relations
\begin{equation}
\label{eq:CCR}
    [\mathbf{O}_A([\udl{A}]),\mathbf{O}_{A'}([\udl{A}'])] 
    =i\sigma(\udl{A},\udl{A}')\one,
\end{equation}
as the canonical quantization of \eqref{eq:symp_pois_bra}, as well as an analogue of the relation \eqref{eq:sym_field_shift},
\begin{equation}
\label{eq:Q_symp_shift}
    \mathbf{O}_{A-\udl{A'}}([\udl{A}])=\mathbf{O}_{A}([\udl{A}])+\sigma(\udl{A},\udl{A'})\one.
\end{equation}
Setting $J=0$, these observables are then such that for each regular Cauchy surface $\ol{\Sigma}\subset\ol{N}$ we may define the smeared (equal time) fields $\mathbf{A}(F)$, $\mathbf{E}(H)$ for $F,H\in \Omega^1(\ol{\Sigma})$ with $H$ co-closed by the relation
\begin{equation}
\label{eq:smearedAE}
    \mathbf{E}(F)-\mathbf{A}(H)=\mathbf{O}_0([\udl{A}]),
\end{equation}
where $F=\udl{A}\restriction_{\ol{\Sigma}}$ and $H=-\nml_{\ol{\Sigma}}\diff\udl{A}$, which then satisfy the canonical equal time commutation relations
\begin{equation}
    [\mathbf{A}(H),\mathbf{E}(F)]=i\ipc{H}{F}_{\ol{\Sigma}}.
\end{equation}
We may then interpret the operator valued distributions $\mathbf{A}$ and $\mathbf{E}$ as the vector potential and electric field at $\ol{\Sigma}$.

As any non-trivial representation of these observables on a Hilbert space necessarily gives rise to unbounded operators, it is more convenient to work with formally exponentiated observables, giving rise to a $C^*$-algebra.

\subsection{A \texorpdfstring{$C^*$-algebra}{C*-algebra} of semi-local observables}\label{sec:semilocal obs}
\subsubsection{Definition of the algebra}
We introduce a $C^*$-algebra\footnote{Recall that a (unital) $C^*$-algebra is a (unital) associative algebra $\mathcal{A}$ over $\CC$ is which admits an anti-linear involution $a\mapsto a^*$ and a norm $a\mapsto \Vert a\Vert$, such that
\begin{equation}
    (ab)^*=b^*a^*,\qquad \Vert ab\Vert\leq \Vert a\Vert\Vert b\Vert,\qquad \Vert aa^*\Vert=\Vert a\Vert^2,
\end{equation}
and $\mathcal{A}$ is complete w.r.t.~the norm topology \cite{dixmierCalgebras1977}.} of semi-local observables that is closely related to the Weyl algebra associated with a symplectic space. We recall some relevant background.
Let $(V,\sigma)$ be a pre-symplectic space and let $\Delta(V,\sigma)$ be the complex linear hull of the free complex vector space over $V$, writing the basis vector corresponding to $v\in V$ as $W(v)$. Then $\Delta(V,\sigma)$ becomes a unital $*$-algebra with
\begin{align}
        W(0)=&\one\label{eq:weyl1},\\
        W(v)^*=&W(-v)\label{eq:weyl2},\\
        W(v)W(w)=&\exp\left(-\frac{i}{2}\sigma(v,w)\ \right)W(v+w)\label{eq:weyl3}.
    \end{align}
If $(V,\sigma)$ is a symplectic space then there is a unique $C^*$-norm on $\Delta(V,\sigma)$ whose completion defines a $C^*$-algebra $\Weyl(V,\sigma)$. 
In the general pre-symplectic case a $C^*$-norm on $\Delta(V,\sigma)$ with good properties has been constructed in~\cite{binzConstructionUniquenessWeyl2004} and again, we take the completion in that norm to form $\Weyl(V,\sigma)$, which is specified uniquely by $(V,\sigma)$ up to unital ${}^*$-isomorphisms. 

Any symplectic map $\iota:(V,\sigma)\to (V',\sigma')$ between pre-symplectic spaces (i.e.~a linear map preserving the pre-symplectic structure) determines a unital $*$-homomorphism of $*$-algebras $\Delta(\iota):\Delta(V,\sigma)\to\Delta(V',\sigma')$ which extends uniquely\footnote{This statement does not appear in~\cite{binzConstructionUniquenessWeyl2004} but follows by the first two sentences of the proof of Prop.~3.8 of that reference.} to a unital $C^*$-homomorphism $\Weyl(\iota):\Weyl(V,\sigma)\to\Weyl(V',\sigma')$ by $\Weyl(\iota)W_{V}(v)=W_{V'}(\iota v)$ (adding subscripts for emphasis) which makes
$\Weyl$ into a functor between the categories of pre-symplectic spaces with symplectic maps and unital $C^*$-algebras with unit-preserving ${}^*$-homomorphisms. In particular, if $\iota$ is a symplectomorphism [i.e., has a symplectic inverse] then $\Weyl(\iota)$ is a unital ${}^*$-isomorphism and, in fact, an isometry (see Prop.~3.8 of~\cite{binzConstructionUniquenessWeyl2004}).

In the symplectic case, it is known that $\Weyl(V,\sigma)$ is \emph{nuclear} \cite[Thm.~10.10]{evansDilationsIrreversibleEvolutions1977}, and hence 
has a uniquely specified $C^*$-tensor product with any other $C^*$-algebra. Specifically, one has
\begin{equation}
    \Weyl(V\oplus W,\sigma_V\oplus \sigma_W) \cong
    \Weyl(V\oplus \sigma_V)\otimes \Weyl(W,\sigma_W)
\end{equation}
under the mapping $W_{V\oplus W}(v\oplus w)\mapsto W_V(v)\otimes W_W(w)$, writing subscripts for emphasis.

For linear theories, such as the free scalar field \cite{dimockAlgebrasLocalObservables1980} or free electromagnetism \cite{dimockQuantizedElectromagneticField1992,pfenningQuantizationMaxwellField2009,sandersElectromagnetismLocalCovariance2014}, a $C^*$-algebra of local observables can be realised through a Weyl algebra. Here we construct a similar algebra for semi-local observables on the affine theory of electromagnetism with a background current on finite Cauchy lens.

\begin{proposition}
\label{prop:CCR_BPI}
   Let $\ol{N}$ a finite Cauchy lens and $J\in \Omega^1_{\diff^*}(\ol{N})$ a background current. Then there exists a non-zero nuclear unital $C^*$-algebra $\Af(\ol{N})$ (determined uniquely up to unital ${}^*$-isomorphisms) generated by elements
\begin{equation}
\{W_{A}([\udl{A}]):(A,\udl{A})\in T\Sol^J(\ol{N})\}
\end{equation} 
satisfying the relations
\begin{align}
        W_{A}(0)=&\one,\\
        W_{A}([\udl{A}])^*=&W_{A}(-[\udl{A}]),\\
        W_{A}([\udl{A}])W_{A}([\udl{A}'])=&\exp\left(-\frac{i}{2}\sigma([\udl{A}],[\udl{A}'])\right)W_{A}([\udl{A}+\udl{A}']),\\
        W_{A-\udl{A}'}([\udl{A}])=&\exp(i\sigma([\udl{A}],[\udl{A}']))W_{A}([\udl{A}]).
\end{align}
\end{proposition}
\begin{proof}
    For a fixed background $A\in \Sol^J(\ol{N})$, the objects $\{W_A([\udl{A}]):[\udl{A}]\in \Sol_{\mathscr{G}}(\ol{N})$ satisfy the relations \eqref{eq:weyl1} to \eqref{eq:weyl3} for the symplectic space $(\Sol_{\mathscr{G}}(\ol{N}),\sigma)$. These therefore generate an algebra $\Weyl(\Sol_{\mathscr{G}}(\ol{N}),\sigma)$ defined uniquely up to unital ${}^*$-isomorphisms. 
    
    As for any other background $A'\in \Sol^J(\ol{N})$, one sets
    \begin{equation}
        W_{A'}([\udl{A}])=\exp(i\sigma([\udl{A}],[A-A']))W_{A}([\udl{A}])\in \Weyl(\Sol_{\mathscr{G}}(\ol{N}),\sigma),
    \end{equation}
    whereupon the relations~\eqref{eq:weyl1} to~\eqref{eq:weyl3} automatically hold for these generators. Therefore $\Af(\ol{N})=\Weyl(\Sol_{\mathscr{G}}(\ol{N}),\sigma)$. 
\end{proof}

\begin{rem}
    As can be seen from the proof of Prop.~\ref{prop:CCR_BPI}, the algebra $\Af(\ol{N})\cong \Weyl(\Sol_{\mathscr{G}}(\ol{N}),\sigma)$ is (up to isomorphisms) independent of the chosen background current $J\in \Omega^1_{\diff^*}(\ol{N})$. This is in line with observations made on the quantised scalar field with classical source, see e.g.~\cite{fewsterLocallyCovariantQuantum2015}. The current does play a role in the labelling of the generators $W_A([\udl{A}])$, where $J=-\diff^*\diff A$, which is used to provide an interpretation of these operators.
\end{rem}

The defining relations of $W_A([\udl{A}])$ are chosen such that, when formally identifying
\begin{equation}
    W_A([\udl{A}])\simeq \exp(i\mathbf{O}_A([\udl{A}])),
\end{equation}
they follow from Eq.~\eqref{eq:CCR} and \eqref{eq:Q_symp_shift}. This identification is made more precise at the level of sufficiently regular representations, see Sec.~\ref{sec:reps}.

For any $\udl{A}'\in \Sol(\ol{N})$, the map 
\begin{equation}
    W_A([\udl{A}])\mapsto W_{A-\udl{A}'}([\udl{A}])=\exp(i\sigma([\udl{A}],[\udl{A}']))W_{A}([\udl{A}])
\end{equation} defines a unital ${}^*$-isomorphism on $\Af(\ol{N})$ that we shall denote as the \emph{quantum background shift} $\alpha([\udl{A}'])\in \Aut(\Af(\ol{N}))$. We view these automorphisms as the quantum analogue of the background shifts $\xi([\udl{A}])^*\in \Aut(\Fsloc(\ol{N}))$, with large gauge transformations corresponding to the case $[\udl{A}']\in\mathscr{LG}(\ol{N})$.
In particular, we have the following.
\begin{proposition}
\label{prop:field_shift}
    The quantum background shifts $\alpha:\Sol_{\mathscr{G}}(\ol{N})\to \Aut(\Af(\ol{N}))$ define a group representation of $\Sol_{\mathscr{G}}(\ol{N})$ as an additive group. Furthermore, given a background $A\in \Sol^J(\ol{N})$, we have
    \begin{equation}
    \label{eq:field_shift_unitary}
        \alpha([\udl{A}])(a)=W_A([\udl{A}])aW_A([\udl{A}])^*,
    \end{equation}
    for each $a\in \Af(\ol{N})$.
\end{proposition}
\begin{proof}
    The group law is immediate from the definition. Furthermore, for any $(A,\udl{A}),(A',\udl{A}')\in T\Sol^J(\ol{N})$, we have
    \begin{align}
        W_A([\udl{A}])W_{A'}([\udl{A}'])W_{A}([\udl{A}])^*&=\exp(-i\sigma(\udl{A},\udl{A}'))W_{A'}([\udl{A}'])=W_{A'-\udl{A}}([\udl{A}']) \nonumber\\
        &=\alpha([\udl{A}])(W_{A'}([\udl{A}'])).
    \end{align}
    By continuity and linearity of $\alpha([\udl{A}])$, this proves Eq.~\eqref{eq:field_shift_unitary}. 
\end{proof}

In what follows, we shall simply refer to the automorphisms $\alpha([\udl{A}])$ as background shifts. The background shifts shall be used to define a notion of localisation on the algebra $\Af(\ol{N})$ at the end of this section.

Just as for $\Fsloc$, we can decompose the algebra $\Af(\ol{N})$ in closed loop and surface observables. Recall that for any regular Cauchy surface with boundaries $\ol{\Sigma}\subset \ol{N}$, we defined the symplectic spaces $(V^{C/S},\sigma^{C/S})$ in Sec.~\ref{sec:init_dat_hodge}.
\begin{definition}
\label{def:bulk_corn_weyls}
     For finite Cauchy lens $\ol{N}$ with regular Cauchy surface with boundary $\ol{\Sigma}\subset\ol{N}$,
     let 
    \begin{equation}
        \Af^{C}(\ol{\Sigma})=\Weyl\left(V^{C}(\ol{\Sigma}),\sigma_{\ol{\Sigma}}^C\right), \quad \Af^S(\ol{N})=\Weyl\left(V^S(\ol{N}),\sigma_{\ol{N}}^S\right),
    \end{equation}     
    writing the respective generators as $W^{C}(\udl{v}^C)$ (for $\udl{v}^C\in V^{C}(\ol{\Sigma})$) and $W^{S}(\udl{v}^S)$ (for $\udl{v}^S\in V^S(\ol{N})$). Lastly, we define the algebra $\Af^{C\oplus S}(\ol{\Sigma})=\Af^{C}(\ol{\Sigma})\otimes \Af^S(\ol{N})$.
\end{definition}
The Weyl relations of $W^{\bullet}(\udl{v}^{\bullet})$ (for $\bullet\in \{C,S\}$) are consistent with the identification (made precise at the level of sufficiently regular representations)
\begin{equation}
\label{eq:Weyl-exponentials}
    W^{\bullet}(\udl{v}^{\bullet})\simeq\exp(i\mathbf{o}^{\bullet}(\udl{v}^{\bullet})),
\end{equation}
where $\mathbf{o}^{\bullet}(\udl{v}^{\bullet})$ are interpreted as quantizations of $o^{\bullet}(\udl{v}^{\bullet})\in \Fsloc^{\bullet}(\ol{\Sigma})$, satisfying the canonical commutation relations
\begin{equation}
    [\mathbf{o}^{\bullet}(\udl{v}_1^{\bullet}),\mathbf{o}^{\bullet}(\udl{v}_2^{\bullet})]=i\sigma^{\bullet}(\udl{v}_1^{\bullet},\udl{v}_2^{\bullet})\one.
\end{equation}

Since $(\Sol_{\mathscr{G}}(\ol{N}),\sigma)\cong \left(V^{C}(\ol{\Sigma})\oplus V^S(\ol{N}),\sigma_{\ol{N}}^C\oplus \sigma_{\ol{\Sigma}}^S\right)$, it follows that $\Af(\ol{N})\cong \Af^{C\oplus S}(\ol{\Sigma})$. In particular, we can define the following isomorphisms.
\begin{definition}
\label{def:kin_transf}
    Let $\ol{N}$ be a finite Cauchy lens, $J\in \Omega^1_{\diff^*}(\ol{N})$ and $\ol{\Sigma}\subset\ol{N}$ a regular Cauchy surface with boundaries. Let $[A]\in\Sol_{\mathscr{G}}^J(\ol{N})$ be the background (depending on $\ol{\Sigma}$) satisfying
    \begin{equation}
        \mathfrak{K}^J_{\ol{\Sigma}}([A])=0.
    \end{equation}
    We define the \emph{quantum CHH isomorphism} for $\ol{\Sigma}$, $\varphi_{\ol{\Sigma}}:\Af(\ol{N})\to\Af^{C\oplus S}(\ol{\Sigma})$, by the relation
    \begin{equation}
    \label{eq:kin-iso}
        \varphi_{\ol{\Sigma}}(W_{A}([\udl{A}]))=W^C(F\oplus H)\otimes W^S(f\oplus h).
    \end{equation}
    for all $[\udl{A}]\in \Sol_{\mathscr{G}}(\ol{N})$ where $(F\oplus H\oplus f\oplus h)=\mathfrak{K}_{\ol{\Sigma}}([\udl{A}])$
    (see Prop.~\ref{prop:sympl_decomp} for the classical version).
    \end{definition}
That Eq.~\eqref{eq:kin-iso} indeed defines an isomorphism, follows directly from the fact that $\mathfrak{K}_{\ol{\Sigma}}$ is a symplectomorphism. The quantum CHH-isomorphisms yields a quantized analogue to Lem.~\ref{lem:symp_to_CS}. 
\begin{proposition}
Let $\ol{N}$ be a finite Cauchy lens, $J\in \Omega^1_{\diff^*}(\ol{N})$ and $\ol{\Sigma}\subset\ol{N}$ a regular Cauchy surface with boundaries. Then for any $[A']\in \Sol_{\mathscr{G}}^J(\ol{N})$, $[\udl{A}]\in \Sol_{\mathscr{G}}(\ol{N})$, we have
\begin{equation}
\label{eq:kin-iso-gen}
    \varphi_{\ol{\Sigma}}(W_{A'}([\udl{A}]))=\exp(-i(o^C(F\oplus H)([A'])+o^S(f\oplus h)([A'])))W^C(F\oplus H)\otimes W^S(f\oplus h),
\end{equation}
where $(F\oplus H\oplus f\oplus h)=\mathfrak{K}_{\ol{\Sigma}}([\udl{A}])$.

\end{proposition}
\begin{proof}
    For general background $[A']\in\Sol_{\mathscr{G}}^J(\ol{N})$, we have
    \begin{equation}
    \varphi_{\ol{\Sigma}}(W_{A'}([\udl{A}]))=\exp(-i\sigma(\udl{A},A'-A))\varphi_{\ol{\Sigma}}(W_{A}([\udl{A}])),
\end{equation}
where $[A]=(\mathfrak{K}^J_{\ol{\Sigma}})^{-1}(0)$. Using Lem.~\ref{lem:symp_to_CS}, we have
\begin{equation}
    \sigma(\udl{A},A'-A)=o^C(F\oplus H)([A'])+o^S(f\oplus h)([A']),
\end{equation}
for $(F\oplus H\oplus f\oplus h)=\mathfrak{K}_{\ol{\Sigma}}([\udl{A}])$. Eq.~\eqref{eq:kin-iso-gen} now follows directly from Def.~\ref{def:kin_transf}.
\end{proof}
It should be pointed out that, as Cauchy-Hodge-Helmholtz decomposition of gauge orbits of classical solutions of the Maxwell equations depend on the choice of Cauchy surface $\ol{\Sigma}$, so does the decomposition of $\Af(\ol{N})$ into tensor factors $\Af^{C}(\ol{\Sigma})\otimes \Af^S(\ol{N})$ via a quantum CHH-isomorphism. In particular, for two arbitrary regular Cauchy surfaces $\ol{\Sigma}$ and $\ol{\Sigma}'$ of $\ol{N}$ the map 
\begin{equation}
    \varphi_{\ol{\Sigma}}\circ\varphi_{\ol{\Sigma}'}^{-1}:\Af^{C}(\ol{\Sigma}')\otimes \Af^S(\ol{N})\to \Af^{C}(\ol{\Sigma})\otimes \Af^S(\ol{N})
\end{equation} will in general not act diagonally on the tensor factors. Nevertheless, the $\ol{\Sigma}$-dependent decomposition of $\Af(\ol{N})$ has an compelling geometric interpretation when observables are interpreted as acting on initial data at $\ol{\Sigma}$, as discussed for classical observables following Thm.~\ref{thm:class_obs_decomp}. Furthermore, it will prove a useful tool in various constructions discussed in the remainder of this paper, including state and gluing constructions.

\subsubsection{Regaining an algebra of local observables}
The algebra $\Af(\ol{N})$ extends the usual Weyl algebra of local observables of the electromagnetic field associated with a region $N\subset M$ (see e.g.~\cite{dimockQuantizedElectromagneticField1992}), in the same way that $\Floc(N)\subset\Fsloc(\ol{N})$.
\begin{definition}
\label{def:loc_obs_qft}
Let $\ol{N}$ be a finite Cauchy lens, $J\in \Omega^1_{\diff^*}(\ol{N})$ and $A\in \Sol^J(\ol{N})$. We define for $f\in\Omega^1_{0\diff^*}(\ol{N})$
\begin{equation}
    W_{\textup{loc}}^J(f)=\exp(i\ipc{f}{A})W_{A}(-[G^{\PJ}(f)]),
\end{equation} 
and denote the algebra generated by these observables as $\Af_{\textup{loc}}(\ol{N})\subset \Af(\ol{N})$. 
\end{definition}
These generators indeed satisfy the usual relations of local Weyl generators for electromagnetism in presence of a source term.
\begin{proposition}
    The observables $W_{\textup{loc}}^J(f)$ defined in Def.~\ref{def:loc_obs_qft} are independent of the choice of $A\in \Sol^J(\ol{N})$, and satisfy the relations
    \begin{align}
        W_{\textup{loc}}^J(0)&=1,\\
        W_{\textup{loc}}^J(f)^*&=W_{\textup{loc}}(-f),\\
        W_{\textup{loc}}^J(f)W_{\textup{loc}}^J(h)&=\exp\left(\frac{i}{2}\ipc{f}{G^{\PJ}(h)}\right)W_{\textup{loc}}^J(f+h),\\
        W_{\textup{loc}}^J(-\diff^*\diff f')&=\exp(i\ipc{f'}{J})\one,
    \end{align}
    for any $f,h\in \Omega^1_{0\diff^*}(\ol{N})$, $f'\in \Omega^1_0(N)$, as well as transform under the background shift
    \begin{equation}
        \alpha([\udl{A}])(W_{\textup{loc}}^J(f))=\exp(i\ipc{f}{\udl{A}})W_{\textup{loc}}^J(f).
    \end{equation}
\end{proposition}
\begin{proof}
    We show that Def.~\ref{def:loc_obs_qft} is independent of the choice of $A\in \Sol^J(\ol{N})$. For $A,A'\in \Sol^J(\ol{N})$, we have
    \begin{align}
        \exp(i\ipc{f}{A'})W_{A'}(-[G^{\PJ}(f)])=&\exp(i\ipc{f}{A'})\exp(-i\sigma(G^{\PJ}(f),A-A'))W_{A}(-[G^{\PJ}(f)])\nonumber\\
        =&\exp(i\ipc{f}{A'})\exp(i\ipc{f}{A-A'})W_{A}(-[G^{\PJ}(f)])\nonumber\\
        =&\exp(i\ipc{f}{A})W_{A}(-[G^{\PJ}(f)]),
    \end{align}
    where we have used Lem.~\ref{lem:st_to_symp_smear}. The remaining relations follow from straightforward computation.
\end{proof}

\subsubsection{Large gauge transformations}
In analogy with the the classical setting, we take large gauge transformations on the algebra $\Af(\ol{N})$ as subgroup of the background shifts characterised in Prop.~\ref{prop:field_shift}. We shall use the parametrisation of large gauge directions $\mathfrak{G}:\mathscr{LG}_\angle(\ol{N})\to \mathscr{LG}(\ol{N})$ given above Prop.~\ref{prop:large_gauge_boundary}. 
\begin{definition}
\label{def:q_lar_gauge_transf}
    Let $\ol{N}$ be a finite Cauchy lens. For $\lambda\in \mathscr{LG}_\angle(\ol{N})$,  $\alpha_{\mathscr{LG}}(\lambda)\in \Aut(\Af(\ol{N}))$ 
    denotes the \emph{large gauge automorphism}
    \begin{equation}
        \alpha_{\mathscr{LG}}(\lambda)=\alpha(\mathfrak{G}(\lambda)),
    \end{equation}
    which is unitarily implemented as $\alpha_{\mathscr{LG}}(\lambda)=\Ad U_{\mathscr{LG}}(\lambda)$, where 
    \begin{equation}
        U_{\mathscr{LG}}(\lambda)=W_A(\mathfrak{G}(\lambda))\in \Af(\ol{N}),
    \end{equation}
    for some choice of background $A\in \Sol^J(\ol{N})$.
\end{definition}
For fixed background, the map $\lambda\mapsto U_{\mathscr{LG}}(\lambda)$
 is a homomorphism from the additive group $\mathscr{LG}_\angle(\ol{N})$ to the group of unitary elements in $\Af(\ol{N})$, which is background-independent up to a phase.
 Thus $\alpha_{\mathscr{LG}}$ defines a group action of inner automorphisms on $\Af(\ol{N})$ that is independent of the background and implements large gauge transformations
\begin{equation}
    \alpha_{\mathscr{LG}}(\lambda)(W_{A}([\udl{A}]))=W_{A-\diff\Lambda}([\udl{A}])=\exp(i\sigma(\udl{A},\diff\Lambda))W_{A}([\udl{A}])=\exp\left(i\ipc{\nml_{\partial\ol{\Sigma}}\nml_{\ol{\Sigma}}\diff\udl{A}}{\lambda}_{\partial\ol{\Sigma}}\right)W_{A}([\udl{A}]),
\end{equation}
for $(A,\udl{A})\in T\Sol^J(\ol{N})$, $\lambda\in \mathscr{LG}_\angle(\ol{N})$ and $[\diff\Lambda]=\mathfrak{G}(\lambda)$. One finds a direct analogue to Prop.~\ref{prop:c_loc_inv}
\begin{proposition}
\label{prop:q_loc_inv}
    Let $\ol{N}$ be a finite Cauchy lens. Then for each $b\in \Af_{\textup{loc}}(\ol{N})$ and $\lambda\in \mathscr{LG}_\angle(\lambda)$, we have
    \begin{equation}
        \alpha_{\mathscr{LG}}(\lambda)(b)=b.
    \end{equation}
\end{proposition}
\begin{proof}
    For $f\in \Omega^1_{0\diff^*}(\ol{N})$, we compute
    \begin{equation}
        \alpha_{\mathscr{LG}}(\lambda)W_{\textup{loc}}(f)=\exp(-i\ipc{\nml_{\partial\ol{\Sigma}}\nml_{\Sigma}\diff G^{\PJ}f}{\lambda})W_{\textup{loc}}(f)=W_{\textup{loc}}(f),
    \end{equation}
    where we have used Eq.~\eqref{eq:sigmadLambdaAp}. Since operators of this form generate $\Af_{\textup{loc}}(\ol{N})$, the proposition follows.
\end{proof}
We say that $\Af_{\textup{loc}}(\ol{N})\subset \Af(\ol{N})^{\alpha_{\mathscr{LG}}}$, where $\Af(\ol{N})^{\alpha_{\mathscr{LG}}}$ denotes the algebra of invariants w.r.t. the action $\alpha_{\mathscr{LG}}:\mathscr{LG}_\angle(\ol{N})\to \Aut(\Af(\ol{N})$.

Under the usual interpretation 
\begin{equation}
    U_{\mathscr{LG}}(\lambda)=W_A([\diff\Lambda])\simeq \exp(i\mathbf{O}_A([\diff\Lambda]))
\end{equation}
for the chosen background $A$. As a classical observable, $O_A([\diff\Lambda])$ evaluates on $[\tilde{A}]\in \Sol_{\mathscr{G}}^J(\ol{N})$ as
\begin{equation}
    O_A([\diff\Lambda])(\tilde{A})=\ipc{\lambda}{\nml_{\partial\ol{\Sigma}}(\tilde{\mathbf{E}}-\mathbf{E}_{\textnormal{bg}})}_{\partial\ol{\Sigma}},
\end{equation}
for $\tilde{\mathbf{E}}=-\nml_{\ol{\Sigma}}\diff \tilde{A}$ and $\mathbf{E}_{bg}=-\nml_{\ol{\Sigma}}\diff A$ the electric fields at $\ol{\Sigma}$ associated with $\tilde{A}$ and the background $A$. Thus, loosely speaking, large gauge transformations $\alpha_{\mathscr{LG}}(\lambda)$ are generated by observables given by smearings of electric flux density through the corner $\angle\ol{N}=\partial\ol{\Sigma}$ of $\ol{N}$, which can be made precise at the level of sufficiently regular representations such as discussed in Sec.~\ref{sec:reps}. The fact that large gauge transformations are generated by electric flux observables is in line with \cite{donnellyLocalSubsystemsGauge2016,stromingerLecturesInfraredStructure2018,rielloHamiltonianGaugeTheory2024}.

\subsubsection{Localisability of observables}
\label{sec:localise}
Using the background shifts, we define a notion of localisability of observables in $\Af(\ol{N})$. Here, 
a \emph{region} in $\ol{N}$ is a relatively open subset of $\ol{N}$ with causally convex interior.
\begin{definition}
\label{def:localisability}
    Let $\ol{N}$ be a finite Cauchy lens.  
    For any compact $K\subset\ol{N}$, let 
    $\Af(\ol{N};K)$ be the subalgebra of $\Af(\ol{N})$ 
    consisting of fixed points under all background shifts
    $\alpha([\udl{A}])$ for $\udl{A}\in \Sol(\ol{N})$ with $\supp(\udl{A})\cap K=\emptyset$. For any region $U$ in $\ol{N}$, let 
\begin{equation}
        \Af(\ol{N};U)=C^*\left(\bigcup_{K\subset U\text{ compact}}\Af(\ol{N};K)\right)
    \end{equation}
    be the smallest $C^*$-subalgebra containing all $\Af(\ol{N};K)$ for $K\subset U$ compact.
    An element $a\in\Af(\ol{N})$ is said to be \emph{localisable} in $K$ (resp., $U$) if 
    $a\in\Af(\ol{N};K)$ (resp.,  $a\in\Af(\ol{N};U)$).     
    Finally, the algebra of local observables localisable in $U$ is defined by
    \begin{equation}
        \Af_{\textup{loc}}(\ol{N};U)=\Af(\ol{N};U)\cap\Af_{\textup{loc}}(\ol{N}).
    \end{equation}
\end{definition}
    The net structure determined by this notion of localisability is preserved by background shifts.
\begin{proposition}
    Let $\ol{N}$ be a finite Cauchy lens, $U\subset \ol{N}$ a region, and $[\udl{A}]\in \Sol_\mathscr{G}(\ol{N})$, then 
    \begin{equation}
        \alpha([\udl{A}])(\Af(\ol{N};U))=\Af(\ol{N};U).
    \end{equation}
\end{proposition}
\begin{proof}
   Suppose that $a\in \Af(\ol{N};U)$, then for any $\udl{A}'\in \Sol(\ol{N})$ with $U\cap \supp(\udl{A}')=\emptyset$ we have 
    \begin{equation}
        \alpha([\udl{A}'])(\alpha([\udl{A}])(a))=\alpha([\udl{A}])(\alpha([\udl{A}'])(a))=\alpha([\udl{A}])(a),
    \end{equation}
    because $\alpha$ is an Abelian group action, and therefore $\alpha([\udl{A}])(\Af(\ol{N};U))\subset \Af(\ol{N};U)$
    for each $[\udl{A}]\in \Sol_\mathscr{G}(\ol{N})$. The reverse inclusion follows because $\alpha([\udl{A}])^{-1}=\alpha(-[\udl{A}])$.
\end{proof}

For a sufficiently well behaved region $U$, the observables localisable in $U$ coincide with the algebra generated by fields smeared with forms supported in $U$. This is made precise as follows.
\begin{proposition}
    \label{prop:loc_localise}
    Let $\ol{N}$ be a finite Cauchy lens and $U\subset\ol{N}$ be a region such that
    \begin{itemize}
        \item $U\cap \partial{\ol{N}}=\emptyset$,
        \item for each compact set $K\subset U$ there exists a pre-compact $U'\subset U$ (i.e.~$\overline{U'}\subset U$ compact) with $K\subset U'$ and a regular Cauchy surface with boundaries $\ol{\Sigma}\subset \ol{N}$ such that $\ol{N}\setminus \mathscr{I}(U')$ is a compact manifold with corners and $\ol{\Sigma}\setminus\mathscr{I}(U')$ is a compact manifold with smooth boundaries.
    \end{itemize}
    Then both $\Af_{\textup{loc}}(\ol{N};U)$ and $\Af(\ol{N};U)$ coincide with the smallest $C^*$-subalgebra of $\Af(\ol{N})$ containing all $W_{\textup{loc}}^J(f)$ with $\supp(f)\subset U$.
\end{proposition}
The proof relies on some technical lemmas given in Appx~\ref{apx:localise}.
\begin{proof}
For each $f\in \Omega^1_{0\diff^*}(U)$  and $\udl{A}\in \Sol(\ol{N})$ with $\supp(\udl{A})\cap \supp(f)=\emptyset$ we have
    \begin{equation}
        \alpha([\udl{A}])W_{\textup{loc}}^J(f)=\exp(i\ipc{f}{[\udl{A}]})W_{\textup{loc}}^J(f)=W_{\textup{loc}}^J(f),
    \end{equation}
    hence $W_{\textup{loc}}^J(f)\in \Af(\ol{N};\supp(f))\cap\Af_{\textup{loc}}(\ol{N})\subset \Af_{\textup{loc}}(\ol{N};U)$.

    Conversely, it suffices to show that for each compact $K\subset U$, the algebra $\Af(\ol{N};K)$
    is generated by operators of the form $W_{\textup{loc}}^J(f)$ with $f\in \Omega^1_{0\diff^*}(U)$. Using
    Lem.~\ref{lem:Weyl_fixedpoint}, which characterises fixed point subalgebras of Weyl algebras, we may infer that $\Af(\ol{N};K)\subset \Af(\ol{N})$ is generated by 
    $W_A([\udl{A}])$ with $(A,\udl{A})\in T\Sol^J(\ol{N};K)$, where
    \begin{equation}
     T\Sol^J(\ol{N};K):=   \{(A,\udl{A})\in T\Sol^J(\ol{N}):\sigma(\udl{A}',
    \udl{A})=0\text{ for all }\udl{A}'\in \Sol(\ol{N})\text{ with }\supp(\udl{A}')\cap K=\emptyset\}.
    \end{equation}
    
    To complete the proof, we show that for
    each $(A,\udl{A})\in T\Sol^J(\ol{N};K)$, there exists an $f\in \Omega^1_{0\diff^*}(U)$ such that $[\udl{A}]=[G^{PJ}f]$,
    which implies 
    \begin{equation}
        W_A([\udl{A}])=W_A([G^{PJ}f])=\exp(i\ipc{f}{A})W_{\textup{loc}}(-f),
    \end{equation}
    thus establishing the result. Accordingly,
    let $(A,\udl{A})\in T\Sol^J(\ol{N};K)$. Denoting $V=\ol{N}\setminus \mathscr{J}(K)$, for each $h\in \Omega^1_{0\diff^*}(V)$, we have
    \begin{equation}
        \ipc{h}{\udl{A}\restriction_V}_{V}= \ipc{h}{\udl{A}}_{\ol{N}}=-\sigma(G^{\PJ}h,\udl{A})=
        0
    \end{equation}
    because $\supp(G^{\PJ}h)\cap K\subset \mathscr{J}(\supp (h))\cap K=\emptyset$. 
    By Lem.~\ref{lem:poincare_dual} (a consequence of Poincar\'e duality), this means that $\udl{A}\restriction_V=\diff\Lambda$ for some $\Lambda\in \Omega^0(\Intr(V))$. Choose $K\subset U'\subset U$ as above and denote $V'=\ol{N}\setminus\mathscr{I}(U')\subset V$. Applying Lem.~\ref{lem:0form_ext} to $\Lambda\restriction_{\Intr(V')}$, means that there exists a $\ol{\Lambda}\in \Omega^0(V')$ such that $\udl{A}\restriction_{V'}=\diff\ol{\Lambda}$. By smoothly extending $\ol{\Lambda}$, one may find a $\tilde{\Lambda}\in \Omega^0(\ol{N})$ such that $\tilde{\Lambda}\restriction_{V'}=\ol{\Lambda}$. We thus conclude that $\tilde{\udl{A}}:=\udl{A}-\diff\tilde{\Lambda}$ satisfies $\supp(\tilde{\udl{A}})\subset \mathscr{J}(\overline{U'})$, and thus by Lem.~\ref{lem:green_sc}, there exists an $f_1\in \Omega^1_{0\diff^*}(U)$ such that $[\tilde{\udl{A}}]=[G^{\PJ}f_1]$.
    
    What remains to be shown, is that there exists an $f_2\in \Omega^1_{0\diff^*}(U)$ such that $[\diff\tilde{\Lambda}]=[G^{\PJ}f_2]$, 
    whereupon $[\udl{A}]=[G^\PJ(f_1+f_2)]\in [G^\PJ\Omega_{0\diff^*}^1(U)]$. Note that for each $\udl{A}'\in \Sol(\ol{N})$ with $\supp(\udl{A}')\cap \overline{U'}=\emptyset$, we have
    \begin{equation}
        \ipc{\nml_{\partial\ol{\Sigma}}\nml_{\ol{\Sigma}}\diff\udl{A}'}{\tilde{\Lambda}}_{\partial\ol{\Sigma}}=\sigma(\udl{A}',\diff\tilde{\Lambda})=\sigma(\udl{A}',\udl{A})=0,
    \end{equation}
    for any regular Cauchy surface with boundaries $\ol{\Sigma}\subset\ol{N}$. Choose $\ol{\Sigma}$ such that $\ol{S}=\ol{\Sigma}\setminus\mathscr{I}(U')$ is a compact manifold with boundaries, we apply Lem.~\ref{lem:form_deform} to find for each $f\in \nml_{\partial \ol{S}}\Omega^1_{\diff^*}(\ol{S})$ with $\supp(f)\subset \partial\ol{S}\cap \partial\ol{\Sigma}$ a form $\mathbf{E}\in \Omega^1_{\diff^*}(\ol{\Sigma})$ with $\nml_{\partial\ol{S}}\mathbf{E}=f$ and $\supp(\mathbf{E})\subset \ol{S}$. Now let $\udl{A}'\in \Sol(\ol{N})$ such that $\mathfrak{I}_{\ol{\Sigma}}([\udl{A}'])=0\oplus \mathbf{E}$, which by Lem.~\ref{lem:E_sol_support} can always be chosen such that $\supp(\udl{A}')\subset \mathscr{J}(\supp(\mathbf{E}))$, and thus $\supp(\udl{A}')\cap \ol{U'}=\emptyset$. As a consequence $\sigma(\udl{A}',\diff\tilde{\Lambda})=-\ipc{f}{\tilde{\Lambda}}_{\partial\ol{\Sigma}}=0$. For any $f\in \Omega^0(\partial\ol{S}\cap\partial\ol{\Sigma})$, we  have $f\in \nml_{\partial \ol{S}}\Omega^1_{\diff^*}(\ol{S})$ if and only if $\ipc{f}{1}_{\partial\ol{S}_i}$ for each connected component $\ol{S}_i\subset\ol{S}$. Consequently, it follows that $\tilde{\Lambda}$ is constant on $\partial\ol{\Sigma}\cap \partial\ol{S}_i$ for each connected component $\ol{S}_i\subset\ol{S}$. As a result, one can always choose a $\Lambda'\in \Omega^0(\ol{N})$ with $\diff\Lambda'\subset \mathscr{J}(\overline{U'})$ (and thus in particular locally constant on $\ol{S}$) such that $\Lambda'\restriction_{\angle \ol{N}}=\tilde{\Lambda}\restriction_{\angle \ol{N}}$. By Lem.~\ref{lem:green_sc}, there thus exists an $f_2\in \Omega^1_{0\diff^*}(\ol{N})$ with $\supp(f_2)\subset \overline{U'}\subset U$, for which $[\diff\Lambda']=[\diff\tilde{\Lambda}]=[G^{\PJ} f_2]$, which was the remaining statement required. 
    \end{proof}

\subsection{States and representations of the semi-local algebra and background covariance}
\label{sec:reps}
We now turn to the question of whether $\Af(\ol{N})$ (or equivalently $\Af^{C\oplus S}(\ol{\Sigma})$ given a choice of regular Cauchy surface $\ol{\Sigma}$) admits any physically meaningful representations. In particular, do these Weyl algebras admit Fock representations? (For examples of Weyl algebras with no Fock representations, see \cite{robinsonSYMPLECTICPATHOLOGY1993}.) 

Fock representations are uniquely specified by a \emph{one-particle structure} on the underlying symplectic space $\Sol_\mathscr{G}(\ol{N})\cong V^{C}(\ol{\Sigma})\oplus V^S(\ol{N})$ (see e.g.~\cite{kayLinearSpinzeroQuantum1978,kayTheoremsUniquenessThermal1991,khavkineAlgebraicQFTCurved2015}), or equivalently by a quasi-free state \cite{manuceauQuasifreeStatesCCRalgebra1968}.
\begin{definition}
\label{def:quasi-free-kin}
    A state $\omega:\Af^{C\oplus S}(\ol{\Sigma})\to \CC$ (i.e.~a continuous positive definite linear unital map) is called \emph{quasi-free} if there exists a real inner product
    \begin{equation}
        \mu:\left(V^{C}(\ol{\Sigma})\oplus V^S(\ol{N})\right)^2\to \RR,
    \end{equation}
    defining a norm $\Vert .\Vert_\mu$ satisfying for $\udl{v}_1^{C/S},\udl{v}_2^{C/S}\in V^{C/S}(\ol{\Sigma}/\ol{N})$
    \begin{equation}\label{eq:norm_domination}
        \vert \sigma_{\ol{\Sigma}}^C(\udl{v}_1^C,\udl{v}_2^C)+\sigma_{\ol{\Sigma}}^S(\udl{v}_1^S,\udl{v}_2^S)\vert\leq 2\Vert \udl{v}_1^C\oplus \udl{v}_1^S\Vert_\mu\Vert \udl{v}_2^C\oplus \udl{v}_2^S\Vert_\mu,
    \end{equation}
    such that 
    \begin{equation}
    \label{eq:quasi-free-kin}
        \omega(W^C(\udl{v}^C)\otimes W^S(\udl{v}^S))=\exp\left(-\frac{1}{2}\Vert \udl{v}^C\oplus \udl{v}^S\Vert_\mu^2\right),
    \end{equation}
    for all $\udl{v}^{C/S}\in V^{C/S}(\ol{\Sigma}/\ol{N})$.
\end{definition}
Concretely, a quasi-free state $\omega$ is uniquely specified by its symmetrised two-point function $\mu$ on the relevant symplectic space, such that the $n$-point correlation functions defined through these states satisfy the Isserlis--Wick theorem.

Given a norm $\mu$ as above, we will typically denote the associated quasi-free state uniquely determined by Eq.~\eqref{eq:quasi-free-kin} as $\omega_\mu$. By \cite[Thm.~3.1]{kayTheoremsUniquenessThermal1991}, a quasi-free state $\omega_\mu$ associated with a norm $\mu$ uniquely corresponds (up to unitary equivalence) to a one-particle structure $K:V^{C}(\ol{\Sigma})\oplus V^S(\ol{N})\to \KK$, where $\KK$ a (complex) Hilbert space, $K$ a real linear (not necessarily injective) map such that
\begin{itemize}
    \item $K(V^{C}(\ol{\Sigma})\oplus V^S(\ol{N}))+iK(V^{C}(\ol{\Sigma})\oplus V^S(\ol{N}))$ is dense in $\KK$,
    \item For any  $\udl{v}_1^{C/S},\udl{v}_2^{C/S}\in V^{C/S}(\ol{\Sigma}/\ol{N})$ we have
    \begin{equation}
        \ipc{K(\udl{v}^C_1\oplus \udl{v}_1^S)}{K(\udl{v}^C_2\oplus \udl{v}_2^S)}_{\KK}=\mu(\udl{v}^C_1\oplus \udl{v}_1^S,\udl{v}^C_2\oplus \udl{v}_2^S)+\frac{i}{2}\left(\sigma^C(\udl{v}_1^C,\udl{v}_2^C)+\sigma^S(\udl{v}_1^S,\udl{v}_2^S)\right).
    \end{equation}
\end{itemize}

Each state on $\omega $ on $\Af^{C\oplus S}(\ol{\Sigma})$ defines a representation via the GNS construction, see e.g.~\cite[Thm.~2.3.16]{bratteliOperatorAlgebrasQuantum1987}. Namely, there exists (uniquely up to unitary equivalence) a triple $(\HH,\pi,\Omega)$ where $\HH$ is a complex Hilbert space, $\pi:\Af^{C\oplus S}(\ol{\Sigma})\to B(\HH)$ is a unital ${}^*$-homomorphism (i.e.~$(\HH,\pi)$ is a Hilbert space representation of $\Af^{C\oplus S}(\ol{\Sigma})$), and $\Omega\in \HH$, such that for each $a\in \Af^{C\oplus S}(\ol{\Sigma})$
\begin{equation}
    \ipc{\Omega}{\pi(a)\Omega}=\omega(a),
\end{equation}
and $\pi(\Af^{C\oplus S}(\ol{\Sigma}))\Omega$ is dense in $\HH$. Applying this construction to a quasi-free state $\omega_\mu$ yields a (bosonic) Fock representation $(\HH_\mu,\pi_\mu,\Omega_\mu)$, see \cite[Sec.~3.2]{kayTheoremsUniquenessThermal1991}
given (up to equivalence) as follows,
using the conventions of \cite[Thm.~5.2.24]{khavkineAlgebraicQFTCurved2015}. The GNS Hilbert space is
\begin{equation}
    \HH_\mu=\overline{\bigoplus_{n=0}^\infty \KK^{\otimes_S n}},
\end{equation}
where $\otimes_S$ denotes the symmetrized tensor product and is $\KK$ the one-particle Hilbert space of $\mu$, while the GNS vector is the Fock vacuum vector
\begin{equation}
    \Omega_\mu=1\oplus 0\oplus0....
\end{equation}
The representation $\pi_\mu$ is defined in terms of the
annihilation and creation operators $a(\psi)$ and $a^*(\psi)$ for $\psi\in \KK$, densely defined on $\HH_\mu$ and satisfying $a(\psi)\Omega_\mu=0$, $[a(\psi),a(\psi')]=[a^*(\psi),a^*(\psi')]=0$ and $[a(\psi),a^*(\psi')]=\ipc{\psi}{\psi'}_{\KK}\cdot 1_{\HH}$. Using the one-particle structure $K:V^{C}(\ol{\Sigma})\oplus V^S(\ol{N})\to \KK$, there then exist essentially self-adjoint operators
\begin{equation}
    \mathbf{o}^C(\udl{v}^C)=a(K(\udl{v}^C\oplus 0))+a^*(K(\udl{v}^C\oplus 0)), \qquad \mathbf{o}^S(\udl{v}^s)=a(K(0\oplus \udl{v}^S))+a^*(K(0\oplus \udl{v}^S)),
\end{equation}
on $\HH_\mu$ for $\udl{v}^C\in V^C(\ol{\Sigma})$ and $\udl{v}^S\in V^S(\ol{N})$, whereupon $\pi_\mu$ is defined by 
\begin{equation}
    \pi_\mu(W^C(\udl{v}^C)\otimes W^S(\udl{v}^S))=\exp\left(i\left( \mathbf{o}^C(\udl{v}^C)+\mathbf{o}^S(\udl{v}^S)\right)\right),
\end{equation}
defined using functional calculus. In particular,  Eq.~\eqref{eq:Weyl-exponentials} is given a precise meaning in such representations.

In line with the discussion on background dependent generators for $\Af(\ol{N})$ in Sec.~\ref{sec:quant_alg}, we introduce the following terminology to describe quasi-free states of $\Af(\ol{N})$.
\begin{definition}
\label{def:quasi-free-sloc}
    A state $\omega:\Af(\ol{N})\to \CC$  is called \emph{quasi-free around a background $A\in \Sol^J(\ol{N})$} if there exists a real inner product
    \begin{equation}
        \mu:\left(\Sol_\mathscr{G}(\ol{N})\right)^2\to \RR,
    \end{equation}
    defining a norm $\Vert .\Vert_\mu$ satisfying for $[\udl{A}_1],[\udl{A}_2]\in \Sol_{\mathscr{G}}(\ol{N})$
    \begin{equation}
        \vert \sigma(\udl{A}_1,\udl{A}_2)\vert\leq 2\Vert [\udl{A}_1]\Vert_\mu\Vert [\udl{A}_2]\Vert_\mu,
    \end{equation}
    such that 
    \begin{equation}
    \label{eq:quasi-free-A}
        \omega(W_{A}([\udl{A}]))=\exp\left(-\frac{1}{2}\Vert [\udl{A}]\Vert_\mu^2\right),
    \end{equation}
    for all $[\udl{A}]\in \Sol_{\mathscr{G}}(\ol{N})$.
\end{definition}

Given a current $J\in \Omega^1_{\diff^*}(\ol{N})$, recall that a choice of Cauchy surface with boundary $\ol{\Sigma}\subset \ol{N}$ singles out a background $[A]=(\mathfrak{K}^J_{\ol{\Sigma}})^{-1}(0)\in \Sol_\mathscr{G}^J(\ol{N})$. A quasi-free state $\omega_{\ol{\Sigma}}$ on $\Af^{C\oplus S}(\ol{\Sigma})$ defines a quasi-free state $\omega$ on $\Af(\ol{N})$ around $[A]$ given by 
\begin{equation}
    \omega=\omega_{\ol{\Sigma}}\circ \varphi_{\ol{\Sigma}},
\end{equation}
where $\varphi_{\ol{\Sigma}}$ is the quantum CHH-isomorphism of Def.~\ref{def:kin_transf}.

Given any quasi-free state $\omega_\mu$ around a background $[A]\in \Sol_\mathscr{G}^J(\ol{N})$ on $\Af(\ol{N})$, evaluating $\omega$ on the generators $W_{\textup{loc}}^J(f)$ of $\Af_{\textup{loc}}(\ol{N})\subset\Af(\ol{N})$ (with $f\in \Omega^1_{0\diff^*}(\ol{N})$) yields
\begin{equation}
    \omega_\mu(W_{\textup{loc}}^J(f))=\exp\left(i\ipc{f}{A}-\frac{1}{2}\Vert [G^\PJ f]\Vert^2_{\mu}\right).
\end{equation}
This means a quasi-free state on $\Af(\ol{N})$ around a background $[A]\in \Sol_{\mathscr{G}}^J(\ol{N})$ defines a Gaussian state with one-point function $[A]$ when restricted to the algebra of local observables $\Af_{\textup{loc}}(\ol{N})$.

One can see from Def.~\ref{def:quasi-free-kin}, that existence of a quasi-free state on $\Af^{C\oplus S}(\ol{\Sigma})$ requires the existence of a norm on $V^C(\ol{\Sigma})\oplus V^S(\ol{N})$ for which $\sigma^C\oplus\sigma^S$ is continuous (a so-called moderating norm for $\sigma^C\oplus\sigma^S$). For general infinite dimensional symplectic spaces, such a norm need not exist \cite{robinsonSYMPLECTICPATHOLOGY1993}. Here, however, examples are easily given.
\begin{theorem}
\label{thm:L2rep}
    Let $\ol{\Sigma}\subset \ol{N}$ be a regular Cauchy surface with boundaries and let \begin{equation}
        \mu_{L^2}:\left(V^{C}(\ol{\Sigma})\oplus V^S(\ol{N})\right)^2\to \RR
    \end{equation} be given by
    \begin{equation}
        \mu_{L^2}(F\oplus H\oplus f\oplus h,F'\oplus H'\oplus f'\oplus h')=\frac{1}{2}\left(\ipc{F}{F'}_{\ol{\Sigma}}+\ipc{H}{H'}_{\ol{\Sigma}}+\ipc{f}{f'}_{\angle\ol{N}}+\ipc{h}{h'}_{\angle\ol{N}}\right).
    \end{equation}
    Then there exists a unique quasi-free state $\omega_{L^2}$ on $\Af^{C\oplus S}(\ol{\Sigma})$ whose associated inner product as in Def.~\ref{def:quasi-free-kin} is given by $\mu_{L^2}$.
\end{theorem}
\begin{proof}
    By \cite[Prop.~10]{manuceauQuasifreeStatesCCRalgebra1968}, it is sufficient to show that Eq.~\eqref{eq:norm_domination} holds for all $\udl{v}^{C/S}\in V^{C/S}(\ol{\Sigma}/\ol{N})$. This follows directly from the Cauchy-Schwarz inequality.
\end{proof}
We denote the Fock representation of $\Af^{C\oplus S}(\ol{\Sigma})$ associated with $\omega_{L^2}$ by $(\HH_{L^2},\pi_{L^2},\Omega_{L^2})$, which we call the \emph{$L^2$-representation} of $\Af^{C\oplus S}(\ol{\Sigma})$. Given $\varphi_{\ol{\Sigma}}:\Af(\ol{N})\to \Af^{C\oplus S}(\ol{\Sigma})$ the quantum CHH-isomorphism for $\ol{\Sigma}$ of Def.~\ref{def:kin_transf}, we say that $(\HH_{L^2},\pi_{L^2}\circ \varphi_{\ol{\Sigma}},\Omega_{L^2})$ is the \emph{$L^2$-representation at $\ol{\Sigma}$} of $\Af(\ol{N})$.
These representations are the analogue of the trivial Fock representations of a scalar field such as considered for instance in \cite{dimockAlgebrasLocalObservables1980}. Representations of this type demonstrate the existence of Fock representations and are thus technically useful. However they are usually of limited physical utility, as the associated state is not expected to satisfy the Hadamard condition in the interior of $N$. The Hadamard condition plays an important role in the construction of interacting (algebraic) quantum field theories through perturbation theory \cite{brunettiMicrolocalAnalysisInteracting2000,hollandsExistenceLocalCovariant2002,rejznerPerturbativeAlgebraicQuantum2016}, and furthermore singles out a class of preferred Hilbert space representations that are locally equivalent \cite{verchContinuitySymplecticallyAdjoint1997}. By contrast, $L^2$-representations of $\Af(\ol{N})$ associated with different Cauchy surfaces may in general be (locally) unitarily inequivalent. The challenge in constructing a Hadamard state is as follows. Although existence of quasi-free Hadamard states for local observables of electromagnetism is well-established, see e.g.\cite{fewsterQuantumWeakEnergy2003,dappiaggiHadamardStatesVector2013a}, a quasi-free Hadamard state on $\Af_{\loc}(\ol{N})\subset\Af(\ol{N})$ need not be extendable to a quasi-free state on $\Af(\ol{N})$.\footnote{Although it is true that for any state $\omega_{\loc}$ on $\Af_{\loc}(\ol{N})$ there exists a state $\omega$ on $\Af(\ol{N})$ such that $\omega_{\loc}=\omega\restriction_{\Af_{\loc}(\ol{N})}$ as a consequence of the Hahn-Banach theorem, this does not mean that for $\omega_{\loc}$ quasi-free one can always choose $\omega$ to be quasi-free.} One thus needs to construct a quasi-free Hadamard state on $\Af_{\loc}(\ol{N})$ whose behaviour near the corner $\angle\ol{N}$ is sufficiently under control to extend to a quasi-free state on $\Af(\ol{N})$.  Such a construction is beyond the scope of this text, but will be discussed in detail in \cite{fewsterHadamardStatesSemilocal}.

A representation $(\HH,\pi)$ of a $C^*$-algebra $\mathcal{A}$ is \emph{irreducible} if the commutant \begin{equation}
    \pi(\mathcal{A})'=\{b\in B(\HH):[b,\pi(a)]=0\text{ for all }a\in \mathcal{A}\}
\end{equation} 
is trivial (i.e.\ $\pi(\mathcal{A})'=\CC\cdot 1_{\HH}$); a representation is reducible if and only if there exists a non-trivial subspace $\KK\subset \HH$ such that $\pi(\mathcal{A})\KK\subset \KK$, defining a sub-representation $(\pi_\KK,\KK)$ with $\pi_\KK(a)=\pi(a)\restriction_{\KK}$ for every $a\in \mathcal{A}$. In particular the $L^2$ representations defined above provide an example of irreducible representations of $\Af(\ol{N})$, due to \cite[Lem.~A.2]{kayTheoremsUniquenessThermal1991}. An irreducible representation of $\Af(\ol{N})$ need of course not restrict to $\Af_{\textup{loc}}(\ol{N})$ as an irreducible representation, but may instead decompose into (superselection) sectors \cite{buchholzPhysicalStateSpace1982}, that semi-local (edge mode) observables interpolate between. We can see this explicitly through a continuation of Example~\ref{Ex:Pois_deg}. 
\begin{example}
\label{Ex:superselection_decomp}
Consider $\ol{N}$ a finite Cauchy lens with Cauchy surface $\ol{\Sigma}\cong S^2\times [0,1]$ such that $\partial\ol{\Sigma}=C_1\sqcup C_2$ with $C_i\cong S^2$ as in Ex.~\ref{Ex:Pois_deg}, and as illustrated by Fig.~\ref{fig:degen_obs}. Let $\Lambda\in \Omega^0(\ol{N})$ such that $\Lambda\restriction_{C_1}=0$ and $\Lambda\restriction_{C_2}=1$ and $[\diff\Lambda]\in \mathscr{LG}(\ol{N})\cap \Sol_{sc,\mathscr{G}}(\ol{N})$, such that the Weyl operator $W([\diff\Lambda])\in \Af_{loc}(\ol{N})$ can be interpreted as the exponentiated total outgoing electric flux observable at $C_2$ (see Eq.~\eqref{eq:C2-flux}). By Eq.~\eqref{eq:sigmadLambdaAp} it follows that $W([\diff\Lambda])$ commutes with all operators in $\Af_{loc}(\ol{N})$  and is thus an element of the centre of $\Af_{loc}(\ol{N})$. However, one easily finds elements of $\Af(\ol{N})$ that do not commute with $W([\diff\Lambda])$. Choose some $[A]\in \Sol_{\mathscr{G}}(\ol{N})$ for which $\mathbf{E}=-\nml_{\ol{\Sigma}}\diff A$ is such that $\ipc{1}{\nml_{C_2}\mathbf{E}}_{C_2}=1$, then for each $t,s\in \RR$ we have
\begin{equation}
    W(t[\diff\Lambda])W(s[A])=\exp(-its)W(s[A])W(t[\diff\Lambda]).
\end{equation}
It is thus clear that for any non-zero representation $(\pi,\HH)$ of $\Af(\ol{N})$, the 1-parameter family of operators $\pi(W(t[\diff\Lambda]))\in \pi(\Af_{\loc}(\ol{N}))'\cap \pi(\Af_{\loc}(\ol{N}))$ cannot be given by multiples of the unit (except for $t=0$). Moreover, if $(\pi,\HH)$ is a Fock representation (or more generally any representation satisfying the regularity assumptions of the Stone-von Neumann theorem \cite[5.2.15-16]{bratteliOperatorAlgebrasQuantum1997}), we can find a unitary $U:\HH\to L^2(\RR)\otimes \KK$ such that
\begin{equation}
    U\pi(W(t[\diff\Lambda]))U^*=M(t)\otimes 1,\qquad U\pi(W(s[A]))U^*=\rho(s)\otimes 1,
\end{equation}
where $M(t)$ and $\rho(s)$ the phase multiplication operator and right translation action on $L^2(\RR)$ respectively, i.e.~
\begin{equation}
    (M(t)\psi)(\xi)=\exp(it\xi)\psi(\xi),\qquad (\rho(s)\psi)(\xi)=\psi(\xi+s).
\end{equation}
Identifying $L^2(\RR)\otimes \KK=\int_{\RR}^\oplus \KK \diff \xi$ (see \cite[Def.~IV.8.14]{takesakiTheoryOperatorAlgebras2001}), a direct consequence of the centrality of $W(t[\diff\Lambda])\in \Af_{\loc}(\ol{N})$ and \cite[Thm.~IV.8.25]{takesakiTheoryOperatorAlgebras2001} is that there exists a Borel measurable field of representations $\pi_\xi:\Af_{\loc}(\ol{N})\to B(\KK)$ for $\xi\in \RR$ such that
\begin{equation}
    \Ad U\circ\pi\restriction_{\Af_{\loc}(\ol{N})}=\int_{\RR}^\oplus \pi_\xi\diff \xi,
\end{equation}
and $\pi_\xi(W(t[\diff\Lambda]))=\exp(it\xi)1_\KK$ almost everywhere. The individual representations $\pi_\xi$ can be understood as superselected on the total electric flux through the corner component $C_2$, but the individual representations $\pi_\xi$ do not extend to the full algebra of semi-local observables $\Af(\ol{N})$. In particular, the operator $W([A])$ shifts the electric flux through $C_2$, and we have
\begin{equation}
    \Ad U\circ\pi\circ(\Ad W(s[A]))\restriction_{\Af_{\loc}(\ol{N})}=\int_{\RR}^\oplus \pi_{(\xi+s)}\diff \xi.
\end{equation}
\end{example}
This example illustrates that representations of $\Af(\ol{N})$ give rise to a continuum of sectors associated with corner fluxes on $\Af_{\loc}(\ol{N})$. The exact nature of the decomposition into flux sectors depends on the continuity properties of the representation $\pi$, in particular, whether $\pi(W([\diff\Lambda]))$ belongs to the double commutant/weak closure $\pi(\Af_{\loc}(\ol{N}))''$ of $\pi(\Af_{\loc}(\ol{N}))$ for $[\diff\Lambda]\in \mathscr{LG}(\ol{N})$. We shall leave detailed discussions of these analytic features to a future work.

Conversely, example \ref{Ex:superselection_decomp} illustrates that representations of $\Af_{\loc}(\ol{N})$ constituting a single electric flux sector, i.e., where the electric flux densities on $\angle\ol{N}$ agree for all normal states in the representation, in general do not extend to representations of $\Af(\ol{N})$. In fact, we can give a simple necessary condition for extendibility of representations of $\Af_{\loc}(\ol{N})$ to  $\Af(\ol{N})$ that rules out single sector representations, using the concept of background shifts introduced in Prop.~\ref{prop:field_shift}.

\begin{definition}
    Let $\ol{N}$ be a finite Cauchy lens. A representation $(\HH,\pi)$ of $\Af_{\textup{loc}}(\ol{N})$ is \emph{background covariant}
    if background shifts are unitarily implementable, i.e., 
    for each $[\underline{A}]\in \Sol_\mathscr{G}(\ol{N})$ there exists a unitary operator $U_{[\underline{A}]}\in B(\HH)$ such that for each $a\in \Af_{\textup{loc}}(\ol{N})$ we have
    \begin{equation}
        \pi(\alpha([\underline{A}])(a))=U_{[\underline{A}]}\pi(a)U_{[\underline{A}]}^*.
    \end{equation}
\end{definition}
Due to the fact that background shifts are inner automorphisms of $\Af(\ol{N})$, we have the following:
\begin{proposition}\label{prop:int fluxes}
     Let $\ol{N}$ be a finite Cauchy lens, and suppose $(\HH,\pi)$ is a representation of $\Af_{\textup{loc}}(\ol{N})$ which is a restriction of a representation $(\HH,\tilde{\pi})$ of $\Af(\ol{N})$. Then $(\HH,\pi)$ is background covariant.
\end{proposition}
\begin{proof}
Fix $A\in \Sol^J(\ol{N})$; then, as $\alpha([\underline{A}])(a)=W_A([\underline{A}])aW_A([\underline{A}])^*$, the implementation is given by $U_{[\underline{A}]}=\tilde{\pi}(W_A([\underline{A}]))$. 
\end{proof}
Importantly, for finite Cauchy lenses with non-empty corners, fixed flux sectors on $\Af_{\loc}(\ol{N})$ fail to be background covariant, as the shift $[A]$ can always be chosen to have non-vanishing electric flux. Fixed flux sectors have a long history of being used in modelling infrared behaviour of electrodynamics in Minkowski spacetime via infra-vacuum representations \cite{krausModelsInfraredDynamics1977}. The semi-local perspective offered by this paper however clearly favours representations built from a large number of flux sectors. Put more plainly, there is a priori no superselection rule associated with electric fluxes. This is similar to the proposal made in \cite{herdegenSemidirectProductCCR1998} in the context of asymptotic QED, where asymptotic degrees of freedom are quantised such that there is no superselection rule associated with asymptotic electric fluxes. Both our semi-local perspective and the asymptotic perspective of \cite{herdegenSemidirectProductCCR1998} are in contrast with the standard local quantum physics philosophy \cite{buchholzPhysicalStateSpace1982}, as well as the superselection sector structure proposed in  \cite{rielloHamiltonianGaugeTheory2024,rielloNullHamiltonianYangMills2025,rielloSymplecticReductionYangMills2021} (in the classical analogue). We will come back to this discussion in section~\ref{sec: superselection}.

\section{Semi-local observables and quantum reference frames}\label{sec:qrfs}

The existence of edge mode observables implies that large gauge transformations are not redundancies in our formulation of electromagnetism. However, it was shown in~\cite{donnellyLocalSubsystemsGauge2016} that  
large gauge transformations become redundancies in an enlarged theory, which we call
the \emph{surface field extension},
incorporating additional fields on the corner of spacetime that 
transform non-trivially under large gauge transformations and contribute to the pre-symplectic form via a corner term. The large gauge transformations constitute a symmetry of the enlarged theory that can be removed by pre-symplectic reduction. Subsequent work has systematised this observation using geometric and homological tools, see e.g.~\cite{gomesUnifiedGeometricFramework2019,mathieuHomologicalPerspectiveEdge2020,rielloHamiltonianGaugeTheory2024}, and various authors have identified the 
corner fields for electromagnetism and other theories as \emph{(quantum) reference frames} \cite{carrozzaEdgeModesReference2022,kabelQuantumReferenceFrames2023,araujo-regadoSoftEdgesMany2025,araujo-regadoRelationalEntanglementEntropies2025a}.\footnote{Related ideas, invoking an observer, appear in \cite{gomesObserversGhostFieldspace2017}.} This section will make the link precise
using the \emph{operational} QRF framework \cite{caretteOperationalQuantumReference2025}, which we have also employed in a recent study of measurement schemes for QFT in curved spacetimes with symmetries~\cite{fewsterQuantumReferenceFrames2025}.
Other formulations of QRFs in the literature include the \emph{information theoretic} \cite{bartlettReferenceFramesSuperselection2007},  \emph{perspective neutral} \cite{vanrietveldeChangePerspectiveSwitching2020}, and \emph{purely perspectival} \cite{giacominiQuantumMechanicsCovariance2019} approaches.

Our goal is to show that a quantised version of the surface field extension theory can be understood as the result of using the surface fields as a QRF for the large gauge symmetries. To do this, we propose a generalised definition of QRFs adapted to $C^*$-algebras rather than the usual von Neumann algebraic setting. First, however, we describe the classical surface field extension.
Throughout this section, we restrict to the case of vanishing external current, $J=0$, for simplicity.

\subsection{The classical surface field extension}
\label{sec:classical_sfe}
We give a brief account of the classical surface field extension for electromagnetism, showing -- in line with observations made in \cite{rielloSymplecticReductionYangMills2021,gomesQuasilocalDegreesFreedom2021} -- that it is mathematically equivalent to the description of Section~\ref{sec:cov_phas}.  

With $J=0$, the reduced phase space of pure electromagnetism $\Sol_\mathscr{G}(\ol{N})$ on a finite Cauchy lens $\ol{N}$ consists of gauge orbits of solutions in $\Sol(\ol{N})$ modulo $\mathscr{G}(\ol{N})$. 
We introduce an additional surface field $\varphi\in \Omega^0(\angle\ol{N})$ and surface charge density\footnote{The volume density on $\angle\ol{N}$ is used to identify densities and $0$-forms.} $\tau\in  \mathscr{LG}_\angle(\ol{N})$. Recall that, by Prop.~\ref{prop:large_gauge_boundary}, $\tau$ is orthogonal to the space of restrictions to $\angle\ol{N}$ of locally constant functions on $\ol{N}$. Equivalently, $\star_{\angle\ol{N}}\tau$ has vanishing integral over the corner of each connected component of $\ol{N}$; each component carries zero net surface charge. 

The surface field $\varphi$ will function as a reference frame for the large gauge transformations, which in particular requires it to transform non-trivially under large gauge transformations. $\tau$ will function as the generator of the action of large gauge transformations on the surface field (seen more explicitly in the Weyl quantised version of the surface field extension discussed in Sec.~\ref{sec:quantum_sfe}) and as such is sometimes referred to as the charge aspect of the large gauge transformations \cite{freidelExtendedCornerSymmetry2021}. Treated separately, the transformation of the field $\varphi$ is often referred to as a frame reorientation \cite{carrozzaEdgeModesReference2022}. Concretely, for $[A]\in \Sol_{\mathscr{G}}(\ol{N})$, $\varphi \in \Omega^0(\angle\ol{N})$ and $\tau\in  \mathscr{LG}_\angle(\ol{N})$, we distinguish the following three transformations involving a transformation of the electromagnetic field $[A]$ and/or surface field $\varphi$ labelled by a large gauge parameter $\lambda\in \mathscr{LG}_\angle({\ol{N}})$, 
\begin{align}
\label{eq:ac_pas}
    [A]\oplus\varphi\oplus \tau&\mapsto [A]+\mathfrak{G}(\lambda)\oplus \varphi\oplus \tau&\qquad\text{active large gauge transformation}\nonumber\\ 
    [A]\oplus\varphi\oplus \tau&\mapsto [A]\oplus \varphi+\lambda\oplus \tau&\qquad\text{frame reorientation or passive LGT}\nonumber\\
    [A]\oplus\varphi\oplus \tau&\mapsto [A]+\mathfrak{G}(\lambda)\oplus \varphi+\lambda\oplus \tau&\qquad\text{joint large gauge transformation}.
\end{align}
A similar distinction is also made in \cite{donnellyLocalSubsystemsGauge2016} but using different terminology.

The space of surface fields and charge densities,  
$\Omega^0(\angle\ol{N})\oplus \mathscr{LG}_\angle(\ol{N})$, may be equipped with a pre-symplectic form
\begin{equation}\label{eq:surface_sigma}
    (\varphi_1\oplus\tau_1,\varphi_2\oplus\tau_2)\mapsto\ipc{\varphi_1}{\tau_2}_{\angle\ol{N}}-\ipc{\varphi_2}{\tau_1}_{\angle\ol{N}},
\end{equation}
which is degenerate due to the constraint on surface charges. The surface field configuration space is obtained by a pre-symplectic (or gauge) reduction in the same vein as described in Sec.~\ref{sec:cov_phas}, the gauge (or degenerate) directions being given by $\Omega^0_{\diff}(\ol{N})\restriction_{\angle\ol{N}}\oplus \{0\}$. Since by definition $\Omega^0(\angle\ol{N})/\left(\Omega^0_{\diff}(\ol{N})\restriction_{\angle\ol{N}}\right)\cong \mathscr{LG}_\angle(\ol{N})$ by Prop.~\ref{prop:large_gauge_boundary}, the reduced phase space is a copy of $(V^S(\ol{N}),\sigma^S)$.

We shall denote
\begin{equation}
\label{eq:soltilde}
    (\widetilde{\Sol}(\ol{N}),\widetilde{\sigma})=(\Sol_\mathscr{G}(\ol{N})\oplus V^S(\ol{N}), \sigma\oplus\sigma^S).
\end{equation}
Both $\Sol_\mathscr{G}(\ol{N})$ and $V^S(\ol{N})$ transform nontrivially under large gauge transformations. 
Following \cite{donnellyLocalSubsystemsGauge2016}, we pass to 
a combined symplectic space $\Sol_\textnormal{ext}(\ol{N})$ of gauge orbits in $ \widetilde{\Sol}(\ol{N})$ under the \emph{joint large gauge transformations} 
\begin{equation}
\label{eq:jlgd}
    \widetilde{\mathscr{G}}(\ol{N}):=\{(\mathfrak{G}(\lambda)\oplus\lambda\oplus 0)
    \in \widetilde{\Sol}(\ol{N}): \lambda\in \mathscr{LG}_\angle(\ol{N})\}.
\end{equation}
The construction proceeds in two steps. We first define a pre-symplectic space $(\widetilde{\mathscr{G}}(\ol{N})^\perp,\widetilde{\sigma})$ as the symplectic complement of $\widetilde{\mathscr{G}}(\ol{N})$ in $(\widetilde{\Sol}(\ol{N}),\widetilde{\sigma})$,
\begin{multline}
    \widetilde{\mathscr{G}}(\ol{N})^\perp=\{([\udl{A}]\oplus\varphi\oplus\tau)\in\widetilde{\Sol}(\ol{N}):
    \sigma([\udl{A}],\mathfrak{G}(\lambda))+\sigma^S(\varphi\oplus\tau,\lambda\oplus0)=0\\\text{ for all }\lambda\in \mathscr{LG}_\angle(\ol{N})\},
\end{multline}
which can be described explicitly as 
\begin{equation}
\label{eq:Gperp_set}
    \widetilde{\mathscr{G}}(\ol{N})^\perp=\{[\udl{A}]\oplus\varphi\oplus\nml^{(2)}_{\angle\ol{N}}\diff\udl{A}:[\udl{A}]\in \Sol_\mathscr{G}(\ol{N}),\varphi\in \mathscr{LG}_\angle(\ol{N})\},
\end{equation}
where $\nml^{(2)}_{\angle\ol{N}}$ is the corner normal given by Eq.~\eqref{eq:corner_nml}.
Choosing a regular Cauchy surface $\ol{\Sigma}\subset\ol{N}$, setting $\mathbf{E}=-\nml_{\ol{\Sigma}}\diff \udl{A}$ and interpreting $\tau$ as a physical electromagnetic charge on $\angle \ol{N}$, the interface relation $\tau=\nml^{(2)}_{\angle\ol{N}}\diff\udl{A}=-\nml_{\partial\ol{\Sigma}}\mathbf{E}$
implies by Gauss' law that $\nml_{\partial(\tilde{\Sigma}\setminus\ol{\Sigma})}\mathbf{E}=0$, where $\tilde{\Sigma}$ is any Cauchy surface extending $\ol{\Sigma}$ in a globally hyperbolic extension of $\ol{N}$.

Arguing as in Prop.~\ref{prop:sym_red}, we find that the space of degenerate directions for $(\widetilde{\mathscr{G}}(\ol{N})^\perp,\tilde{\sigma})$ is $\widetilde{\mathscr{G}}(\ol{N})$,\footnote{This amounts to the statement that $\widetilde{\mathscr{G}}(\ol{N})$ is a Lagrangian subspace of $\widetilde{\Sol}(\ol{N})$.} and we define the \emph{surface field extended phase space} as the gauge-reduced phase space 
\begin{equation}
\label{eq:sfeps}
    \Sol_\textnormal{ext}(\ol{N}) = \widetilde{\mathscr{G}}(\ol{N})^\perp/\widetilde{\mathscr{G}}(\ol{N}).
\end{equation}
The space $\Sol_\textnormal{ext}(\ol{N})$ is referred to in~\cite{donnellyLocalSubsystemsGauge2016} as a \emph{fusion product} of $\Sol_\mathscr{G}(\ol{N})$ and $V^S(\ol{N})$ with respect to $\widetilde{\mathscr{G}}(\ol{N})$. Alternatively, it can be obtained as a single symplectic leaf of the quotient $\widetilde{\Sol}(\ol{N})/\widetilde{\mathscr{G}}(\ol{N})$.
The gauge orbit of $[\udl{A}]\oplus\varphi\oplus\tau\in\widetilde{\mathscr{G}}(\ol{N})^\perp$ is denoted $[[\udl{A}]\oplus \varphi\oplus \tau]\in\Sol_\textnormal{ext}(\ol{N})$. Although the construction of $\Sol_\textnormal{ext}(\ol{N})$ introduced additional degrees of freedom at intermediate steps, the symplectic spaces $(\Sol_\textnormal{ext}(\ol{N}),\widetilde{\sigma})$ and $(\Sol_\mathscr{G}(\ol{N}),\sigma)$ are equivalent. 

\begin{proposition}\label{prop:class_edge_equiv}
    Define a linear map $\hat{E}:\Sol_{\mathscr{G}}(\ol{N})\to \widetilde{\Sol}(\ol{N})$ by
    \begin{equation}
    \hat{E}([\udl{A}])= [\udl{A}]\oplus 0\oplus\nml^{(2)}_{\angle\ol{N}}\diff\udl{A}.
    \end{equation}
    Then $\hat{E}$ is symplectic and its range $\widetilde{\Sol}_0(\ol{N}):=\hat{E}(\Sol_{\mathscr{G}}(\ol{N}))$ (equipped with the restriction of $\widetilde{\sigma}$) is contained in $\widetilde{\mathscr{G}}(\ol{N})^\perp$. Consequently, there are unique symplectic maps $E$ and $\check{E}$ so that the diagram 
    \begin{equation}
\label{eq:Emap_diag2}
    \begin{tikzcd}
        && \widetilde{\Sol}(\ol{N})&\\
        \Sol_{\mathscr{G}}(\ol{N})\arrow[r,"E"] \arrow[drr,bend right=10,"\check{E}"',"\cong"] \arrow[urr,bend left=10,"\hat{E}"]&\widetilde{\mathscr{G}}(\ol{N})^\perp\arrow[ur,"\hat{\iota}"] \arrow[rr,"q"] &&\Sol_\textnormal{ext}(\ol{N}) \\
        &&\widetilde{\Sol}_{0}(\ol{N})\arrow[ul,"\check{\iota}"]\arrow[ur,"q\circ\check{\iota}","\cong"']&
    \end{tikzcd}
\end{equation}
commutes, where $\hat{\iota}$ and $\check{\iota}$ are subspace inclusions, $\check{E}$ is an isomorphism and $q:\widetilde{\mathscr{G}}(\ol{N})^\perp\to \Sol_\textnormal{ext}(\ol{N})$ is the quotient map. Moreover, $\check{E}^{-1}$ is the restriction of the canonical projection $\Sol_\mathscr{G}(\ol{N})\oplus V^S(\ol{N})\to \Sol_\mathscr{G}(\ol{N})$ to $\widetilde{\Sol}_0(\ol{N})$, and $q\circ\check{\iota}$ is an isomorphism with inverse $c$ given by
\begin{equation}\label{eq:cdef}
        c([[\udl{A}]\oplus\varphi\oplus \nml^{(2)}_{\angle\ol{N}}\diff\udl{A}])
        = [\udl{A}]\oplus\varphi\oplus \nml^{(2)}_{\angle\ol{N}}\diff\udl{A} - 
        \widetilde{\mathfrak{G}}(\varphi) =
        ([\udl{A}]-\mathfrak{G}(\varphi))\oplus 0\oplus \nml^{(2)}_{\angle\ol{N}}\diff\udl{A},
    \end{equation}
    where $\widetilde{\mathfrak{G}}(\lambda) = (\mathfrak{G}(\lambda)\oplus\lambda\oplus 0)\in\widetilde{\mathscr{G}}(\ol{N})\subset\widetilde{\mathscr{G}}(\ol{N})^\perp$.
    In other words, $c$ maps any gauge orbit in $\Sol_\textnormal{ext}(\ol{N})$ to its unique representative with vanishing surface field, while
    $c\circ q$ is a gauge-fixing map that gauges away the surface field.
\end{proposition}
\begin{proof}
    Direct calculation shows that $\hat{E}$ is symplectic, and by inspection
    we see that it maps bijectively onto
    \begin{equation}
        \widetilde{\Sol}_{0}(\ol{N})=\{[\udl{A}]\oplus 0\oplus\nml^{(2)}_{\angle\ol{N}}\diff\udl{A}:[\udl{A}]\in \Sol_{\mathscr{G}}(\ol{N})\},
    \end{equation} i.e.~the subspace of $\widetilde{\Sol}(\ol{N})$ with vanishing surface field. The existence and uniqueness of $E$ and $\check{E}$ is immediate, as is commutativity of the diagram, and $E$, $\check{E}$, $\hat{\iota}$ and $\check{\iota}$ are all symplectic once $\widetilde{\Sol}_0(\ol{N})$ is endowed with the restriction of $\widetilde{\sigma}$. The invertibility of $\check{E}$ with the stated inverse is straightforward. It remains to check that $c$ is indeed the inverse of $q\circ\check{\iota}$. Let $([\udl{A}]\oplus 0\oplus\nml^{(2)}_{\angle\ol{N}}\diff\udl{A})\in  \widetilde{\Sol}_0(\ol{N})$, then
    \begin{equation}
        (c\circ q\circ\check{\iota})([\udl{A}]\oplus 0\oplus\nml^{(2)}_{\angle\ol{N}}\diff\udl{A})=c[[\udl{A}]\oplus 0\oplus\nml^{(2)}_{\angle\ol{N}}\diff\udl{A}]=[\udl{A}]\oplus 0\oplus\nml^{(2)}_{\angle\ol{N}}\diff\udl{A}-\widetilde{\mathfrak{G}}(0)=[\udl{A}]\oplus 0\oplus\nml^{(2)}_{\angle\ol{N}}\diff\udl{A},
    \end{equation}
    and let $[[\udl{A}]\oplus \varphi\oplus\nml^{(2)}_{\angle\ol{N}}\diff\udl{A}]\in  \Sol_{\textnormal{ext}}(\ol{N})$, then
    \begin{equation}
        (q\circ\check{\iota}\circ c)[[\udl{A}]\oplus \varphi\oplus\nml^{(2)}_{\angle\ol{N}}\diff\udl{A}]=[[\udl{A}]\oplus \varphi\oplus\nml^{(2)}_{\angle\ol{N}}\diff\udl{A}-\widetilde{\mathfrak{G}}(\varphi)]=[[\udl{A}]\oplus \varphi\oplus\nml^{(2)}_{\angle\ol{N}}\diff\udl{A}],
    \end{equation}
    where we have used that $\widetilde{\mathfrak{G}}(\varphi)\in \widetilde{\mathscr{G}}(\ol{N})$.
    We conclude that $c=(q\circ\check{\iota})^{-1}$.
\end{proof}
Since $q\circ E=q\circ\check{\iota}\circ\check{E}$ is a symplectomorphism, the classical description of pure electromagnetism on a finite Cauchy lens is mathematically equivalent to its surface field extension.

Choosing a regular Cauchy surface $\ol{\Sigma}\subset\ol{N}$, note that the composition of $\hat{E}$ with the inverse CHH-isomorphism takes the simple form
\begin{equation}\label{eq:EhatCHH}
    (\hat{E}\circ\mathfrak{K}_{\ol{\Sigma}}^{-1})(F\oplus H\oplus f\oplus h) = 
    \mathfrak{K}_{\ol{\Sigma}}^{-1}(F\oplus H\oplus f\oplus h) \oplus 0 \oplus - h
\end{equation}
for all $F\oplus H\in V^C(\ol{\Sigma})$ and $f\oplus h\in V^S(\ol{N})$. 
This is because the last component of $\mathfrak{K}_{\ol{\Sigma}}(\udl{A})$ is  $h=-\nml^{(2)}_{\angle\ol{N}}\diff\udl{A}=\nml_{\partial\ol{\Sigma}}\mathbf{E}$ where $\mathbf{E}=-\nml_{\ol{\Sigma}}\diff \udl{A}$ -- see Eq.~\eqref{eq:KJSigma}.

\subsection{The quantised surface field extension}\label{sec:quantum_sfe}

We now formulate the surface field extension at the level of quantised observables. The quantum analogue of the classical observables obtained by symplectic smearings with $\widetilde{\mathscr{G}}(\ol{N})^\perp\subset \Sol_\mathscr{G}(\ol{N})\oplus V^S(\ol{N})$ will be a subalgebra $\widetilde{\Af}^{\mathscr{LG}} (\ol{N})\subset\Af(\ol{N})\otimes\Af^S(\ol{N})$. Here, 
$\Af^S(\ol{N})=\Weyl(V^S(\ol{N}),\sigma^S)$ as in Def.~\ref{def:bulk_corn_weyls}, but with the generators $W^S(f\oplus h)$ for $(f\oplus h)\in V^S(\ol{N})$ interpreted via
\begin{equation}
    W^S(f\oplus h)\simeq \exp(i(\tau(f)-\varphi(h))),
\end{equation}
where $\varphi(h)$ and $\tau(f)$ are smearings of the surface field and surface charge density respectively.

The joint large gauge transformations, given classically by $[\udl{A}]\oplus\varphi\oplus\tau\mapsto ([\udl{A}]+\mathfrak{G}(\lambda))\oplus(\varphi+\lambda)\oplus\tau$ for $\lambda\in\mathscr{LG}_\angle(\ol{N})$, are implemented 
on $\widetilde{\Af}(\ol{N}):=\Af(\ol{N})\otimes \Af^S(\ol{N})$
as follows.
\begin{definition}
\label{def:bound_gauge_gen}
    Let $\ol{\Sigma}\subset \ol{N}$ be a regular Cauchy surface with boundary in a finite Cauchy lens $\ol{N}$. For each $\lambda\in \mathscr{LG}_\angle(\ol{N})$ 
    the \emph{frame reorientation} $\alpha_{\mathscr{LG}}^S(\lambda)=\Ad U_{\mathscr{LG}}^S(\lambda)\in \Aut(\Af^S(\ol{N}))$ is unitarily implemented by 
    \begin{equation}
        U_{\mathscr{LG}}^S(\lambda)=W^S(\lambda\oplus 0)\in \Af^S(\ol{N}),
    \end{equation}
    and the \emph{joint large gauge automorphism} 
    \begin{equation}
        \widetilde{\alpha}_{\mathscr{LG}}(\lambda)=\alpha_{\mathscr{LG}}(\lambda)\otimes \alpha_{\mathscr{LG}}^S(\lambda)=\Ad \widetilde{U}_{\mathscr{LG}}(\lambda) \in \Aut(\widetilde{\Af}(\ol{N}))
    \end{equation}
    is unitarily implemented by $\widetilde{U}_{\mathscr{LG}}(\lambda)=U_{\mathscr{LG}}(\lambda)\otimes U_{\mathscr{LG}}^S(\lambda)$ with $U_{\mathscr{LG}}(\lambda)$ as in Def. \ref{def:q_lar_gauge_transf}. The subalgebra of invariants under joint large gauge automorphisms is 
    \begin{equation}
        \widetilde{\Af}^{\mathscr{LG}}(\ol{N}):=(\widetilde{\Af}(\ol{N}))^{\widetilde{\alpha}_{\mathscr{LG}}}=\{a\in \widetilde{\Af}(\ol{N}):\widetilde{\alpha}_{\mathscr{LG}}(\lambda)(a)=a\text{ for all }\lambda\in \mathscr{LG}_\angle(\ol{N})\}.
    \end{equation}
\end{definition}
Note that due to the invariance of any $a\in \widetilde{\Af}^{\mathscr{LG}}(\ol{N})$ under $\widetilde{\alpha}_{\mathscr{LG}}(\lambda)$, we have 
\begin{equation}
    (\alpha_{\mathscr{LG}}(\lambda)\otimes \id)(a)=(\id\otimes \alpha_{\mathscr{LG}}^S(-\lambda))(a).
\end{equation}
In other words, the adjoint action of the unitaries $U_{\mathscr{LG}}(\lambda)\otimes 1$ and $1\otimes U_{\mathscr{LG}}^S(-\lambda)$ on $a$ coincide. This is simply an instance of the equivalence of \emph{active} and \emph{passive} transformations (in this case of the boundary gauge potential) for physical quantities defined relative to a reference frame. In this language, we view the large gauge transformations on the electromagnetic field as active transformations, and the frame reorientations as passive following Eq.~\eqref{eq:ac_pas}.

The unitaries $U_{\mathscr{LG}}(\lambda)$ and $U_{\mathscr{LG}}^S(\lambda)$ can be interpreted, respectively, as the exponential of the electric flux through the corner and the exponential of the surface charge density, smeared with $\lambda$ in each case. Their combination $\widetilde{U}_{\mathscr{LG}}(\lambda)$ can be interpreted as the exponential of the smeared electric flux outside the boundary. 
The algebra $\widetilde{\Af}^{\mathscr{LG}}(\ol{N})$ comprises the joint large gauge invariant observables of the electromagnetic field on $\ol{N}$ and the surface field and charge density on $\angle{\ol{N}}$.

We will characterise $\widetilde{\Af}^{\mathscr{LG}}(\ol{N})$ using the language of quantum reference frames, with the idea that the surface fields provide a QRF for the joint large-gauge transformations. We show that each electromagnetic field observable $a\in\Af(\ol{N})$ has a corresponding invariant joint observable $\yen(a)\in\widetilde{\Af}^{\mathscr{LG}}(\ol{N})$ that may be regarded as a dressing of $a$ with surface observables, or a relativisation of $a$ with respect to large gauge transformations on $\widetilde{\Af}(\ol{N})$. The relativisation map $\yen$ is an injective $*$-homomorphism and defines an isomorphism between $\Af(\ol{N})$ and its image
$\Af_{\rel}(\ol{N})=\yen(\Af(\ol{N}))\subset \widetilde{\Af}^{\mathscr{LG}}(\ol{N})$. 
Heuristically, the edge mode observables (i.e.~non large gauge invariant observables) in $\Af(\ol{N})$ can be understood as smeared Wilson lines with endpoints on $\angle\ol{N}$ (see Sec.~\ref{sec:edge_mode_obs}). In that perspective, the relativisation map $\yen$ turns edge mode observables into joint large-gauge invariant observables of the electromagnetic and surface field, by attaching the appropriate surface field observables to the end-points of of these Wilson lines.

To construct $\yen$, we concretely identify $\Af(\ol{N})$ with $\Weyl(\Sol_\mathscr{G}(\ol{N}))$, so that a standard Weyl generator $W([\udl{A}])$ in the latter is identified with $W_0([\udl{A}])$ in the former, exploiting the fact that the vanishing solution provides a natural background when $J=0$. Then let  
$\hat{\yen}=\Weyl(\hat{E})$ be the quantisation of the symplectic map $\hat{E}:\Sol_{\mathscr{G}}(\ol{N})\to \Sol_{\mathscr{G}}(\ol{N})\oplus V^S(\ol{N})$ from  Prop.~\ref{prop:class_edge_equiv}. Thus $\hat{\yen}$ is
an injective ${}^*$-homomorphism from the algebra of bulk observables $\Af(\ol{N})=\Weyl(\Sol_{\mathscr{G}}(\ol{N}),\sigma)$ to the joint algebra 
\begin{equation}
\label{eq:atilde}
\widetilde{\Af}(\ol{N})=\Weyl(\widetilde{\Sol}(\ol{N}),\widetilde{\sigma})= \Weyl(\Sol_{\mathscr{G}}(\ol{N}),\sigma)\otimes\Weyl(V^S(\ol{N}), \sigma^S).
\end{equation}
Denoting the generators of $\widetilde{\Af}(\ol{N})$ by $\widetilde{W}(\udl{v})$ for $\udl{v}\in \widetilde{\Sol}(\ol{N})$, 
the map $\hat{\yen}$ acts on generators as
\begin{equation}
\hat{\yen}(W([\udl{A}]))= \widetilde{W}(\hat{E}([\udl{A}])) = \widetilde{W}([\udl{A}]\oplus 0\oplus \nml^{(2)}_{\angle\ol{N}}\diff\udl{A}) ,
\end{equation}
or, in tensor product form (and reinserting a general background),
\begin{equation}
\label{eq:C*QRF}
    \hat{\yen}(W_A([\udl{A}]))=W_A([\udl{A}])\otimes W^S(0\oplus \nml^{(2)}_{\angle\ol{N}}\diff\udl{A}),
\end{equation}
for $(A,\udl{A})\in T\Sol^J(\ol{N})$.  In particular, it follows that $\hat{\yen}(U_{\mathscr{LG}}(\lambda))=U_{\mathscr{LG}}(\lambda)\otimes 1$ for $\lambda\in\mathscr{LG}_{\angle}(\ol{N})$.
Crucially, every 
$\hat{\yen}(a)\in \widetilde{\Af}(\ol{N})$ is invariant under the action $\tilde{\alpha}_{\mathscr{LG}}$.
\begin{lemma}
\label{lem:relativise}
    The range $\Af_{\rel}(\ol{N}):=\hat{\yen}(\Af(\ol{N}))$ is a subalgebra of $\widetilde{\Af}^{\mathscr{LG}}(\ol{N})$.
    Consequently there are unique unital $*$-homomorphisms $\yen$ and $\check{\yen}$ so that the diagram
    \begin{equation}
\label{eq:Ymap_diag}
    \begin{tikzcd}
        && \widetilde{\Af}(\ol{N})\\
        \Af(\ol{N})\arrow[r,"\yen"] \arrow[drr,bend right=10,"\check{\yen}","\cong"'] \arrow[urr,bend left=10,"\hat{\yen}"]&\widetilde{\Af}^{\mathscr{LG}}(\ol{N})\arrow[ur,"\hat{\jmath}"] &\\
        &&\Af_{\rel}(\ol{N})\arrow[ul,"\check{\jmath}"]
    \end{tikzcd}
\end{equation} 
commutes, in which $\check{\yen}$ is a ${}^*$-isomorphism, and $\hat{\jmath}$ and $\check{\jmath}$ are the subalgebra inclusions.  
\end{lemma}
\begin{proof}
    It us sufficient to check that for each $\lambda\in \mathscr{LG}_\angle(\ol{N})$ and $(A,\udl{A})\in T\Sol^J(\ol{N})$, we have $\tilde{\alpha}_{\mathscr{LG}}(\lambda)(\hat{\yen}(W_A([\udl{A}])))=W_A([\udl{A}])$.
    Note that for general $(f\oplus h)\in V^S(\ol{N})$, we have
    \begin{equation}
    \label{eq:LG_gen_trasf}
        \widetilde{\alpha}_{\mathscr{LG}}(\lambda)(W_A([\udl{A}])\otimes W^S(f\oplus h) )=\exp(i\ipc{\lambda}{\nml^{(2)}_{\angle\ol{N}}\diff\udl{A}-h})W_A([\udl{A}])\otimes W^S(f\oplus h).
    \end{equation}
    Hence in particular
    \begin{align}
        \tilde{\alpha}_{\mathscr{LG}}(\lambda)(\hat{\yen}(W_A([\udl{A}])))=&\tilde{\alpha}_{\mathscr{LG}}(\lambda)(W_A[\udl{A}]\otimes W^S(0\oplus \nml^{(2)}_{\angle\ol{N}}\diff\udl{A}))\nonumber\\=&W_A([\udl{A}])\otimes W^S(0\oplus \nml^{(2)}_{\angle\ol{N}}\diff\udl{A})=\hat{\yen}(W_A([\udl{A}])).
    \end{align}
    The existence of $\yen$ and $\check{\yen}$ and the commutativity of~\eqref{eq:LG_gen_trasf} follow immediately.
\end{proof}

In Sec~\ref{sec:POVMs} we make precise in what sense $\yen:\Af(\ol{N})\to\widetilde{\Af}^{\mathscr{LG}}(\ol{N})$ is a relativisation map, showing that this map can be recovered as a limit of relativisation maps of the type defined in \cite{loveridgeSymmetryReferenceFrames2018}. We refer to its image $\Af_{\rel}(\ol{N})$ as the \emph{algebra of relativised observables}.
Although defined at the level of abstract Weyl generators, the interpretation of $\yen$ becomes more apparent when we recall that formally,
\begin{equation}
    W([\udl{A}])\simeq \exp\left(i(\mathbf{E}(H)-\mathbf{A}(F))\right),\qquad W^S(f,h)\simeq \exp\left(i(\tau(f)-\varphi(h))\right),
\end{equation}
where, given some choice of regular Cauchy surface $\ol{\Sigma}\subset\ol{N}$,  $H=\udl{A}\restriction_{\ol{\Sigma}}$ and $F=-\nml_{\ol{\Sigma}}\diff\udl{A}$ as in Eq.~\eqref{eq:smearedAE}, $\mathbf{A}$ and $\mathbf{E}$ the vector potential and electric field at $\ol{\Sigma}$ and $\varphi$ and $\tau$ the surface field and surface charge density and $\partial\ol{\Sigma}=\angle\ol{N}$.
We may then formally write
\begin{equation}
    \yen\left(\mathbf{E}(H)-\mathbf{A}(F)\right)=\mathbf{E}(H)-\mathbf{A}(F)+\varphi(\nml_{\partial\ol{\Sigma}}F),
\end{equation}
which can be made precise at the level of sufficiently regular representations of $\Af(\ol{N})$ and $\Af^S(\ol{N})$. Viewing $\mathbf{A}(F)$ as a smeared analogue of a Wilson line as discussed in Sec.~\ref{sec:semlocsymp}, we may say that at each boundary endpoint of a Wilson line, $\yen$ dresses the line with surface a field operator $\varphi$, as sketched in Fig.~\ref{fig:Wilson_relativise_endpoints}. We stress, however, that while these interpretative comments may help build intuition and comparisons with other works, they are not assumed within the rigorous development.

Having obtained the relativisation map $\yen$, we can characterise the algebra of joint large gauge invariant observables as the subalgebra generated by the relativised observables and the implementers of frame reorientations. The following result is analogous to the description of invariant von Neumann algebras in the presence of quantum reference frames as discussed in \cite{fewsterQuantumReferenceFrames2025}. This analogy will be made more precise in  Sec.~\ref{sec:POVMs},
\begin{theorem}
\label{thm:Weyl_comm_thm}
    The algebra $\widetilde{\Af}^{\mathscr{LG}}(\ol{N})\subset \widetilde{\Af}(\ol{N})$ is the smallest algebra containing all observables of the form
    \begin{equation}
        \yen(W_A([\udl{A}]))(1\otimes U_{\mathscr{LG}}^S(\lambda)),
    \end{equation}
    for $(A,\udl{A})\in T\Sol^J(\ol{N})$, $\lambda\in \mathscr{LG}_\angle(\ol{N})$. 
    Equivalently, $\widetilde{\Af}^{\mathscr{LG}}(\ol{N})$ is generated by 
    $\{\widetilde{W}(\udl{v}):\udl{v}\in \widetilde{\mathscr{G}}(\ol{N})^\perp\}$.
\end{theorem}
\begin{proof}
    Noting that, for $\udl{v}\in \widetilde{\Sol}(\ol{N})$ and $\lambda\in \mathscr{LG}_\angle(\ol{N})$, we have
    \begin{equation}
        \widetilde{\alpha}_{\mathscr{LG}}(\lambda)(\widetilde{W}(\udl{v}))=\exp(-i\widetilde{\sigma}(\mathfrak{G}(\lambda)\oplus\lambda\oplus 0,\udl{v}))\widetilde{W}(\udl{v}),
    \end{equation}
    it follows from Lem.~\ref{lem:Weyl_fixedpoint} that $\widetilde{\Af}^{\mathscr{LG}}(\ol{N})=(\widetilde{\Af}(\ol{N}))^{\widetilde{\alpha}_{\mathscr{LG}}}$ is generated the set of Weyl generators $\widetilde{W}(\udl{v})$ for $\udl{v}\in \mathscr{G}^\perp(\ol{N})$. Writing these Weyl generators in tensor product form and using Eq.~\eqref{eq:Gperp_set}, we find that $\widetilde{\Af}^{\mathscr{LG}}(\ol{N})$ is the smallest algebra containing all operators of the form $W_A([\udl{A}])\otimes W^S(f\oplus h)$ for which $(A,\udl{A})\in T\Sol^J(\ol{N})$, $f\in \mathscr{LG}_\angle(\ol{N})$ and $h=\nml^{(2)}_{\angle\ol{N}}\diff\udl{A}$. 
    We now compute
    \begin{align}
        W_A([\udl{A}])\otimes W^S(f\oplus h)=&\exp\left(-\frac{i}{2}\ipc{f}{h}\right)W_A([\udl{A}])\otimes (W^S(0\oplus h) W^S(f\oplus 0))\nonumber\\
        =&\exp\left(-\frac{i}{2}\ipc{f}{h}\right)(W_A([\udl{A}])\otimes W^S(0\oplus h) )(1\otimes W^S(f\oplus 0))\nonumber\\
        =&\exp\left(-\frac{i}{2}\ipc{f}{h}\right)\yen(W_A([\udl{A}])) (1\otimes U_{\mathscr{LG}}^S(f)),
    \end{align}
    which shows that $\widetilde{\Af}^{\mathscr{LG}}(\ol{N})$ is the smallest algebra containing all operators of the form $\yen(W_A([\udl{A}]))(1\otimes U_{\mathscr{LG}}^S(f))$ for $(A,\udl{A})\in T\Sol^J(\ol{N})$, $f\in \mathscr{LG}_\angle(\ol{N})$.
\end{proof}
 
The algebras $\Af_{\rel}(\ol{N})$ and $\widetilde{\Af}^{\mathscr{LG}}(\ol{N})$ may be identified with Weyl algebras of the pre-symplectic spaces appearing in Prop.~\ref{prop:class_edge_equiv}.
\begin{theorem}
    \label{thm:At_Weyl_q}
    There are unique $C^*$-isomorphisms $\hat{\xi}$ and $\check{\xi}$ making the following diagram commute:
    \begin{equation}  
    \begin{tikzcd}
      \Af_{\rel}(\ol{N})\arrow[r,"\check{\jmath}"] & \widetilde{\Af}^{\mathscr{LG}}(\ol{N}) \arrow[r,"\hat{j}"]  & \widetilde{\Af}(\ol{N})\arrow[d,equal]\\
      \Weyl(\widetilde{\Sol}_0(\ol{N}))\arrow[u,"\check{\xi}"]\arrow[r,"\Weyl(\check{\iota})"] &  \Weyl(\widetilde{\mathscr{G}}(\ol{N})^\perp)\arrow[u,"\hat{\xi}"]\arrow[r,"\Weyl(\hat{\iota})"] &  
      \Weyl(\widetilde{\Sol}(\ol{N})),
    \end{tikzcd}
\end{equation} 
    where $\check{\iota}$, $\hat{\iota}$ are the inclusion maps as in Prop.~\ref{prop:class_edge_equiv} and we have suppressed the pre-symplectic forms from the notation.
\end{theorem}
\begin{proof}
    Theorem~\ref{thm:Weyl_comm_thm} shows that $\widetilde{\Af}^{\mathscr{LG}}(\ol{N})$ is the closure of 
    \begin{equation}
    \Delta(\hat{\iota})(\Delta(\widetilde{\mathscr{G}}(\ol{N})^\perp,\widetilde{\sigma}))\subset \Delta(\widetilde{\Sol}(\ol{N}), \widetilde{\sigma})
    \end{equation}
    in the Weyl algebra norm.
    
    If $(V,\sigma)$ is a pre-symplectic space, the $2$-norm on $\Delta(V,\sigma)$ is defined by $\|\sum_j z_j W(v_j)\|_2^2=\sum_j |z_j|^2$ for finite sets of distinct elements $v_j\in V$ and arbitrary coefficients $z_j\in\CC$. By Prop.~3.9 of~\cite{binzConstructionUniquenessWeyl2004}, the norm of $\Weyl(\widetilde{\Sol}(\ol{N}), \widetilde{\sigma})$ dominates the $2$-norm on $\Delta(\widetilde{\Sol}(\ol{N}), \widetilde{\sigma})$ and hence the same is true for the inherited norm pulled back to $\Delta(\widetilde{\mathscr{G}}(\ol{N})^\perp,\widetilde{\sigma})$ by $\Delta(\hat{\iota})$, which consequently coincides with the norm of $\Weyl(\widetilde{\mathscr{G}}(\ol{N})^\perp,\widetilde{\sigma})$ by Cor.~3.11 of~\cite{binzConstructionUniquenessWeyl2004}.
    It follows that there is a unique unital $C^*$-isomorphism $\hat{\xi}:\Weyl(\widetilde{\mathscr{G}}(\ol{N})^\perp,\widetilde{\sigma})\to \widetilde{\Af}^{\mathscr{LG}}(\ol{N})$ so that
    $\Weyl(\hat{\iota})=\hat{\jmath}\circ\hat{\xi}$, given by the continuous extension of $\Delta(\hat{\iota})$.

    Finally, $\Af_\rel(\ol{N})$ is defined as the range of $\yen$ and is therefore generated by $\{\widetilde{W}(v):v\in\widetilde{\Sol}_0(\ol{N})\}$. The existence of the required map $\check{\xi}$ is established in the same way as for $\hat{\xi}$.
\end{proof}

The diagram~\eqref{eq:Ymap_diag} has the same structure as part of~\eqref{eq:Emap_diag2}, with $\Af_{\rel}(\ol{N})$ playing the role of $\widetilde{\Sol}_0(\ol{N})$.
We now show that the remainder of~\eqref{eq:Emap_diag2} also admits a quantum analogue, yielding a commuting diagram of unital $C^*$-algebras
\begin{equation}
\label{eq:Ymap_diag2}
    \begin{tikzcd}
        && \Af(\ol{N})\otimes \Af^S(\ol{N})&\\
        \Af(\ol{N})\arrow[r,"\yen"] \arrow[drr,bend right=10,"\cong"',"\check{\yen}"] \arrow[urr,bend left=10,"\hat{\yen}"]&\widetilde{\Af}^{\mathscr{LG}}(\ol{N})\arrow[ur,"\hat{\jmath}"] 
        \arrow[rr,"\won"]&&\widetilde{\Af}^{\mathscr{LG}}(\ol{N})/I \\
        &&\Af_{\rel}(\ol{N})\arrow[ul,"\check{\jmath}"]\arrow[ur,"\won\circ \check{\jmath}","\cong"']&
    \end{tikzcd}\,,
\end{equation}
where $I$ is a closed two-sided $*$-ideal, $\won$ is the quotient map and $\won\circ\check{\jmath}$ is an isomorphism. To this end, we construct a surjective $*$-homomorphism
\[
\Gamma:\widetilde{\Af}^{\mathscr{LG}}(\ol{N})\to 
\Af_{\rel}(\ol{N})\]
such that $\Gamma\circ\check{\jmath}=\id$. The claim then follows with $I=\ker \Gamma$ by the first isomorphism theorem.

We will demonstrate in Sec.~\ref{sec: superselection} that $\Gamma$ can be interpreted as a map onto the observables associated with a superselection sector of $\widetilde{\Af}^{\mathscr{LG}}(\ol{N})$.

To construct $\Gamma$, note that the gauge-fixing map $c\circ q$ (gauging away the surface degrees of freedom) from Prop.~\ref{prop:class_edge_equiv} is a symplectic map from $\widetilde{\mathscr{G}}(\ol{N})^\perp$ to $\widetilde{\Sol}_0(\ol{N})$ providing a left inverse to the inclusion $\check{\iota}:\widetilde{\Sol}_0(\ol{N})\to \widetilde{\mathscr{G}}(\ol{N})^\perp$. Using the quantisation of this gauge-fixing map, we can caonstruct $\Gamma$ by
\begin{equation}
    \Gamma  =  \check{\xi}\circ \Weyl(c\circ q)\circ\hat{\xi}^{-1} ,
\end{equation}
and loosely speaking we can think of it as a ``quantum gauge fixing map'', whose main properties are given as follows.
\begin{theorem}
\label{thm:Gamma}
The map $\Gamma  =  \check{\xi}\circ \Weyl(c\circ q)\circ\hat{\xi}^{-1}: \widetilde{\Af}^{\mathscr{LG}}(\ol{N})\to 
\Af_{\rel}(\ol{N})$ is a surjective unital $C^*$-homomorphism obeying $\Gamma\circ\check{\jmath}=\id$, whose kernel is the closed two-sided $*$-ideal
generated by $\{\widetilde{U}_{\mathscr{LG}}(\lambda)-1\otimes 1:\lambda\in \mathscr{LG}_\angle(\ol{N})\}$.
In particular, for every $\lambda\in \mathscr{LG}_\angle(\ol{N})$, one has $\Gamma(\widetilde{U}_{\mathscr{LG}}(\lambda))=1$. 
\end{theorem}
\begin{proof}
    Surjectivity of $\Gamma$ follows from that of $c\circ q$, and the left-inverse property follows because
    \begin{equation} 
    \Gamma\circ \check{\jmath} = \check{\xi}\circ \Weyl(c\circ q)\circ\hat{\xi}^{-1} \circ \check{\jmath} = \check{\xi}\circ \Weyl(c\circ q) \circ \Weyl(\check{\iota}) \circ \check{\xi}^{-1}= \id,
    \end{equation}
    where we have used $c\circ q\circ \check{\iota}=\id$.

    Now let $I$ be the smallest closed two-sided ideal of $\widetilde{\Af}^{\mathscr{LG}}(\ol{N})$ containing $\{\widetilde{U}_{\mathscr{LG}}(\lambda)-1\otimes 1:\lambda\in \mathscr{LG}_\angle(\ol{N})\}$. We demonstrate $I=\ker \Gamma$ by showing that each side is a subspace of the other.
    To verify that $I\subset\ker(\Gamma)$, note first that $\widetilde{U}_{\mathscr{LG}}(\lambda) = \widetilde{W}(\mathfrak{G}(\lambda)\oplus\lambda\oplus 0)= \hat{\xi}(\widetilde{W}^\perp(\widetilde{\mathfrak{G}}(\lambda)))$,
    where $\widetilde{\mathfrak{G}}:\mathscr{LG}_\angle(\ol{N})\to\widetilde{\mathscr{G}}(\ol{N})^\perp$ by $\widetilde{\mathfrak{G}}(\lambda)=\mathfrak{G}(\lambda)\oplus\lambda\oplus 0$ as in Prop.~\ref{prop:class_edge_equiv} and we use $\widetilde{W}^\perp$ to label the Weyl generators of $\Weyl(\widetilde{\mathscr{G}}(\ol{N})^\perp)$
    As 
    $(c\circ q)(\widetilde{\mathfrak{G}}(\lambda))=0$, we obtain $\Gamma (\widetilde{U}_{\mathscr{LG}}(\lambda))=\check{\xi}(\widetilde{W}(0))=1$. By linearity, it follows that 
    \begin{equation}\label{eq:GammaUtilde}
        \Gamma(\widetilde{U}_{\mathscr{LG}}(\lambda)-1\otimes 1)=0
    \end{equation}
    and hence $I\subset \ker(\Gamma)$ follows because $I$ is an ideal and $\Gamma$ is a continuous ${}^*$-homomorphism. 
    
    To show $\ker(\Gamma)\subset I$, note that the Weyl generators of $\widetilde{\Af}^{\mathscr{LG}}(\ol{N})$ may be written
    \begin{align}
        \widetilde{W}([\udl{A}]\oplus f\oplus h)=&(W([\udl{A}]-\mathfrak{G}(f))\oplus W(0\oplus h))(\widetilde{W}(\widetilde{\mathfrak{G}}(f))-1\otimes 1)+W([\udl{A}]-\mathfrak{G}(f))\oplus W(0\oplus h)\nonumber\\
        =&\yen(W([\udl{A}]-\mathfrak{G}(f)))(\widetilde{U}_{\mathscr{LG}}(f)-1\otimes 1)+(\check{\jmath}\circ\check{\yen})(W([\udl{A}]-\mathfrak{G}(f))),
    \end{align}
    where $[\udl{A}]\in \Sol_{\mathscr{G}}(\ol{N})$, $f\in \mathscr{LG}_\angle(\ol{N})$ and $h=\nml^{(2)}_{\angle\ol{N}}\diff\udl{A}$.
    Using $\Gamma\circ\check{\jmath}=\id$ and~\eqref{eq:GammaUtilde}, we thus find 
    \begin{equation}
        \Gamma(\widetilde{W}([\udl{A}]\oplus f\oplus h))=\check{\yen}(W([\udl{A}]-\mathfrak{G}(f))),
    \end{equation}
    and consequently
    \begin{equation}
        (\id-(\check{\jmath}\circ\Gamma))(\widetilde{W}([\udl{A}]\oplus f\oplus h)))\in I.
    \end{equation}
    By linearity and continuity, we find for all $a\in \widetilde{\Af}^{\mathscr{LG}}(\ol{N})$
    \begin{equation}
        (\id-(\check{\jmath}\circ\Gamma))(a)\in I.
    \end{equation}
    Thus $\Gamma(a)=0$ implies $a=\id(a)\in I$, so $\ker(\Gamma)\subset I$. Finally, $\Gamma(\widetilde{U}_{\mathscr{LG}}(\lambda))$ was given in~\eqref{eq:GammaUtilde}.
\end{proof} 

In analogy with the classical surface field extension, we denote the extended algebra
\begin{equation}
\label{eq:aext}
    \Af_{\textnormal{ext}}(\ol{N}):=\widetilde{\Af}^{\mathscr{LG}}(\ol{N})/I
\end{equation}
As in the classical case, we have established equivalence between the surface field extension theory and the original formulation of electromagnetism,
which we summarise as follows.
\begin{corollary}
\label{cor:constr_iso}
    There are ${}^*$-isomorphisms
    \begin{equation}
        \Af(\ol{N})\cong \Af_{\rel}(\ol{N})\cong \Af_{\textnormal{ext}}(\ol{N}),
    \end{equation}
    given by $\check{\yen}:\Af(\ol{N})\to \Af_{\rel}(\ol{N})$, and $\won\circ\check{\jmath}:\Af_{\rel}(\ol{N})\to \Af_{\textnormal{ext}}(\ol{N})$.
\end{corollary}
On one hand we have the \emph{semi-local picture} with observables given by $\Af(\ol{N})$, on the other the \emph{surface field extended picture} with observables given by $\Af_{\textnormal{ext}}(\ol{N})$. The equivalence is given by the isomorphism $\won\circ\yen:\Af(\ol{N})\to \Af_{\textnormal{ext}}(\ol{N})$. These results mirror the classical equivalence between the electromagnetism in $\ol{N}$ and its surface field extension established in Section~\ref{sec:classical_sfe}.  Note that the quotient by $I$ identifies the
unitary implementations of the active and (reverse) passive large gauge transformations on $\Af_{\textnormal{ext}}(\ol{N})$, i.e., 
$U_{\mathscr{LG}}(\lambda)\otimes 1+I= 1\otimes U_{\mathscr{LG}}^S(-\lambda)+I$. The adjoint action of these operators agrees with the image of the large gauge transformations under 
the isomorphism $\won\circ\yen:\Af(\ol{N})\to \Af_{\textnormal{ext}}(\ol{N})$,
\begin{equation}
    \won\circ\yen\circ\alpha_{\mathscr{LG}}(\lambda)=\Ad(U_{\mathscr{LG}}(\lambda)\otimes 1+I)=\Ad(1\otimes U_{\mathscr{LG}}^S(-\lambda)+I).
\end{equation} 

In the next section we show that the surface field extended picture may be regarded as a description of electromagnetism inside a charged shell, where the electric field outside the shell is set to vanish. This forms part of a broader discussion of superselection sectors of $\widetilde{\Af}^{\mathscr{LG}}(\ol{N})$.

\subsection{Superselection sectors of \texorpdfstring{$\widetilde{\Af}^{\mathscr{LG}}(\ol{N})$}{ALG(N)}}\label{sec: superselection}
We consider Hilbert space representations of $\widetilde{\Af}^{\mathscr{LG}}(\ol{N})$, the algebra of joint large gauge invariant electromagnetic and surface field observables. The unitary equivalence classes of its irreducible representations are known as \emph{superselection sectors} of the algebra, see e.g.~\cite{doplicherFieldsObservablesGauge1969,buchholzLocalityStructureParticle1982,buchholzPhysicalStateSpace1982}. The algebra $\widetilde{\Af}^{\mathscr{LG}}(\ol{N})$ has a non-trivial center, since, for example, the unitary operators $\widetilde{U}_{\mathscr{LG}}(\lambda)\in \widetilde{\Af}^{\mathscr{LG}}(\ol{N})$ for $\lambda\in \mathscr{LG}_\angle(\ol{N})$ are central.\footnote{By identifying $\widetilde{\Af}^{\mathscr{LG}}(\ol{N})$ with $\Weyl(\widetilde{\mathscr{G}}(\ol{N})^\perp,\widetilde{\sigma})$ as in Thm.~\ref{thm:At_Weyl_q}, we see that this is a consequence of $\widetilde{\mathscr{G}}(\ol{N})\subset \widetilde{\mathscr{G}}(\ol{N})^\perp$ being a degenerate subspace.} This means that for each irreducible representation, the operators $\widetilde{U}_{\mathscr{LG}}(\lambda)$ act through multiplication by a phase. Moreover, for a sufficiently regular irreducible representation $(\tilde{\pi},\tilde{\HH})$ of $\widetilde{\Af}^{\mathscr{LG}}(\ol{N})$, there exists a function $\Phi\in \mathscr{LG}_\angle(\ol{N})$ such that
\begin{equation}
    \tilde{\pi}(\widetilde{U}_{\mathscr{LG}}(\lambda))=\exp(i\ipc{\Phi}{\lambda}_{\angle\ol{N}})\one_{\tilde{\HH}}.
\end{equation}
We may interpret this phase function $\Phi$ as the electric flux density on the outside of the corner $\angle{\ol{N}}$, as such we refer to $\Phi$ as the \emph{external flux function}. Clearly, any two representations whose associated external fluxes differ are unitarily inequivalent. We show that, when restricting attention to sectors obeying the superselection criterion 
\begin{equation}\label{eq: sups criterion}
 \tilde{\pi}\restriction_{\Af_{\textup{rel}}(\ol{N})}=\pi \quad \textrm{ (up to unitary equivalence)}\,,
\end{equation}
for some fixed representation $(\pi,\HH)$ of $\Af_{\textup{rel}}$, the superselection sectors are uniquely fixed by the external flux functions. In particular, we construct $*$-homomorphisms $\Gamma_\Phi:\widetilde{\Af}^{\mathscr{LG}}(\ol{N})\to \Af_{\textup{rel}}(\ol{N})$ mapping onto the sector labelled by $\Phi$, so that $\Gamma_0$ is the map $\Gamma$ studied in Theorem~\ref{thm:Gamma}. Moreover, we describe how the irreducible representations of $\Af_{\textnormal{ext}}(\ol{N})$ may be identified with $\Phi=0$ irreducible representations of $\widetilde{\Af}^{\mathscr{LG}}(\ol{N})$.

For the remainder of this section, we fix a non-zero irreducible representation $(\pi,\HH)$ of $\Af_{\textnormal{rel}}(\ol{N})$. For each $\Phi\in \mathscr{LG}_\angle(\ol{N})$, we can construct a superselection sector of $\widetilde{\Af}^{\mathscr{LG}}(\ol{N})$ for which $\Phi$ is the external flux function. We denote by $\beta_\Phi\in\Aut(\widetilde{\Af}^{\mathscr{LG}}(\ol{N}))$ the map
    \begin{equation}
        \beta_\Phi=\Ad \widetilde{W}(0\oplus 0\oplus \Phi)\restriction_{\widetilde{\Af}^{\mathscr{LG}}(\ol{N})},
    \end{equation}
    where $\Ad \widetilde{W}(0\oplus 0\oplus \Phi)\in \Aut(\widetilde{\Af}(\ol{N}))$ is a background shift on $\widetilde{\Af}(\ol{N})$. Using the Weyl relations, it is straightforward to check that $\beta_\Phi\circ\check{\jmath}=\check{\jmath}$, where $\check{\jmath}:\Af_{\textup{rel}}(\ol{N})\to \widetilde{\Af}^{\mathscr{LG}}(\ol{N})$ is the inclusion map, and that \begin{equation}
        \label{eq:Ext_flux_field_shift}
            \beta_\Phi(\widetilde{U}_{\mathscr{LG}}(\lambda))=\exp(i\ipc{\Phi}{\lambda}_{\angle\ol{N}})\widetilde{U}_{\mathscr{LG}}(\lambda).
    \end{equation} Since $\widetilde{U}_{\mathscr{LG}}(\lambda)$ is an element of the center of $\widetilde{\Af}^{\mathscr{LG}}(\ol{N})$, it is clear that $\beta_\Phi$ is an inner automorphism if and only if $\Phi=0$. Using $\beta_\Phi$, we define a surjective homomorphism $\Gamma_\Phi:\widetilde{\Af}^{\mathscr{LG}}(\ol{N})\to \Af_{\textnormal{rel}}(\ol{N})$ by
    \begin{equation}
        \Gamma_\Phi:=\Gamma\circ\beta_\Phi,
    \end{equation}
    such that $\Gamma=\Gamma_0$. These maps allow us to define a pulled back representation $(\pi_\Phi,\HH)$ of $\widetilde{\Af}^{\mathscr{LG}}(\ol{N})$ by
    \begin{equation}
        \pi_\Phi=\pi\circ\Gamma_\Phi.
    \end{equation}
    The representations $(\pi_\Phi,\HH)$ have the following properties.
    \begin{theorem}
    \label{thm:sup_sel}
         Up to unitary equivalence, the representation $(\pi_\Phi,\HH)$ is the unique representation of $\widetilde{\Af}^{\mathscr{LG}}(\ol{N})$ such that $\pi_\Phi\circ \check{\jmath}$ is unitarily equivalent to $\pi$
        and where for all $\lambda\in \mathscr{LG}_\angle(\ol{N})$ we have
        \begin{equation}
            \pi_\Phi(\widetilde{U}_{\mathscr{LG}}(\lambda))=\exp(i\ipc{\Phi}{\lambda}_{\angle\ol{N}})\one.
        \end{equation}
        
        Furthermore, $(\pi_\Phi,\HH)$ is irreducible. Lastly, let $\Phi'\in \mathscr{LG}_\angle(\ol{N})$. Then $(\pi_\Phi,\HH)$ is unitarily equivalent to $(\pi_{\Phi'},\HH)$ if and only if $\Phi=\Phi'$.
    \end{theorem}
    \begin{proof}
        Using $\beta_\Phi\circ\check{\jmath}=\check{\jmath}$ and Thm.~\ref{thm:Gamma}, we have
        \begin{equation}
            \pi_\Phi\circ \check{\jmath}=\pi\circ\Gamma\circ\beta_\Phi\circ \check{\jmath}=\pi\circ\Gamma\circ\check{\jmath}=\pi.
        \end{equation}
        Furthermore, using Thm.~\ref{thm:Gamma} and Eq.~\eqref{eq:Ext_flux_field_shift}, we find
        \begin{equation}
            \pi_\Phi(\widetilde{U}_{\mathscr{LG}}(\lambda))=\exp(i\ipc{\Phi}{\lambda}_{\angle\ol{N}})\pi(\Gamma(\widetilde{U}_{\mathscr{LG}}(\lambda))=\exp(i\ipc{\Phi}{\lambda}_{\angle\ol{N}})\pi(1)=\exp(i\ipc{\Phi}{\lambda}_{\angle\ol{N}})\one_\HH.
        \end{equation}
        
        Now let $(\tilde{\pi},\tilde{\HH})$ be a representation of $\widetilde{\Af}^{\mathscr{LG}}(\ol{N})$ such that $\tilde{\pi}\circ\check{\jmath}$ is unitarily equivalent to $\pi$ and $\tilde{\pi}(\widetilde{U}_{\mathscr{LG}}(\lambda))=\exp(i\ipc{\Phi}{\lambda}_{\angle\ol{N}})\one_{\tilde{\HH}}$. Let $V:\HH\to \tilde{\HH}$ be the unitary operator such that
        \begin{equation}
            \tilde{\pi}(a)=V^*\pi(a)V=V^*\pi_\Phi(a)V,
        \end{equation}
        for each $a\in \Af_{\textup{rel}}(\ol{N})$. We show that for $[\udl{A}]\in \Sol_\mathscr{G}(\ol{N})$ and $\lambda\in \mathscr{LG}_\angle(\ol{N})$, we have
         \begin{equation}
            \tilde{\pi}(\yen(W([\udl{A}]))(1\otimes U_{\mathscr{LG}}^S(\lambda)))=V\pi_\Phi(\yen(W([\udl{A}]))(1\otimes U_{\mathscr{LG}}^S(\lambda)))V^*.
        \end{equation}
        Due to Thm.~\ref{thm:Weyl_comm_thm}, this relation uniquely fixes $\tilde{\pi}=\Ad V\circ\pi_\Phi$.

        Note that by Eq.~\eqref{eq:C*QRF}, we have $\yen(W([\udl{A}]))(1\otimes U_{\mathscr{LG}}^S(\lambda))=\yen(W([\udl{A}])U_{\mathscr{LG}}(-\lambda))\widetilde{U}_{\mathscr{LG}}(\lambda)$
        and thus
        \begin{equation}
            \tilde{\pi}(\yen(W([\udl{A}]))(1\otimes U_{\mathscr{LG}}^S(\lambda)))=\tilde{\pi}(\yen(W([\udl{A}])U_{\mathscr{LG}}(-\lambda)))\tilde{\pi}(\widetilde{U}_{\mathscr{LG}}(\lambda)).
        \end{equation}
        Since $\yen(W([\udl{A}])U_{\mathscr{LG}}(-\lambda))\in \Af_{\textup{rel}}(\ol{N})$, we have
        \begin{equation}
            \tilde{\pi}(\yen(W([\udl{A}])U_{\mathscr{LG}}(-\lambda)))=V\pi_\Phi(\yen(W([\udl{A}])U_{\mathscr{LG}}(-\lambda)))V^*.
        \end{equation}
        Furthermore,
        \begin{equation}
            \tilde{\pi}(\widetilde{U}_{\mathscr{LG}}(\lambda))=\exp(i\ipc{\Phi}{\lambda}_{\angle\ol{N}})\one_{\tilde{\HH}}=V\pi_\Phi(\widetilde{U}_{\mathscr{LG}}(\lambda))V^*.
        \end{equation}
        We conclude that
        \begin{align}
            \tilde{\pi}(\yen(W([\udl{A}]))(1\otimes U_{\mathscr{LG}}^S(\lambda)))=&V\pi_\Phi(\yen(W([\udl{A}])U_{\mathscr{LG}}(-\lambda)))\pi_\Phi(\widetilde{U}_{\mathscr{LG}}(\lambda))V^*\nonumber\\
            =&V\pi_\Phi(\yen(W([\udl{A}]))(1\otimes U_{\mathscr{LG}}^S(\lambda)))V^*.
        \end{align}
        This proves uniqueness of $(\pi_\Phi,\HH)$ up to unitary equivalence.

        Irreducibility of $\pi_\Phi$ follows from the fact that
        \begin{equation}
            \pi_\Phi(\widetilde{\Af}^{\mathscr{LG}}(\ol{N}))''=\pi(\Af_{\textup{rel}}(\ol{N}))''=B(\HH).
        \end{equation}
        
        Lastly, if $\pi_\Phi$ and $\pi_{\Phi'}$ are unitarily equivalent, $\exp(i\ipc{\Phi'}{\lambda}_{\angle\ol{N}}) =\exp(i\ipc{\Phi}{\lambda}_{\angle\ol{N}})$ 
        for all $\lambda\in\mathscr{LG}_\angle(\ol{N})$ and hence $\Phi=\Phi'$.
    \end{proof}
Consequently, irreducible representations of the form $(\pi_\Phi,\HH)$ for the algebra $\widetilde{\Af}^{\mathscr{LG}}(\ol{N})$ yield sectors with external flux function $\Phi$. In particular, the $\Phi=0$ sector is given by the unitary equivalence class of $\pi_0=\pi\circ \Gamma$. This $\Phi=0$ sector can alternatively by constructed from representations of $\Af_{\textup{ext}}(\ol{N})$ using the homomorphisms of diagram \eqref{eq:Ymap_diag2}. Since $\Af_{\textup{rel}}(\ol{N})\cong \Af_{\textup{ext}}(\ol{N})$, the irreducible  representation $(\pi,\HH)$ of $\Af_{\textup{rel}}(\ol{N})$ corresponds uniquely to an irreducible representation $(\pi_{\textup{ext}},\HH)$ of $\Af_{\textup{ext}}(\ol{N})$ through $\pi=\pi_{\textup{ext}}\circ\won\circ \check{\jmath}$. Comparing $\pi_{\textup{ext}}$ to $\pi_0$, we find the following.
\begin{corollary}
    Any irreducible representations $(\pi_{\textup{ext}},\HH)$ of $\Af_{\textup{ext}}(\ol{N})$ and $(\pi_0,\HH)$ of $\widetilde{\Af}^{\mathscr{LG}}(\ol{N})$ built from an irreducible representation $(\pi,\HH)$ of $\Af_{\textup{rel}}(\ol{N})$ satisfy $\pi_{\textup{ext}}\circ\won=\pi_0$.
\end{corollary}
\begin{proof}
    Let $\tilde{\pi}=\pi_{\textup{ext}}\circ\won$, then we compute
    \begin{equation}
        \tilde{\pi}\circ\check{\jmath}=\pi_{\textup{ext}}\circ\won\circ\check{\jmath}=\pi.
    \end{equation}
    Furthermore, we compute for $\lambda\in \mathscr{LG}_\angle(\ol{N})$ that
    \begin{equation}
        \tilde{\pi}(\widetilde{U}_{\mathscr{LG}}(\lambda))=\pi_{\textup{ext}}(\won(\widetilde{U}_{\mathscr{LG}}(\lambda)))=\pi_{\textup{ext}}(\won(1))=\one.
    \end{equation}
    By Thm.~\ref{thm:sup_sel}, $(\tilde{\pi},\HH)$ is unitarily equivalent to $(\pi_0,\HH)$, where equality follows due to equality on $\Af_{\textup{rel}}(\ol{N})\subset \widetilde{\Af}^{\mathscr{LG}}(\ol{N})$, since $\pi(\Af_{\textup{rel}}(\ol{N}))''=B(\HH)$.
\end{proof}
We may therefore think of the algebra $\Af_{\textup{ext}}(\ol{N})$ as corresponding to the $\Phi=0$ sector of $\widetilde{\Af}^{\mathscr{LG}}(\ol{N})$.

So far we have not discussed to what extent the observables in $\widetilde{\Af}^{\mathscr{LG}}(\ol{N})$ correspond to physically measurable quantities in realistic models. This is relevant for the question of whether the sectors $\pi_\Phi$ are physically distinguishable for different $\Phi$. If $\widetilde{\Af}^{\mathscr{LG}}(\ol{N})$ indeed models the observables of the electromagnetic field inside a charged shell, and the internal electric flux $\nml_{\partial\ol{\Sigma}}\mathbf{E}$ and the surface charge density $\tau$ can be measured separately, then we have for a sufficiently regular representation $\pi_\Phi$ that
\begin{equation}
    \nml_{\partial\ol{\Sigma}}\mathbf{E}+\tau=\Phi\one.
\end{equation}
In settings where the only operationally accessible observables lie in (weak closures of) $\Af_{\rel}(\ol{N})\subset \widetilde{\Af}^{\mathscr{LG}}(\ol{N})$, we are not able to measure $\tau$ and the superselection criterion of Eq.~\eqref{eq: sups criterion} entails that representations $\pi_\Phi$ and $\pi_0$ are operationally indistinguishable. This would then allow us to perform all computations in the $\Phi=0$ sector, using the algebra $\Af_{\textup{ext}}(\ol{N})$. To resolve distinctions of this type, one should formulate a measurement theory for quantum fields on spacetimes with boundaries and corners, for instance by extending \cite{fewsterQuantumFieldsLocal2020}. This can then be used to make more precise statements about operational accessibility of observables in models where the charged surface fields couple to the electromagnetic field, see e.g.~\cite{mathieuHomologicalPerspectiveEdge2020,carrozzaEdgeModesReference2022}. At the same time, such models would shine a light on the measureability of edge mode observables (i.e. operators in $\Af_{\rel}(\ol{N})$ transforming non-trivially under frame reorientations). Edge mode observables being measurable would in particular rule out any superselection rules on either $\tau$ or $\nml_{\partial\ol{\Sigma}}\mathbf{E}$ individually. As discussed at the end of Sec.~\ref{sec:reps}, this goes against the superselection structure proposed in \cite{rielloHamiltonianGaugeTheory2024,rielloNullHamiltonianYangMills2025,rielloSymplecticReductionYangMills2021}.

\subsection{Relativisation, POVMs and crossed products}
\label{sec:POVMs}

So far, we have not made precise in which sense the map $\yen:\Af(\ol{N})\to\widetilde{\Af}^{\mathscr{LG}}(\ol{N})$ as defined by the commutative diagram  \eqref{eq:Ymap_diag} may be regarded as a relativisation map. Note that $\yen$ satisfies an equivariance relation
\begin{equation}
\label{eq:yen-covariance}
    (1\otimes \alpha_{\mathscr{LG}}^S(\lambda))(\yen(a))=\yen(\alpha_{\mathscr{LG}}(-\lambda)(a))
\end{equation}
for $a\in \Af(\ol{N})$ and $\lambda\in \mathscr{LG}_\angle(\ol{N})$, which is easily checked on the generators. 
This relation mimics the covariance property of the relativisation map in the context of (principal) QRFs for locally compact abelian groups acting on von Neumann algebras, see e.g.~\cite{caretteOperationalQuantumReference2025,fewsterQuantumReferenceFrames2025}. In that framework, relativisation maps are constructed from covariant positive operator valued measures (POVMs), yielding completely positive linear maps that are not necessarily ${}^*$-homomorphisms. 
As we show below, it turns out one can in fact construct the map $\yen:\Af(\ol{N})\to\widetilde{\Af}^{\mathscr{LG}}(\ol{N})$ as given above from relativisation maps defined using covariant POVMs when working at the level of sufficiently regular representations of the relevant $C^*$-algebras. More precisely, we show that $\yen$ (which yields invariants for the infinite dimensional group $\mathscr{LG}(\ol{N})$) can be arbitrarily well approximated by relativisation maps constructed from covariant POVMs for finite dimensional (and hence locally compact) subgroups of $\mathscr{LG}(\ol{N})$.

\begin{theorem}
\label{thm:edge_mode_QRF}
    Consider faithful representations $(\pi^{\mathcal{S}},\HH^{\mathcal{S}})$ and $(\pi^{\mathcal{R}},\HH^{\mathcal{R}})$ 
    of $\Af(\ol{N})$ and $\Af^S(\ol{N})$, where $\HH^{\mathcal{S}/\mathcal{R}}$ are separable and the map
     \begin{equation}
        \RR\ni t \mapsto \pi^{\mathcal{S}}(W_A([t\udl{A}]))\otimes \pi^{\mathcal{R}}(W^{\partial}(t\udl{v})).
    \end{equation}
    is strongly continuous for each fixed   $[\udl{A}]\in \Sol_\mathscr{G}(\ol{N})$, $\udl{v}\in V^S(\ol{N})$.
    Denote system and reference von Neumann algebras by
    \begin{equation}
        \MM^{\mathcal{S}}=\pi^{\mathcal{S}}(\Af(\ol{N}))'',\qquad \MM^{\mathcal{R}}=\pi^{\mathcal{R}}(\Af^S(\ol{N}))'',
    \end{equation} 
    and let 
    $\MM=\MM^{\mathcal{S}}\otimes \MM^{\mathcal{R}}\subset B(\HH^{\mathcal{S}}\otimes \HH^{\mathcal{R}})$. Now let $U_\pi^{\mathcal{S}/\mathcal{R}}:\mathscr{LG}_\angle(\ol{N})\to\MM^{\mathcal{S}/\mathcal{R}}$ be unitary representations given by
    \begin{equation}
        U^\mathcal{S}(\lambda)=\pi^{\mathcal{S}}(U_{\mathscr{LG}}(\lambda)),\qquad U^\mathcal{R}(\lambda)=\pi^{\mathcal{R}}(U_{\mathscr{LG}}^S(\lambda)),
    \end{equation}
    such that $(\pi^{\mathcal{S}}\otimes \pi^{\mathcal{R}}) (\widetilde{U}_{\mathscr{LG}}(\lambda))=U^\mathcal{S}(\lambda)\otimes U^\mathcal{R}(\lambda)$.
    Assume furthermore that the map
    \begin{equation}
        \mathscr{LG}_\angle(\ol{N})\ni \lambda\mapsto \pi^{\mathcal{R}}(W^{\partial}(\lambda\oplus 0)),
    \end{equation}
    is strongly continuous w.r.t.~a $W^{2,s}$ Sobolev topology on $\mathscr{LG}_\angle(\ol{N}) \subset\Omega^0(\angle\ol{N})$, for some $s\in \RR$. Then there exist
    \begin{itemize}
        \item a sequence $G_n\subset \mathscr{LG}_\angle(\ol{N})$ of groups with $G_n\subset G_{n+1}$ and where $\bigcup_n G_n$ is $W^{2,s}$ dense in $\mathscr{LG}_\angle(\ol{N})$,        \item a sequence of triples $(\HH,P_n,U^\mathcal{R}_n)$ for each group $G_n$, where $U^{\mathcal{R}}_n=U^\mathcal{R}\restriction_{G_n}$ and $P_n:\Bor(G_n)\to\MM^{\mathcal{R}}$ a projection valued measure satisfying  for each $X\in \Bor(G_n)$, $\lambda\in G_n$ 
        \begin{equation}
        U^{\mathcal{R}}_n(\lambda)P_n(X)U^{\mathcal{R}}_n(\lambda)^*=P_n(X+\lambda),
        \end{equation}
        \item a relativisation map $\yen_n:\MM^{\mathcal{S}}\to \MM^{\Ad (U^\mathcal{S}\otimes U^\mathcal{R})\restriction_{G_n}}$ given by the (weak) integral
        \begin{equation}
        \yen_n(a)=\int_{G_n}  U^\mathcal{S}(f) a U^\mathcal{S}(f)^*\otimes \diff P_n(f),
    \end{equation}
    \end{itemize}
    such that for $a\in \Af(\ol{N})$ and $\yen:\Af(\ol{N})\to \widetilde{\Af}^{\mathscr{LG}}(\ol{N})$ defined by Eq.~\eqref{eq:C*QRF}, one has 
    \begin{equation}
        (\pi^\mathcal{S}\otimes\pi^\mathcal{R})(\yen(a))=\slim_{n\to\infty}\yen_n(\pi^\mathcal{S}(a)).
    \end{equation}
\end{theorem}
This is proven in appendix \ref{apx:edge_mode_QRF}. Any faithful Fock representation of $\Af(\ol{N})$ (such as the $L^2$ representation of Thm.~\ref{thm:L2rep}) tensored with an additional $L^2$ representation of $\Af^S(\ol{N})$ provides an example of a faithful representation satisfying the conditions of Thm.~\ref{thm:edge_mode_QRF}, where the required Sobolev continuity holds for the $W^{2,0}$ (or $L^2$) topology on $\mathscr{LG}_\angle(\ol{N})$. 

In the proof, the sequence of topological groups $G_n\subset \mathscr{LG}_\angle(\ol{N})$ come about through spectral decomposition of $\mathscr{LG}_\angle(\ol{N})$ with respect to an essentially self-adjoint corner Laplacian $\Delta_\angle:\Omega^0(\angle \ol{N})\to\Omega^0(\angle \ol{N})$. For each group $G_n$, the triple $(\HH,U^\mathcal{R}_n,P_n)$, i.e.~a unitary representation $(\HH,U^\mathcal{R}_n)$ of $G_n$ and a $G_n$-covariant projection valued measure $P_n:\Bor(G_n)\to B(\HH)$ goes under the name of a principal quantum reference frame \cite{caretteOperationalQuantumReference2025} and constitutes a special case of a system of imprimitivity \cite{mackeyTheoryUnitaryGroup1976}. Each $G_n\cong \RR^m$ for some $m\in \mathbb{N}$, and as a result the unitary representation $(\HH,U_n^{\mathcal{R}})$ arises from a representation of Weyl-algebra over $\RR^{2m}$, which are well-understood through the Stone-von Neumann theorem \cite[5.2.15-16]{bratteliOperatorAlgebrasQuantum1997}. This allows us to construct the PVM $P_n$ explicitly in the Schr\"odinger representation of this Weyl algebra on $L^2(\RR^{m})$. Given the principal QRF $(\HH,U^\mathcal{R}_n,P_n)$, there is a standard definition of a relativisation map $\yen_n:\MM^{\mathcal{S}}\to \MM^{\Ad (U^\mathcal{S}\otimes U^\mathcal{R})\restriction_{G_n}}$, see e.g.~\cite{loveridgeSymmetryReferenceFrames2018,fewsterQuantumReferenceFrames2025}. It's agreement with the $C^*$-algebraic relativisation map $\yen$ in the limit $n\to \infty$ is verified through explicit computatation.

As discussed in \cite{fewsterQuantumReferenceFrames2025}, the relativisation map $\yen_n$, and more generally the algebra of invariants $\MM^{\Ad (U^\mathcal{S}\otimes U^\mathcal{R})\restriction_{G_n}}$ are closely related to the crossed product algebra $\MM^{\mathcal{S}}\rtimes_{\Ad U^\mathcal{S}} G_n$, see e.g.~\cite{vanDaele:1978,takesakiTheoryOperatorAlgebras2002} for the definition of crossed product algebras for locally compact groups. Concretely, there exists a (spatial) isomorphism
\begin{equation}
\label{eq:inv_to_CP}
    \MM^{\Ad (U^\mathcal{S}\otimes U^\mathcal{R})\restriction_{G_n}}\cong (\MM^{\mathcal{S}}\rtimes_{\Ad U^\mathcal{S}} G_n)\otimes \tilde{\MM}_n,
\end{equation}
for some residual von Neumann algebra $\tilde{\MM}_n\subset B(\KK_n)$ for which $\HH^\mathcal{R}\cong L^2(G_n)\otimes \KK_n$ (see \cite[Thm.~4.9]{fewsterQuantumReferenceFrames2025}). In fact, this isomorphism can be used to define the relativisation map $\yen_n$, as seen in the proof of Lem.~\ref{lem:emq2}. Furthermore, from the isomorphism above one can show that
\begin{equation}
\label{eq:inv_gen}
    \MM^{\Ad (U^\mathcal{S}\otimes U^\mathcal{R})\restriction_{G_n}}=\{\yen_n(a)(1_{\HH_S}\otimes b):a\in \MM^{\mathcal{S}},\, b\in (\MM^{\mathcal{R}})^{\Ad  U^\mathcal{R}\restriction_{G_n}}\}''.
\end{equation}
It is not immediately clear if there is any direct generalisation of the statement of Eq.~\eqref{eq:inv_to_CP} for the full algebra of invariants $\MM^{\Ad (U^\mathcal{S}\otimes U^\mathcal{R})}$. In the first place, it is unclear what an appropriate generalisation of a crossed product algebra would be with respect to the group action of $\mathscr{LG}_\angle(\ol{N})$ on $\MM^{\mathcal{S}}$, as this group is not locally compact. Nevertheless, through the relativisation map $\yen$ we can give an analogue of Eq.~\eqref{eq:inv_gen} for the full group $\mathscr{LG}_\angle(\ol{N})$.
\begin{corollary}
    Given the assumptions of Thm.~\ref{thm:edge_mode_QRF}, we have the following inclusions.
    \begin{align}
    (\pi^\mathcal{S}\otimes \pi^\mathcal{R})(\widetilde{\Af}^{\mathscr{LG}}(\ol{N}))''\subset&
    \{\slim_{n\to \infty}\yen_n(a)(\one_{\HH_S}\otimes b):a\in \pi^{\mathcal{S}}(\Af(\ol{N})),\, b\in (\MM^{\mathcal{R}})^{\Ad  U^\mathcal{R}}\}''\nonumber\\\subset& \MM^{\Ad (U^\mathcal{S}\otimes U^\mathcal{R})}.
\end{align}
\end{corollary}
Clearly, equality follows if $(\pi^\mathcal{S}\otimes \pi^\mathcal{R})(\widetilde{\Af}^{\mathscr{LG}}(\ol{N}))''=\MM^{\Ad (U^\mathcal{S}\otimes U^\mathcal{R})}$, or loosely speaking, if restricting to invariants commutes taking weak closures.
\begin{proof}
    By Thm.~\ref{thm:Weyl_comm_thm}, we have that
    \begin{multline}
        (\pi^\mathcal{S}\otimes \pi^\mathcal{R})(\widetilde{\Af}^{\mathscr{LG}}(\ol{N}))''=\\\{(\pi^\mathcal{S}\otimes \pi^\mathcal{R})(\yen(W_A([\udl{A}])))(\one_{\HH_S}\otimes U^\mathcal{R}(\lambda)):(A,\udl{A})\in T\Sol^J(\ol{N}),\, \lambda\in \mathscr{LG}_\angle(\ol{N})\}''.
    \end{multline}
    Since $U^\mathcal{R}(\lambda))\in \MM^{U^\mathcal{R}}\}$ due to the fact that $\mathscr{LG}_\angle(\ol{N})$ is an abelian group, and since by Thm.~\ref{thm:edge_mode_QRF} we have
    \begin{equation}
        (\pi^\mathcal{S}\otimes \pi^\mathcal{R})(\yen(W_A([\udl{A}])))=\slim_{n\to\infty}\yen_n(\pi^\mathcal{S}(W_A([\udl{A}])))\in \MM^{\Ad (U^\mathcal{S}\otimes U^\mathcal{R})},
    \end{equation}
    we thus have 
    \begin{align}
        \{(\pi^\mathcal{S}\otimes \pi^\mathcal{R})(\yen(W_A([\udl{A}])))(\one_{\HH_S}\otimes U^\mathcal{R}(\lambda)):(A,\udl{A})\in T\Sol^J(\ol{N}),\, \lambda\in \mathscr{LG}_\angle(\ol{N})\}\nonumber\\
        \subset \{\slim_{n\to \infty}\yen_n(a)(\one_{\HH_S}\otimes b):a\in \pi^{\mathcal{S}}(\Af(\ol{N})),\, b\in (\MM^{\mathcal{R}})^{\Ad  U^\mathcal{R}}\}\nonumber\\
        \subset \MM^{\Ad (U^\mathcal{S}\otimes U^\mathcal{R})}. 
    \end{align}
    By taking the double commutant, the claim follows.
\end{proof}

\section{Application of the relativisation map to gluing procedures}
\label{sec:gluing}
We now describe how the relativization map discussed above can be used in a gluing procedure of algebras and states. Consider, for example, two connected finite Cauchy lenses $\ol{N}_1$ and $\ol{N}_2$ with a common corner, suitably embedded in an ambient globally hyperbolic spacetime $M=\mathscr{D}(N_1\cup N_2)$. Each Cauchy lens has an associated algebra of semi-local observables; we will describe how a glued algebra may be constructed which incorporates all the local observables on $M$ together with additional corner observables that are singular at $\angle\ol{N}_i$. This setting provides an alternative interpretation of the relativistation map, which can now be realised as a homomorphism from $\Af(\ol{N}_1)$ to the glued algebra. The precise form of this map depends on an identification of 
$\Af^S(\ol{N}_1)$ with a subalgebra of $\Af(\ol{N}_2)$, which can be provided by a quantum CHH-isomorphism depending on a choice of Cauchy surface for $\ol{N}_2$. Using the relativisation map in this form, one can construct states on the glued algebra from states on the individual Cauchy lenses.

\subsection{General setting}
Gluing procedures for classical electromagnetism and/or Yang-Mills theories have been described for instance in \cite{donnellyLocalSubsystemsGauge2016,gomesQuasilocalDegreesFreedom2021,cattaneoNoteGluingFiber2023}. In this paper, we understand a classical gluing procedure for electromagnetism as follows. Let $M$ be globally hyperbolic and suppose 
$M=\mathscr{D}(N_1\cup N_2)$, where $N_1,N_2\subset M$ are well-behaved causally convex submanifolds, possibly with boundaries and corners.  The aim of a gluing procedure is to construct a symplectic space of configurations on $M$ from the symplectic spaces $\Sol_{\mathscr{G}}(N_1)$ and $\Sol_{\mathscr{G}}(N_2)$. An example is described in \cite{donnellyLocalSubsystemsGauge2016}, using a fusion product in a manner similar to the charged edge field construction described above (up to some regularity issues discussed below). We shall use the relativisation map to formulate a quantum version of gluing on the algebra of observables, which will allow us to construct a quasi-free state on an algebra of observables associated with $M$ from compatible quasi-free states on algebras associated with finite Cauchy lenses $\ol{N}_1, \ol{N}_2\subset M$. This latter result is reminiscent of a gluing construction for quasi-free states of scalar fields  in a different geometric context~\cite{janssenQuantumFieldsSemiglobally2022}.

For simplicity, we take vanishing background current $J=0$. Let $M$ be a globally hyperbolic spacetime with a compact boundaryless smooth Cauchy surface $\Sigma\subset M$
which can be decomposed as $\Sigma=\ol{\Sigma}_1\cup\ol{\Sigma}_2$, where $\ol{\Sigma}_i\subset\ol{N}_i$ are regular Cauchy surfaces with boundaries of connected finite Cauchy lenses $\ol{N}_i$ so that $M=\mathscr{D}(\ol{N}_1\cup\ol{N}_2)$. We also assume that $\ol{N}_1\cap\ol{N}_2=\angle\ol{N}_1\cap\angle\ol{N}_2=\angle\ol{N}_1=\angle\ol{N}_2$. 
The choices $M$, $\Sigma$, $\ol{N}_i$ and $\ol{\Sigma}_i$ will be kept fixed throughout this subsection. Sometimes it is useful to write the common boundary as $\partial\ol{\Sigma}_*=\angle\ol{N}_*$; similarly, other objects that are independent of $i=1,2$ will be given a subscript $*$.
Due to the fact that each $\ol{\Sigma}_i$ is connected, one has $\nml_{\partial\ol{\Sigma}_1}\Omega^1_{\diff^*}(\ol{\Sigma}_1)=\nml_{\partial\ol{\Sigma}_2}\Omega^1_{\diff^*}(\ol{\Sigma}_2)=(\RR\cdot 1)^\perp$ as a subset of $\Omega^0(\ol{\Sigma}_1\cap\ol{\Sigma}_2)=\Omega^0(\partial\ol{\Sigma}_*)$.

Applying the covariant phase space formalism to $M$, the $J=0$ electromagnetic phase space is 
\begin{equation}
    \Sol_{\mathscr{G}}(M)=\{A\in \Omega^1(M):-\diff^*\diff A=0\}/\diff\Omega^0(M),
\end{equation}
with the usual symplectic structure $\sigma$ given by Eq.~\eqref{eq:symp_struct} (here with $\ol{\Sigma}=\Sigma$). In the absence of a boundary or corner, the (quantised) algebra of local observables is simply given by $\Af(M)=\Weyl(\Sol_{\mathscr{G}}(M),\sigma)$, see e.g.~\cite{dimockQuantizedElectromagneticField1992}, where we denote the Weyl generators by $W^M([\udl{A}])$ with $[\udl{A}]\in\Sol_{\mathscr{G}}(M)$. For any open set $U\subset M$, the observables localisable in $U$ are defined, as usual, to be the smallest $C^*$-subalgebra $\Af(M;U)\subset\Af(M)$ containing
\begin{equation}
    \{W^M([G^{\PJ}(f)]):f\in \Omega^1_{0\diff^*}(M),\, \supp(f)\subset U\}.
\end{equation} 
Below we show that the algebra $\Af(M)$ can be obtained by a gluing procedure applied to the algebras $\Af(\ol{N}_1)$ and $\Af(\ol{N}_2)$.

\subsection{Gluing of algebras}
\label{sec:glue_alg}

As already mentioned, the classical gluing construction involves a fusion product of $\Sol_\mathscr{G}(\ol{N}_1)$ with $\Sol_\mathscr{G}(\ol{N}_2)$ relative to joint large gauge transformations. By analogy with the discussion in Sec.~\ref{sec:quantum_sfe}, the glued QFT algebra is a quotient 
\begin{equation}\label{eq:glued_def}
    \Af_\textnormal{glue}(\ol{N}_1; \ol{N}_2)  = (\Af(\ol{N}_1)\otimes \Af(\ol{N}_2))^{\alpha_\mathscr{LG}\otimes\alpha_\mathscr{LG}}/I,
\end{equation}
where $I$ is the two-sided closed $*$-ideal generated by elements $U^{(1)}_\mathscr{LG}(\lambda)\otimes U^{(2)}_\mathscr{LG}(\lambda)-1\otimes 1$ for $\lambda\in \mathscr{LG}_\angle(\ol{N}_*)$. Following our discussion in Sec.~\ref{sec: superselection}, we may say that dividing out this ideal amounts to restricting to the sector of $(\Af(\ol{N}_1)\otimes \Af(\ol{N}_2))^{\alpha_\mathscr{LG}\otimes\alpha_\mathscr{LG}}$ for which the electric fluxes $\nml_{\partial\ol{\Sigma}_1}\mathbf{E}(\lambda)$, $\nml_{\partial\ol{\Sigma}_2}\mathbf{E}(\lambda)$ (as formal generators of $U^{(1)}_\mathscr{LG}(\lambda)$ and $U^{(2)}_\mathscr{LG}(\lambda)$ respectively) are such that
\begin{equation}
    \nml_{\partial\ol{\Sigma}_1}\mathbf{E}+\nml_{\partial\ol{\Sigma}_2}\mathbf{E}=0.
\end{equation}
This can be interpreted as imposing the Gauss constraint under the assumption that no additional electric charges are present on the gluing surface $\partial\ol{\Sigma}_*$. As such this gluing construction is closely related to the perspective neutral approach to quantum reference frames which is rooted in constraint quantisation. For a treatment of quantum reference frames in lattice gauge theories using this approach, see \cite{araujo-regadoRelationalEntanglementEntropies2025a}.

Note that the glued algebra $\Af_\textnormal{glue}(\ol{N}_1; \ol{N}_2)$ does not depend on the choices of Cauchy surfaces $\ol{\Sigma}_i\subset \ol{N}_i$, but only on the finite Cauchy lenses $\ol{N}_i$ that share a corner $\angle\ol{N}_*$. One may however parametrize the glued algebra by the closed-loop and surface observables arising from CHH-decompositions of observables at $\ol{\Sigma}_i$ as in Def.~\ref{def:kin_transf}. This parametrization is non-unique, but we will highlight two particular isomorphisms parametrizing $\Af_\textnormal{glue}(\ol{N}_1; \ol{N}_2)$ below. To simplify notation, we will from now on identify the algebra of semi-local observables $\Af(\ol{N}_i)$ associated with the electromagnetic field on $\ol{N}_i$ with $\Af^C(\ol{\Sigma}_i)\otimes \Af^S(\ol{N}_*)$ using the quantum CHH-isomorphisms of Def.~\ref{def:kin_transf}, also noting that the surface algebra depends only on the common corner $\angle\ol{N}_*$. The Weyl generators of $\Af^{C/S}(\ol{\Sigma}_i/\ol{N}_i)$ will be labelled $W^{C/S}_i(\udl{v})$ for $\udl{v}\in V^{C/S}(\ol{\Sigma}_i/\ol{N}_i)$.

First note that, given our choices of $\ol{\Sigma}_i$ ($i=1,2$), there are mutually inverse $C^*$-isomorphisms $\Xi_{12}:\Af^C(\ol{\Sigma}_1)\otimes \Af (\ol{N}_2)\to
\Af(\ol{N}_1)\otimes \Af^C(\ol{\Sigma}_2)$ and $\Xi_{21}:
\Af(\ol{N}_1)\otimes \Af^C(\ol{\Sigma}_2)\to \Af^C(\ol{\Sigma}_1)\otimes \Af (\ol{N}_2)$ so that
\begin{equation}
    \label{eq:Xi21}
    \Xi_{21} (W_1^C(\udl{v}_1)\otimes W_1^S(f\oplus h)\otimes W^C_2(\udl{v}_2)) = W_1^C(\udl{v}_1)\otimes W^C_2(\udl{v}_2)\otimes W_1^S(-f\oplus -h),
\end{equation}
arising as Weyl quantisations of obvious symplectomorphisms.

\begin{theorem}\label{thm:glue}
Given connected regular Cauchy surfaces $\ol{\Sigma}_i\subset\ol{N}_i\subset M=\mathscr{D}(\ol{N}_1\cup \ol{N}_2)$ (for $i=1,2$), such that $\ol{\Sigma}_1\cap\ol{\Sigma}_2=\partial\ol{\Sigma}_*$ and $\Sigma=\ol{\Sigma}_1\cup\ol{\Sigma}_2$ is a Cauchy surface of $M$,
    there are isomorphisms $\Af^C(\ol{\Sigma}_1)\otimes \Af (\ol{N}_2)\stackrel{\overset{\circ\bullet}{\psi}}{\rightarrow} \Af_\textnormal{glue}(\ol{N}_1; \ol{N}_2)\stackrel{\overset{\bullet\circ}{\psi}}{\leftarrow} 
         \Af (\ol{N}_1)\otimes\Af^C(\ol{\Sigma}_2)$  with 
    \begin{align}
     \overset{\bullet\circ}{\psi}(W^C_1(\udl{v}_1)\otimes W^S_1(f\oplus h)\otimes W^C_2(\udl{v}_2))&=\won(W^C_1(\udl{v}_1)\otimes W^S_1(f\oplus h)\otimes W^C_2(\udl{v}_2) \otimes W^S_2(0\oplus -h))\nonumber\\
    \overset{\circ\bullet}{\psi}(W^C_1(\udl{v}_1)\otimes  W^C_2(\udl{v}_2)\otimes W^S_2(f\oplus h))&=\won(W^C_1(\udl{v}_1)\otimes W^S_1(0\oplus -h) \otimes W^C_2(\udl{v}_2) \otimes W^S_2(f\oplus h)) ,
    \end{align}
    for all $\udl{v}_i\in V^C(\ol{\Sigma}_i)$ and $f\oplus h\in V^S(\ol{N}_*)$, 
    where $\won$ is the quotient map in~\eqref{eq:glued_def}, and so that the diagram 
    \begin{equation}\label{eq:glue_isos}
    \begin{tikzcd}
         \Af^C(\ol{\Sigma}_1)\otimes \Af (\ol{N}_2)
         \arrow[r,"\overset{\circ\bullet}{\psi}"]\arrow[rr,bend left=20,"\Xi_{12}"] &  \Af_\textnormal{glue}(\ol{N}_1; \ol{N}_2) &   
         \Af (\ol{N}_1)\otimes\Af^C(\ol{\Sigma}_2) \arrow[l,"\overset{\bullet\circ}{\psi}"']\arrow[ll,bend left=20,"\Xi_{21}"']
    \end{tikzcd}     
    \end{equation}
    commutes. In particular, $\Af_\textnormal{glue}(\ol{N}_1; \ol{N}_2) \cong \Af^C(\ol{\Sigma}_1)\otimes \Af^S(\ol{N}_*)\otimes \Af^C(\ol{\Sigma}_2)$.
\end{theorem}
Here, the notation $\bullet\circ$ indicates the involvement of the full algebra of the first Cauchy lens but only the closed loop factor from the second.
The first step towards proving Theorem~\ref{thm:glue} is to compute $(\Af(\ol{N}_1)\otimes \Af(\ol{N}_2))^{\alpha_\mathscr{LG}\otimes\alpha_\mathscr{LG}}$. 
In what follows, the notation $\pi_\sigma$ ($\sigma\in S_k$) will denote the isomorphism rearranging factors in an $k$-fold tensor product according to the permutation $\sigma$ expressed in cycle notation.
\begin{lemma}
    The algebra $(\Af(\ol{N}_1)\otimes \Af(\ol{N}_2))^{\alpha_\mathscr{LG}\otimes\alpha_\mathscr{LG}}$ is the smallest subalgebra of $\Af(\ol{N}_1)\otimes \Af(\ol{N}_2)$ 
    containing
    \begin{equation}
        W^C_1(\udl{v}_1)\otimes W_1^S(f_1\oplus h)\otimes W^C_2(\udl{v}_2) \otimes W^S_2(f_2\oplus -h)
    \end{equation}
    for all $\udl{v}_i\in V^C(\ol{\Sigma}_i)$, $f_1,f_2,h\in \mathscr{LG}_\angle(\ol{N}_*)$.
Consequently, there are isomorphisms 
\begin{equation}
    \Af^C(\ol{\Sigma}_1)\otimes\widetilde{\Af}(\ol{N}_2) \stackrel{\overset{\circ\bullet}{\upsilon}}{\longleftarrow}
    (\Af(\ol{N}_1)\otimes \Af(\ol{N}_2))^{\alpha_\mathscr{LG}\otimes\alpha_\mathscr{LG}}
    \stackrel{\overset{\bullet\circ}{\upsilon}}{\longrightarrow} \widetilde{\Af}(\ol{N}_1)\otimes \Af^C(\ol{\Sigma}_2)
\end{equation}
given by $\overset{\circ\bullet}{\upsilon}=\pi_{(324)}$ and $\overset{\bullet\circ}{\upsilon}=\pi_{(34)}$ restricted to the subalgebra $(\Af(\ol{N}_1)\otimes \Af(\ol{N}_2))^{\alpha_\mathscr{LG}\otimes\alpha_\mathscr{LG}}$ in the four-fold tensor product $\Af^C(\ol{\Sigma}_1)\otimes\Af^S(\ol{\Sigma}_1)\otimes \Af^C(\ol{\Sigma}_2)\otimes\Af^S(\ol{\Sigma}_2)$.
\end{lemma}
The isomorphism $\overset{\bullet\circ}{\upsilon}$ (resp., $\overset{\circ\bullet}{\upsilon}$) is designed to carry the $\Af(\ol{N}_1)$ (resp., $\Af(\ol{N}_2)$) factor in $(\Af(\ol{N}_1)\otimes \Af(\ol{N}_2))^{\alpha_\mathscr{LG}\otimes\alpha_\mathscr{LG}}$ to the bulk factor within $\widetilde{\Af}(\ol{N}_1)$ (resp., $\widetilde{\Af}(\ol{N}_1)$) in the surface field extension picture.
\begin{proof} 
Under the identification of $\Af(\ol{N}_i)$ with $\Weyl(V^C(\ol{\Sigma}_i))\otimes \Weyl(V^S(\ol{N}_i))= \Weyl(V^C(\ol{\Sigma}_i)\oplus V^S(\ol{N}_i))$ given by the quantum CHH-isomorphism, one has
\begin{equation}
    \alpha_{\mathscr{LG}}(\lambda) = \Ad W^{C\oplus S}_i( 0\oplus (\lambda\oplus 0)).
\end{equation}
Identifying $\Af(\ol{N}_1)\otimes \Af(\ol{N}_2)$ with $\Weyl(V^C(\ol{\Sigma}_1)\oplus V^S(\ol{N}_1)\oplus V^C(\ol{\Sigma}_2)\oplus V^S(\ol{N}_2))$, it follows that
the automorphism group $\alpha_{\mathscr{LG}}\otimes  \alpha_{\mathscr{LG}})$ is implemented by the adjoint action of Weyl generators labelled by elements of
\begin{equation}
    \mathscr{G} = \{ 0\oplus(\lambda\oplus 0)\oplus 0\oplus (\lambda\oplus 0):\lambda\in\mathscr{LG}_\angle(\ol{N}_*)\}.
\end{equation}
By Lemma~\ref{lem:Weyl_fixedpoint} the invariant subalgebra $(\Af(\ol{N}_1)\otimes \Af(\ol{N}_2))^{\alpha_\mathscr{LG}\otimes\alpha_\mathscr{LG}}$ is generated by Weyl generators labelled by elements of $\mathscr{G}^\perp$, 
taking the symplectic complement with respect to $\sigma_1^C\oplus \sigma^S_1\oplus \sigma^C_2\oplus \sigma^S_2$, given explicitly as
\begin{equation}\label{eq:Gperp}
\mathscr{G}^\perp=\{\udl{v}_1\oplus (f_1\oplus h)\oplus \udl{v}_2\oplus (f_2\oplus -h):
\udl{v}_i\in V^C(\ol{\Sigma}_i),~ f_1,f_2,h\in \mathscr{LG}_\angle(\ol{N}_*)\}.
\end{equation}
Returning to the tensor product form, we have proved the first statement. 

To complete the proof, we may use results from the classical surface field extension. In CHH-form, one has $\widetilde{\Sol}(\ol{N}_i)=V^C(\ol{\Sigma}_i)\oplus
V^S(\ol{N}_*)\oplus V^S(\ol{N}_*)$ and the joint large gauge directions are
\begin{equation}
    \mathscr{G}_i = \{ 0\oplus(\lambda\oplus 0)\oplus (\lambda\oplus 0):\lambda\in\mathscr{LG}_\angle(\ol{N}_*)\}
\end{equation}
with symplectic complement
\begin{equation}
    \mathscr{G}_i^\perp =\{\udl{v}_i\oplus (f_1\oplus h) \oplus (f_2\oplus -h):
\udl{v}_i\in V^C(\ol{\Sigma}_i),~ f_1,f_2,h\in \mathscr{LG}_\angle(\ol{N}_*)\}.
\end{equation}
Thus $V^C(\ol{\Sigma}_1)\oplus \mathscr{G}_2^\perp \overset{\cong}{\leftarrow} \mathscr{G}^\perp \overset{\cong}{\rightarrow} \mathscr{G}_1^\perp\oplus V^C(\ol{\Sigma}_2)$ where the isomorphisms are given by permutations $(324)$ or $(34)$ on the ordering of summands. Passing to the quantisation, these become the stated permutations of tensor factors. 
\end{proof}
Next, let $\yen_i: \Af(\ol{N}_i)\to \widetilde{\Af}^{\mathscr{LG}}(\ol{N}_i)$, $\Gamma_i: \widetilde{\Af}^{\mathscr{LG}}(\ol{N}_i)\to \Af_\textup{rel}(\ol{N}_i)$ and
$\won_i:\widetilde{\Af}^{\mathscr{LG}}(\ol{N}_i)\to \widetilde{\Af}^{\mathscr{LG}}(\ol{N}_i)/\ker\Gamma_i$ be the 
relativisation, superselection and quotient maps as constructed in Section~\ref{sec:quantum_sfe}, recalling that $\won_i\circ\yen_i$ is an isomorphism by diagram~\eqref{eq:Ymap_diag2}.
\begin{figure}
    \begin{center}
    \begin{tikzcd}
     \Af^C(\ol{\Sigma}_1)\otimes\Af(\ol{N}_2)\arrow[d,"\id\otimes\yen_2"] &
     \Af^C(\ol{\Sigma}_1)\otimes \Af^S(\ol{N}_*)\otimes \Af^C(\ol{\Sigma}_2)
     \arrow[l,"\pi_{(23)}"']\arrow[r,"\id"]
     & \Af(\ol{N}_1)\otimes \Af^C(\ol{\Sigma}_2) \arrow[d,"\yen_1\otimes\id"]
    \\
    \Af^C(\ol{\Sigma}_1)\otimes\widetilde{\Af}^{\mathscr{LG}}(\ol{N_2})
    \arrow[d,"\id\otimes\won_2"] 
    & (\Af(\ol{N}_1)\otimes\Af(\ol{N}_2))^{\alpha_\mathscr{LG}\otimes \alpha_\mathscr{LG}}\arrow[d,"\won"]
    \arrow[r,"\overset{\bullet\circ}{\upsilon}"]\arrow[l,"\overset{\circ\bullet}{\upsilon}"']& 
    \widetilde{\Af}^{\mathscr{LG}}(\ol{N}_1)\otimes \Af^C(\ol{\Sigma}_2)
    \arrow[d,"\won_1\otimes\id"] \\
    \Af^C(\ol{\Sigma}_1)\otimes \widetilde{\Af}^{\mathscr{LG}}(\ol{N}_2)/I_2
    &
    \Af_\textnormal{glue}(\ol{N}_1; \ol{N}_2)\arrow[l,"\overset{\circ\bullet}{\gamma}"']\arrow[r,"\overset{\bullet\circ}{\gamma}"]
    & \widetilde{\Af}^{\mathscr{LG}}(\ol{N}_1)/I_1\otimes \Af^C(\ol{\Sigma}_2) 
\end{tikzcd}
\end{center}
\caption{Diagram for Theorem~\ref{thm:glue}. In the top line we make use of the identifications $\Af^C(\ol{\Sigma}_i)\otimes \Af^S(\ol{\Sigma}_i)=\Af(\ol{N}_i)$ for $i=1,2$, in the bottom line we have set $I_i=\ker\Gamma_i$.
\label{fig:glue}}
\end{figure}
\begin{proof}[Proof of Theorem~\ref{thm:glue}.] 
Continuing with the notation from the previous proof, Theorem~\ref{thm:Gamma} shows that
$I_i=\ker\Gamma_i$ is the closed two-sided $*$-ideal generated by differences of Weyl operators 
labelled by elements of $\mathscr{G}_i$. On the other hand, $I$ is generated by differences of Weyl operators labelled by $\mathscr{G}$. In tensor product form, such Weyl operators are written $1\otimes W^S_1(\lambda\oplus 0)\otimes 1\otimes W^S_2(\lambda\oplus 0)$ and are mapped to $1\otimes \widetilde{W}_2(0\oplus(\lambda\oplus 0)\oplus (\lambda\oplus 0))$ and
$\widetilde{W}_1(0\oplus(\lambda\oplus 0)\oplus (\lambda\oplus 0))\otimes 1$ by  $\overset{\circ\bullet}{\upsilon}$ and $\overset{\bullet\circ}{\upsilon}$. 
It follows that $\overset{\circ\bullet}{\upsilon}(I) = \Af^C(\ol{\Sigma}_1)\otimes I_2=\ker(\Gamma_2\otimes\id)$ and $\overset{\bullet\circ}{\upsilon}(I) = I_1\otimes \Af^C(\ol{\Sigma}_2)=
\ker(\Gamma_1\otimes\id)$. Then, because 
    $\overset{\bullet\circ}{\upsilon}$ is an isomorphism, there is a unique isomorphism $\overset{\bullet\circ}{\gamma}$ making the lower-right square in Fig.~\ref{fig:glue} commute.  As $\won_1\circ\yen_1$ is an isomorphism, the required isomorphism $\Af^C(\ol{\Sigma}_1)\otimes \Af (\ol{N}_2)\to \Af_\textnormal{glue}(\ol{N}_1; \ol{N}_2)$ is 
    \begin{equation}
    \overset{\bullet\circ}{\psi}=
    \overset{\bullet\circ}{\gamma}{}^{-1}\circ ((\won_1\circ\yen_1)\otimes\id)
    \end{equation}
    obtained by following arrows in the right-hand half of the diagram in Fig.~\ref{fig:glue}. An analogous argument establishes the other isomorphism in~\eqref{eq:glue_isos}, using the left-hand half of Fig.~\ref{fig:glue} as
    \begin{equation}
    \overset{\circ\bullet}{\psi}=\overset{\circ\bullet}{\gamma}{}^{-1}\circ (\id\otimes (\won_2\circ\yen_2)).
    \end{equation}
    Using commutativity of the lower right square of Fig.~\ref{fig:glue}, we compute for any $\udl{v}_i\in V^C(\ol{\Sigma}_i)$ and $f\oplus h\in V^S(\ol{N}_*)$
    \begin{align}\label{eq:psibo}
        \MoveEqLeft[6] \overset{\bullet\circ}{\psi} (W^C_1(\udl{v}_1)\otimes W^S_1(f\oplus h)\otimes W^C_2(\udl{v}_2))\nonumber \\ 
         &=\won\circ \overset{\bullet\circ}{\upsilon}{}^{-1}( W^C_1(\udl{v}_1)\otimes  W^S_1(f\oplus h)\otimes W^S_2(0\oplus -h)\otimes W^C_2(\udl{v}_2) ) \nonumber\\
        &=\won( W^C_1(\udl{v}_1)\otimes  W^S_1(f\oplus h)\otimes W^C_2(\udl{v}_2)\otimes W^S_2(0\oplus -h) ).
    \end{align}
    A similar computation yields
    \begin{align}\label{eq:psiob}
        \MoveEqLeft[6] \overset{\circ\bullet}{\psi}(W^C_1(\udl{v}_1)\otimes W^C_2(\udl{v}_2)\otimes W^S_2(-f\oplus -h)) \nonumber \\  
        &=\won( W^C_1(\udl{v}_1)\otimes  W^S_1(0\oplus h)\otimes W^C_2(\udl{v}_2)\otimes W^S_2(-f\oplus -h) ).
    \end{align}
    
    To prove that the diagram~\eqref{eq:glue_isos} commutes, we must show that $\overset{\bullet\circ}{\psi}=\overset{\circ\bullet}{\psi}\circ\Xi_{21}$, or equivalently that the right hand sides of~\eqref{eq:psibo} and~\eqref{eq:psiob}
    agree for all $\udl{v}_i\in V^C(\ol{\Sigma}_i)$ and $f\oplus h\in V^S(\ol{N}_*)$.
    To show this, note that due to the Weyl relations, we have
    \begin{equation}
        W^S_1(0\oplus h)\otimes W^S_2(-f\oplus -h)=(W^S_1(-f\oplus 0)\otimes W^S_2(-f\oplus 0))(W^S_1(f\oplus h)\otimes W^S_2(0\oplus -h)),
    \end{equation}
    and that, by definition of the ideal $I=\ker\won$ we have
    \begin{equation}
        \won(1\otimes W^S_1(-f\oplus 0)\otimes 1\otimes W^S_2(-f\oplus 0))=\won(1\otimes 1\otimes 1\otimes 1).
    \end{equation}
    The equality of~\eqref{eq:psibo} and~\eqref{eq:psiob} follows by the fact that $\won$ is a homomorphism. We thus have $\overset{\circ\bullet}{\psi}\circ\Xi_{21}=\overset{\bullet\circ}{\psi}$, and hence also 
    $\overset{\bullet\circ}{\psi}\circ\Xi_{12}=\overset{\bullet\circ}{\psi}\circ\Xi_{21}^{-1}=\overset{\circ\bullet}{\psi}$.
\end{proof}
The isomorphisms $\overset{\bullet\circ}{\psi}$ and $\overset{\circ\bullet}{\psi}$ (constructed using relativisation maps) can be interpreted as follows: the map $a\mapsto \overset{\circ\bullet}{\psi}(a\otimes 1)$ from
$\Af(\ol{N_1})$ to $\Af_\textnormal{glue}(\ol{N}_1;\ol{N}_2)$ maps semi-local observables  
on $\ol{N}_1$ to large gauge invariant observables on $\ol{N}_1\cup\ol{N}_2$. Following the analogy with semi-local observables and Wilson lines described in Sec.~\ref{sec:class_semi-loc}, one may loosely think of the relativisation map $\yen_1$ as mapping (smeared) Wilson lines on $\ol{\Sigma}_1$ with endpoints on the shared boundary $\partial\ol{\Sigma}_1=\partial\ol{\Sigma}_2$ to closed loop observables on $\ol{\Sigma}$, closing up the lines with a (smeared analogue of a) geodesic in $\ol{\Sigma}_2$ as illustrated in Fig.~\ref{fig:Wilson_relativise}. The manner in which these lines are closed up of course depends on the choice of Cauchy surface. More precisely, the map $\overset{\bullet\circ}{\psi}$ depends on the choices of Cauchy surfaces $\ol{\Sigma}_2\subset\ol{N}_2$. This is simply an instance of the fact that relativisation maps (or dressings) mapping into invariant observables are highly non-unique and rely on some choice of quantum reference frame, see e.g.~\cite{vanrietveldeChangePerspectiveSwitching2020,caretteOperationalQuantumReference2025}.
\begin{figure}[ht]
\includegraphics[width=\textwidth]{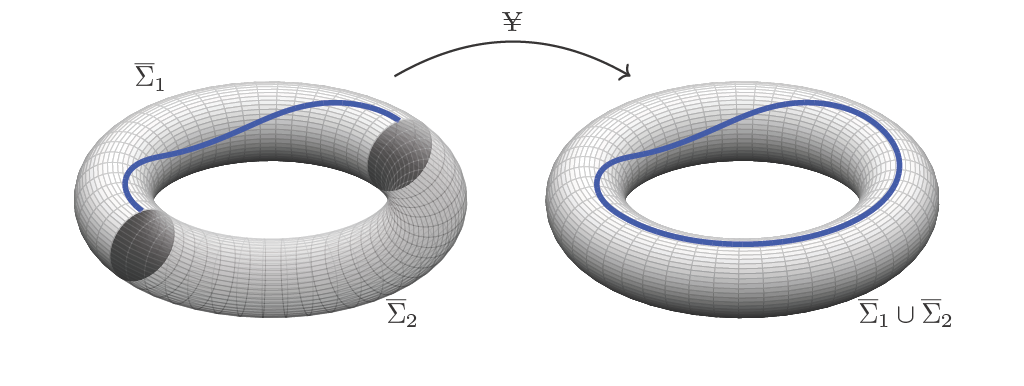}
\caption{Closing up Wilson lines by means of relativisation.}
\label{fig:Wilson_relativise}
\end{figure}

The algebra $\Af_{\textup{glue}}(\ol{N}_1,\ol{N}_2)$ can also be interpreted as 
an algebra of gauge invariant observables of the electromagnetic field on $M=\mathscr{D}(\ol{N}_1\cup\ol{N_2})$, extending $\Af(M)$. 
\begin{theorem}\label{thm:Mtoglue}
There exists a unique injective ${}^*$-homomorphism $\iota:\Af(M)\to \Af_{\textup{glue}}(\ol{N}_1,\ol{N}_2)$ such that for each $\udl{A}\in \Sol_{\mathscr{G}}(M)$ with $\udl{A}_i=\udl{A}\restriction_{\ol{N}_i}\in \Sol(\ol{N}_i)$ for $i\in \{1,2\}$ we have
\begin{equation}
\label{eq:Mtoglue}
    \iota(W^M([\udl{A}]))=\won(W([\udl{A}_1])\otimes W([\udl{A}_2])).
\end{equation}
Furthermore, for any regular Cauchy surfaces $\ol{\Sigma}_i\subset \ol{N}_i$ as in Thm.~\ref{thm:glue} and $f_i\in \Omega^1_{0\diff^*}(M)$ for $i\in \{1,2\}$ with $\supp(f_i)\subset N_i$, we have
\begin{align}
    \iota(W^M([G^\PJ f_1]))=\overset{\bullet\circ}{\psi}(W([G^\PJ f_1\restriction_{\ol{N}_1}])\otimes 1)=\won(W([G^\PJ f_1\restriction_{\ol{N}_1}])\otimes 1),\nonumber\\
    \iota(W^M([G^\PJ f_2]))=\overset{\circ\bullet}{\psi}(1\otimes W([G^\PJ f_2\restriction_{\ol{N}_2}]))=\won(1\otimes W([G^\PJ f_2\restriction_{\ol{N}_2}])).
    \label{eq:Mloctoglue}
\end{align}
\end{theorem}
\begin{proof}
We first define an injective symplectic map $\mathfrak{i}:\Sol_{\mathscr{G}}(M)\to \Sol_{\mathscr{G}}(\ol{N}_1)\oplus \Sol_{\mathscr{G}}(\ol{N}_2)$, which relies on some choice of gauge fixing in $\Sol_{\mathscr{G}}(M)$. We then show that the map $\iota=\won\circ \Weyl(\mathfrak{i})$ is well-defined and independent of this choice of gauge fixing and satisfies Eq.~\eqref{eq:Mtoglue}.

To fix a gauge, choose a Cauchy surface $\Sigma=\ol{\Sigma}_1\cup\ol{\Sigma}_2$ of $M$ such that $\ol{\Sigma}_i$ is a regular Cauchy surfaces of $\ol{N}_i$ for $i\in \{1,2\}$. Let $[\udl{A}]\in \Sol_{\mathscr{G}}(M)$ and use the Hodge-Helmholtz decomposition on $\ol{\Sigma}$ to choose a representative $\udl{A}\in [\udl{A}]$ such that $\udl{A}\restriction_{\ol{\Sigma}}\in \Omega^1_{\diff^*}(\ol{\Sigma})$. Note that this uniquely determines $\mathbf{A}=\udl{A}\restriction_{\ol{\Sigma}}$ from $[\udl{A}]$. Similarly, $\mathbf{E}=-\nml_{\ol{\Sigma}}\diff\udl{A}$ is uniquely defined from $[\udl{A}]$. This allows us to define 
    \begin{equation}
    \label{eq:Mtoprod}
        \mathfrak{i}([\udl{A}])=\left(\mathfrak{I}_{\ol{\Sigma}_1}\right)^{-1}([\mathbf{A}\restriction_{\ol{\Sigma}_1}]\oplus \mathbf{E}\restriction_{\ol{\Sigma}_1})\oplus \left(\mathfrak{I}_{\ol{\Sigma}_2}\right)^{-1}([\mathbf{A}\restriction_{\ol{\Sigma}_2}]\oplus \mathbf{E}\restriction_{\ol{\Sigma}_2}).
    \end{equation} 
    Note that $\mathfrak{i}([\udl{A}])=0$ implies $\mathbf{E}=0$ and $\mathbf{A}\restriction_{\ol{\Sigma}_i}\in \mathscr{G}(\ol{\Sigma}_i)$, which by smoothness of $\mathbf{A}$ implies $\mathbf{A}\in \diff\Omega^0(\Sigma)$. Thus $[\udl{A}]=0$ and $\mathfrak{i}$ is injective. By explicit computation, one readily sees that the map $\mathfrak{i}$ is symplectic w.r.t.\ the obvious symplectic structures. We can thus define an injective ${}^*$-homo{\-}morphism 
    \begin{equation}
        \Weyl(\mathfrak{i}):\Af(M)\to\Af(\ol{N}_1)\otimes\Af(\ol{N}_2).
    \end{equation}
    It is straightforward to check that
    \begin{equation}
        \Weyl(\mathfrak{i})(\Af(M))\subset (\Af(\ol{N}_1)\otimes\Af(\ol{N}_2))^{\alpha_{\mathscr{LG}}\otimes\alpha_{\mathscr{LG}}}.
    \end{equation}
    As a result, we may define $\iota=\won\circ \Weyl(\mathfrak{i})$. It remains to be shown that $\iota$ is injective. 
    As $\Af(M)$ is the Weyl algebra of a symplectic space, it is simple \cite{slawnyFactorRepresentationsCalgebra1972}. Then  $\iota(1)=\won(1)\neq 0$ implies that $\ker(\iota)=\{0\}$.

    To show that $\iota$ is uniquely determined by Eq.~\eqref{eq:Mtoglue}, and thus independent of the choice of $\Sigma$, let $\udl{A}\in \Sol(M)$ with $\udl{A}\restriction_{\Sigma}\in \Omega^1_{\diff^*}(\ol{\Sigma})$ as above. Then for any representative $\udl{A}'\in[\udl{A}]$ with $\mathbf{A}'=\udl{A}'\restriction_{\Sigma}$ and $\mathbf{E}'=-\nml_{\ol{\Sigma}}\diff\udl{A}'$, we have $\mathbf{E}=\mathbf{E}'$ and $\mathbf{A}'=\mathbf{A}+\diff\Lambda$, where $\Lambda\in \Omega^0(\Sigma)$ can be chosen such that $\lambda:=\Lambda\restriction_{\partial\ol{\Sigma}_*}\in\mathscr{LG}_\angle(\ol{N}_*)$. Now let $\udl{A}_i'=\udl{A}\restriction_{\ol{N}_i}\in \Sol(\ol{N}_i)$, such that $\udl{A}_i'\restriction_{\ol{\Sigma}_i}=\mathbf{A}\restriction_{\ol{\Sigma}_i}+\diff\Lambda\restriction_{\ol{\Sigma}_1}$ and $\nml_{\ol{\Sigma}_i}\diff\udl{A}_i'=-\mathbf{E}\restriction_{\ol{\Sigma}_i}$. It follows that \begin{equation}
        [\udl{A}_1']\oplus [\udl{A}_2']=\mathfrak{i}([\udl{A}])+\mathfrak{G}(\lambda)\oplus \mathfrak{G}(\lambda),
    \end{equation}
    and thus
    \begin{equation}
        W([\udl{A}_1'])\otimes W([\udl{A}_2'])=(\Weyl(\mathfrak{i})(W^M[\udl{A}]))(U^{(1)}_{\mathscr{LG}}(\lambda)\otimes U^{(2)}_{\mathscr{LG}}(\lambda)).
    \end{equation}
    From the definition of $I=\ker(\won)$, it follows that 
    \begin{equation}
        \won(W([\udl{A}_1'])\otimes W([\udl{A}_2']))=\iota(W^M([\udl{A}])).
    \end{equation}
    Since the operators of the form $W^M([\udl{A}])$ generate $\Af(M)$, uniqueness of $\iota$ follows.

    We now prove Eq.~\eqref{eq:Mloctoglue}. It suffices to consider the case for $f_1\in \Omega^1_{0\diff^*}(M)$ with $\supp(f_1)\subset N_1$. Note that since $N_2$ and $N_1$ are spacelike separated, we have $G^{\PJ}f_1\restriction_{\ol{N}_2}=0$.
    
    Fix regular Cauchy surfaces $\ol{\Sigma}_i\subset\ol{N}_i$ as in Thm.~\ref{thm:glue} defining the isomorphism $
         \overset{\bullet\circ}{\psi}:\Af (\ol{N}_1)\otimes\Af^C(\ol{\Sigma}_2)\to \Af_\textnormal{glue}(\ol{N}_1; \ol{N}_2)$.  Without loss of generality, we may take the map $\mathfrak{i}:\Sol_{\mathscr{G}}(M)\to \Sol_{\mathscr{G}}(\ol{N}_1)\oplus \Sol_{\mathscr{G}}(\ol{N}_2)$ to be defined through the gauge fixing of initial data at $\Sigma=\ol{\Sigma}_1\cup\ol{\Sigma}_2$. We wish to show that
         \begin{equation}
             \overset{\bullet\circ}{\psi}(W([G^{\PJ}f_1\restriction_{\ol{N}_1}])\otimes 1_{\Af^C(\ol{\Sigma}_2)})=(\won\circ\Weyl(\mathfrak{i}))(W^M([G^{\PJ}f_1])).
         \end{equation}
         Using the Hodge-Helmholtz decomposition, we decompose
     \begin{equation}
         (G^\PJ f_1)\restriction_{\Sigma}=\mathbf{A}+\diff\alpha,
     \end{equation} where $\mathbf{A}\in \Omega^1_{\diff^*}(\Sigma)$ and $\alpha\in \Omega^0(\Sigma)$. Using the CHH decomposition, we find
    \begin{equation}
        (\mathfrak{K}_{\ol{\Sigma}_1}\oplus\mathfrak{K}_{\ol{\Sigma}_2})(\mathfrak{i}([G^\PJ f_1]))=F_1\oplus H_1\oplus \alpha_1\restriction_{\partial\ol{\Sigma}_1}\oplus 0\oplus 0\oplus 0\oplus \alpha_2\restriction_{\partial\ol{\Sigma}_2}\oplus 0,
    \end{equation}
    where $H_1=(\nml_{\Sigma}\diff G^\PJ f_1)\restriction_{\ol{\Sigma}_1}$ and
    \begin{equation}
        \mathbf{A}\restriction_{\ol{\Sigma}_1}=F_1+\diff\alpha_1,\qquad \mathbf{A}\restriction_{\ol{\Sigma}_2}=\diff\alpha_2
    \end{equation}
   with $\alpha_i\in \Omega^0(\ol{\Sigma}_i)$ with $\Delta\alpha_i=0$ and $\alpha_i\restriction_{\partial\ol{\Sigma}_i}=\mathscr{LG}_\angle(\ol{N}_i)$. As a result, we find
    \begin{align}
        (\won\circ\Weyl(\mathfrak{i}))(W^M([G^\PJ f_1]))=&\won(W^C_1(F_1\oplus H_1)\otimes W^S_1(\alpha_1\restriction_{\partial\ol{\Sigma}_1}\oplus 0)\otimes 1\otimes W^S_2(\alpha_2\restriction_{\partial\ol{\Sigma}_2}\oplus 0))\nonumber\\
        =&\won(W^C_1(F_1\oplus H_1)\otimes W^S_1(\alpha_1\restriction_{\partial\ol{\Sigma}_1}-\alpha_2\restriction_{\partial\ol{\Sigma}_2}\oplus 0)\otimes 1\otimes 1)\nonumber\\
        =&\overset{\bullet\circ}{\psi}(W^C_1(F_1\oplus H_1)\otimes W^S_1(\alpha_1\restriction_{\partial\ol{\Sigma}_1}-\alpha_2\restriction_{\partial\ol{\Sigma}_2}\oplus 0)\otimes 1),\label{eq:Mloctoglue_calc}
    \end{align}
    having identified $\Af(\ol{N}_i)=\Af^{C\oplus S}(\ol{\Sigma}_i)$ using the CHH isomorphisms, the defining property of $I=\ker(\won)$ and Eq.~\eqref{eq:psibo}.

    Since $(G^\PJ f_1)\restriction_{\ol{\Sigma}_2}=0$, it follows that
    \begin{equation}
        \diff\alpha\restriction_{\ol{\Sigma}_2}=-\diff\alpha_2,
    \end{equation}
    In particular, we may choose $\alpha$ such that $\alpha\restriction_{\ol{\Sigma}_2}=-\alpha_2$, we find
    \begin{equation}
        (G^\PJ f_1)\restriction_{\ol{\Sigma}_1}=F_1+\diff(\alpha_1+\alpha),
    \end{equation}
    where $(\alpha_1+\alpha)\restriction_{\partial\ol{\Sigma}_1}=\alpha_1\restriction_{\partial\ol{\Sigma}_1}-\alpha_2\restriction_{\partial\ol{\Sigma}_2}$. As a result, we find that
    \begin{equation}
        \mathfrak{K}_{\ol{\Sigma}_1}([G^\PJ f_1\restriction_{\ol{N}_1}])=F_1\oplus H_1\oplus \alpha_1\restriction_{\partial\ol{\Sigma}_1}-\alpha_2\restriction_{\partial\ol{\Sigma}_2}\oplus 0,
    \end{equation}
    and thus
    \begin{equation}
        W([G^\PJ f_1\restriction_{\ol{N}_1}])=W^C_1(F_1\oplus H_1)\otimes W^S_1(\alpha_1\restriction_{\partial\ol{\Sigma}_1}-\alpha_2\restriction_{\partial\ol{\Sigma}_2}\oplus 0).
    \end{equation}
    Combined with Eq.~\eqref{eq:Mloctoglue_calc}, the desired result follows.
\end{proof}

Thm.~\ref{thm:Mtoglue} shows that we can recover the algebra $\Af(M)$ of observables on $M$ from the glued algebra $\Af_{\textup{glue}}(\ol{N}_1,\ol{N}_2)$ by restricting to observables that are sufficiently regular at the gluing interface $\angle\ol{N}_*$. Concretely,  
$\Af(M)$ is isomorphic to the algebra generated by operators of the form
\begin{equation}
    \won(W([\udl{A}_1])\otimes W([\udl{A}_2]))\in \Af_{\textup{glue}}(\ol{N}_1,\ol{N}_2),
\end{equation}
for which $\udl{A}_i\in \Sol(\ol{N}_i)$ are such that there exists a solution $\udl{A}\in \Sol(M)$ with $\udl{A}\restriction_{\ol{N}_i}=\udl{A}_i$. This is only possible because we have included semi-local observables. Restricting to local (or large gauge invariant) observables in the gluing procedure would yield a glued algebra without the holonomy observables that are relevant to descriptions of the Aharonov-Bohm effect, see e.g.~\cite{sandersElectromagnetismLocalCovariance2014}. 

In general $\Af_{\textup{glue}}(\ol{N}_1,\ol{N}_2)$ contains many observables that lie outside $\iota(\Af(M))$. Examples include operators of the form $\won(W([\diff\Lambda])\otimes 1)$ for $[\diff\Lambda]\in \mathscr{LG}(\ol{N})$ with $\diff\Lambda\restriction_{\angle\ol{N}}\neq 0$. These correspond to exponentiations of observables measuring the smeared electric flux densities through the shared corner $\angle\ol{N}_*$, and thus correspond to fields smeared along a co-dimension 2 spacelike hypersurface. Such smearings are typically regarded as too singular to be included in the bulk description of a quantum field theory (for a result in this direction, see Theorem II in~\cite{woronowiczGeneralizationTheoremWightman1968}). In this approach to gluing, one must therefore discard singular observables from the glued algebra of observables to recover a sensible algebra of semi-local observables on the glued spacetime.

We had previously seen in Thm.~\ref{thm:glue} that the glued algebra $\Af_{\textup{glue}}(\ol{N}_1,\ol{N}_2)$ can be parametrized by $\Af^C(\ol{\Sigma}_1)\otimes \Af^S(\ol{N}_*)\otimes \Af^C(\ol{\Sigma}_2)$ using either the isomorphisms $\overset{\bullet\circ}{\psi}$ or $\overset{\circ\bullet}{\psi}$ (depending on a choice of Cauchy surface $\Sigma\subset \ol{N}_1\cup\ol{N}_2\subset M$). What we have now seen in Thm.~\ref{thm:Mtoglue} is that when these maps are applied to local observables in $\Af_{\loc}(\ol{N}_1)$ or $\Af_{\loc}(\ol{N}_2)$ respectively, their image is independent of the choice of $\Sigma$. Concretely, as a consequence of Eq.~\eqref{eq:Mloctoglue}, the diagram of Fig.~\ref{fig:locglue} commutes, in which the maps $\chi$ and $\iota$ are $\Sigma$-independent. Note that we have  reinstated the CHH-isomorphism $\varphi_{\ol{\Sigma}_1}$ in Fig.~\ref{fig:locglue} rather than identifying $\Af(\ol{N}_1)$ with $\Af^C(\ol{\Sigma}_1)\otimes\Af^S(\ol{N}_1)$.
\begin{figure}
    \begin{center}
    \begin{tikzcd}
     \Af_{\loc}(\ol{N}_1)\arrow[r,"\chi"] \arrow[d,"\varphi_{\ol{\Sigma}_1}\otimes 1_{\Af^C(\ol{\Sigma}_2)}"] &
     \Af(M)
     \arrow[d,"\iota"]
    \\
    \Af^C(\ol{\Sigma}_1)\otimes \Af^S(\ol{N}_*)\otimes \Af^C(\ol{\Sigma}_2)
    \arrow[r,"\overset{\bullet\circ}{\psi}"] 
    &  \Af_{\textup{glue}}(\ol{N}_1,\ol{N}_2).
\end{tikzcd}
\end{center}
\caption{Commuting diagram following from Eq.~\eqref{eq:Mloctoglue}, showing two equivalent ways of mapping local observables in $\Af_{\loc}(\ol{N}_1)$ into the glued algebra $\Af_{\textup{glue}}(\ol{N}_1,\ol{N}_2)$. Here $\chi$ is the injective $*$-homomorphism for which $\chi(W_{\loc}(f\restriction_{\overline{N}_1}))=W^M([-G^\PJ f])$ for all $f\in \Omega^1_{0\diff^*}(M)$ with $\supp(f)\subset N_1$, $\varphi_{\ol{\Sigma}_1}\otimes 1_{\Af^C(\ol{\Sigma}_2)}:a\mapsto \varphi_{\ol{\Sigma}_1}(a)\otimes 1$ with $\varphi_{\ol{\Sigma}_1}:\Af(\ol{N}_1)\to \Af^C(\ol{\Sigma}_1)\otimes \Af^S(\ol{N}_1)$ is the quantum CHH isomorphism as in Def.~\ref{def:kin_transf} and $\iota$ is as in Thm.~\ref{thm:Mtoglue}. Note that $\varphi_{\ol{\Sigma}_1}$ was suppressed in Eq.~\eqref{eq:Mloctoglue}.
\label{fig:locglue}}
\end{figure}

\subsection{Gluing of states}

So far, we have discussed how the observable algebras from $\ol{N}_1$ and $\ol{N}_2$ can be glued together. To conclude this section, we consider how gluing may be implemented at the level of states. Suppose one is given states $\omega_i$ on $\Af(\ol{N}_i)$ for $i=1,2$, perhaps obeying some compatibility condition, we aim to construct one or more glued states on $\Af_{\textnormal{glue}}(\ol{N}_1;\ol{N}_2)$ that represents a gluing of $\omega_1$ and $\omega_2$ at the shared corner $\angle\ol{N}_*$. To do this, we impose some minimal conditions on the glued state in terms of its partial traces to algebras associated with Cauchy surfaces $\Sigma\subset\ol{N}_1\cup\ol{N}_2\subset M$. 
\begin{definition}
    \label{def:glue_states}
   Let $\Sigma=\ol{\Sigma}_1\cup\ol{\Sigma}_2\subset M$ be a Cauchy surface as in Thm.~\ref{thm:glue}. A state $\omega$ on $\Af_\textnormal{glue}(\ol{N}_1;\ol{N}_2)$ is a \emph{gluing} at $\Sigma$ of
    states $\omega_i$ on $\Af(\ol{N}_i)$ ($i=1,2$) if the pulled back states $\overset{\circ\bullet}{\omega}=\omega\circ\overset{\circ\bullet}{\psi}$ 
    on $\Af^C(\ol{\Sigma}_1)\otimes\Af(\ol{N}_2)$ and  $\overset{\bullet\circ}{\omega}=\omega\circ\overset{\bullet\circ}{\psi}$ on $\Af(\ol{N}_1)\otimes\Af^C(\ol{\Sigma}_2)$ have partial traces
    \begin{equation}
        \overset{\circ\bullet}{\omega}{}^{C}_i = \omega^C_i =\overset{\bullet\circ}{\omega}{}^{C}_i, \qquad
        \overset{\bullet\circ}{\omega}{}^{S}_1 = \omega^S_1, \qquad \overset{\circ\bullet}{\omega}{}^{S}_2 = \omega^S_2.
    \end{equation}
    for $i=1,2$, where the notation $\varphi^{C/S}_{i}$ is a partial trace of a state $\varphi$ onto $\Af^{C/S}(\ol{\Sigma}_i)$, assuming that $\varphi$ is a state on a tensor product of algebras including a single tensor factor of $\Af^{C/S}(\ol{\Sigma}_i)$. We also abbreviate $(\omega_i){}^{C/S}_{i}$ to $\omega^{C/S}_i$.
\end{definition}
In fact, it can be shown that $\overset{\circ\bullet}{\omega}{}^{C}_i= \overset{\bullet\circ}{\omega}{}^{C}_i$, so not all of these conditions have to be checked. Note that the notion of gluing given above depends on a choice of Cauchy surface $\Sigma$ and in general it is not expected that states constituting a gluing at one Cauchy surface $\Sigma$ will constitute a gluing at every other appropriate Cauchy surface $\Sigma'$ of $M$. Nevertheless, this notion of gluing provides an interesting starting point for constructing states on the glued algebra $\Af_\textnormal{glue}(\ol{N}_1;\ol{N}_2)$ (and hence by Thm.~\ref{thm:Mtoglue} on $\Af(M)$) from states on $\Af(\ol{N}_1)$ and $\Af(\ol{N}_2)$. For the remainder of this section, we take the Cauchy surface $\Sigma\subset M$ as in Thm.~\ref{thm:glue} fixed.

Gluings only exist for states that are compatible at the boundary in the following sense. 
\begin{lemma} \label{lem:comp}
If a gluing of $\omega_1$ and $\omega_2$ exists then $\omega_1^S(W(f\oplus h)) = \omega_2^S(W(-f\oplus -h))$ for
all $f\oplus h\in V^S(\ol{N}_*)$.
\end{lemma}
\begin{proof}
    Suppose $\omega$ is a gluing of $\omega_1$ and $\omega_2$ as above. Let $f\oplus h\in V^S(\ol{N}_*)$, then \begin{equation}
        \omega_1^S(W(f\oplus h))=\overset{\bullet\circ}{\omega}{}^{S}_1(W(f\oplus h))=\overset{\bullet\circ}{\omega}{}(1\otimes W_1(f\oplus h)\otimes 1),
    \end{equation}
    and
    \begin{equation}
        \omega_2^S(W(-f\oplus -h))=\overset{\circ\bullet}{\omega}{}^{S}_2(W(-f\oplus -h))=\overset{\circ\bullet}{\omega}{}(1\otimes 1 \otimes W_2(-f\oplus -h)).
    \end{equation}
    Now note that by Thm.~\ref{thm:glue}, we have
    \begin{equation}
        \overset{\bullet\circ}{\omega}=\omega \circ \overset{\bullet\circ}{\psi}=\overset{\circ\bullet}{\omega}\circ\overset{\circ\bullet}{\psi}{}^{-1} \circ \overset{\bullet\circ}{\psi}=\overset{\circ\bullet}{\omega}\circ\Xi_{21}.
    \end{equation} 
    The result follows using the definition of $\Xi_{21}$ in Eq.~\eqref{eq:Xi21}. 
\end{proof}

Examples of glued states constructed from compatible $\omega_i$ are easily given,\footnote{These states could be constructed even if $\omega_1$ and $\omega_2$ do not obey the compatibility condition of Lemma~\ref{lem:comp} but the sense in which they represent gluings is less clear.}  such as
\begin{equation}
\omega_{1(2)}=(\omega_1\otimes \omega_2^{C})\circ \overset{\bullet\circ}{\psi}{}^{-1}, \qquad 
\omega_{(1)2}=(\omega_1^C\otimes \omega_2)\circ \overset{\circ\bullet}{\psi}{}^{-1}.
\end{equation} 
Moreover, any mixture of these states is a gluing of
$\omega_1$ and $\omega_2$, including the equal-weight mixture $\omega_{12}=\tfrac{1}{2}(\omega_{(1)2}+\omega_{1(2)})$ which treats $\omega_1$ and $\omega_2$ symmetrically. 

We can describe a number of properties of $\omega_{(1)2}$ and $ \omega_{1(2)}$.
\begin{proposition}
    \label{prop:glue_states}
    Let $\omega_{1,2}$ be compatible states on $\Af(\ol{N}_i)$ and $\omega_{(1)2}$ and $\omega_{1(2)}$ as above.
    Let $\omega_{(1)2}'=\omega_{(1)2}\circ \iota$ and $\omega_{1(2)}'=\omega_{1(2)}\circ \iota$ be induced states on $\Af(M)$, where $\iota:\Af(M)\to \Af_{\textnormal{glue}}(\ol{N}_1;\ol{N_2})$ is as in Thm.~\ref{thm:Mtoglue}. 
    \begin{enumerate}
        \item $\omega_{(1)2}=\omega_{1(2)}$ if and only if $\omega_i=\omega_i^C\otimes \omega_i^S$ for all $i\in\{1,2\}$. 
        \item If $\omega_i$ are quasi-free for both $i\in\{1,2\}$, then $\omega_{(1)2}'$ and $\omega_{1(2)}'$ are quasifree. 
        \item For $i\in\{1,2\}$, let $f_i\in \Omega^1_{0\diff^*}(\ol{M})$ such that $\supp(f_i)\subset N_i$. Then 
        \begin{equation}
            \omega_{1(2)}'(W^M([G^\PJ f_1]))=\omega_1(W([G^\PJ f_1\restriction_{\ol{N}_1}])),\qquad \omega_{(1)2}'(W^M([G^\PJ f_2]))=\omega_2(W([G^\PJ f_2\restriction_{\ol{N}_2}])).
        \end{equation}
    \end{enumerate}
\end{proposition}
\begin{proof}
    1. Note that $\omega_{(1)2}=\omega_{1(2)}$ if and only if $(\omega_1^C\otimes \omega_2)=(\omega_1\otimes \omega_2^C)\circ \Xi_{12}$ if and only if $(\omega_1^C\otimes \omega_2)\circ \Xi_{21}=(\omega_1\otimes \omega_2^C)$. Let $\udl{v}_i^C\in V^C(\ol{\Sigma}_i)$ and $\udl{v}^S\in V^S(\ol{N}_*)$ then
    \begin{align}
        \Xi_{12}(1\otimes W_2^C(\udl{v}_2^C)\otimes W_2^S(\udl{v}^S))=1\otimes W_1^S(-\udl{v}^S)\otimes W_2^C(\udl{v}_2^C),\nonumber\\
        \Xi_{21}(W_1^C(\udl{v}_1^C)\otimes W_1^S(\udl{v}^S)\otimes 1)=1\otimes W_2^C(\udl{v}_2^C)\otimes W_2^S(-\udl{v}^S).
    \end{align}
    Thus suppose $\omega_{(1)2}=\omega_{1(2)}$, then
    \begin{align}
        \omega_2(W_2^C(\udl{v}_2^C)\otimes W_2^S(\udl{v}^S))=&((\omega_1\otimes \omega_2^C)\circ\Xi_{12})(1\otimes W_2^C(\udl{v}_2^C)\otimes W_2^S(\udl{v}^S))\nonumber\\
        =&(\omega_1\otimes \omega_2^C)(1\otimes W_1^S(-\udl{v}^S)\otimes W_2^C(\udl{v}_2^C))\nonumber\\=&\omega_1^S(W_1^S(-\udl{v}^S))\omega_2^C(W_2^C(\udl{v}_2^C))=\omega_2^C(W_2^C(\udl{v}_2^C))\omega_2^S(W_2^S(\udl{v}^S)).
    \end{align}
    Similarly 
    \begin{equation}
        \omega_1(W_1^C(\udl{v}_1^C)\otimes W_1^S(\udl{v}^S))=\omega_1^C(W_1^C(\udl{v}_1^C))\omega_1^S(W_1^S(\udl{v}^S)).
    \end{equation}
    It follows that $\omega_i=\omega_i^C\otimes \omega_i^S$. To show the reverse, one simply shows that $(\omega_1^C\otimes \omega_2^C\otimes \omega_2^S)=(\omega_1^C\otimes \omega_1^S\otimes \omega_2^C)\circ \Xi_{12}$ holds on all Weyl generators.

    2. For each $[\udl{A}]\in \Sol_\mathscr{G}(M)$, let 
    \begin{equation}
        \udl{v}_1^C\oplus (f_1\oplus h_1)\oplus \udl{v}_2^C\oplus (f_2\oplus -h_1)=(\mathfrak{K}_{\ol{\Sigma}_1}\oplus \mathfrak{K}_{\ol{\Sigma}_2})(\mathfrak{i}([\udl{A}])),
    \end{equation}
    we then define a linear map $\mathfrak{i}_{1(2)}:\Sol_\mathscr{G}(M)\to V^C(\ol{\Sigma}_1)\oplus V^S(\ol{N}_1)\oplus V^C(\ol{\Sigma}_2)$
    by
    \begin{equation}
        \mathfrak{i}_{1(2)}([\udl{A}])=\udl{v}_1^C\oplus(f_1-f_2\oplus h_1)\oplus\udl{v}_2^C.
    \end{equation}
    We compute
    \begin{align}
        \iota(W^M([\udl{A}]))&=\won(W_1^C(\udl{v}_1^C)\otimes W_1^S(f_1-f_2\oplus h_1)\otimes W_2^C(\udl{v}_2^C)\otimes W_2^S(0\oplus -h_1)) \nonumber\\
        &=\overset{\bullet\circ}{\psi}(W_1^C(\udl{v}_1^C)\otimes W_1^S(f_1-f_2\oplus h_1)\otimes W_2^C(\udl{v}_2^C)).
    \end{align}
    Denoting by $\mu_1$ the inner product on $V^C(\ol{\Sigma}_1)\oplus V^S(\ol{N}_1)$ associated with $\omega_1$ and by $\mu_2^C$ the inner product on $V^C(\ol{\Sigma}_2)$ associated with $\omega_2^C$, we find 
    \begin{align}
        \omega_{1(2)}(\iota(W^M([\udl{A}])))=&\exp\left(-\frac{1}{2}\left(\Vert \udl{v}_1^C\oplus(f_1-f_2\oplus h_1)\Vert_{\mu_1}^2+\Vert \udl{v}_2^C\Vert_{\mu_2^C}^2\right)\right)\nonumber\\
        =&\exp\left(-\frac{1}{2}\Vert \mathfrak{i}_{1(2)}([\udl{A}])\Vert_{\mu_1\oplus\mu_2^C}^2\right).
    \end{align}
    Thus $\omega_{1(2)}'=\omega_{1(2)}\circ\iota$ is quasi-free with associated inner product 
    \begin{equation}
        \mu_{1(2)}=(\mu_1\oplus \mu_2^C)\circ(\mathfrak{i}_{1(2)})^{\otimes 2}.
    \end{equation}
    One finds a similar result for $\omega_{(1)2}'$.

    3. Using Thm.~\ref{thm:Mtoglue}, we find
    \begin{align}
        \omega_{1(2)}'(W^M([G^\PJ f_1]))=&((\omega_1\otimes \omega_2^{C})\circ \overset{\bullet\circ}{\psi}{}^{-1})(\iota(W^M([G^\PJ f_1])))\nonumber\\
        =&((\omega_1\otimes \omega_2^{C})\circ \overset{\bullet\circ}{\psi}{}^{-1})(\overset{\bullet\circ}{\psi}(W^M([G^\PJ f_1\restriction_{\ol{N}_1}])\otimes 1))\nonumber\\
        =&(\omega_1\otimes \omega_2^{C})(W^M([G^\PJ f_1\restriction_{\ol{N}_1}])\otimes 1)=\omega_1(W^M([G^\PJ f_1\restriction_{\ol{N}_1}])).
    \end{align}
    Similarly, one finds the analogous result for $\omega_{(1)2}'$.
\end{proof}

A major shortcoming of this procedure is that the glued state
may exhibit singular behaviour at the boundary, such as the failure of the Hadamard property (developed for semilocal observables in~\cite{fewsterHadamardStatesSemilocal}) -- perhaps indicating the existence of a surface defect. This is  related to the presence of observables in $\Af_{\textnormal{glue}}(\ol{N}_1;\ol{N}_2)$ that are too singular to be included in $\Af(M)$ as discussed in Sec.~\ref{sec:glue_alg}, which raises the question of whether there even exist quasi-free on $\Af_{\textnormal{glue}}(\ol{N}_1;\ol{N}_2)$ that pull back to Hadamard states on $\Af(M)$. In future work, it would be interesting to consider smoothed gluing prescriptions that might preserve the Hadamard property. For example, 
Hadamard states of scalar fields can be glued using partition of unity methods~\cite{gerardConstructionHadamardStates2014}, which can be regarded as gluing over a smoothed boundary. Alternatively, one could consider regularisation procedures analogous to~\cite{brumVacuumlikeHadamardStates2014}. Such prescriptions might sacrifice a small neighbourhood of the boundary, in which the state does not match either of the original states, resulting in a family of approximate gluings rather than a unique glued state.

\section{Conclusions and Outlook}\label{sec:conclusion}

In this work, we have initiated the development of a systematic algebraic framework for describing quantum electromagnetism on spacetimes with boundaries and corners, adopting the perspective that semi-local observables provide the correct language to capture the physics of QFT with boundary. Focusing on finite Cauchy lenses, we constructed the classical covariant phase space, exhibited a  symplectic decomposition into closed-loop and surface components, and identified a corresponding Poisson algebra of semi-local observables extending the traditional algebra of local observables. Quantisation via a Weyl $C^*$-algebra led to a semi-local quantum field theory whose boundary-sensitive degrees of freedom transform non-trivially under large gauge transformations. 

One of the central conceptual insights of this paper is to interpret edge modes (in this case corner surface fields) as operational \emph{quantum reference frames} (QRFs) for large gauge transformations. Accordingly, we extended the electromagnetic system by a surface field and constructed a \emph{relativisation map} that dresses semi-local observables into fully gauge-invariant ones. In this extension, the relativised observables $\Af_\rel(\ol{N})$ can be used to formulate a superselection criterion \eqref{eq: sups criterion} for the algebra of joint large gauge invariant observables $\widetilde{\Af}^{\mathscr{LG}}(\ol{N})$, whose corresponding sectors are labelled by external electric fluxes. 
In contrast to previous literature on operational QRFs, our construction of a relativisation map proceeds on the level of $C^*$-algebras and hence is state-independent.

We have also shown how the QRF viewpoint can be applied to glue theories on two finite Cauchy lenses that share a common boundary. This novel approach to gluing emphasizes the physical idea that smeared Wilson lines on Cauchy lenses that touch the common boundary should be glued together to form Wilson loops in the algebra of the glued Cauchy lens. Finally, the operational interpretation of the relativisation map via covariant projection-valued measures further emphasizes the connection of our current results to with our previous work \cite{fewsterQuantumReferenceFrames2025}.

\paragraph{Future research directions} 

Our constructions suggest a number of promising avenues for further research. A key outstanding problem is to construct physically motivated states on $\Af(\ol{N})$ — in particular, quasi-free Hadamard states compatible with the presence of boundaries and with gluing procedures. Such states would also enable the definition of Wick polynomials and the study of renormalisation in concrete settings. We will address this problem in a forthcoming paper~\cite{fewsterHadamardStatesSemilocal}.

Furthermore, we would like to use our current setting to study superselection sectors of the theory, by considering appropriate classes of representations and potentially provide a link to Doplicher-Haag-Roberts or Buchholz-Fredenhagen superselection sector analysis (see e.g. chapter IV of \cite{haagLocalQuantumPhysics1996}). This could also shed some light on infrared problem in QED, where infrared degrees of freedom and corresponding superselection sectors play an important role.

It is also highly desirable to generalise our framework beyond finite Cauchy lenses, to a broader class of not necessarily compact Lorentzian manifolds with (asymptotic) boundary and corner, allowing for boundaries that may be timelike or null.

Another obvious research direction is extension to non-abelian and gravitational gauge theories. Many qualitative features considered here — particularly the emergence of edge modes and the role of large gauge transformations — are expected to extend to Yang–Mills theories and general relativity. Adapting our semi-local framework to non-abelian gauge symmetry and diffeomorphism symmetry would provide a powerful tool for the quantisation of gauge theories on manifolds with corners. We also hope to make contact with the framework for Lorentzian cobordisms and gluing provided in \cite{bunkLorentzianBordismsAlgebraic2025} and develop a functorial formulation of semi-local quantum physics, in analogy to locally covariant QFT~\cite{brunettiGenerallyCovariantLocality2003}. 
 
More broadly, semi-local quantum physics should provide a natural language for subsystem decomposition in gauge theories. We anticipate applications to entanglement entropy, relative entropy and quantum information-theoretic quantities, where edge mode contributions are known to be crucial.
 
Finally, the name `semi-local observable' begs the question of how such objects could be measured. It is therefore important to extend the measurement framework of~\cite{fewsterQuantumFieldsLocal2020} to the semi-local arena. 

In summary, our analysis positions semi-local observables and quantum reference frames as central organising principles for gauge theories on manifolds with boundary. We anticipate that these tools will serve as a foundation for a broad range of developments in mathematical physics, from quantum gravity and infrared QFT to quantum information theory.

\paragraph{Acknowledgments} We thank Aldo Riello, Michele Schiavina, Laurent Freidel and Filippo Nava for valuable conversations during the course of this work. We also thank Philipp H\"ohn and the anonymous referees for detailed comments on earlier versions of the paper. Part of the research was conducted during visits of DJ and KR to the Perimeter Institute, which we thank for its hospitality.

\section*{Statements and declarations}

\subsection*{Funding}
 This research has been supported by the Engineering \& Physical Sciences Research Council grant \verb!EP/Y000099/1! to the University of York.
For the purpose of open access a Creative Commons Attribution (CC BY) licence is applied to any Author Accepted Manuscript version arising from this submission.

\subsection*{Competing interests}
The authors have no competing interests to declare that are relevant to the content of this article.
\subsection*{Data availability} No datasets were generated or analysed in this work.

\begin{appendices}
\section{Table of relevant symbols and spaces}\label{apx:symbols}
Below $\ol{M}$ is used for a generic manifold $\ol{M}$ (possibly with boundaries) with $M=\Intr\ol{M}$ a manifold without boundaries, $\ol{N}$ for a finite Cauchy lens, $\ol{\Sigma}\subset \ol{N}$ for a regular Cauchy surface of $\ol{N}$ and $S$ for a generic Riemannian manifold (possibly with boundaries). 

\renewcommand{\arraystretch}{1.2}
\begin{longtable}{||r|c|l||} 
 \hline
 Symbol & Sect./Eq. & Space  \\ [0.5ex] 
 \hline\hline
 $\Omega^k(\ol{M})$& \S\ref{sec:forms} &differential $k$-forms on manifold $\ol{M}$ \\
 $\Omega^k_{\diff}(\ol{M})$ & \S\ref{sec:forms} & closed $k$-forms on manifold $\ol{M}$ \\
 $\Omega^k_{\diff^*}(\ol{M})$ & \S\ref{sec:forms}& co-closed $k$-forms on manifold $\ol{M}$ \\
 $\Omega^k_0(\ol{M})$& \S\ref{sec:forms}&  differential $k$-forms with compact support on manifold $M$ \\
 $\Omega^k_{0\diff}(\ol{M})$& \S\ref{sec:forms} & closed $k$-forms with compact support on manifold $M$ \\
 $\Omega^k_{0\diff^*}(\ol{M})$& \S\ref{sec:forms} & co-closed $k$-forms with compact support on manifold $M$ \\
 $\Omega^k_{\tang\diff^*}(\ol{M})$& \eqref{eq:Omegaktdstar_def} & co-closed $k$-forms $\alpha$ on manifold $\ol{M}$ with $\nml_{\partial\ol{M}}\alpha=0$ \\
 $\Sol^J(\ol{N})$& \eqref{eq:SolJN_def} & affine space of $A\in \Omega^k(\ol{N})$ with $-\diff^*\diff A=J$\\
 $\Sol(\ol{N})$& \eqref{eq:solspace_lin} & vector space of $A\in \Omega^k(\ol{N})$ with $-\diff^*\diff A=0$\\
 $T\Sol^J(\ol{N})$& \eqref{eq:TSolJNbar_def} & $\Sol^J(\ol{N})\times \Sol(\ol{N})$\\
 $\mathscr{G}(\ol{N})$& \eqref{eq:GNbardef}&  vector space of $\diff \Lambda$ for $\Lambda\in \Omega^0(\ol{N})$ with $\Lambda\restriction_{\angle\ol{N}}=0$\\
$\mathscr{G}_{\angle}(\ol{N})$ & \eqref{eq:G_angle} & $\Omega^0_{\diff}(\ol{N})\restriction_{\angle\ol{N}}$ \\ 
 $\Sol^J_{\mathscr{G}}(\ol{N})$& \eqref{eq:SolJGN_def} & $\Sol^J(\ol{N})/\mathscr{G}(\ol{N})$\\
 $\Sol_{\mathscr{G}}(\ol{N})$& \eqref{eq:SolJGN_def} & $\Sol(\ol{N})/\mathscr{G}(\ol{N}) = \Sol^0_{\mathscr{G}}(\ol{N})$\\
 $T\Sol^J_{\mathscr{G}}(\ol{N})$& Def.~\ref{def:red_phas} & $\Sol^J_{\mathscr{G}}(\ol{N})\times \Sol_{\mathscr{G}}(\ol{N})$\\
 $\mathscr{G}(\ol{\Sigma})$& \eqref{eq:Gsigma_def} & $\mathscr{G}(\ol{N})\restriction_{\ol{\Sigma}}$\\
$\mathcal{A}(\ol{\Sigma})$ & \eqref{eq:Asigma_def} & $\Omega^1(\ol{\Sigma})/\mathscr{G}(\ol{\Sigma})$ \\
$\mathcal{E}^\rho(\ol{\Sigma})$ & \eqref{eq:Erho_def} & $\{\mathbf{E}\in \Omega^1(\ol{\Sigma}):-\diff^*_{\ol{\Sigma}}\mathbf{E}=\rho\}$ \\
 $\mathscr{LG}(\ol{N})$ & \eqref{eq:large_gauge_def} & 
    $\{[\diff\Lambda]\in \Sol_\mathscr{G}(\ol{N}):\Lambda\in \Omega^0(\ol{N})\}$\\
    $\mathscr{LG}_\angle(\ol{N})$ & Prop.~\ref{prop:large_gauge_boundary} & 
    orthogonal complement of $\mathscr{G}_{\angle}(\ol{N})$ in $\Omega^0(\angle\ol{N})$\\
 $V^C(\ol{\Sigma})$& \eqref{eq:VCVS_def} & $\Omega^1_{\tang\diff^*}(\ol{\Sigma})^{\oplus 2}$\\
 $V^S(\ol{N})$& \eqref{eq:VCVS_def} & $\mathscr{LG}_\angle(\ol{N})^{\oplus 2}$\\
 $\Floc^{}
    (\ol{N})$&Def.~\ref{def:class-loc-obs}& classical algebra of local observables\\
$\Fsloc^{}
    (\ol{N})$&Def.~\ref{def:class-sloc-obs}& classical algebra of semi-local observables\\
    $\Fsloc^{C/S}
    (\ol{\Sigma})$&Def.~\ref{def:class_obs_decomp}& classical algebra closed loop/surface observables\\
 $\Af(\ol{N})$&Prop.~\ref{prop:CCR_BPI}& quantum algebra of semi-local observables \\
 $\Af^{C/S}(\ol{\Sigma})$&Def.~\ref{def:bulk_corn_weyls}& quantum algebra of closed loop/surface observables \\
 $\Af_{\textup{loc}}(\ol{N})$&Def.~\ref{def:loc_obs_qft}& quantum algebra of local observables \\
 $\Af(\ol{N};U)$&Def.~\ref{def:localisability}& subalgebra of semi-local observables localisable in $U\subset\ol{N}$\\
 $\widetilde{\Sol}(\ol{N})$&\eqref{eq:soltilde}&$\Sol_\mathscr{G}(\ol{N})\oplus V^S(\ol{N})$\\
 $\widetilde{\mathscr{G}}(\ol{N})$&\eqref{eq:jlgd}&joint large gauge directions in $\widetilde{\Sol}(\ol{N})$\\
 $\Sol_{\textnormal{ext}}(\ol{N})$&\eqref{eq:sfeps}&surface field extended phase space $ \widetilde{\mathscr{G}}(\ol{N})^\perp/\widetilde{\mathscr{G}}(\ol{N})$\\
 $\widetilde{\Af}(\ol{N})$&\eqref{eq:atilde}&$\Weyl(\widetilde{\Sol}(\ol{N}),\widetilde{\sigma})$\\
 $\widetilde{\Af}^{\mathscr{LG}}(\ol{N})$&\eqref{eq:atilde}&joint large gauge invariant observables in $\widetilde{\Af}(\ol{N})$\\
 $\Af_\rel(\ol{N})$&Lem.~\ref{lem:relativise}& quantum algebra of relativised observables\\
 $\Af_{\textnormal{ext}}(\ol{N})$&Def.~\ref{def:bound_gauge_gen}& quantum surface field extended algebra\\
 $\Af_{\textup{glue}}(\ol{N}_1,\ol{N}_2)$&\eqref{eq:glued_def}& glued algebra of finite Cauchy lensens $\ol{N}_{1,2}$ touching at the corner
 \\[1ex] 
 \hline
\end{longtable} 

\section{Manifolds with corners and Lorentzian geometry}
\label{apx:corners}
As mentioned in Sec.~\ref{sec:set-up}, manifolds with corners are locally modelled on $\ol{\RR}_+^n:=[0,\infty)^n$ (see e.g.~\cite{leeIntroductionSmoothManifolds2012}). By convention, a function on (a relatively open subset of) $\ol{\RR}_+^n$ is said to be smooth if it can be smoothly extended to an open $\RR^n$-neighbourhood of any point in its domain. A chart with corners on an $n$-dimensional topological manifold $\ol{M}$ is a pair $(U,\varphi)$ with $\varphi$ a homeomorphism of open subset $U\subset \ol{M}$ onto a relatively open subset of $\ol{\RR}_+^n$. 
A smooth $n$-manifold with corners is an $n$-dimensional topological manifold admitting an maximal atlas $\{(U_i,\varphi_i):i\in I\}$ of smoothly compatible charts with corners, where
smooth compatibility is adapted to smoothness on $\ol{\RR}_+^n$.
Any smooth $n$-manifold with corners $\ol{M}$ admits \emph{smooth extensions}, i.e., homeomorphic embedding $\iota:\ol{M}\to\tilde{M}$ in a smooth $n$-manifold $\tilde{M}$ (without boundaries or corners) from which its maximal atlas may be obtained by restriction and pull-back -- see e.g.~\cite[Prop.~3.1]{douadyArrondissementVarietesCoins1973}.
As such, the smooth functions $C^\infty(\ol{M},\RR)$ can be defined intrinsically using the maximal atlas, or equivalently as the range of a pullback $C^\infty(\ol{M},\RR)=\iota^*(C^\infty(\tilde{M},\RR))$ for any $\iota:\ol{M}\to\tilde{M}$ as above, see e.g.~\cite{gurerDifferentialFormsManifolds2019}. Smooth maps between manifolds with corners are defined in terms of smooth maps between charts with corners in the obvious way. 

If $\pi_k:\ol{\RR}_+^n\to\RR$ is the $k$'th coordinate projection, $F^n_k:=\pi_k^{-1}(\{0\})$ ($1\le k\le n)$ and $E^n_{k,l}:=F^n_k\cap F^n_{l}$ ($1\le k<l\le n$) are the faces and edges of $\ol{\RR}_+^n$.
 Then the boundary of a smooth $n$-manifold with corners $\ol{M}$ with atlas $\{(U_i,\varphi_i)\}_{i\in I}$ is
\begin{equation}
    \partial\ol{M}:=\bigcup_{i\in I}\varphi_i^{-1}\left(\bigcup_{k=1}^n F^n_k\right)=\bigcup_{i\in I}\varphi_i^{-1}\left(\partial\ol{\RR}_+^n\right),
\end{equation}
and the (first-order) corner is
\begin{equation}
    \angle\ol{M}:=\bigcup_{i\in I}\varphi_i^{-1}\left(\bigcup_{1\leq k<j\leq n}E^n_{k,j}\right)=\bigcup_{i\in I}\varphi_i^{-1}\left(\angle\ol{\RR}_+^n\right).
\end{equation}
Higher order corners may also be defined, but will not be used here. The boundary of a manifold with corners decomposes as $\partial\ol{M}=\bigcup_{j\in J}\ol{S}_j$ for some family of submanifolds with corners $\ol{S}_j\subset \ol{M}$; however, $\partial\ol{M}$ is not necessarily in itself a smooth manifold with corners.

As any real interval $I$ (whether open, closed, or semiclosed) is trivially a manifold with (possibly empty) boundary there is a clear notion of a smooth curve in $\ol{M}$ as a smooth map $\gamma:I\to \ol{M}$.
The (co)tangent bundle and their tensor products $T^{(k,l)}\ol{M}:=(T\ol{M})^{\otimes k}\otimes (T^*\ol{M})^{\otimes l}$, as well as their smooth sections, (co)vector- and tensor fields $\mathfrak{X}^{(k,l)}(\ol{M})$, can be defined either intrinsically or, as for the space of smooth functions, via a pullback from a smooth extension $\iota:\ol{M}\to\tilde{M}$. 
In particular, a metric tensor on $\ol{M}$ is a smooth symmetric tensor field $g\in \mathfrak{X}^{(0,2)}(\ol{M})$ that is nondegenerate and therefore of fixed, possibly indefinite, signature. Lorentzian metrics will have mostly minus signature. Notions of (time)-orientations of $\ol{M}$ can be defined as on manifolds without boundary.

On a (time-oriented) Lorentzian manifold with corners $(\ol{M},g)$, a smooth curve $\gamma:I\to\ol{M}$ on real interval $I$ is timelike\slash null\slash causal\slash  spacelike\slash past-directed\slash future-directed if its tangent vector is everywhere of the given type,
including at any endpoints of $I$, by which we mean points of $I\cap\partial I$. A piecewise smooth curve on an interval $[a,b]$ is a function $\gamma:[a,b]\to \ol{M}$ together with a finite partition $a=t_0<\cdots<t_n=b$ so that each $\gamma|[t_r,t_{r+1}]$ is smooth; for general interval $I\subset \RR$, $\gamma:I\to\ol{M}$ is piecewise smooth if $\gamma|_{[a,b]}$ is piecewise smooth for all $a,b\in I$ with $a<b$. 

To describe the causal structure of a Lorentzian time-oriented spacetime with corners $(\ol{M},g)$, we follow Penrose~\cite{Penrose:1972} and define a \emph{trip} to be a piecewise smooth curve whose segments are future-directed timelike geodesics, and a \emph{trip from $p$ to $q$} to be a trip $\gamma:[a,b]\to \ol{M}$ with $p=\gamma(a)$, $q=\gamma(b)$. \emph{Causal trips} (between points) are defined in the same way but allowing possibly degenerate causal geodesics in place of timelike geodesics. Then the chronological (resp., causal) future/past $\mathscr{I}^\pm(U)$ (resp., $\mathscr{J}^\pm(U)$) of $U\subset \ol{M}$ are defined so that $p\in\mathscr{I}^+(U)$ (resp., $p\in \mathscr{J}^+(U)$) if and only if there is a trip (resp., causal trip) starting in $U$ and ending at $p$, and with analogous definitions for the chronological and causal pasts. We also write 
$\mathscr{J}(U)=\mathscr{J}^+(U)\cup \mathscr{J}^-(U)$. 

A future-directed piecewise smooth causal curve
$\gamma:I\to \ol{M}$ is past-extendible (resp., future-extendible) if $\gamma(t)$ has a limit $p\in\ol{M}$ as $t\to \inf I$ (resp., $t\to \sup I$) whose causal past (resp., future) is strictly larger than $\{p\}$. The extra condition is not needed for manifolds without boundary, but here we may encounter curves that have endpoints at the boundary or corners of $\ol{M}$. 
The domain of dependence $\mathscr{D}(U)$ of
$U\subset \ol{M}$ is the set of points $p\in\ol{M}$ so that $U$ is met by every future-directed causal trip through $p$ that is neither future-extendible nor past-extendible.

A submanifold $S\subset M$ (possibly with corners) is space-like if $g_S=-\iota^*g$ defines an \emph{induced Riemannian metric} on $S$, where $\iota:S\to \ol{M}$ is the identity embedding. A Cauchy surface in a time-oriented Lorentzian manifold without boundary is an achronal subset whose domain of dependence is the full spacetime, and which is necessarily closed; if a Cauchy surface exists then the spacetime is globally hyperbolic and can be foliated by Cauchy surfaces that are smooth spacelike hypersurfaces (see e.g.~\cite{bernalSmoothCauchyHypersurfaces2003,minguzziCausalHierarchySpacetimes2008}). 
\section{Gauge transformations and the initial value problem for pure electromagnetism on Cauchy lenses}
\label{Ap:IV_prob_lens}

\subsection{Initial value problem and gauge transformations}
\label{apx:IV_prob_lens}
Here, we will compute the space of gauge transformations $\mathscr{G}(\ol{N})$ defined in~\eqref{eq:GNbardef} to establish~\eqref{eq:GNbar}, and also prove Proposition~\ref{prop:ext_ini_dat} on the initial value problem in the reduced phase space of electromagnetism on a Cauchy lens.

It will be useful to start with a version of the initial value problem in globally hyperbolic spacetimes without boundary. If $\Sigma$ is a smooth spacelike hypersurface in smooth Lorentzian manifold $M$ (possibly with corners), we define $\Data^p_\Sigma:\Omega^p(M)\to \Omega^p(\Sigma)\times\Omega^p(\Sigma)$ by
\begin{equation}
    \Data^p_\Sigma(A) = (A\restriction_\Sigma,(-1)^p\nml_\Sigma \diff A).
\end{equation}
The main use will be $p=1$, and we drop the superscript in this case. 
\begin{proposition}\label{prop:IV_simple}
    Let $\Sigma$ be a Cauchy surface in globally hyperbolic spacetime $M$. For any $J\in\Omega^1_{\diff^*}(M)$, 
    \begin{equation}
        \Data_\Sigma(\Sol^J(M))=
        \{(\mathbf{A},\mathbf{E})\in 
        \Omega^1(\Sigma)^{\times 2}: -\diff^*_\Sigma\mathbf{E}=\nml_\Sigma J\}. 
    \end{equation} 
    Furthermore, whenever $\Data_{\ol{\Sigma}}(A_1)=\Data_{\ol{\Sigma}}(A_2)$ for $A_1,A_2\in \Sol^J(M)$, there exists a $\Lambda\in \Omega^0(M)$ with $\Lambda\restriction_{\Sigma}=0$ such that $A_1-A_2=\diff\Lambda$.
\end{proposition}
Much more can be said about the continuity of solution on initial data,  extensions to distributional data and other support systems -- see e.g., \cite{sandersElectromagnetismLocalCovariance2014}. The above result is proved by solving the normally hyperbolic initial value problem (see ~\cite{barWaveEquationsLorentzian2007}) $-(\diff^*\diff + \diff\diff^*)A=J$
subject to Cauchy data comprising $\Data_\Sigma^1(A)$ together with data $(\nml_\Sigma A,\nml_\Sigma \nabla_n A)\in \Omega^0(\Sigma)^{\times 2}$, where $n$ is the future-directed unit normal vector to $\Sigma$ (extended to a neighbourhood thereof). The additional data is chosen so that $\Data_\Sigma^0(\diff^*A)=(0,0)$, which is arranged by solving the Gauss constraint. As $\diff^*\diff (\diff^*A)=0$, it follows that $\diff^*A=0$ globally and so $A$ is a Lorenz gauge 
solution to the Maxwell equations with the required data. Conversely, the constraint on $\mathbf{E}$ is a necessary condition for solvability: if $-\diff^*\diff A = J$ then
$\diff_\Sigma^*\nml_\Sigma\diff A=-\nml_\Sigma\diff^*\diff A = \nml_\Sigma J$. Uniqueness can be shown by finding for $A\in \Sol(M)$ with $\Data_{\Sigma}(A)=0$, a $\Lambda\in \Omega^0(M)$ with  $-\diff^*\diff\Lambda=-\diff^*A$ and where for $\tilde{A}=A-\diff\Lambda$ one has $\Data_{\Sigma}(\tilde{A})=0$ and $(\nml_{\Sigma}\tilde{A},\nml_{\Sigma}\nabla_{n}\tilde{A})=0$. For such $\Lambda$, it follows that $\tilde{A}=0$ or equivalently $A=\diff\Lambda$. For a detailed proof, see e.g.~\cite{pfenningQuantizationMaxwellField2009}.
 
From now on, let $\ol{N}$ be a fixed finite Cauchy lens, $J\in \Omega^1_{\diff^*}(\ol{N})$ a background current and $\ol{\Sigma}\subset\ol{N}$ a regular Cauchy surface with boundaries, so the pre-symplectic structure $\sigma$ on $\Sol(\ol{N})$ defined in Sec.~\ref{sec:cov_phas} is 
\begin{equation}
    \sigma(\udl{A},\udl{A}')=\int_{\ol{\Sigma}}A'\wedge \star\diff A-A\wedge \star\diff A'.
\end{equation}
We also write $\rho=\nml_{\ol{\Sigma}}J\in \Omega^0(\ol{\Sigma})$ and set
    \begin{equation}
        \mathcal{E}^\rho(\ol{\Sigma})=\{\mathbf{E}\in \Omega^1(\ol{\Sigma}):-\diff^*_{\ol{\Sigma}}\mathbf{E}=\rho\}.
    \end{equation}

Let us first prove existence of solutions given appropriate initial data.
\begin{lemma}
\label{lem:IV_constr} 
With the above definitions,
\begin{equation}
\Data_{\ol{\Sigma}} (\Sol^J(\ol{N}))= \Omega^1(\ol{\Sigma})\times \mathcal{E}^\rho(\ol{\Sigma}).
\end{equation}
\end{lemma}
\begin{proof} 
The inclusion of the left-hand side in the right holds by same argument as for Proposition~\ref{prop:IV_simple}.      
    Now consider an isometric causally convex embedding $\iota:\ol{N}\to M$ such that $M$ a globally hyperbolic manifold without boundaries and $\iota(\ol{\Sigma})\subset \tilde{\Sigma}$ with $\tilde{\Sigma}\subset M$ a Cauchy surface. Given $(\mathbf{A},\mathbf{E})\in \Omega^1(\ol{\Sigma})\times\mathcal{E}^\rho(\ol{\Sigma})$, choose  $(\tilde{\mathbf{A}},\tilde{\mathbf{E}})\in \Omega^1_0(\tilde{\Sigma})^{\times 2}$ 
    and $\tilde{J}\in \Omega^0_0(M)$ such that 
    \begin{equation}
        (\mathbf{A},\mathbf{E})=(\iota^*\tilde{\mathbf{A}},\iota^*\tilde{\mathbf{E}}), \qquad J=\iota^*\tilde{J},
    \end{equation}
    noting that $\tilde{J}$ is not necessarily co-closed. To remedy this, let $g\in \Omega^0_{\sc}(M)$ solve the initial value problem (see \cite[Thm.~3.2.11]{barWaveEquationsLorentzian2007})
    \begin{equation}
        -\diff^*\diff g= \diff^*\tilde{J}\in \Omega^0_0(M),\qquad g\restriction_{\tilde{\Sigma}}=0,\qquad\nml_{\ol{\Sigma}}\diff g=\tilde{\rho}-\nml_{\ol{\Sigma}} \tilde{J}\in \Omega_0^0(\tilde{\Sigma}),
    \end{equation}
    where $\tilde{\rho}=-\diff^*_{\ol{\Sigma}}\tilde{\mathbf{E}}\in \Omega_0^0(\tilde{\Sigma})$, noting that $g$ vanishes in $\ol{N}$ because its data vanish on $\ol{\Sigma}$ the source $\diff^*\tilde{J}$ vanishes on $\ol{N}$.
    Then $\tilde{J}':=\tilde{J}+\diff g\in \Omega^1_{\diff^*,sc}(M)$, and  $-\diff^*_{\ol{\Sigma}}\tilde{\mathbf{E}}=
    \tilde{\rho}=\nml_{\tilde{\Sigma}}\tilde{J}'$. 
Accordingly Proposition~\ref{prop:IV_simple} yields $\tilde{A}\in \Sol^{\tilde{J}'}(M)$ with $\Data_{\tilde{\Sigma}}(A)=(\tilde{\mathbf{A}},\tilde{\mathbf{E}})$. 
    As $\iota^* \tilde{J}'= J$, we now have
    $A=\iota^*\tilde{A}\in \Sol^J(\ol{N})$ with $\Data_{\ol{\Sigma}}(A)=(\mathbf{A},\mathbf{E})$ as required.
\end{proof}

Using this existence result, we can give a description of all degenerate directions of the symplectic structure $\sigma$ on $\Sol(\ol{N})$.
\begin{proposition}
\label{prop:sym_red}
The set $\mathscr{G}(\ol{N})=\{\udl{A}\in \Sol(\ol{N}):\sigma(\udl{A},\udl{A}')=0\text{ for all }\udl{A}'\in \Sol(\ol{N})\}$ is given by 
\begin{equation}\label{eq:GNbar}
    \mathscr{G}(\ol{N})=\{\diff\Lambda:\Lambda\in \Omega^0(\ol{N}),\,\Lambda\restriction_{\angle \ol{N}} =0\}.
\end{equation}
\end{proposition}
If $\diff\Lambda\in\mathscr{G}(\ol{N})$ then $\Lambda$ may of course be modified by any locally constant function on $\ol{N}$ while yielding the same gauge direction. The proof relies on the Hodge decomposition given in Lem.~\ref{lem:Hodge} as well as the kernel of the $\Data$ map on $\Sol(N)$ described in Prop.~\ref{prop:IV_simple}. 

\begin{proof}Suppose $\udl{A}\in\Sol(\ol{N})$. Then $\udl{A}\in \mathscr{G}(\ol{N})$ if and only if
\begin{equation}
\ipc{\udl{A}}{\nml_{\ol{\Sigma}}\diff \udl{A'}}- \ipc{\nml_{\ol{\Sigma}}\diff\udl{A}}{\udl{A}'}=0
\end{equation}
for all $\udl{A}'\in\Sol(\ol{N})$. By Lemma~\ref{lem:IV_constr}, $\Data_\Sigma(\Sol(\ol{N}))$ contains $\{0\}\times \Omega_{\diff^*}^1(\ol{\Sigma})$ and
$\Omega^1(\ol{\Sigma})\times \{0\}$. Hence, $\udl{A}\in \mathscr{G}(\ol{N})$ if and only if $\udl{A}\restriction_{\ol{\Sigma}}\in 
\Omega^1_{\diff^*}(\ol{\Sigma})^\perp$ and $\nml_{\ol{\Sigma}}\diff\udl{A}=0$. Using Lemma~\ref{lem:Hodge}, we now have
$\udl{A}\in \mathscr{G}(\ol{N})$ if and only if $\nml_{\ol{\Sigma}}\diff\udl{A}=0$ and $\udl{A}\restriction_{\ol{\Sigma}}=
 \diff\lambda$ for some $\lambda\in\Omega^0(\ol{\Sigma})$ with $\lambda|_{\partial\ol{\Sigma}}=0$. In particular, this shows immediately that the
 right-hand side of~\eqref{eq:GNbar} contains the left by taking $\lambda=\Lambda|_{\ol{\Sigma}}$.

Conversely, suppose that $\udl{A}\in \mathscr{G}(\ol{N})$; then $\udl{A}\restriction_{\ol{\Sigma}}=
 \diff\lambda$ for some $\lambda\in\Omega^0(\ol{\Sigma})$ with $\lambda|_{\partial\ol{\Sigma}}=0$. Choose $\tilde{\Lambda}\in\Omega^0(\ol{N})$ with 
$\tilde{\Lambda}\restriction_{\ol{\Sigma}}=\lambda$, noting that $\diff\tilde{\Lambda}\in \mathscr{G}(\ol{N})$. 
Then $\tilde{\udl{A}}=\udl{A}-\diff\tilde\Lambda\in \Sol(\ol{N})$ has $\Data_{\ol{\Sigma}}(\tilde{\udl{A}})=(0,0)$. By Proposition~\ref{prop:IV_simple}, we conclude that $\tilde{\udl{A}}\restriction_{N}=\diff\Lambda$ for some $\Lambda\in \Omega^0(N)$ with $\Lambda\restriction_{\Sigma}=0$ (where $N=\Intr(\ol{N})$ and $\Sigma=\Intr(\ol{\Sigma})$). Since $\udl{\tilde{A}}\in \Omega^1(\ol{N})$, 
it follows that we can smoothly extend $\diff\Lambda$ to $\ol{N}$, from which it follows that $\Lambda$ can be smoothly extended to $\ol{N}$. It follows that $\udl{\tilde{A}}=\diff\Lambda$ for some $\Lambda\in \Omega^0(\ol{N})$ with $\Lambda\restriction_{\ol{\Sigma}}=0$. Hence $\udl{A}=\diff(\Lambda+\tilde{\Lambda})$ and $(\Lambda+\tilde{\Lambda})|_{\angle \ol{N}}=(\Lambda+\tilde{\Lambda})|_{\partial \ol{\Sigma}}=0$.
The reverse inclusion is proved. 
\end{proof}

\subsection{Boundary value problems and initial data decompositions} 
\label{apx:data_hodge_lemmas}

Finally, we prove some results relevant to the decomposition of initial data described in Sec.~\ref{sec:init_dat_hodge}. These rely on solvability of the Poisson equations with Neumann and Dirichlet boundary conditions.
\begin{proposition}\label{prop:BVP}
   Let $\ol{\Sigma}$ be a compact Riemannian manifold with boundary. Then the Neumann boundary problem
    \begin{equation}\label{eq:NeumannBVP}
        -\diff^*\diff \varphi = \rho,\qquad \nml_{\partial\ol{\Sigma}}\diff\varphi=f,
    \end{equation}
    for $\rho\in \Omega^0(\ol{\Sigma})$ and $f\in \Omega^0(\partial\ol{\Sigma})$ is solvable for $\varphi\in \Omega^0(\ol{\Sigma})$ if and only if $\ipc{\kappa}{f}_{\partial\ol{\Sigma}}=\ipc{\kappa}{\rho}_{\ol{\Sigma}}$ for all locally constant
    $\kappa\in \Omega^0_{\diff}(\ol{\Sigma})$. For any two solutions $\varphi,\varphi'$ to this problem, one has $\diff(\varphi-\varphi')=0$.

    Furthermore, the Dirichlet boundary problem
    \begin{equation}\label{eq:DirBVP}
        -\diff^*\diff \varphi = \rho,\qquad \varphi\restriction_{\partial\ol{\Sigma}}=f,
    \end{equation}
    for $\rho\in \Omega^0(\ol{\Sigma})$ and $f\in \Omega^0(\partial\ol{\Sigma})$ has a unique solution $\varphi\in \Omega^0(\ol{\Sigma})$.
\end{proposition}
For the proof, we refer to \cite[Cor.~3.4.8]{schwarzHodgeDecompositionMethod1995} and \cite[Thm.~3.4.10]{schwarzHodgeDecompositionMethod1995}. These boundary value problems can now be used to show the following two lemmas.
\begin{lemma}\label{lem:nmlcoclosed}
Let $\ol{\Sigma}$ be a compact Riemannian manifold with nonempty boundary. Then $\Omega^0(\partial\ol{\Sigma})$ admits the orthogonal decomposition
\begin{equation}
    \Omega^0(\partial\ol{\Sigma})=\nml_{\partial\ol{\Sigma}}\Omega^1_{\diff^*}(\ol{\Sigma})\oplus \left(\Omega^0_{\diff}(\ol{\Sigma})\restriction_{\partial\ol{\Sigma}}\right).
\end{equation} 
\end{lemma}
\begin{proof}
    For any $\mathbf{F}\in \Omega^1_{\diff^*}(\ol{\Sigma})$ and $f\in \Omega^0_{\diff}(\ol{\Sigma})$, Stokes' theorem gives
    \begin{equation}
        \ipc{f\restriction_{\partial\ol{\Sigma}}}{\nml_{\partial\ol{\Sigma}}\mathbf{F}}_{\partial\ol{\Sigma}}=\ipc{\diff f}{\mathbf{F}}_{\ol{\Sigma}}-\ipc{f}{\diff^*\mathbf{F}}_{\ol{\Sigma}}=0,
    \end{equation} thus $\nml_{\partial\ol{\Sigma}}\Omega^1_{\diff^*}(\ol{\Sigma})$ and $\Omega^0_{\diff}(\ol{\Sigma})\restriction_{\partial\ol{\Sigma}}$ are orthogonal. 
    Given $h\in \Omega^0(\partial\ol{\Sigma})$, let $f\in \Omega^0_{\diff}(\ol{\Sigma})$
    be the unique locally constant $0$-form so that $f\restriction_{\partial\ol{\Sigma}}$ is the orthogonal projection of $h$ onto $\left(\Omega^0_{\diff}(\ol{\Sigma})\restriction_{\partial\ol{\Sigma}}\right)$
    (a closed subspace of $L^2(\partial\ol{\Sigma})$)
    whereupon $h-f\restriction_{\partial\ol{\Sigma}}\in \left(\Omega^0_{\diff}(\ol{\Sigma})\restriction_{\partial\ol{\Sigma}}\right)^\perp$.
    By Prop.~\ref{prop:BVP}, it follows that the boundary value problem $-\diff^*\diff\varphi=0$ and $\nml_{\partial\ol{\Sigma}}\diff\varphi=h-f\restriction_{\partial\ol{\Sigma}}$ has a solution $\varphi\in \Omega^0(\ol{\Sigma})$, and hence $h-f\restriction_{\partial\ol{\Sigma}}\in \nml_{\partial\ol{\Sigma}}\Omega^1_{\diff^*}(\ol{\Sigma})$.
    \end{proof}

\begin{lemma}
\label{lem:A_boundary_uni}
    Let $\ol{N}$ be a finite Cauchy lens, $\ol{\Sigma}\subset \ol{N}$ a regular Cauchy surface with boundary. For each $\alpha\in \Omega^0(\ol{\Sigma})$, there exists a unique $\tilde{\alpha}\in \Omega^0(\ol{\Sigma})$ such that $\Delta\tilde{\alpha}=0$, $\tilde{\alpha}\restriction_{\partial\ol{\Sigma}}\in \nml_{\partial\ol{\Sigma}}\Omega^1_{\diff^*}(\ol{\Sigma})$ and $\diff(\alpha-\tilde{\alpha})\in \mathscr{G}(\ol{\Sigma})=\mathscr{G}(\ol{N})\restriction_{\ol{\Sigma}}$.
\end{lemma}
\begin{proof}
By subtracting a locally constant function from $\alpha$, one may obtain $\alpha'\in\Omega^0(\ol{\Sigma})$ such that $\diff\alpha'=\diff\alpha$ and $\alpha'\restriction_{\partial\ol{\Sigma}}\in \left(\Omega^0_{\diff}(\ol{\Sigma})\restriction_{\partial\ol{\Sigma}}\right)^\perp$, and thus $\alpha'\in \nml_{\partial\ol{\Sigma}}\Omega^1_{\diff^*}(\ol{\Sigma})$ by Lem.~\ref{lem:nmlcoclosed}.
    Now let $\alpha''$ be the unique solution to $\alpha''\restriction_{\partial\ol{\Sigma}}=0$ and $\Delta\alpha''=\Delta\alpha=\Delta \alpha'$ (see \eqref{eq:DirBVP} in Prop.~\ref{prop:BVP}). Setting $\tilde{\alpha}=\alpha'-\alpha''$, we find that $\Delta\tilde{\alpha}=0$, $\tilde{\alpha}\in \nml_{\partial\ol{\Sigma}}\Omega^1_{\diff^*}(\ol{\Sigma})$ and $\diff(\alpha-\tilde{\alpha})=\diff \alpha''\in \mathscr{G}(\ol{\Sigma})$.

    That $\tilde{\alpha}$ is specified uniquely follows from the fact that $\diff(\tilde{\alpha}_1-\tilde{\alpha}_2)\in \mathscr{G}(\ol{\Sigma})$ implies $(\tilde{\alpha}_1-\tilde{\alpha}_2)\restriction_{\partial\ol{\Sigma}}
    =c\restriction_{\partial{\ol{\Sigma}}}$ for some locally constant $c\in\Omega^0_{\diff}(\ol{\Sigma})$. 
If, additionally, $(\tilde{\alpha}_1-\tilde{\alpha}_2)\restriction_{\partial\ol{\Sigma}}\in \nml_{\partial\ol{\Sigma}}\Omega^1_{\diff^*}(\ol{\Sigma})$, then $c=0$ by Lem.~\ref{lem:nmlcoclosed}, so $(\tilde{\alpha}_1-\tilde{\alpha}_2)\restriction_{\partial\ol{\Sigma}}=0$. If, additionally, $\Delta(\tilde{\alpha}_1-\tilde{\alpha}_2)=0$, this now implies $\tilde{\alpha}_1-\tilde{\alpha}_2=0$.
\end{proof}

\subsection{Proof of Prop.~\ref{prop:sympl_decomp}}\label{apx:proof_of_VCVS_decomp}

Prop.~\ref{prop:sympl_decomp} is proved by constructing $\mathfrak{K}^J_{\ol{\Sigma}}$, using the map $\mathfrak{I}^J_{\ol{\Sigma}}$ from Prop.~\ref{prop:ext_ini_dat} and a Hodge decomposition, together with some lemmas given in Appx~\ref{apx:data_hodge_lemmas}. 

\begin{proof}[Proof of Prop.~\ref{prop:sympl_decomp}]
We will start at the level of initial data, establishing an affine bijection
\begin{equation}
    \mathfrak{H}_{\ol{\Sigma}}^\rho:\mathcal{A}(\ol{\Sigma})\times \mathcal{E}^\rho(\ol{\Sigma})\to \left(\Omega^1_{\tang\diff^*}(\ol{\Sigma})\right)^2\times \left(\nml_{\partial\ol{\Sigma}}\Omega^1_{\diff^*}(\ol{\Sigma})\right)^2,
\end{equation}
whose linearisation is a linear bijection
\begin{equation}
    \mathfrak{H}_{\ol{\Sigma}}:\mathcal{A}(\ol{\Sigma})\oplus \mathcal{E}(\ol{\Sigma})\to \left(\Omega^1_{\tang\diff^*}(\ol{\Sigma})\right)^{\oplus 2}\oplus \left(\nml_{\partial\ol{\Sigma}}\Omega^1_{\diff^*}(\ol{\Sigma})\right)^{\oplus 2}.
\end{equation}
For $(\mathbf{A}+\mathscr{G}(\ol{\Sigma}),\mathbf{E})\in \mathcal{A}(\ol{\Sigma})\times \mathcal{E}^\rho(\ol{\Sigma})$, the tangential Hodge--Helmholtz decomposition gives
\begin{align}
    \mathbf{A}+\mathscr{G}(\ol{\Sigma})=&\mathbf{A}^{\tang}+\diff\alpha+\mathscr{G}(\ol{\Sigma}),\nonumber\\
    \mathbf{E}=&\mathbf{E}^{\tang}+\diff \varepsilon,
\end{align}
where $\mathbf{A}^{\tang},\mathbf{E}^{\tang}\in \Omega^1_{\tang\diff^*}(\ol{\Sigma})$, $\diff\alpha+\mathscr{G}(\ol{\Sigma})\in \diff\Omega^0(\ol{\Sigma})/\mathscr{G}(\ol{\Sigma})$ and $\varepsilon\in\{\varphi\in \Omega^0(\ol{\Sigma}):-\diff^*\diff\varphi=\rho\}$. The pair $(\mathbf{A}^{\tang},\mathbf{E}^{\tang})$ will constitute the component of $\mathfrak{H}_{\ol{\Sigma}}^\rho(\mathbf{A}+\mathscr{G}(\ol{\Sigma}),\mathbf{E})$ in $V^C(\ol{\Sigma})$.

By Lem.~\ref{lem:A_boundary_uni}, we can choose for each equivalence class $\diff\alpha+ \mathscr{G}(\ol{\Sigma})$ a unique $\alpha\in \Omega^0(\ol{\Sigma})$ such that $\alpha\restriction_{\partial\ol{\Sigma}}\in\nml_{\partial\ol{\Sigma}}\Omega^1_{\diff^*}(\ol{\Sigma})$ and $\Delta \alpha=0$. As such functions are uniquely specified by their boundary values $\alpha\restriction_{\partial\ol{\Sigma}}$, the map $\diff\alpha+ \mathscr{G}(\ol{\Sigma}) \mapsto \alpha\restriction_{\partial\ol{\Sigma}}$ defines a linear isomorphism $\diff\Omega^0(\ol{\Sigma})/\mathscr{G}(\ol{\Sigma})\cong \nml_{\partial\ol{\Sigma}}\Omega^1_{\diff^*}(\ol{\Sigma})$.

Let $\varphi\in\Omega^0(\ol{\Sigma})$ solve the Neumann boundary problem for the Poisson equation
\begin{equation}
    -\diff^*\diff\varphi=\rho\in \diff^*\Omega^1(\ol{\Sigma}),\qquad \nml_{\partial\ol{\Sigma}}\diff\varphi=f^\rho\in \Omega^0(\partial\ol{\Sigma}),
\end{equation}
where, on the boundary of each component of $\ol{\Sigma}$, $f^\rho$ is constant and equal to the total charge in that component divided by the area of its boundary, from which one has $\ipc{c\restriction_{\ol{\Sigma}}}{f^\rho}_{\partial\ol{\Sigma}}=\ipc{c}{\rho}_{\ol{\Sigma}}$ for every locally constant $c\in\Omega^0(\ol{\Sigma})$.
By Prop.~\ref{prop:BVP}, the solution $\varphi$ exists and is unique up to the addition of locally constant 0-forms. One furthermore finds
\begin{equation}
    \diff\Omega^0(\ol{\Sigma})\cap \mathcal{E}^\rho(\ol{\Sigma})=\diff\varphi +\diff\{h\in \Omega^0(\ol{\Sigma}):\diff^*\diff h=0\}\cong \nml_{\partial\ol{\Sigma}}\Omega^1_{\diff^*}(\ol{\Sigma}),
\end{equation}
via the map $\diff\varepsilon\mapsto \nml_{\partial\ol{\Sigma}}\diff\varepsilon- f^\rho$ (surjectivity following from Prop.~\ref{prop:BVP}). The pair $(\alpha\restriction_{\partial\ol{\Sigma}}
,\nml_{\partial\ol{\Sigma}}\diff\varepsilon- f^\rho)$ will be the component of
 $\mathfrak{H}_{\ol{\Sigma}}^\rho(\mathbf{A}+\mathscr{G}(\ol{\Sigma}),\mathbf{E})$ in $V^S(\ol{N})$.

Combining these observations, the required map $\mathfrak{H}_{\ol{\Sigma}}^\rho$ is defined by
\begin{equation}
    \mathfrak{H}_{\ol{\Sigma}}^\rho(\mathbf{A}+\mathscr{G}(\ol{\Sigma}),\mathbf{E})=\left(\mathbf{A}^{\tang},\mathbf{E}^{\tang};\alpha\restriction_{\partial\ol{\Sigma}},\nml_{\partial\ol{\Sigma}}\diff\varepsilon-f^\rho\right),
\end{equation}
and is an affine bijection whose linearisation defines the required linear bijection $\mathfrak{H}_{\ol{\Sigma}}$.

Using the map $\mathfrak{I}_{\ol{\Sigma}}^J:\Sol_\mathscr{G}^J(\ol{N})\to \mathcal{A}(\ol{\Sigma})\times \mathcal{E}^\rho(\ol{\Sigma})$ from Prop.~\ref{prop:ext_ini_dat}, we now define the bijection \begin{equation}
    \mathfrak{K}^J_{\ol{\Sigma}}:=\mathfrak{H}_{\ol{\Sigma}}^\rho\circ \mathfrak{I}_{\ol{\Sigma}}^J,
\end{equation}
which is an affine map with linearisation 
\begin{equation}
    \mathfrak{K}_{\ol{\Sigma}}:=\mathfrak{H}_{\ol{\Sigma}}\circ \mathfrak{I}_{\ol{\Sigma}}.
\end{equation}
What remains to be shown is that $\mathfrak{K}_{\ol{\Sigma}}:\Sol_{\mathscr{G}}(\ol{N})\to V^C(\ol{\Sigma})\oplus V^S(\ol{N})$ is a linear symplectomorphism w.r.t.\ the given symplectic forms. This follows from orthogonality of the tangential Hodge--Helmholtz decomposition. Let $[A_1],[A_2]\in \Sol_{\mathscr{G}}(\ol{N})$ and, for $i=1,2$, write
\begin{equation}
\label{eq:K_lin_map}
    \mathfrak{K}_{\ol{\Sigma}}([A_i])=\left(\mathbf{A}_i^{\tang},\mathbf{E}_i^{\tang};\alpha_i\restriction_{\partial\ol{\Sigma}},\nml_{\partial\ol{\Sigma}}\diff\varepsilon_i\right),
\end{equation}
where $[A_i]\restriction_{\ol{\Sigma}}=\mathbf{A}_i^{\tang}+\diff\alpha+\mathscr{G}(\ol{\Sigma})$ and $\nml_{\ol{\Sigma}}\diff A_i=\mathbf{E}_i^{\tang}+\diff\varepsilon_i$. Using orthogonality, one then computes
\begin{align}
    \sigma([A_1],[A_2])=&\ipc{\mathbf{A}^{\tang}_1}{\mathbf{E}^{\tang}_2}_{\ol{\Sigma}}-\ipc{\mathbf{A}^{\tang}_2}{\mathbf{E}^{\tang}_1}_{\ol{\Sigma}}+\ipc{\diff\alpha_1}{\diff\varepsilon_2}_{\ol{\Sigma}}-\ipc{\diff\alpha_2}{\diff\varepsilon_1}_{\ol{\Sigma}}\nonumber\\
    =&\ipc{\mathbf{A}^{\tang}_1}{\mathbf{E}^{\tang}_2}_{\ol{\Sigma}}-\ipc{\mathbf{A}^{\tang}_2}{\mathbf{E}^{\tang}_1}_{\ol{\Sigma}}+\ipc{\alpha_1\restriction_{\partial\ol{\Sigma}}}{\nml_{\partial\ol{\Sigma}}\diff\varepsilon_2}_{\partial\ol{\Sigma}}-\ipc{\alpha_2\restriction_{\partial\ol{\Sigma}}}{\nml_{\partial\ol{\Sigma}}\diff\varepsilon_1}_{\partial\ol{\Sigma}}\\
    &=(\sigma^C\oplus\sigma^S)(\mathfrak{K}_{\ol{\Sigma}}([A_1]),\mathfrak{K}_{\ol{\Sigma}}([A_2])) \nonumber,
\end{align}
having used~\eqref{eq:diffcodiffbound} and $-\diff^*\diff\varepsilon_i=0$.
\end{proof}

\section{Localisability of observables}
\label{apx:localise}
Here we prove some technical lemmas required for the proof of Prop.~\ref{prop:loc_localise}.

The first lemma concerns invariant subalgebras of Weyl algebras.
\begin{lemma}\label{lem:Weyl_fixedpoint}
    Let $(V,\sigma)$ be a symplectic space and let $Z\subset V^*$ be a subspace of the dual. Regarding $Z$ as an abelian group, define a $Z$-action $\gamma:Z\to\Aut(\Weyl(V,\sigma))$ by
    \begin{equation}
        \gamma(z)(W(v))=e^{iz(v)} W(v)
    \end{equation}
    for each $v\in V$ and $z\in Z$. Then the fixed-point subalgebra 
    \begin{equation} 
    \Weyl(V,\sigma)^{\gamma}=\{a\in\Weyl(V,\sigma): \gamma(z)a=a~\textnormal{for all}~z\in Z\}
    \end{equation}
    is the subalgebra generated by $\{W(v): W(v)\in  \Weyl(V,\sigma)^{\gamma}\}=\{W(v): v\in\bigcap_{z\in Z} \ker z\}$.
\end{lemma}
\begin{proof}
    For each $\varepsilon>0$ there exists a minimal $m\in \mathbb{N}$, and an $a_m\in \Weyl(V,\sigma)$ of the form
    \begin{equation}
        a_m=\sum_{n=1}^m c_nW(v_n),
    \end{equation}
    with $v_n\in V$ and $c_n\in \CC$, such that
    $\Vert a-a_m\Vert<\varepsilon$.
    For any $z\in Z$ and $s\in\mathbb{R}$, we have
    \begin{equation}
        a_m(s):= \frac{1}{2}\left(\gamma(sz)(a_m) + \gamma(-sz)(a_m)
        \right) = \sum_{n=1}^m c_n \cos (s z(v_n))
        W (v_n).
    \end{equation}
    As $\gamma(sz)(a)=a$ for all $s$, 
    we may estimate
    \begin{align}
        \Vert a-a_m(s)\Vert\leq& \frac{1}{2}\Vert a-\gamma(sz)(a_m)\Vert+\frac{1}{2}\Vert a- \gamma(-sz)(a_m)\Vert\nonumber\\
        =& \frac{1}{2}\Vert \gamma(sz)(a-a_m)\Vert+\frac{1}{2}\Vert  \gamma(-sz)(a-a_m)\Vert \nonumber\\
        =& \Vert a-a_m\Vert<\varepsilon,
    \end{align}
    using the fact that any automorphism is norm-preserving.
   By minimality of $m$, we have $\cos (sz(v_n))\neq 0$ for all $s\in \mathbb{R}$ and hence
$z(v_n)=0$ for all $n$ and also all $z\in Z$. Thus $\Weyl(V,\sigma)^\gamma$ is contained in, and indeed equal to, the subalgebra generated by $\{W(v): W(v)\in  \Weyl(V,\sigma)^{\gamma}\}$. 
\end{proof} 

We also require the following consequence of Poincar\'e duality.
\begin{lemma}
\label{lem:poincare_dual}
    Let $M$ be an oriented manifold without boundaries equipped with a non-degenerate metric. If $F\in \Omega^k(M)$ obeys $\ipc{f}{F}_M=0$ 
    for all $f\in \Omega^k_{0\diff^*}(M)$, then $F\in \diff\Omega^{k-1}(M)$.
\end{lemma}
\begin{proof}
    Consider first $f=\diff^*h$ for $h\in \Omega_0^{k+1}(M)$, then
    \begin{equation}
        \ipc{h}{\diff F}_M=\ipc{\diff^*h}{F}_M=0.
    \end{equation}
    Since this holds for arbitrary compactly supported forms $h$, it follows that $\diff F=0$, so $F$ belongs to a de Rham cohomology class 
    \begin{equation}
        \tilde{F}:=F+\diff\Omega^{k-1}(M)\in \Omega^k_{\diff}(M)/(\diff\Omega^{k-1}(M)).
    \end{equation}
    The cohomology classes $\Omega^k_{\diff}(M)/(\diff\Omega^{k-1}(M))$ and $\Omega^k_{0\diff^*}(M)/(\diff^*\Omega_0^{k+1}(M))$ have a well defined pairing through
    \begin{equation}
        \ipc{h+\diff^*\Omega_0^{k+1}(M)}{H+\diff\Omega^{k-1}(M)}=\ipc{h}{H}_M,
    \end{equation}
    for $h\in \Omega^k_{0\diff^*}(M)$ and $H\in \Omega^k_{\diff}(M)$. In particular, for any $\tilde{H}\in \Omega^k_{0\diff}(M)/(\diff\Omega_0^{k-1}(M))$, this pairing defines a linear functional $D_{\tilde{H}}:\Omega^k_{0\diff^*}(M)/(\diff^*\Omega_0^{k+1}(M))\to \RR$, with
    \begin{equation}
        D_{\tilde{H}}\tilde{h}=\ipc{\tilde{h}}{\tilde{H}};
    \end{equation}
    moreover Poincar\'e duality (see e.g.~\cite[Sec.~5.12]{greubConnectionsCurvatureCohomology1972}) implies that the map $\tilde{H}\mapsto D_{\tilde{H}}$ is injective. In our case, $D_{\tilde{F}}=0$ and therefore the cohomology class $\tilde{F}$ is trivial, i.e., $F\in \diff\Omega^{k-1}(M)$.
\end{proof}

Furthermore, we need the following lemma on extending forms on the interior of a manifold with corners to its boundary.
\begin{lemma}
\label{lem:0form_ext}
    Let $\ol{M}$ an $n$-manifold with corners and non-degenerate smooth metric, and $M=\Intr(\ol{M})$. Let $\alpha\in \Omega^0(M)$ with $\beta\in \Omega^{1}(\ol{M})$ such that
    \begin{equation}
        \beta\restriction_{M}=\diff\alpha.
    \end{equation}
    Then $\alpha$ can be extended to a smooth form $\ol{\alpha}\in \Omega^0(\ol{M})$.
\end{lemma}
\begin{proof}
    It suffices to prove this for $\ol{M}=U\cap \ol{\RR}_+^n$ for some $U\subset \RR^n$ open and convex. Let $\alpha\in \Omega^0(M)$ such $\diff\alpha$ can be smoothly extended to $\ol{M}$, i.e. all derivatives of the functions $\{\partial_i\alpha:i\in\{1,...,n\})\in C^\infty(M)\}$ are continuously extendible to $\ol{M}$. Choose $x_0\in M$ and denote for each $x\in \ol{M}$ the function $\gamma_x:[0,1]\to \ol{M}$ by 
    \begin{equation}
        \gamma_x(t) = x_0+t(x-x_0).
    \end{equation}
    Now define 
    \begin{equation}
        \ol{\alpha}(x)=\alpha(\gamma_x(0))+\int_0^1 \iota_{\dot{\gamma}_x(t)}\beta(\gamma_x(t))\diff t.
    \end{equation}
    One readily checks that for $x\in M$, we have $\gamma_x([0,1])\in M$ and thus
    \begin{equation}
        \ol{\alpha}(x)=\alpha(\gamma_x(0))+\int_0^1 \iota_{\dot{\gamma}_x(t)}\diff\alpha(\gamma_x(t))\diff t=\alpha(x).
    \end{equation}
    Furthermore, for each $x,x'\in \ol{M}$, we have
    \begin{align}
        \ol{\alpha}(x)-\ol{\alpha}(x')=&\int_0^1 \left(\iota_{\dot{\gamma}_x(t)}\beta(\gamma_x(t))\iota_{\dot{\gamma}_{x'}(t)}\beta(\gamma_{x'}(t))\right)\diff t\nonumber\\
        =&\alpha(\gamma_x(s))-\alpha(\gamma_{x'}(s))+\int_s^1 \left(\iota_{\dot{\gamma}_x(t)}\beta(\gamma_x(t))-\iota_{\dot{\gamma}_{x'}(t)}\beta(\gamma_{x'}(t))\right)\diff t\nonumber\\
        =&\int_0^1 \iota_{\dot{\gamma}_{x,x',s}(t)}\beta(\gamma_{x,x',s}(t))\diff t+\int_s^1 \left(\iota_{\dot{\gamma}_x(t)}\beta(\gamma_x(t))-\iota_{\dot{\gamma}_{x'}(t)}\beta(\gamma_{x'}(t))\right)\diff t,
    \end{align}
    for each $s\in [0,1)$, where $\gamma_{x,x',s}(t)=x_0+s(x-x_0)+ts(x'-x)$. Now let $K\subset \ol{N}$ compact and convex such that $x_0,x,x'\in K$, then we may estimate
    \begin{align}
        \vert\ol{\alpha}(x)-\ol{\alpha}(x')\vert\leq&\left(s\Vert x-x'\Vert +(1-s)(\Vert x-x_0\Vert+\Vert x'-x_0\Vert\right)C_K,
    \end{align}
    where 
    \begin{equation}
        C_K=\sup_{x\in K}\sqrt{\sum_{i=1}^n(\beta_i(x))^2},
    \end{equation}
    where $\beta_i\in C^\infty(M)$ the coefficients of $\beta$ in $\RR^n$-coordinates.
    Thus $\vert\ol{\alpha}(x)-\ol{\alpha}(x')\vert\leq\Vert x-x'\Vert C_K$, which implies that $\ol{\alpha}$ is continuous on $\ol{M}$. Since $\diff\alpha$ is smoothly extendible to $\ol{M}$, it thus follows that all derivatives of $\ol{\alpha}$ on $M$ are continuously extendible to $\ol{M}$ and thus $\ol{\alpha}\in \Omega^0(\ol{M})$.
\end{proof}

Now, we prove the following result on existence certain co-closed forms with specified normals.
\begin{lemma}
\label{lem:form_deform}
    Let $\ol{M}$ be a $d$-dimensional compact oriented manifold with smooth Riemannian metric $g$ and a smooth boundary that decomposes into compact connected components $\partial\ol{M}=\bigcup_{i\in I}S_i$. Let $J\subset I$ and suppose $f\in \nml_{\partial\ol{M}}\Omega^1_{\diff^*}(\ol{M})$ vanishes identically on $S_J:=\bigcup_{i\in J} S_i$. 
    Then there exists a neighbourhood $V$ of $S_J$ and an $F\in \Omega^1_{\diff^*}(\ol{M})$ such that $\nml_{\partial\ol{M}}F=f$ and $F\restriction_V=0$.
\end{lemma}
\begin{proof}
    Let $U\subset \ol{M}$ be a geodesic collar neighbourhood of $S_{J}$, parametrised by $(s,p)$ for $s\in[0,\varepsilon)$ and $p\in S_J$, where $(s,p)$ is a point at geodesic distance $s$ from $p\in S_J$ along the normal geodesic. Thus $U\cong [0,\varepsilon)\times S_J$ and we also have a smooth `distance to boundary' map $s:U\to [0,\epsilon)$. By slight abuse of notation, we can write for the metric $g$
    \begin{equation}
        g\restriction_{U}=\diff s^2\oplus g_s,
    \end{equation}
    where $g_s$ is an $s$-dependent metric on $S_J$. 

    By assumption, $f$ may be written $f=\nml_{\partial\ol{M}}\tilde{F}$ for some $\tilde{F}\in \Omega^1_{\diff^*}(\ol{M})$, written locally as
    \begin{equation}
        \tilde{F}\restriction_{U}=\tilde{F}_0\oplus \tilde{H}_{s},
    \end{equation}
    where $\tilde{F}_0\in \Omega^0([0,\varepsilon)\times S_J)$ and $\tilde{H}_s$ an $s$-dependent $1$-form on $S_J$. Choose a smooth function $h:[0,\varepsilon)\to [0,\varepsilon)$ obeying 
    \begin{equation}
        h(s)=\begin{cases}
            0&s<\varepsilon/3\\
            s&s>2\varepsilon/3.
        \end{cases}
    \end{equation}
    We then define $F\in \Omega^1(\ol{M})$ 
    to agree with $\tilde{F}$ on $\ol{M}\setminus U$, whereas $F\restriction_{U}=F_0\oplus H_s$, where given local coordinates $\{x^i\}$ in a neighbourhood of $p\in S_J$, we set
    \begin{align}
       F_0(s,p)=\sqrt{\frac{\det( g_{h(s)}(p))}{\det( g_{s}(p))}}\tilde{F}_0(h(s),p),\\H_s(p)=h'(s)\sqrt{\frac{\det( g_{h(s)}(p))}{\det( g_{s}(p))}}\left(g_s\right)_{ij}(p) \left(g_{h(s)}\right)^{jk}(p)\left(\tilde{H}_{h(s)}(p)\right)_k\diff x^i.
    \end{align}
    Noting that
    $F_0\oplus H_s$ and $\tilde{F}_0\oplus \tilde{H}_s$
    agree for $s>2\varepsilon/3$, we see that $F$ is smooth. Moreover, $F$ vanishes on 
    \begin{equation}
        V=\{x\in U:s(x)<\varepsilon/3\}.
    \end{equation}
    It remains to show that $\diff^*F=0$, which clearly holds on $\ol{M}\setminus U$, while on $U$ we calculate
    \begin{align}
        (\diff^*F)(s,p)=&-\frac{1}{\sqrt{\det(g_s)}}\left(\partial_s \sqrt{\det(g_s)} F_0(s,p)+\partial_i\sqrt{\det(g_s)}\left(g_s\right)^{ij}\left(H_s\right)_j\right)\nonumber\\
        =&-\frac{1}{\sqrt{\det(g_s)}}\left(\partial_s \sqrt{\det(g_{h(s)})} \tilde{F}_0(h(s),p)+h'(s)\partial_i\det(g_{h(s)})\left(g_{h(s)}\right)^{ij}\left(\tilde{H}_{h(s)}\right)_j\right)\nonumber\\
        =&\sqrt{\frac{\det(g_{h(s)})}{\det(g_s)}}h'(s)(\diff^*\tilde{F})(h(s),p)=0.
    \end{align}
\end{proof}
Lastly, we require the following result on support of solutions to the Maxwell equations.
\begin{lemma}
\label{lem:E_sol_support}
    Let $\ol{\Sigma}\subset\ol{N}$ a regular Cauchy surface with boundaries of a finite Cauchy lens. Let $[\udl{A}]\in \Sol_{\mathscr{G}}(\ol{N})$ the unique gauge orbit of linear solutions such that $\mathfrak{I}_{\ol{\Sigma}}([\udl{A}])=0\oplus \mathbf{E}$ for $\mathbf{E}\in \mathcal{E}(\ol{\Sigma})$ as in Prop.~\ref{prop:ext_ini_dat}. Then there exists an $\udl{A}\in[\udl{A}]$ such that $\supp(\udl{A})\subset \mathscr{J}(\supp(\mathbf{E}))$.
\end{lemma}
\begin{proof}
    Let $n$ a vectorfield normal to $\ol{\Sigma}$ and choose $\udl{A}\in \Omega^1(\ol{N})$ such that 
    \begin{equation}
        -(\diff^*\diff+\diff\diff^*)\udl{A}=0,\qquad (\nml_{\ol{\Sigma}}\udl{A},\nml_{\ol{\Sigma}} \nabla_n \udl{A})=0,\qquad \Data_{\ol{\Sigma}}(\udl{A})=(0,\mathbf{E}),
    \end{equation} 
    using the map $\Data_{\ol{\Sigma}}:\Omega^1(\ol{N})\to \Omega^1(\ol{\Sigma})\times \Omega^1(\ol{\Sigma})$ as in Prop.~\ref{prop:IV_simple}. By \cite[Thm.~3.2.11]{barWaveEquationsLorentzian2007}, it follows that $\supp(\udl{A})\subset \mathscr{J}(\supp(\mathbf{E}))$. Since $\diff^*_{\ol{\Sigma}}\mathbf{E}=0$, we compute  $\diff^*\udl{A}\restriction_{\Sigma}=\nabla_n\diff^*\udl{A}\restriction_{\Sigma}=0$, which implies $\diff^*\udl{A}=0$ and thus $\udl{A}\in \Sol(\ol{N})$. Hence $[\udl{A}]\in \Sol_{\mathscr{G}}(\ol{N})$ and $\mathfrak{I}_{\ol{\Sigma}}([\udl{A}])=0\oplus \mathbf{E}$.
\end{proof}
\section{A relativisation map for large gauge transformations}
\label{apx:edge_mode_QRF}

Throughout, $\ol{N}$ is a fixed finite Cauchy lens and $\ol{\Sigma}$ is a fixed
regular Cauchy surface with boundary of $\ol{N}$. Thus $\partial\ol{\Sigma}=\angle\ol{N}$
is a compact smooth Riemannian manifold without boundary and the Laplace--de Rham operator $\Delta_\angle$ on $\partial\ol{\Sigma}$ is essentially self-adjoint and nonnegative on $\Omega^0(\angle\ol{N})\otimes\CC\subset L^2(\angle\ol{N})$, whose closure is also denoted $\Delta_\angle$. The $W^{2,s}$ Sobolev topology on $\Omega^0(\angle\ol{N})$ can be defined by the inner product
    \begin{equation}
        \ipc{f}{g}_{W^{2,s}}=\ipc{(1+\Delta_\angle)^{\frac{s}{2}}f}{(1+\Delta_\angle)^{\frac{s}{2}}g}_{L^2}.
    \end{equation}
The spectrum $\sigma(\Delta_\angle)\subset[0,\infty)$ is a discrete unbounded set; the eigenvalues have finite degeneracy and any eigenfunction is smooth, see e.g.~\cite[Ch.~8]{taylorPartialDifferentialEquations2023}. Moreover, if $\Pi_\mu$ is the spectral projector onto $[0,\mu]$ then
$\Pi_\mu\Omega^0(\angle\ol{N})\subset\Omega^0(\angle\ol{N})$. As $\mu\to\infty$,   $\Pi_\mu$ tends strongly to the identity on $\Omega^0(\angle\ol{N})$ with respect to the $L^2$-topology and, since $\Pi_\mu$ commutes with $(1+\Delta_\angle)$, the same holds true in the $W^{2,s}$-topology for every $s\in\RR$.
With respect to the $L^2$-inner product, there is an orthogonal decomposition 
 \begin{equation}
        \Omega^0(\ol{\Sigma})=\mathscr{LG}_\angle(\ol{N})\oplus \mathscr{G}_\angle(\ol{N}),
    \end{equation}
 (see Lem.~\ref{lem:nmlcoclosed}), where $\mathscr{LG}_\angle(\ol{N})$ is defined in Prop.~\ref{prop:large_gauge_boundary} and we recall that
        \begin{equation}
            \mathscr{G}_\angle(\ol{N})=\left\{\Lambda\restriction_{\angle\ol{N}}:\Lambda\in \Omega^0_{\diff}(\ol{N})\right\}
        \end{equation} is finite-dimensional.
As $\mathscr{G}_\angle(\ol{N})\subset \ker(\Delta_\angle)$, the decomposition of $\Omega^0(\angle\ol{N})$ is orthogonal in every $W^{2,s}$-inner product and one has $\Pi_\mu \mathscr{LG}_\angle(\ol{N})\subset \mathscr{LG}_\angle(\ol{N})$. 

\begin{lemma}
\label{lem:LG_filt}
    The subspaces $G_\mu=\Pi_\mu \mathscr{LG}_\angle(\ol{N})$, indexed by
    $\mu\in\sigma(\Delta_\angle)$, form an ascending filtration of $\mathscr{LG}_\angle(\ol{N})$ by finite-dimensional real inner product spaces. Equipping $G_\mu$ with its unique Hausdorff vector topology, the inclusion of $G_\mu$ in $G_{\mu'}$ for $\mu'>\mu$ is continuous, and $\bigcup_{\mu\in\sigma(\Delta_\angle)} G_{\mu}$ is dense in $\mathscr{LG}_\angle(\ol{N})$ with respect to the $W^{2,s}$-topology for every $s\in\RR$.
    \end{lemma}
    \begin{proof}
        For each $\mu\in \sigma(\Delta_\angle)$, the space $\Pi_\mu \Omega^0(\angle\ol{N})$ is finite dimensional by \cite[Thm.~III.B.18]{berardSpectralGeometryDirect1986}, and thus so is $G_\mu$. For $\mu'>\mu$, we have $\Pi_{\mu'}> \Pi_\mu$ and thus \begin{equation}
            G_\mu=\Pi_{\mu}G_{\mu'}\subset G_{\mu'}\subset \mathscr{LG}_{\angle}(\ol{N}).
        \end{equation}
        Equipping each $G_\mu$ by the restriction of the $L^2$-inner product, we find that the indexed family $\{G_\mu\}_{\mu\in \sigma(\Delta_\angle)}$ forms an ascending filtration of $\mathscr{LG}_{\angle}(\ol{N})$ by finite-dimensional real inner product spaces.
        
        For each $s\in \RR$, the space $\mathscr{LG}_\angle(\ol{N})$ equipped with the  $W^{2,s}$-topology is a Hausdorff topological vector space. Restricting to $G_\mu$, these Sobolev topologies are therefore equivalent to the unique Hausdorff vector topology on $G_\mu$, i.e., the Euclidean topology~\cite[Thm 3.2]{SchaeferWolff:1999}.  Similarly, for $\mu'>\mu$, the unique Hausdorff vector topology of $G_{\mu'}$ restricts to the Hausdorff vector topology on $G_\mu$. Hence the inclusion $G_\mu\hookrightarrow G_{\mu'}$ is continuous with respect to these topologies.

        Now let $s\in \RR$. Since $\lim_{\mu \to \infty}\Pi_\mu=1$ in the strong operator topology on $W^{2,s}(\angle\ol{N})$, it follows that for each $f\in \mathscr{LG}_{\angle}(\ol{N})$ the sequence $\{\Pi_\mu f\}_{\mu\in \sigma(\Delta_\angle)}$ converges to $f$ in the $W^{2,s}$ topology. Since $\Pi_\mu f\in G_\mu$, we find that $\bigcup_{\mu\in\sigma(\Delta_\angle)} G_{\mu}$ is $W^{2,s}$-dense in $\mathscr{LG}_{\angle}(\ol{N})$.
    \end{proof}

As a topological group, we may identify $(G_\mu,+)$ with $(\RR^n,+)$ for some $n\in \NN$. Using such an identification, the Weyl generators $W^S$, represented on some Hilbert space, yield a reprentation of a Weyl algebra over $\RR^{2n}$. Up to some regularity assumptions, the Stone-von Neumann theorem tells us that there is a unique irreducible representation (up to unitary equivalence) of such a Weyl algebra. We shall use this fact to construct a principal quantum reference frame, i.e.~a covariant projection valued measure, for the groups $G_\mu$.

Let $(\pi^\mathcal{R},\HH^\mathcal{R})$ be a regular faithful representation of $\Af^S(\ol{N})$, i.e.~for each $(f\oplus h)\in V^S(\ol{N})=\mathscr{LG}_{\angle}(\ol{N})^{\oplus 2}$, the map
\begin{equation}
    \RR\ni t\mapsto \pi^\mathcal{R}(W^S(tf\oplus th))
\end{equation}
is strongly continuous. For $f\in\mathscr{LG}_\angle(\ol{N})$, let 
\begin{equation}
U^{\mathcal{R}}(f)=\pi^\mathcal{R}(W^S(f\oplus 0)), \qquad V^\mathcal{R}(f) = \pi^\mathcal{R}(W^S(0\oplus f))
\end{equation}
which provides two unitary representations of $\mathscr{LG}_\angle(\ol{N})$ on $\HH^\mathcal{R}$ that are strongly continuous when restricted to any $G_\mu$,
and obey 
\begin{equation}
\label{eq:UV_Weyl}
    U^\mathcal{R}(f) V^\mathcal{R}(h) U^\mathcal{R}(f)^{-1} = e^{-i\langle f,h\rangle_{\angle\ol{N}}}V^\mathcal{R}(h).
\end{equation}
by the Weyl relations, for all $f,h\in\mathscr{LG}_{\angle}(\ol{N})$.

For each $\mu\in \sigma(\Delta_{\angle})$, we construct a projection valued measure $P_\mu:\Bor(G_\mu)\to B(\HH^\mathcal{R})$ that is covariant w.r.t. the unitary representation $U^\mathcal{R}$ of $G_\mu$.
\begin{lemma}\label{lem:emq1}
    Suppose $\Af^S(\ol{N})$ admits a regular faithful representation $(\pi^\mathcal{R},\HH^{\mathcal{R}})$ on separable Hilbert space $\HH^\mathcal{R}$ defining $\MM^\mathcal{R}\subset B(\HH^\mathcal{R})$ and $U^\mathcal{R}$ as in Thm.~\ref{thm:edge_mode_QRF}. Then for any $\mu\in\sigma(\Delta_\angle)$ 
    there is a unique regular projection-valued measure $P_\mu:\Bor(G_\mu)\to \MM^\mathcal{R}$ on $G_\mu$ so that
    \begin{equation}
    \label{eq:V_PVM}
        V^\mathcal{R}(f) = \int_{G_\mu} e^{i\langle f,f'\rangle_{\angle\ol{N}}}\diff P_\mu(f')
    \end{equation}
    as a weak integral for each $f\in G_\mu$. The measure has the covariance property
    $U^\mathcal{R}(f)P_\mu(X) U^\mathcal{R}(f)^{-1} = P_\mu(X+f)$
    for $X\in\Bor(G_\mu)$, $f\in G_\mu$. 
\end{lemma}
\begin{proof}
      Let $n=\dim G_\mu$. Treating $G_\mu\subset\mathscr{LG}_\angle(\ol{N})$ with the $L^2$-inner product as an $n$-dimensional 
    real inner-product space, choose a linear orthogonal map $\varphi_n:\RR^{n}\to G_\mu$, where 
    $\RR^{n}$ is equipped with the standard inner product. Then $\varphi_n^{\oplus 2}: \RR^{n}\oplus \RR^{n}\to G_\mu\oplus G_\mu$ is symplectic
    with respect to the symplectic form $\sigma_{n}(\bs\oplus\bt,\bs'\oplus\bt')=\bs\cdot \bt'- \bs'\cdot \bt$ on $\RR^n\oplus \RR^n$, and 
    $(\pi^\mathcal{R}\circ\varphi_n^{\oplus 2},\HH^\mathcal{R})$ is a regular representation of $\Weyl(\RR^n\oplus \RR^n,\sigma_n)$.
    By the Stone--von Neumann theorem, it is therefore a multiple of the Schr\"odinger representation (see~\cite[5.2.15-16]{bratteliOperatorAlgebrasQuantum1997}) and there is a 
    separable Hilbert space $\KK_n$ and unitary map $T_n:\HH^{\mathcal{R}}\to L^2(\RR^n)\otimes \KK_n$, so that 
    \begin{equation}
        T_n \tilde{U}(\bs) T_n^{*}= \lambda_n(\bs)\otimes  1_{\KK_n}, \qquad
        T_n \tilde{V}(\bt) T_n^{*}= M_n(\bt)\otimes  1_{\KK_n}
    \end{equation}
    where 
    \begin{equation} 
    \tilde{U}(\bs)= \pi^\mathcal{R}(W^S(\varphi_n(\bs)\oplus 0)), \qquad 
    \tilde{V}(\bt)= \pi^\mathcal{R}(W^S( 0\oplus \varphi_n(\bt))),
    \end{equation}
    while $\lambda_n(\bs)$ implements the left action of translations on $L^2(\RR^n)$ and $M_n(\bt)$ is a phase multiplication operator, i.e.,
    \begin{equation}
        (\lambda_n(\bs)\psi)(\bx) = \psi(\bx-\bs), \qquad (M_n(\bt)\psi)(\bx) = e^{i\bt\cdot\bx}\psi(\bx)
    \end{equation}
    for all $\psi\in L^2(\RR^n)$. Since the operators $\{\lambda_n(\bs)M_n(\bt):\bs,\bt\in \RR^n\}$ span a weakly dense subset of $B(L^2(\RR^n))$ (see \cite[Prop.~5.2.4]{bratteliOperatorAlgebrasQuantum1997}), it follows that there exists a von Neumann algebra $\tilde{\MM}_n\subset B(\KK_n)$ such that
    \begin{equation}
        T_n\MM^{\mathcal{R}}T_n^*=B(L^2(\RR^n))\otimes \tilde{\mathcal{M}}_n.
    \end{equation}
    Next, let $\tilde{P}_n:\Bor(\RR^n)\to B(\HH^\mathcal{R})$ be the unique projection valued measure such that 
    \begin{equation}
    \label{eq:Vt_PVM}
        \tilde{V}(\bt)=\int_{\RR^n}e^{i\bt\cdot\bx} \diff \tilde{P}_n(\bx),
    \end{equation}
    given explicitly by 
     \begin{equation}
        \tilde{P}_n(Y)=T_n^*(m_Y\otimes 1_\KK)T_n \in \MM^\mathcal{R},
    \end{equation}
    where $m_Y$ is multiplication by the indicator function of $Y\in \Bor(\RR^n)$. Direct calculation gives
    \begin{equation}
    \label{eq:Pt_cov}
        \tilde{U}(\bs)\tilde{P}_n(Y)\tilde{U}(\bs)^{-1}=T_n^*(\lambda(\bs)m_{Y}\lambda_n(-\bs)\otimes 1_{\KK_n})T_n=T_n^*(m_{Y+\bs}\otimes 1_{\KK_n})T=\tilde{P}_n(Y+\bs).
    \end{equation}
    The proof is completed by pulling back $\tilde{P}_n$ from $\RR^n$ back to $G_\mu$, defining $P_\mu=\tilde{P}_n\circ\varphi_n^{-1}$. 
    The formula Eq.~\eqref{eq:V_PVM} and the required covariance relation follows straightforwardly from Eqs.~\eqref{eq:Vt_PVM} and~\eqref{eq:Pt_cov}.
    The uniqueness of $P_\mu$ follows from that of $\tilde{P}_n$.
\end{proof}

The next step is to introduce relativisation maps that characterise invariants of the joint system--reference algebra with respect to the diagonal action of $G_\mu$.
\begin{lemma}
    \label{lem:emq2}  
    Continuing with the notation and assumptions of Lem.~\ref{lem:emq1},
    let $(\pi^\mathcal{S},\HH^\mathcal{S})$ be a regular representation of $\Af(\ol{N})$ defining $\MM^{\mathcal{S}}=\pi^\mathcal{S}(\Af(\ol{N}))''$ as in Thm.~\ref{thm:edge_mode_QRF}. Let $U^\mathcal{S}: \mathscr{LG}_\angle(\ol{N})\to \MM^\mathcal{S}$ be given by $U^\mathcal{S}(f)=\pi^\mathcal{S}(U_{\mathscr{LG}}(f))$ and denote $\MM=\MM^\mathcal{S}\otimes \MM^{\mathcal{R}}$. For each $\mu\in\sigma(\Delta_\angle)$, 
    there exists a unique injective ${}^*$-homomorphism $\yen_\mu:\MM^\mathcal{S}\to \MM^{\Ad (U^\mathcal{S}\otimes U^\mathcal{R})\restriction_{G_\mu}}$ such that 
    \begin{equation}
    \label{eq:rel_map_tensorstate}
        (\omega^{\mathcal{S}}\otimes \omega^{\mathcal{R}})(\yen_\mu(a))=\int_{G_n} \omega^{\mathcal{S}}(U^\mathcal{S}(f)aU^\mathcal{S}(f)^{-1})\,\diff(\omega^{\mathcal{R}}\circ P_\mu)(f),
    \end{equation}
    for each pair of normal states $\omega^{\mathcal{S}/\mathcal{R}}$ on $\MM^{\mathcal{S}/\mathcal{R}}$, thus
    allowing us to write
    \begin{equation}
        \yen_\mu(a)=\int_{G_\mu} U^\mathcal{S}(f)aU^\mathcal{S}(f)^*\otimes \diff  P_\mu(f).
    \end{equation}
\end{lemma}
\begin{proof}
    Fix $\mu\in \sigma(\Delta_\angle)$. Using the isomorphism $\varphi_n:\RR^{n}\to G_\mu$ as in Lem.~\ref{lem:emq1}, define
    group representations of $\RR^n$ on $\HH^{\mathcal{S}}$ and $\HH^{\mathcal{R}}$ by 
    $\tilde{U}^\mathcal{S}=U^\mathcal{S}\circ \varphi_n$ and $\tilde{U}^\mathcal{R}=U^\mathcal{R}\circ \varphi_n$ respectively. It is sufficient to show that there exists a unique injective ${}^*$-homomorphism $\tilde{\yen}_n:\MM^{\mathcal{S}}\to \MM^{\Ad(\tilde{U}^\mathcal{S}\otimes \tilde{U}^\mathcal{R})}$ such that for each pair of normal states $\omega^{\mathcal{S}/\mathcal{R}}$ on $\MM^{\mathcal{S}/\mathcal{R}}$ we have 
    \begin{equation}
    \label{eq:ytmap_states}     (\omega^{\mathcal{S}}\otimes\omega^\mathcal{R})(\tilde{\yen}_n(a))=\int_{\RR^n}\omega^S(\tilde{U}^\mathcal{S}(\bx)a\tilde{U}^\mathcal{S}(\bx)^{-1})\,\diff(\omega^\mathcal{R}\circ \tilde{P}_n)(\bx),
    \end{equation}
    where $\tilde{P}_n:\Bor(\RR^{n})\to \MM^{\mathcal{R}}$ as constructed in the proof of Lem.~\ref{lem:emq1}.

    First we consider the map $\kappa:\MM^\mathcal{S}\to \MM^\mathcal{S}\otimes L^2(\RR^{n})$ defined by its action on $\psi\in \HH^\mathcal{S}\otimes L^2(\RR^{n})$ (viewed as an $L^2(\RR^{n},\HH^\mathcal{S})$ function)
    \begin{equation}
        (\kappa(a)\psi)(\bx)=\tilde{U}^\mathcal{S}(\bx)a\tilde{U}^\mathcal{S}(\bx)^{-1}\psi(\bx).
    \end{equation}
    By \cite[Prop.~2.5]{vanDaele:1978}, $\kappa$ is an injective ${}^*$-homomorphism. By \cite[Thm.~3.11]{vanDaele:1978} we furthermore have $\kappa(\MM^{\mathcal{S}})\subset (\MM^{\mathcal{S}}\otimes B(L^2(\RR^{n})))^{\Ad(\tilde{U}^{\mathcal{R}}\otimes \rho)}$, in fact, denoting the right translation action $\rho_n:\RR^n\to B(L^2(\RR^{n}))$, the algebra
    \begin{equation}
        \MM^{\mathcal{S}}\rtimes_{\Ad \tilde{U}^{\mathcal{S}}}\RR^{n}=\{\kappa(a)\rho_n(\bs):a\in \MM^S,\bs\in \RR^{n}\}''
    \end{equation}
    coincides with $(\MM^{\mathcal{S}}\otimes B(L^2(\RR^n)))^{\Ad(\tilde{U}^{\mathcal{S}}\otimes \lambda)}$.
    
    Using the isomorphism $T:\HH^\mathcal{R}\to L^2(\RR^n)\otimes \KK_n$, we define $\tilde{\yen}_n:\MM^\mathcal{S}\to B(\HH^\mathcal{S}\otimes \HH^\mathcal{R})$ as
    \begin{equation}
        \tilde{\yen}_n(a)=(1_{\HH^{\mathcal{S}}}\otimes T)^*(\kappa(a)\otimes 1_{\KK_n})(1_{\HH^{\mathcal{S}}}\otimes T).
    \end{equation}
    Since $\kappa$ is an injective ${}^*$-homomorphism, so is $\tilde{\yen}_n$. Furthermore, since $T_n \tilde{U}^\mathcal{R}(\bs) T_n^*=\lambda_n(\bs)$, we find
    \begin{equation}
        \tilde{\yen}_n(a)\in \MM^{\Ad(\tilde{U}^{\mathcal{S}}\otimes \tilde{U}^{\mathcal{R}})}.
    \end{equation}
    This map agrees with \cite[Def.~4.13]{fewsterQuantumReferenceFrames2025} applied to our setting. Eq.~\eqref{eq:ytmap_states} then follows by \cite[Prop.~4.15]{fewsterQuantumReferenceFrames2025}. Moreover, by \cite[Thm.~3.7]{glowackiWalgebraicIntegrationTheory2026}, $\tilde{\yen}_n$ is the unique map satisfying Eq.~\eqref{eq:ytmap_states}. The required map is $\yen_\mu = \tilde{\yen}_n$, where we note that
    \begin{align}    (\omega^{\mathcal{S}}\otimes\omega^\mathcal{R})(\yen_\mu(a))=&\int_{\RR^n}\omega^S(\tilde{U}^\mathcal{S}(\bx)a\tilde{U}^\mathcal{S}(\bx)^{-1})\,\diff(\omega^\mathcal{R}\circ \tilde{P}_n)(\bx)\nonumber\\
    =&\int_{G_\mu}\omega^S(\tilde{U}^\mathcal{S}(\varphi_n^{-1}(f))a\tilde{U}^\mathcal{S}(\varphi_n^{-1}(f))^{-1})\,\diff(\omega^\mathcal{R}\circ \tilde{P}_n)(\varphi_n^{-1}(f))\nonumber\\
    =&\int_{G_\mu}\omega^S(U^\mathcal{S}(f)aU^\mathcal{S}(f)^{-1})\,\diff(\omega^\mathcal{R}\circ P_\mu)(f).
    \end{align}
\end{proof}

It is instructive to compute $\yen_\mu(\pi^\mathcal{S}(W_A([\udl{A}])))$ for $(A,\udl{A})\in T\Sol^J(\ol{N})$. We find
\begin{align}
    \yen_\mu(\pi^\mathcal{S}(W_A([\udl{A}]))) &=  
    \int_{G_\mu} \pi^\mathcal{S}(U_{\mathscr{LG}}(f)W_A([\udl{A}])U_{\mathscr{LG}}(f)^{-1})\otimes \diff  P_\mu(f) \\
    &= \pi^\mathcal{S}( W_A([\udl{A}]) )\otimes \int_{G_\mu} e^{i\langle \nu,f\rangle}
    \diff  P_\mu(f) \\ 
    &= \pi^\mathcal{S}( W_A([\udl{A}]) )\otimes \int_{G_\mu} e^{i\langle \Pi_\mu\nu,f\rangle}
    \diff  P_\mu(f)
\end{align}
where we have $\nu=\nml^{(2)}_{\angle\ol{N}}\diff\udl{A}\in \mathscr{LG}_\angle(\ol{N})$ and $\Pi_\mu:\mathscr{LG}_\angle(\ol{N})\to G_\mu$ the $L^2$-orthogonal projection. By Lem.~\ref{lem:emq1}, we have
\begin{equation}
\label{eq:yenmu_gen}
    \yen_\mu(\pi^\mathcal{S}(W_A([\udl{A}])))= \pi^\mathcal{S}( W_A([\udl{A}]) )\otimes V^{\mathcal{R}}(\Pi_\mu\nu)= \pi^\mathcal{S}( W_A([\udl{A}]) )\otimes \pi^\mathcal{R}(W^S({0\oplus\Pi_\mu \nu})).
\end{equation}
Given suitable continuity assumptions one can take the limit $\mu\to\infty$ using the fact that $\Pi_\mu \nu\to \nu$ in any $W^{2,s}$-topology on $\mathscr{LG}_\angle(\ol{N})$.
\begin{lemma}
\label{lem:emq3}
    Continuing with the notation and assumptions from Lem.~\ref{lem:emq2}, assume that the map
    \begin{equation}
        \mathscr{LG}_\angle(\ol{N})\ni f\mapsto \pi^{\mathcal{S}}(W^{\partial}(0\oplus f)),
    \end{equation}
    is strongly continuous w.r.t. the $W^{2,s}$-topology for some $s\in \RR$. Then there is a norm continuous injective ${}^*$-homomorphism    
    $\yen_\infty:\Af(\ol{N})\to \MM^{\Ad U^{\mathcal{S}}\otimes U^\mathcal{R}}$ given by
    \begin{equation}
        \yen_\infty(a)=\slim_{\mu\to\infty}\yen_\mu(\pi^\mathcal{S}(a)).
    \end{equation}
    In particular, for $(A,\udl{A})\in T\Sol^J(\ol{N})$ one has
    \begin{equation}
        \yen_\infty(W_A([\udl{A}]))=\pi^{\mathcal{S}}(W_A([\udl{A}]))\otimes \pi^\mathcal{R}(W^S(0\oplus \nml^{(2)}_{\angle\ol{N}}\diff\udl{A})).
    \end{equation}
\end{lemma}
\begin{proof}
    For $\mu\in \sigma(\Delta_{\angle})$ and $(A,\udl{A})\in T\Sol^J(\ol{N})$, Eq.~\eqref{eq:yenmu_gen} yields
    \begin{equation}
        \yen_\mu(W_A([\udl{A}]))=\pi^{\mathcal{S}}(W_A([\udl{A}]))\otimes \pi^\mathcal{R}(W^S(0\oplus \Pi_\mu\nml^{(2)}_{\angle\ol{N}}\diff\udl{A}))
    \end{equation}

    Let $s\in \RR$ such that the $W^{2,s}$-continuity requirement holds. Due to \ref{lem:LG_filt}, we have
    \begin{equation}
        \lim_{\mu\to\infty}\Pi_\mu\nml^{(2)}_{\angle\ol{N}}\diff\udl{A}=\nml^{(2)}_{\angle\ol{N}}\diff\udl{A},
    \end{equation}
    w.r.t. the $W^{2,s}$-topology. Thus, in the strong operator topology, we find the following limit
    \begin{equation}
        \slim_{\mu\to\infty} \yen_\mu(\pi^{\mathcal{S}}(W_A(\udl{A})))=\pi^{\mathcal{S}}(W_A(\udl{A}))\otimes \pi^{\mathcal{R}}\left(W^S\left(0\oplus \nml^{(2)}_{\angle\ol{N}}\diff\udl{A}\right)\right).
    \end{equation}
    By linearity, this defines a map 
    \begin{equation}
        a\mapsto \yen_\infty(a):=\slim_{\mu\to\infty} \yen_\mu(\pi^{\mathcal{S}}(a))
    \end{equation} on a norm-dense subset of $\Af(\ol{N})$. Large gauge invariance of $\yen_\infty(a)$ follows by invariance of $\yen_\infty(W_A([\udl{A}]))$.
    Since $\yen_\mu$ are injective unital ${}^*$-homomorphisms of von Neumann algebras, they are in particular norm preserving. This means that $a\mapsto \yen_\infty(a)$ is also norm continuous (with $\Vert\yen_\infty(a)\Vert\leq \sup_{\mu}\Vert \yen_\mu(\pi^{\mathcal{S}}(a))\Vert=\Vert a\Vert$) and thus one can extend $\yen_\infty$ to $\Af(\ol{N})$ by continuity. Since for $\psi\in \HH^\mathcal{S}\otimes\HH^\mathcal{R}$ we have
    \begin{align}
        \Vert(\yen_\mu(\pi^{\mathcal{S}}(a))-\yen_\infty(a))\psi\Vert\leq& \Vert(\yen_\mu(\pi^{\mathcal{S}}(\tilde{a}))-\yen_\infty(\tilde{a}))\psi\Vert+\left(\Vert \yen_\mu(\pi^{\mathcal{S}}(\tilde{a}-a))\Vert+\Vert \yen_\infty(\tilde{a}-a)\Vert\right)\Vert \psi\Vert\nonumber\\
        \leq& \Vert(\yen_\mu(\pi^{\mathcal{S}}(\tilde{a}))-\yen_\infty(\tilde{a}))\psi\Vert+2\Vert \tilde{a}-a\Vert\Vert \psi\Vert,
    \end{align}
    for any $\mu\in \sigma(\Delta_{\angle})$, $a,\tilde{a}\in \Af(\ol{N})$. Hence for all $a\in \Af(\ol{N})$, we find
    \begin{equation}
        \yen_\infty(a)=\slim_{\mu\to\infty} \yen_\mu(\pi^{\mathcal{S}}(a)).
    \end{equation}
\end{proof}

Indeed one now finds that, at the level of sufficiently regular representations, the maps $\yen_\mu(\pi^{\mathcal{S}}$ tend to the relativisation map of Eq.~\eqref{eq:C*QRF}.

\begin{proof}[Proof of Thm.~\ref{thm:edge_mode_QRF}]
    Given the assumptions of this theorem, then by Lem.~\ref{lem:LG_filt} to ~\ref{lem:emq2} one has a sequence of principal quantum reference frames $(\HH_R,U^{\mathcal{R}}\restriction_{G_n},P_n)$ and associated relativisation maps $\yen_n:\MM\to \MM^{\Ad(U^{\mathcal{S}}\otimes U^\mathcal{R})\restriction_{G_n}}$ for $\mu\in \sigma(\Delta_{\angle})$, where by slight abuse of notation we have set $G_n=G_{\mu_n}$, $P_n=P_{\mu_n}$ and $\yen_n=\yen_{\mu_n}$ with $\mu_n$ the $n$'th eigenvalue of $\Delta_{\angle}$. By Lem.~\ref{lem:emq3}, these QRFs can be chosen such that for $a\in \mathcal{A}(\ol{N})$ the limit $\yen_\infty(a)=\slim_{n\to\infty}\yen_n(\pi^{\mathcal{S}}(a))$ converges and where for $(A,\udl{A})\in T\Sol^J(\ol{N})$
    \begin{equation}
    \label{eq:pi-yen}
        \yen_\infty(W_A([\udl{A}]))=(\pi^{\mathcal{S}}\otimes\pi^{\mathcal{R}})(\yen(W_A([\udl{A}]))),
    \end{equation}
    with $\yen$ as in Eq.~\eqref{eq:C*QRF}. Since $\pi^{\mathcal{S}}\otimes\pi^{\mathcal{R}}$ is a faithful representation, it is in particular norm preserving. By linearity and continuity expression \eqref{eq:pi-yen}  therefore extends to all $a\in \Af(\ol{N})$.
\end{proof}

\end{appendices}

{\small
\bibliographystyle{alphaurl-ini-jstor}
\bibliography{zotero,additional}
}
\end{document}